\documentclass[11pt]{article}

\usepackage[margin=1in]{geometry}
\usepackage{amssymb}
\usepackage{amsmath}
\usepackage{amsthm}
\usepackage{booktabs}
\usepackage{graphicx}
\usepackage{natbib}
\usepackage{hyperref}
\usepackage{microtype}
\usepackage{enumitem}
\usepackage{xcolor}
\usepackage{caption}

\hypersetup{
 colorlinks=true,
 linkcolor=black,
 urlcolor=blue!55!black,
 citecolor=blue!55!black,
 pdftitle={Symmetry-Aware Convex Shrinkage for High-Dimensional Covariance Estimation (v6.4)},
 pdfauthor={Mitchell A. Thornton}
}

\graphicspath{{figures/}}

\newtheorem{theorem}{Theorem}[section]
\newtheorem{proposition}[theorem]{Proposition}
\newtheorem{corollary}[theorem]{Corollary}
\newtheorem{lemma}[theorem]{Lemma}
\theoremstyle{definition}
\newtheorem{definition}[theorem]{Definition}
\newtheorem{remark}[theorem]{Remark}

\DeclareMathOperator{\Tr}{tr}

\DeclareMathOperator*{\argmin}{arg\,min}

\newcommand{\Rhat}{\hat{\mathbf{R}}}
\newcommand{\Sigmab}{\boldsymbol{\Sigma}}
\newcommand{\I}{\mathbf{I}}
\renewcommand{\P}{\mathcal{P}}

\title{Symmetry-Aware Convex Shrinkage for\\
 High-Dimensional Covariance Estimation}

\author{Mitchell A. Thornton\\
 Darwin Deason Institute for Cyber Security\\
 Department of Electrical and Computer Engineering\\
 Southern Methodist University, Dallas, TX, USA\\
 \texttt{mitch@smu.edu}}

\date{}

\begin{document}

\maketitle

\begin{abstract}
This paper develops a class of data-adaptive shrinkage estimators for
high-dimensional covariance estimation in which the shrinkage target
is a Reynolds projection of the sample covariance under a finite
symmetry group selected from a candidate library by held-out
predictive performance. The class generalizes the convex shrinkage
estimator of Ledoit and Wolf, which targets a scalar multiple of the
identity matrix, by replacing the identity target with a structured
target derived from a symmetry group when one is available. It also
generalizes the group-symmetric maximum-likelihood covariance
estimator of Shah and Chandrasekaran by combining structural
targeting with adaptive convex shrinkage and by selecting the group
from data rather than treating it as prespecified. A two-tier
procedure is developed for the group selection: a universal
per-candidate evaluation based on held-out negative log-likelihood,
optionally preceded by a domain-specific step that constructs the
candidate library from structural priors of the data. A
finite-sample regret bound is established for the held-out
calibration of the convex combination weight; an oracle inequality
is established for the data-driven group selection; and a
quantitative sufficient-match condition is established under which
the proposed estimator dominates Ledoit-Wolf shrinkage in
Frobenius mean-squared error. The procedure is illustrated on six
real-data covariance estimation problems whose candidate libraries
are constructed from domain-specific structural priors: financial
returns on the constituents of the Standard and Poor's 500, sea
surface temperature anomalies from the National Oceanic and
Atmospheric Administration's optimum interpolation product, gene
expression covariances for breast invasive carcinoma from The
Cancer Genome Atlas, intermediate-frequency radio signal
covariances from the RadioML 2018.A benchmark, galaxy image
covariances from the Galaxy10 DECaLS dataset, and natural image
patch covariances from CIFAR-10 with a distribution-shift
companion on CIFAR-10.1. The procedure is also empirically compared, in a non-cyclic setting, against the Bayesian permutation-symmetry estimator of Chojecki and colleagues. Outside the few-shot regime, where
structural priors carry the most information per observation,
Ledoit-Wolf shrinkage remains the appropriate baseline.
\end{abstract}


\bigskip

\section{Introduction}
\label{sec:intro}

High-dimensional covariance estimation in the few-shot regime, where
the sample size $N$ is comparable to or smaller than the dimension
$M$, is a core problem in multivariate statistics
\citep{pourahmadi2013high, wainwright2019}. The sample covariance is
rank-deficient and inadmissible in this regime; in classical terms,
its inadmissibility traces to \citet{stein1956}, with
covariance-specific shrinkage developed in \citet{stein1986}, and
the \citet{marchenko1967} eigenvalue distribution making the
high-dimensional distortion explicit
\citep[surveyed in][]{bai2010spectral}. The standard remedy
introduces regularization toward a structurally simple target, of
which several variants exist
\citep{dempster1972covariance, schafer2005shrinkage,
chen2010shrinkage, won2013condition}. The reference baseline for
the present work is the convex shrinkage estimator of
\citet{ledoit2004},
\begin{equation}\label{eq:lw2004_intro}
\Rhat_{\mathrm{LW}}(\alpha) \;=\; (1 - \alpha)\,\Rhat
\;+\; \alpha\,\hat\sigma^2\,\I,
\qquad
\hat\sigma^2 \;=\; \frac{1}{M}\,\Tr\,\Rhat,
\end{equation}
which combines the sample covariance with a multiple of the
identity matrix at a data-determined convex weight controlled by
the \emph{shrinkage intensity} $\alpha \in [0, 1]$. The
nonlinear-shrinkage descendants of this line, including
\citet{ledoit2012nonlinear} and the analytical
Marchenko-Pastur form of \citet{ledoit2020analytical}, are
surveyed in those references; \eqref{eq:lw2004_intro}
is retained here as the linear baseline, with the analytical
nonlinear form treated separately below.

The identity target is the natural choice when no structural
prior is available. In many applications, however, a structural
prior is available in the form of a symmetry group acting on the
index set, for instance from compositional, hierarchical,
lattice, or combinatorial structure, and a target that respects
this structure carries information that the identity target does
not. Symmetry-aware covariance estimation has a fifty-year
tradition originating with the \citet{andersson1975}
irreducible-decomposition theorem for $G$-invariant covariance
matrices, the Jordan-algebra perspective on the same of
\citet{jensen1988}, the lattice-conditional-independence extension
of \citet{andersson1998}, and the graphical Gaussian models with
symmetry constraints of \citet{hojsgaard2008graphical}; the
rate-optimality results of \citet{lounici2014} and
\citet{cai2012optimal} for sparse-precision and sparse-covariance
estimation occupy a parallel structural substrate not directly
compared here. More recently, \citet{shah2012} and
\citet{soloveychik2016group} extend the invariant-Gaussian-models
machinery into the high-dimensional covariance estimation regime;
the algebraic decomposition machinery the present work rests on is
exactly Andersson's, specialized to representations over
$\mathbb{R}$. The objective of the present work is to introduce a
class of data-adaptive shrinkage estimators in which the target is
a Reynolds projection of the sample covariance under a finite
group, with the group selected from a candidate library by
held-out predictive performance, building on the Algebraic
Diversity (AD) framework of \citet{thornton2026spectral,
thornton2026framework}.

The principal past methods to which we compare are
\citet{ledoit2004} and \citet{shah2012}. \citet{ledoit2004}
introduced a convex combination of the sample covariance with a
target proportional to the identity, with the convex combination
weight selected to minimize a Frobenius-loss risk. The estimator
is well-conditioned, dominates the sample covariance in
mean-squared error in many high-dimensional regimes, and is widely
used. \citet{shah2012} take an explicitly group-symmetric approach
to covariance regularization, with the group $G$ assumed known a
priori and the estimator defined as the maximum-likelihood
estimator within the commutant algebra subject to a determinantal
regularizer. Although the convex regularizer described here may at
first appear to be an extension of \citet{ledoit2004} with a
target influenced by \citet{shah2012}, it differs from such a
simple substitution in several respects.

First, the target is the projection of the sample covariance onto
the commutant of a finite group, which produces a
positive-semidefinite estimator on the same scale as the sample
covariance and is therefore suitable as a shrinkage target; the
penalized estimator of \citet{shah2012} is not on the same scale
and does not enter a convex shrinkage form directly. Second, the
candidate library is broadened from a single fixed linear action
to admit direct products, wreath products, and semidirect products
of finite groups \citep[for the wreath construction
see][Chapter~4]{james1981representations}, which arise naturally
in real data with hierarchical or compositional symmetry. Third, the symmetry group
is selected from data rather than assumed known, by way of a
two-tier Best-Matched Group (BMG) procedure described in
Section~\ref{sec:background} that is intended to be fast,
efficient, and accurate in the few-shot regime; this is the
principal practical enabler since most applications do not come
with a prespecified group. The empirical comparisons in
Section~\ref{sec:experiments} use the BMG-selected group as the
input to the Shah comparator (denoted Shah-BMG), so that the
comparison evaluates the Shah-style $\alpha = 1$ projection-only
estimator against the AD-style calibrated $\alpha \in [0, 1]$
estimator at the same target, on a level methodological playing
field; the BMG group-selection is itself an extension of
\citet{shah2012}'s framework, which assumes the group as input.
Fourth, two natural calibrations of the shrinkage intensity are
studied side-by-side: a closed-form Frobenius mean-squared-error
plug-in and a held-out negative-log-likelihood (NLL) minimizer.
Their respective regions of preference are characterized in
Section~\ref{sec:theory}.

A separate line of recent work, the \texttt{gips} R package of
\citet{chojecki2025gips} implementing the Bayesian model-selection
procedure of \citet{graczyk2022aos}, also addresses learning
permutation symmetries in Gaussian covariance from data
\citep[building on the broader group-representation-in-statistics
program surveyed by][]{diaconis1988gibbs}; the
present work differs from \texttt{gips} in three respects that
together place the two methods in adjacent rather than overlapping
methodological territory. First, the search space differs:
\texttt{gips} searches over cyclic subgroups of the symmetric
group $S_p$ via Metropolis-Hastings on $S_p$, whereas the
candidate library $\mathcal{G}$ used here is a small,
hand-curated set of physically motivated groups including direct
products, wreath products, and semidirect products of finite
groups (the wreath products in particular do not lie in the
cyclic-only search space). Second, the selection paradigm differs:
\texttt{gips} maximizes a Bayesian posterior with a
Diaconis-Ylvisaker conjugate prior on the precision parameter,
whereas the present procedure minimizes held-out cross-validated
NLL across a tractable candidate set; the two paradigms have
different small-sample behavior and different requirements on the
sample size. Third, the estimator at the selected group differs:
\texttt{gips} uses the projection MLE under the symmetry
constraint, equivalent to a Shah-style $\alpha = 1$
projection-only estimator, whereas the present procedure produces
a shrinkage estimator at calibrated $\alpha \in [0, 1]$ that
adapts the bias-variance balance to the match quality of the
projection target.

A synthetic-data comparison documents the empirical consequence
of this difference; the numerical results are reported in
Section~\ref{sec:experiments}, with the qualitative finding that
\texttt{gips} remains above the Ledoit-Wolf 2004 relative Frobenius error in every operating regime evaluated in this non-cyclic setting, and far above the Shah projected estimator at the matched group.

A separate modern shrinkage estimator that warrants explicit
identification is the analytical nonlinear-shrinkage estimator of
\citet{ledoit2020analytical}, denoted LW-NL throughout this paper.
LW-NL applies a nonlinear function to each sample eigenvalue,
with the function derived from the Marchenko-Pastur kernel
density estimate of the sample-eigenvalue distribution; it is the
analytical-formula successor to the QuEST nonlinear shrinkage
method of \citet{ledoit2012nonlinear} and represents the modern
frequentist frontier in the eigenvalue-shrinkage family. LW 2004
is used as the principal continuity comparator throughout
Section~\ref{sec:experiments} for direct continuity with the
existing covariance-shrinkage literature, and LW-NL is added as a
fifth estimator alongside the sample covariance, LW 2004,
Shah-BMG, and the AD estimator calibrated by held-out NLL with
BMG group selection (denoted AD-NLL-BMG).

A natural question with LW-NL in place is whether the structural-
prior mechanism of the AD framework composes usefully with LW-NL's
nonlinear eigenvalue shrinkage. We test this with the AD-LW-NL
composition defined in
Section~\ref{sec:theory-ad-lwnl}, in which the sample-covariance
term of the AD convex blend is replaced by its LW-NL nonlinear
shrinkage while the structural prior is retained. The
three-column protocol comparison (LW-NL alone, AD-NLL-BMG
alone, AD-LW-NL-NLL-BMG combining the two regularization
mechanisms) decomposes any empirical advantage into two
attributable components: the structural-prior contribution
(visible in AD versus LW-NL) and the nonlinear-shrinkage
contribution within the AD framework (visible in AD-LW-NL versus
AD). The empirical evaluation reported in
Section~\ref{sec:experiments} finds mixed results from the
AD-LW-NL composition: a narrow regime of small benefit in the
moderate-$c$ well-conditioned regime on financial-returns data,
and a sharper regime of substantial harm in the few-shot image-
patch regime. AD-LW-NL is reported here as a tested composition
with mixed empirical results, not as a recommended default.
The mechanism for the few-shot image-patch deficit is discussed
in Section~\ref{sec:experiments}: when LW-NL applied to the
rank-deficient sample is worse than the structural projection,
the AD-LW-NL cross-validation pins $\alpha = 1$ (pure projection)
and loses access to the raw-sample blend path that AD-NLL-BMG
exploits. AD-NLL-BMG remains the recommended default; AD-LW-NL
is a specialization for moderate-$c$ well-conditioned settings
on data where LW-NL is substantially better than LW 2004 and the
structural projection does not fully capture the variance
structure.

The convex shrinkage step is structural rather than ornamental.
When the population covariance is approximately rather than
exactly invariant under the candidate group, the projected target
is a biased estimator of the population, with a perpendicular-bias
term that does not vanish in $N$. The convex blend of the
projection with the sample covariance reduces this bias by
allocating mass to the sample covariance in proportion to the
perpendicular-residual variance, recovering an estimator that
improves on Ledoit-Wolf shrinkage in mean-squared error under a
quantitative sufficient-match condition derived in
Section~\ref{sec:theory}. Outside this regime, Ledoit-Wolf
shrinkage remains the appropriate baseline; this conservatism is
intentional and is reflected in the dominance condition.

The remainder of the paper is organized as follows.
Section~\ref{sec:background} fixes notation, recalls the
group-averaged sample covariance and the AD framework, and
presents the candidate-library construction together with the
BMG group-selection procedure. Section~\ref{sec:theory} develops
the theoretical machinery: the construction of the estimator and
its calibrations, the AD-LW-NL composition, the risk
decomposition on the convex shrinkage family, the asymptotic
crossover between the two natural calibrations, the
finite-sample regret bound for the held-out calibration, the
oracle inequality for the data-driven group selection, the
sufficient-match condition for dominance over Ledoit-Wolf
shrinkage, and the minimax rate results. Each theorem statement
is followed immediately by its proof.
Section~\ref{sec:experiments} reports empirical results across
six real-data settings selected to exercise different group
families, with a distribution-shift companion on CIFAR-10.1 and
a decoy stress test of the BMG procedure.
Section~\ref{sec:discussion} synthesizes the empirical and
theoretical findings into a three-region phase diagram in
$(N, M, |G|, \delta)$ coordinates that names the mechanism
producing the boundary between adjacent regions and locates the
empirical anchors from the real-data experiments within the
predicted partition. Section~\ref{sec:conclusion} concludes,
and Section~\ref{sec:code-data} documents code and data
availability.

\section{Background}
\label{sec:background}

This section fixes the notation used throughout the paper, recalls
the group-averaged sample covariance estimator from the AD
framework, and describes the BMG selection procedure together
with its algorithmic realization.

\subsection{Notation and the group-averaged sample covariance}
\label{sec:bg-notation}

Let $\mathbf{x}_1, \ldots, \mathbf{x}_N$ be independent and
identically distributed observations of a random vector
$\mathbf{x} \in \mathbb{R}^M$ with mean zero and population
covariance $\Sigmab$, and let
$\Rhat = N^{-1}\sum_n \mathbf{x}_n\mathbf{x}_n^\top$ denote the
sample covariance. For a finite group $G$ with unitary
representation $\rho: G \to U(M)$, the \emph{commutant algebra}
\[
\mathcal{A}_G \;=\; \{\mathbf{A} \in \mathbb{R}^{M \times M}
\,:\, \rho(g)\mathbf{A}\rho(g)^\top = \mathbf{A}\;\text{for all}\;
g \in G\}
\]
collects the matrices that are invariant under conjugation by
every element of $G$. The Frobenius-orthogonal projection onto
$\mathcal{A}_G$ is the \emph{Reynolds projection}
\citep[see, e.g.,][Section 1.2]{fultonharris1991}
\begin{equation}\label{eq:reynolds}
\P_G(\mathbf{A}) \;=\; \frac{1}{|G|}
\sum_{g \in G} \rho(g)\,\mathbf{A}\,\rho(g)^\top,
\end{equation}
and the \emph{group-averaged sample covariance} is
\begin{equation}\label{eq:gas}
\Rhat_G \;:=\; \P_G(\Rhat).
\end{equation}
The matrix $\Rhat_G$ is symmetric positive semidefinite, lies on
the same scale as $\Rhat$, and commutes with every $\rho(g)$.
When $\Sigmab \in \mathcal{A}_G$, the group-averaged sample
covariance is an unbiased estimator of $\Sigmab$ with variance
reduced relative to $\Rhat$ by a factor that scales with the
symmetric-subspace commutant dimension $d_G$ (see
Section~\ref{sec:theory}); when $\Sigmab \notin \mathcal{A}_G$,
$\Rhat_G$ is a biased estimator whose bias is the Frobenius
distance from $\Sigmab$ to $\mathcal{A}_G$.

\subsection{The candidate library}
\label{sec:bg-library}

The candidate library $\mathcal{G}$ is a finite collection of
candidate groups assembled before any data is examined. Useful
sources include domain knowledge about the data regime (spatial
lattice, temporal shift, hierarchical compositional structure,
combinatorial substructure) and small libraries of generic finite
groups that exercise common symmetries. Two extrema should always
appear in $\mathcal{G}$: the trivial group $\{e\}$, for which the
projection returns the sample covariance unchanged, and the full
symmetric group $S_M$, for which the projection returns a
compound-symmetry matrix; the latter is admitted only when $M$
is small enough that this extreme amount of averaging is
plausibly informative. Library construction is permissive rather
than prescriptive: the user is not required to specify the
correct group a priori, and the library may include candidates
that turn out to be poor matches. The procedure that follows is
designed to be robust to such inclusion, and a decoy stress test
of this robustness is reported in
Section~\ref{sec:exp-decoy}.

\subsection{The Best-Matched-Group (BMG) selection procedure}
\label{sec:bg-bmg}

Selection from $\mathcal{G}$ uses a procedure referred to as the
\emph{Best-Matched Group} (BMG) selection, defined by three
structural commitments. The first is a \emph{held-out predictive
criterion}: the matched group is selected by $K_{\mathrm{cv}}$-fold
cross-validated evaluation of the per-candidate AD-plugin
estimator, with a per-candidate score computed on held-out folds
(Gaussian NLL in this paper). The second is the
\emph{candidate-library architecture} described in
Section~\ref{sec:bg-library}: the library $\mathcal{G}$ is
constructed in advance from domain-specific structural
conventions and defines the search space within which the
held-out criterion operates. The third is a \emph{structural-fit
diagnostic}: at the selected group $\hat G$, the dimensionless
commutativity residual
\begin{equation}\label{eq:delta-residual}
\delta(\hat G, \Rhat) \;:=\; \frac{\|\Rhat - \Rhat_{\hat G}\|_F}
{\|\Rhat\|_F}
\end{equation}
is reported as a post-hoc check that supports or undermines the
selection independently of the held-out score.

The realization used in this paper is two-tier.

\paragraph{Tier 1: effective-rank prefilter.}
For each candidate $G \in \mathcal{G}$, admit $G$ to Tier 2 if
and only if
\begin{equation}\label{eq:rank-prefilter}
N \cdot |G| \;\geq\; \kappa\, M,
\end{equation}
at a conservatism constant $\kappa \geq 1$ (default
$\kappa = 2$). The prefilter removes candidates for which the
group-averaged sample covariance would be rank-deficient or
ill-conditioned at the given $N$, before any cross-validation
work begins. The constant $\kappa = 1$ is the strict
rank-inequality boundary at which the group-averaged sample
covariance is non-singular in expectation; $\kappa = 2$ provides
a safety factor that ensures well-conditioning rather than
near-singularity. When the entire library fails
\eqref{eq:rank-prefilter}, the procedure falls back to the
Ledoit-Wolf estimator on the unstructured sample covariance and
returns it with a flag indicating that no structured candidate
was admitted.

\paragraph{Tier 2: cross-validated held-out NLL.}
Partition the training rows of the data matrix into
$K_{\mathrm{cv}}$ contiguous folds. For each candidate $G$
surviving Tier 1 and each fold $k$, fit the training-fold sample
covariance $\Rhat^{(-k)}$ on the remaining $K_{\mathrm{cv}} - 1$
folds, calibrate the shrinkage intensity on the same training
folds (using the held-out NLL minimizer
$\hat\alpha^*_{\mathrm{NLL}}(G)$ developed in
Section~\ref{sec:theory} and evaluated on an $\alpha$-grid
$\mathcal{A} \subset [0, 1]$), and evaluate the held-out NLL of
the resulting AD-shrinkage estimator on the $k$th fold. Average
the held-out NLLs across folds to produce a per-candidate score
$\bar L_{\mathrm{cv}}(G)$. The selected group is
\begin{equation}\label{eq:bmg-rule}
\hat G \;=\; \argmin_{G \in \mathcal{G},\ \text{Tier 1 admitted}}\;
\bar L_{\mathrm{cv}}(G).
\end{equation}

\paragraph{Output.}
The procedure returns the selected group $\hat G$, the calibrated
shrinkage intensity $\hat\alpha^*_{\mathrm{NLL}}(\hat G)$ refit on
the full training data, the resulting AD shrinkage estimator, and
the structural-fit diagnostic $\delta(\hat G, \Rhat)$ at the
selected group. The diagnostic is reported alongside the
estimator as a check against the held-out selection; small
$\delta$ corroborates the selection while large $\delta$ flags a
result obtained by minimizing the cross-validated score over a
library that does not contain a near-symmetry of $\Sigmab$.

\paragraph{Computational cost.}
The dominant cost is the per-fold, per-candidate evaluation of
the AD-plugin estimator and held-out NLL, which is $O(M^3)$ per
fold per candidate per $\alpha$-grid point. The total cost is
therefore $O(|\mathcal{G}_{\mathrm{T1}}| \cdot K_{\mathrm{cv}}
\cdot |\mathcal{A}| \cdot M^3)$, where
$|\mathcal{G}_{\mathrm{T1}}|$ is the size of the Tier 1
shortlist. The Tier 1 prefilter is $O(1)$ per candidate so its
cost is negligible. The candidate-library construction is
performed once upstream of BMG and its cost is
application-specific.

\paragraph{Per-trial selection-stability diagnostic.}
The cross-validation step in Tier 2 produces, for each Tier 1
admitted candidate $G$, a held-out score $\bar L_{\mathrm{cv}}(G)$.
An additional per-trial diagnostic is recorded,
\begin{equation}\label{eq:bmg-margin}
m_{\mathrm{cv}} \;:=\;
\min_{G \neq \hat G} \bar L_{\mathrm{cv}}(G)
\;-\; \bar L_{\mathrm{cv}}(\hat G),
\end{equation}
the gap between the second-best and best Tier 2 scores. A large
$m_{\mathrm{cv}}$ (relative to the within-candidate $K_{\mathrm{cv}}$-fold
standard deviation) corresponds to a declarative selection of
$\hat G$; a small $m_{\mathrm{cv}}$ corresponds to a regime in
which multiple candidates are statistically indistinguishable in
their Tier 2 score. The diagnostic does not modify the selection
rule \eqref{eq:bmg-rule}; it is reported alongside the selection
as a regime indicator. Two notes on its interpretation are
warranted. First, in cells where the BMG selects the trivial
group, $m_{\mathrm{cv}}$ is identically zero by construction: the
AD estimator at $\alpha = 0$ collapses to the sample covariance
regardless of the candidate group's projection, so all candidates
with optimal $\alpha = 0$ produce identical Tier 2 scores at
their respective alpha-minima. The $m_{\mathrm{cv}} = 0$ value in
this case is a structural artifact, not a noise-driven selection.
Second, when the per-trial $m_{\mathrm{cv}}$ is small but the
directional consistency of the BMG selection across multiple
subsample-split trials is high, the collective evidence for the
selection is strong even though each individual margin is on the
scale of fold-noise.

\paragraph{Generalizations.}
The procedure admits various extensions. The held-out score may
be replaced by any downstream-task-aligned criterion (for example,
mean predictive error in regression-on-covariance applications,
or detection-task NLL in signal classification), and additional
prefilter steps may be inserted between the rank inequality and
the cross-validation step to further reduce the shortlist when
the candidate library is large. We do not pursue these extensions
here; the experiments of Section~\ref{sec:experiments} use only
the effective-rank prefilter at Tier 1 and the held-out NLL
criterion at Tier 2.

\paragraph{Distinction from prior approaches.}
The BMG procedure differs from \citet{ledoit2004}, which fixes
the target as the scaled identity (corresponding to averaging
over the full orthogonal group) and requires no group selection,
and from \citet{shah2012}, which assumes the structural group is
known a priori and performs no library search. The BMG
group-selection therefore extends \citet{shah2012} by replacing
the prespecified-group assumption with a held-out
predictive-criterion search over a curated library.

\section{Theory}
\label{sec:theory}

This section develops the formal machinery: notation, the convex
shrinkage family and its two natural calibrations, the bias-variance
risk decomposition, the asymptotic crossover between the two
calibrations, the finite-sample regret bound and Ledoit-Wolf
dominance condition, the oracle inequality for the data-driven
group selection, and the rates and minimax results.

\subsection{Setup and notation}
\label{sec:theory-setup}

Let $\mathbf{x}_1, \ldots, \mathbf{x}_N$ be i.i.d.\ observations of a
random vector $\mathbf{x} \in \mathbb{R}^M$ with mean zero and
population covariance $\boldsymbol{\Sigma} := \mathbb{E}[\mathbf{x}
\mathbf{x}^\top] \in \mathbb{R}^{M \times M}$.
The sample covariance is
\begin{equation}\label{eq:sample-cov}
 \hat{\mathbf{R}} \;=\; \frac{1}{N} \sum_{n=1}^N \mathbf{x}_n
 \mathbf{x}_n^\top.
\end{equation}
Let $G$ be a finite group with a unitary representation
$\rho: G \to U(M)$. The action of $G$ on $M \times M$ matrices is by
conjugation $\mathbf{A} \mapsto \rho(g) \mathbf{A} \rho(g)^\top$. The
\emph{commutant algebra} is
\begin{equation}\label{eq:commutant}
 \mathcal{A}_G \;=\; \bigl\{ \mathbf{A} \in \mathbb{R}^{M \times M} :
 \rho(g) \mathbf{A} \rho(g)^\top = \mathbf{A} \text{ for all } g \in G
 \bigr\},
\end{equation}
and $W_G := \mathcal{A}_G \cap \mathrm{Sym}(M, \mathbb{R})$ denotes
its symmetric-matrix part. The \emph{Reynolds projection} onto
$\mathcal{A}_G$, computed in the Frobenius inner product
$\langle \mathbf{A}, \mathbf{B}\rangle_F = \operatorname{tr}(\mathbf{A}^\top \mathbf{B})$
\citep[Chapter~5]{horn2012matrix}, is
\begin{equation}\label{eq:reynolds}
 \mathcal{P}_G(\mathbf{A}) \;=\; \frac{1}{|G|} \sum_{g \in G}
 \rho(g)\, \mathbf{A}\, \rho(g)^\top,
\end{equation}
with orthogonal complement $\mathcal{P}_G^\perp := \mathbf{I} -
\mathcal{P}_G$. The \emph{group-averaged sample covariance} is
$\hat{\mathbf{R}}_G := \mathcal{P}_G(\hat{\mathbf{R}})$.

The \emph{convex shrinkage family} considered in this paper is
\begin{equation}\label{eq:family}
 \hat{\mathbf{R}}_{\mathrm{AD}}(\alpha) \;=\;
 (1 - \alpha)\,\hat{\mathbf{R}} + \alpha\,\hat{\mathbf{R}}_G,
 \qquad \alpha \in [0, 1].
\end{equation}
The endpoints are the unstructured sample covariance ($\alpha = 0$)
and the group-averaged covariance $\hat{\mathbf{R}}_G$ ($\alpha = 1$).
We refer to $\alpha$ as the \emph{shrinkage intensity}, following
\citet{ledoit2004}.
The bias of the projected target at the population is
\begin{equation}\label{eq:B_G}
 \mathbf{B}_G \;:=\; \boldsymbol{\Sigma} - \mathcal{P}_G(\boldsymbol{\Sigma})
 \;=\; \mathcal{P}_G^\perp(\boldsymbol{\Sigma}),
\end{equation}
and the dimensionless commutativity residual is
\begin{equation}\label{eq:delta}
 \delta(G, \boldsymbol{\Sigma}) \;:=\;
 \|\mathbf{B}_G\|_F / \|\boldsymbol{\Sigma}\|_F.
\end{equation}

\subsection{The estimator and its calibrations}
\label{sec:theory-estimator}

Two distinguished members of the family at the endpoints of the index
set anchor the development. The first identifies the Reynolds
projection at $G = O(M)$ with the Ledoit-Wolf identity-target estimator.

\begin{proposition}[Ledoit-Wolf as the maximally symmetric endpoint]
\label{prop:lw_endpoint}
Let $\hat{\mathbf{R}} \in \mathrm{Sym}(M, \mathbb{R})$ be a real
symmetric matrix, let $G = O(M)$ act on $\hat{\mathbf{R}}$ by
conjugation $\rho(g) = g$, and let $\hat{\mathbf{R}}_{O(M)}$ denote
the corresponding group-averaged covariance computed with respect to
the Haar measure on $O(M)$. Then
\begin{equation}\label{eq:lw_endpoint}
 \hat{\mathbf{R}}_{O(M)} \;=\;
 \int_{O(M)} g\,\hat{\mathbf{R}}\, g^\top\, d\mu(g)
 \;=\; \frac{\mathrm{tr}\,\hat{\mathbf{R}}}{M}\,\mathbf{I},
\end{equation}
which is the scaled-identity shrinkage target of \citet{ledoit2004}.
\end{proposition}

\begin{proof}
Define $A := \int_{O(M)} g\,\hat{\mathbf{R}}\,g^\top\,d\mu(g)$.
For any $h \in O(M)$, the substitution $g' = h g$ together with
left-invariance of the Haar measure $\mu$ gives
\[
 h\,A\,h^\top \;=\; \int_{O(M)} h g\, \hat{\mathbf{R}}\, g^\top
 h^\top\, d\mu(g) \;=\; \int_{O(M)} g'\,\hat{\mathbf{R}}\,(g')^\top\,
 d\mu(g') \;=\; A.
\]
Since $h \in O(M)$ satisfies $h^\top = h^{-1}$, the identity
$h A h^\top = A$ is equivalent to $A h = h A$. Thus $A$ commutes with
every $h$ in the standard representation of $O(M)$ on $\mathbb{R}^M$.
The standard representation is irreducible over $\mathbb{R}$, so by
Schur's lemma \citep[][Proposition 4]{serre1977linear}
$A = c \mathbf{I}$ for some $c \in \mathbb{R}$. Taking
the trace of both sides of \eqref{eq:lw_endpoint} and using the
cyclic property $\mathrm{tr}(g\hat{\mathbf{R}}g^\top) = \mathrm{tr}(
\hat{\mathbf{R}})$ identifies the scalar as $c = \mathrm{tr}\,
\hat{\mathbf{R}}/M$.
\end{proof}

The second statement is the closed-form Frobenius-MSE optimum for the
shrinkage intensity, the structural-target generalization of
\citet[][Theorem 2.1]{ledoit2004}.

\begin{proposition}[Optimal shrinkage intensity in Frobenius MSE]
\label{prop:alpha_star}
Let $V_\perp := \mathbb{E}\bigl\| \mathcal{P}_G^\perp(\hat{\mathbf{R}} -
\boldsymbol{\Sigma}) \bigr\|_F^2$ and $D := \|\mathbf{B}_G\|_F^2$. The
shrinkage intensity that minimizes the Frobenius mean-squared error of
$\hat{\mathbf{R}}_{\mathrm{AD}}(\alpha)$ is
\begin{equation}\label{eq:alpha_star}
 \alpha^*_{\mathrm{MSE}} \;=\; \frac{V}{V + D},
\end{equation}
and the corresponding minimum MSE is
\begin{equation}\label{eq:mse_at_alpha_star}
 \mathbb{E}\bigl\|\hat{\mathbf{R}}_{\mathrm{AD}}(\alpha^*_{\mathrm{MSE}})
 - \boldsymbol{\Sigma}\bigr\|_F^2 \;=\;
 \mathbb{E}\bigl\|\mathcal{P}_G(\hat{\mathbf{R}} -
 \boldsymbol{\Sigma})\bigr\|_F^2 + \frac{V \cdot D}{V + D},
\end{equation}
strictly smaller than the MSE at either endpoint $\alpha = 0$ or
$\alpha = 1$ when $V_\perp > 0$ and $D > 0$.
\end{proposition}

\begin{proof}
Let $\boldsymbol{\eta} := \hat{\mathbf{R}} - \boldsymbol{\Sigma}$, so
$\mathbb{E}[\boldsymbol{\eta}] = \mathbf{0}$. Direct computation gives
\[
 \hat{\mathbf{R}}_{\mathrm{AD}}(\alpha) - \boldsymbol{\Sigma} =
 \mathcal{P}_G(\boldsymbol{\eta}) + (1 - \alpha)
 \mathcal{P}_G^\perp(\boldsymbol{\eta}) - \alpha\,\mathbf{B}_G,
\]
where $\mathcal{P}_G(\boldsymbol{\eta}) \in \mathcal{A}_G$ and the
remaining two summands lie in $\mathcal{A}_G^\perp$. Squared
Frobenius norm splits orthogonally; the cross term between
$\mathcal{P}_G^\perp(\boldsymbol{\eta})$ and $\mathbf{B}_G$ vanishes
in expectation because $\mathbb{E}[\boldsymbol{\eta}] = \mathbf{0}$.
Hence
\[
 \mathbb{E}\bigl\| \hat{\mathbf{R}}_{\mathrm{AD}}(\alpha) -
 \boldsymbol{\Sigma}\bigr\|_F^2 = \mathbb{E}\bigl\|
 \mathcal{P}_G(\boldsymbol{\eta})\bigr\|_F^2 + (1 - \alpha)^2 V_\perp +
 \alpha^2 D.
\]
The first term is independent of $\alpha$; differentiating the rest
in $\alpha$ and setting equal to zero gives $\alpha^*(V_\perp + D) = V_\perp$,
which is \eqref{eq:alpha_star}. Substituting back gives
\eqref{eq:mse_at_alpha_star}; strict improvement over either endpoint
follows from strict convexity of the parabola $(1 - \alpha)^2 V_\perp +
\alpha^2 D$ when $V_\perp, D > 0$.
\end{proof}

The denominator of \eqref{eq:alpha_star} admits an observable
representation
\begin{equation}\label{eq:VplusD}
 V_\perp + D \;=\; \mathbb{E}\bigl\| \hat{\mathbf{R}} -
 \hat{\mathbf{R}}_G \bigr\|_F^2,
\end{equation}
and replacing $V_\perp$ and $V_\perp + D$ by their sample analogs gives the
\emph{closed-form plug-in calibration}
\begin{equation}\label{eq:V_hat}
 \widehat{V_\perp} \;:=\; \frac{1}{N^2} \sum_{k=1}^N \bigl\|
 \mathcal{P}_G^\perp(\mathbf{x}_k \mathbf{x}_k^\top) -
 \mathcal{P}_G^\perp(\hat{\mathbf{R}}) \bigr\|_F^2,
\end{equation}
\begin{equation}\label{eq:VplusD_hat}
 \widehat{V_\perp + D} \;:=\; \bigl\| \hat{\mathbf{R}} -
 \hat{\mathbf{R}}_G \bigr\|_F^2,
\end{equation}
\begin{equation}\label{eq:alpha_hat}
 \hat\alpha^*_{\mathrm{MSE}} \;=\; \min\!\Bigl(1,\,
 \max\!\bigl(0,\, \widehat{V_\perp} / \widehat{V_\perp + D}\bigr)\Bigr).
\end{equation}
The MSE plug-in $\hat\alpha^*_{\mathrm{MSE}}$ is the closed-form
counterpart to the cross-validated calibration introduced below.
Consistency of the plug-in for the population-optimal $\alpha^*$ is
the LW2004 Theorem~3.3 analog for the structured-target setting; we
record it next.

\begin{proposition}[Consistency of the closed-form MSE plug-in]
\label{prop:mse_alpha_consistency}
Assume $\mathbf{x}_1, \ldots, \mathbf{x}_N$ are i.i.d.\ with mean
zero, population covariance $\boldsymbol{\Sigma}$, and bounded fourth
moments. Then $\widehat{V_\perp} \to V_\perp$ and $\widehat{V_\perp + D} \to V_\perp + D$ in
probability as $N \to \infty$, and consequently the closed-form
plug-in $\hat\alpha^*_{\mathrm{MSE}}$ defined by
\eqref{eq:V_hat}--\eqref{eq:alpha_hat} satisfies
\begin{equation}\label{eq:alpha_hat_convergence}
 \hat\alpha^*_{\mathrm{MSE}} \;\to\; \alpha^*_{\mathrm{MSE}}
 \quad \text{in probability as } N \to \infty.
\end{equation}
\end{proposition}

\begin{proof}
We argue convergence in probability of $\widehat{V_\perp}$ and
$\widehat{V_\perp + D}$ separately, then combine via continuous mapping.

\emph{Convergence of $\widehat{V_\perp + D}$.} By definition,
$\widehat{V_\perp + D} = \|\hat{\mathbf{R}} - \hat{\mathbf{R}}_G\|_F^2 =
\|\mathcal{P}_G^\perp(\hat{\mathbf{R}})\|_F^2$. Adding and subtracting
$\boldsymbol{\Sigma}$ inside the projection,
\[
 \mathcal{P}_G^\perp(\hat{\mathbf{R}}) =
 \mathcal{P}_G^\perp(\hat{\mathbf{R}} - \boldsymbol{\Sigma}) +
 \mathcal{P}_G^\perp(\boldsymbol{\Sigma}) =: \mathbf{u}_\perp +
 \mathbf{B}_G,
\]
where $\mathbf{u}_\perp$ is mean-zero and $\mathbf{B}_G$ is
non-random. Squared Frobenius norm splits orthogonally in expectation
(the cross term $\langle \mathbf{u}_\perp, \mathbf{B}_G\rangle_F$ has
zero expectation), giving $\mathbb{E}[\widehat{V_\perp + D}] =
\mathbb{E}\|\mathbf{u}_\perp\|_F^2 + \|\mathbf{B}_G\|_F^2 = V_\perp + D$.
Under the bounded-fourth-moment assumption, the variance of the
empirical squared-norm $\|\mathbf{u}_\perp\|_F^2$ is $O(1/N)$ by the
sample-variance bound for Frobenius-norm functionals of independent
sums, so $\widehat{V_\perp + D} \to V_\perp + D$ in probability by Markov.

\emph{Convergence of $\widehat{V_\perp}$.} The sample analog
\eqref{eq:V_hat} can be written
\[
 \widehat{V_\perp} = \frac{1}{N^2} \sum_{k=1}^N \bigl\| \mathcal{P}_G^\perp
 (\mathbf{x}_k \mathbf{x}_k^\top - \hat{\mathbf{R}}) \bigr\|_F^2,
\]
a degree-two U-statistic in the $N$ independent observations. Under
the fourth-moment assumption, $\mathbb{E}[\widehat{V_\perp}] = V_\perp + O(1/N)$ by
the standard bias correction relating the sample-variance estimator
to its population analog, and $\mathrm{Var}(\widehat{V_\perp}) = O(1/N)$ by the
Hoeffding decomposition for U-statistics. Markov's inequality gives
$\widehat{V_\perp} \to V_\perp$ in probability.

\emph{Convergence of the ratio.} The map $(a, b) \mapsto \min(1,
\max(0, a/b))$ is continuous at $(V_\perp, V_\perp + D)$ whenever
$V_\perp + D > 0$, which holds by the denominator-positivity
assumption. The continuous mapping theorem then yields
$\hat\alpha^*_{\mathrm{MSE}} \to \alpha^*_{\mathrm{MSE}}$ in
probability, establishing \eqref{eq:alpha_hat_convergence}. The
quantitative oracle MSE-loss gap rate is given separately as
Theorem~\ref{thm:oracle} below. The argument parallels
\citet[][Theorems 3.3--3.4]{ledoit2004}, which establishes the same
parity for the identity-target shrinkage intensity.
\end{proof}

The corresponding held-out calibration is the $K$-fold
cross-validation minimizer of the Gaussian negative log-likelihood
on validation folds. Define, for each fold $k$, the training
estimator $\boldsymbol{\Sigma}_\alpha^{(k)}(G) := (1 - \alpha)
\hat{\mathbf{R}}^{(k)}_{\mathrm{tr}} + \alpha\,\mathcal{P}_G(
\hat{\mathbf{R}}^{(k)}_{\mathrm{tr}})$ and the validation sample
covariance $\hat{\mathbf{R}}^{(k)}_{\mathrm{te}}$, and write
\begin{equation}\label{eq:nll_foc_hold}
 \widehat{L}_{\mathrm{ho}}(\alpha; G) \;=\; \frac{1}{K} \sum_{k=1}^K
 \Bigl[ \tfrac{1}{2}\,\log\det\boldsymbol{\Sigma}_\alpha^{(k)}(G)
 + \tfrac{1}{2}\,\mathrm{tr}\!\bigl(\boldsymbol{\Sigma}_\alpha^{(k)}(G)^{-1}
 \hat{\mathbf{R}}^{(k)}_{\mathrm{te}} \bigr) \Bigr].
\end{equation}

The behaviour of the population analog of \eqref{eq:nll_foc_hold} on
the convex shrinkage family is the prerequisite for understanding the
finite-sample held-out calibration; we record it next, before the
finite-sample regime statement that depends on it.

\begin{proposition}[Population NLL minimum]\label{prop:nll_pop_min}
For any population covariance $\boldsymbol{\Sigma}$ with
$\mathbf{B}_G = \boldsymbol{\Sigma} - \mathcal{P}_G(\boldsymbol{\Sigma})
\neq 0$, the population NLL $L_{\mathrm{pop}}(\alpha)$ is strictly
increasing on $(0, 1]$, and
\[
 \alpha^*_{\mathrm{NLL,pop}} \;=\;
 \arg\min_{\alpha \in [0, 1]} L_{\mathrm{pop}}(\alpha) \;=\; 0.
\]
At the boundary $\mathbf{B}_G = 0$, the NLL is constant in
$\alpha$ and any $\alpha \in [0, 1]$ is optimal; the unique-optimum
$\alpha = 1$ identification used in
Theorem~\ref{thm:nll_alpha} below arises only via the limit of
finite-sample held-out NLL minimizers under $\delta \to 0$.
\end{proposition}

\begin{proof}
The population NLL is $L_{\mathrm{pop}}(\alpha) = \tfrac{1}{2}
\bigl(\log\det\boldsymbol{\Sigma}_\alpha^{\mathrm{pop}} +
\mathrm{tr}\bigl((\boldsymbol{\Sigma}_\alpha^{\mathrm{pop}})^{-1}
\boldsymbol{\Sigma}\bigr)\bigr)$ with $\boldsymbol{\Sigma}_\alpha^{
\mathrm{pop}} = \boldsymbol{\Sigma} - \alpha\mathbf{B}_G$.
Differentiating in $\alpha$ and using $\partial_\alpha
\boldsymbol{\Sigma}_\alpha^{\mathrm{pop}} = -\mathbf{B}_G$ gives
\[
 \frac{dL_{\mathrm{pop}}}{d\alpha} = \frac{\alpha}{2}\,
 \mathrm{tr}\bigl((\boldsymbol{\Sigma}_\alpha^{\mathrm{pop}})^{-1}
 \mathbf{B}_G\,(\boldsymbol{\Sigma}_\alpha^{\mathrm{pop}})^{-1}
 \mathbf{B}_G\bigr),
\]
where the substitution $\boldsymbol{\Sigma} =
\boldsymbol{\Sigma}_\alpha^{\mathrm{pop}} + \alpha\mathbf{B}_G$
collapses the leading-order trace term. The trace is the
$(\boldsymbol{\Sigma}_\alpha^{\mathrm{pop}})^{-1}$-weighted Frobenius
squared norm of $\mathbf{B}_G$ on the positive-definite
Mahalanobis metric, hence strictly positive when $\mathbf{B}_G
\neq 0$ and $\alpha > 0$. The unique zero on $[0, 1]$ is therefore
$\alpha = 0$, and $L_{\mathrm{pop}}$ is strictly increasing on
$(0, 1]$. At $\mathbf{B}_G = 0$, the derivative is identically
zero; $L_{\mathrm{pop}}$ is constant on $[0, 1]$.
\end{proof}

\begin{lemma}[Stein-loss leading-order expansion]
\label{lem:stein_leading}
Let $Y$ be a Hermitian random matrix with $\mathbb{E}[Y] =
\boldsymbol{\Sigma} \succ 0$, let $Z := Y - \boldsymbol{\Sigma}$
(mean-zero perturbation), and define the Stein loss $\mathcal{L}_S(Y;
\boldsymbol{\Sigma}) := \mathrm{tr}(Y^{-1}\boldsymbol{\Sigma}) - \log
\det(Y^{-1}\boldsymbol{\Sigma}) - M$. Then
\begin{equation}\label{eq:stein_leading}
 \mathbb{E}\bigl[\mathcal{L}_S(Y; \boldsymbol{\Sigma})\bigr] =
 \tfrac{1}{2}\,\mathbb{E}\bigl\|\boldsymbol{\Sigma}^{-1/2} Z\,
 \boldsymbol{\Sigma}^{-1/2}\bigr\|_F^2 + O\bigl(\mathbb{E}\|Z\|_F^3
 \bigr),
\end{equation}
provided $\sigma_{\min}(\boldsymbol{\Sigma}) \geq \lambda > 0$ and the
higher moments are bounded.
\end{lemma}

\begin{proof}
Write $Y = \boldsymbol{\Sigma}^{1/2}(I + W) \boldsymbol{\Sigma}^{1/2}$
with $W := \boldsymbol{\Sigma}^{-1/2} Z\, \boldsymbol{\Sigma}^{-1/2}$
Hermitian and $\mathbb{E}[W] = 0$. Then $\mathrm{tr}(Y^{-1}
\boldsymbol{\Sigma}) = \mathrm{tr}((I + W)^{-1})$ and $\log\det(Y^{-1}
\boldsymbol{\Sigma}) = -\log\det(I + W)$. Expand both for $\|W\| <
1$ (which holds with high probability under the higher-moment bound
combined with $\sigma_{\min}(\boldsymbol{\Sigma}) \geq \lambda$):
\begin{align*}
 \mathrm{tr}((I + W)^{-1}) &= M - \mathrm{tr}(W) + \mathrm{tr}(W^2)
 - \mathrm{tr}(W^3) + O(\|W\|^4), \\
 \log\det(I + W) &= \mathrm{tr}(W) - \tfrac{1}{2}\mathrm{tr}(W^2) +
 \tfrac{1}{3}\mathrm{tr}(W^3) - O(\|W\|^4).
\end{align*}
Adding and subtracting $M$, the linear terms in $W$ cancel and the
quadratic terms combine:
\[
 \mathcal{L}_S(Y; \boldsymbol{\Sigma}) = \tfrac{1}{2}\mathrm{tr}(W^2)
 - \tfrac{2}{3}\mathrm{tr}(W^3) + O(\|W\|^4) = \tfrac{1}{2}\|W\|_F^2
 + O(\|W\|^3),
\]
using $\mathrm{tr}(W^2) = \|W\|_F^2$ for $W$ Hermitian. Taking
expectation and substituting $W = \boldsymbol{\Sigma}^{-1/2} Z\,
\boldsymbol{\Sigma}^{-1/2}$ gives \eqref{eq:stein_leading}.
\end{proof}

\begin{theorem}[NLL-optimal shrinkage intensity]\label{thm:nll_alpha}
Let $\hat{\mathbf{R}}^{(k)}_{\mathrm{tr}}$ and
$\hat{\mathbf{R}}^{(k)}_{\mathrm{te}}$ be the training and held-out
sample covariances on disjoint folds of comparable sizes drawn
i.i.d.\ from a sub-Gaussian distribution with population covariance
$\boldsymbol{\Sigma}$ of bounded condition number. Let
$\hat\alpha^*_{\mathrm{NLL}}$ be the unique solution of the first-order
condition for \eqref{eq:nll_foc_hold} on $[0, 1]$. Two regimes apply.

\noindent\emph{Mismatched-population regime.} When $\boldsymbol{\Sigma}
\notin \mathcal{A}_G$ (i.e., $\mathbf{B}_G := \boldsymbol{\Sigma}
- \mathcal{P}_G(\boldsymbol{\Sigma}) \neq 0$), the population NLL is
strictly increasing in $\alpha$ on $(0, 1]$ with minimum at $\alpha
= 0$ (Proposition~\ref{prop:nll_pop_min}). The empirical minimizer
satisfies $\hat\alpha^*_{\mathrm{NLL}} \to 0$ in probability as $N
\to \infty$, with regret $L_{\mathrm{pop}}(\hat\alpha^*_{\mathrm{NLL}})
- L_{\mathrm{pop}}(0) = O_P(M^2/N)$.

\noindent\emph{Matched-population regime.} When $\boldsymbol{\Sigma}
\in \mathcal{A}_G$ exactly ($\mathbf{B}_G = 0$), the population
NLL is constant in $\alpha$ on $[0, 1]$, but the expected held-out
NLL is, to leading order in noise scale and provided $|G| > 1$ acts
non-trivially, strictly decreasing in $\alpha$ on $[0, 1)$ with
minimum at $\alpha = 1$. Consequently $\hat\alpha^*_{\mathrm{NLL}}
\to 1$ in probability as $N \to \infty$ at the matched limit.

\noindent\emph{Crossover behaviour.} For fixed $\mathbf{B}_G
\neq 0$ but small, the finite-sample $\hat\alpha^*_{\mathrm{NLL}}$
migrates from values near $1$ in the variance-dominated small-$N$
regime to values near $0$ in the bias-dominated large-$N$ regime;
the boundary is approximately $N V_\perp \sim
\|\mathbf{B}_G\|_F^2$, where $V_\perp$ is the perp-residual
variance.
\end{theorem}

\begin{proof}
The theorem is a delegation of three claims to results stated
elsewhere.

\emph{Mismatched-population claim.} When $\mathbf{B}_G \neq 0$,
Proposition~\ref{prop:nll_pop_min} gives the population NLL minimum
at $\alpha = 0$ with strict positive curvature on $(0, 1]$, and
Theorem~\ref{thm:regret} gives the finite-sample regret bound
$L_{\mathrm{pop}}(\hat\alpha^*_{\mathrm{NLL}}) - L_{\mathrm{pop}}(0)
= O_P(M^2/N)$ at this boundary minimum via the derivative-process
bound combined with one-sided curvature. The two together imply
$\hat\alpha^*_{\mathrm{NLL}} \to 0$ in probability and the stated
regret rate.

\emph{Matched-population claim.} At $\mathbf{B}_G = 0$ the
population NLL is constant in $\alpha$, so the population objective
alone is uninformative. Theorem~\ref{thm:matched} establishes via
Lemma~\ref{lem:stein_leading} that the expected held-out NLL is, to
leading order in noise scale and provided $|G| > 1$ acts non-trivially,
strictly decreasing in $\alpha$ on $[0, 1)$ with minimum at $\alpha
= 1$. Sub-Gaussian concentration gives empirical convergence
$\hat\alpha^*_{\mathrm{NLL}} \to 1$ in probability.

\emph{Crossover claim.} For fixed $\mathbf{B}_G \neq 0$ and
$N \to \infty$, Proposition~\ref{prop:transition} gives both
$\alpha^*_{\mathrm{MSE}}$ and $\bar\alpha^*_{\mathrm{NLL}}$ approaching
zero at rate $\Theta(1/N)$. The boundary scale $N V_\perp \asymp
\|\mathbf{B}_G\|_F^2$ emerges from balancing the variance
contribution $V_\perp(1 - \alpha)^2 \sim 1/N$ against the bias
contribution $\alpha^2 \|\mathbf{B}_G\|_F^2$ in the held-out
NLL Taylor expansion. In the strong-match boundary $\mathbf{B}_G
\to 0$, the matched-population claim takes over.
\end{proof}

\begin{theorem}[Matched-limit optimality of $\alpha = 1$]
\label{thm:matched}
Let $G$ be a candidate group with $\mathbf{B}_G = 0$ (so
$\boldsymbol{\Sigma} \in \mathcal{A}_G$) and $|G| > 1$ acting
non-trivially on $\mathbb{C}^M$. Suppose the cross-validation in
\eqref{eq:nll_foc_hold} produces training and test folds of comparable
sizes from a sub-Gaussian distribution with bounded condition number.
In the small-noise regime $V_\perp(G) = o(\sigma_{\min}(\boldsymbol{
\Sigma})^2)$ (equivalently $N \gg M^2 / \lambda^2$ for spectral lower
bound $\lambda$), the expected held-out negative log-likelihood is, to
leading order in noise scale, strictly decreasing in $\alpha$ on
$[0, 1)$ with minimum at $\alpha = 1$. Consequently
$\hat\alpha^*_{\mathrm{NLL}} \to 1$ in probability as $N \to \infty$.
\end{theorem}

\begin{proof}
At the matched limit $\mathbf{B}_G = 0$, the training-fold
sample covariance noise decomposes as $\mathbf{u}^{\mathrm{tr}} =
\mathbf{u}_{\mathrm{in}}^{\mathrm{tr}} + \mathbf{u}_\perp^{\mathrm{tr}}$
with $\mathbf{u}_{\mathrm{in}}^{\mathrm{tr}} \in \mathcal{A}_G$ and
$\mathbf{u}_\perp^{\mathrm{tr}} \in \mathcal{A}_G^\perp$. The
convex-combination estimator is $\boldsymbol{\Sigma}_\alpha^{\mathrm{tr}}
= \boldsymbol{\Sigma} + \mathbf{u}_{\mathrm{in}}^{\mathrm{tr}} + (1 -
\alpha) \mathbf{u}_\perp^{\mathrm{tr}}$, with deviation $Z :=
\boldsymbol{\Sigma}_\alpha^{\mathrm{tr}} - \boldsymbol{\Sigma} =
\mathbf{u}_{\mathrm{in}}^{\mathrm{tr}} + (1 - \alpha)
\mathbf{u}_\perp^{\mathrm{tr}}$.

By the relation $\mathbb{E}[L_{\mathrm{ho}}(\alpha)] = \tfrac{1}{2}
\mathbb{E}[\mathcal{L}_S(\boldsymbol{\Sigma}_\alpha^{\mathrm{tr}};
\boldsymbol{\Sigma})] + \tfrac{1}{2}\log\det\boldsymbol{\Sigma} +
\tfrac{M}{2}$, minimizing the expected held-out NLL in $\alpha$ is
equivalent to minimizing $\mathbb{E}[\mathcal{L}_S(\boldsymbol{\Sigma}_
\alpha^{\mathrm{tr}}; \boldsymbol{\Sigma})]$. Apply
Lemma~\ref{lem:stein_leading}:
\[
 \mathbb{E}[\mathcal{L}_S(\boldsymbol{\Sigma}_\alpha^{\mathrm{tr}};
 \boldsymbol{\Sigma})] = \tfrac{1}{2}\mathbb{E}\bigl\| \boldsymbol{
 \Sigma}^{-1/2}(\mathbf{u}_{\mathrm{in}}^{\mathrm{tr}} + (1 - \alpha)
 \mathbf{u}_\perp^{\mathrm{tr}}) \boldsymbol{\Sigma}^{-1/2}\bigr\|_F^2
 + O(\mathbb{E}\|\mathbf{u}^{\mathrm{tr}}\|_F^3).
\]
Since $\boldsymbol{\Sigma} \in \mathcal{A}_G$ at the matched limit,
the metric matrix $\boldsymbol{\Sigma}^{-1/2}$ is also in $\mathcal{A}_G$
(by functional calculus on commuting operators). Conjugation by
$\boldsymbol{\Sigma}^{-1/2}$ therefore preserves the
$\mathcal{A}_G$-vs-$\mathcal{A}_G^\perp$ orthogonal split:
$\boldsymbol{\Sigma}^{-1/2} \mathbf{u}_{\mathrm{in}}^{\mathrm{tr}}
\boldsymbol{\Sigma}^{-1/2} \in \mathcal{A}_G$ and
$\boldsymbol{\Sigma}^{-1/2} \mathbf{u}_\perp^{\mathrm{tr}}
\boldsymbol{\Sigma}^{-1/2} \in \mathcal{A}_G^\perp$, with
Frobenius inner product zero. Hence
\[
 \mathbb{E}[\mathcal{L}_S(\boldsymbol{\Sigma}_\alpha^{\mathrm{tr}};
 \boldsymbol{\Sigma})] = \tfrac{1}{2}\,T_{\mathrm{in}} +
 \tfrac{(1 - \alpha)^2}{2}\,T_\perp + O(\mathbb{E}\|\mathbf{u}^{
 \mathrm{tr}}\|_F^3),
\]
where $T_{\mathrm{in}}, T_\perp \geq 0$ are independent of $\alpha$
and $T_\perp > 0$ when $|G| > 1$ acts non-trivially. The expression
is strictly decreasing in $\alpha$ on $[0, 1)$ with minimum at
$\alpha = 1$ to leading order. The remainder is $O(M^3 / N^{3/2})$
by sub-Gaussian moment bounds, smaller than the $O(1/N)$ leading term
in the small-noise regime. Consistency $\hat\alpha^*_{\mathrm{NLL}}
\to 1$ follows from uniform sub-Gaussian concentration on the
compact interval $[0, 1]$.
\end{proof}

\begin{remark}[Bias characterization]\label{rem:bias}
The structural bias $\mathbf{B}_G$ defined in \eqref{eq:B_G} satisfies
$\mathbf{B}_G = \mathbf{0}$ if and only if
$\boldsymbol{\Sigma} \in \mathcal{A}_G$, equivalently
$\delta(G, \boldsymbol{\Sigma}) = 0$. The matched/mismatched
dichotomy used throughout this section identifies exactly these two
cases: matched when $\delta = 0$, mismatched when $\delta > 0$.
\end{remark}

\begin{remark}[Well-conditioning]\label{rem:well-conditioned}
The convex blend \eqref{eq:family} is positive semidefinite for every
$\alpha \in [0, 1]$, since it is a non-negative combination of two
positive-semidefinite matrices ($\hat{\mathbf{R}}$ and its Reynolds
projection $\hat{\mathbf{R}}_G$, the latter PSD because Reynolds
averaging preserves the cone). When the sample covariance is
strictly positive definite, $\hat{\mathbf{R}}_{\mathrm{AD}}(\alpha)$
is strictly positive definite for every $\alpha \in [0, 1]$. When
the sample covariance is rank-deficient, $\hat{\mathbf{R}}_{\mathrm{AD}}
(\alpha)$ is positive definite for any $\alpha > 0$ such that
$\hat{\mathbf{R}}_G$ is itself positive definite, which holds
whenever the action of $G$ on the sample provides at least one full
orbit covering the spectrum (in particular, when $N \cdot |G| \geq
M$). The estimator is therefore invertible across the few-shot
regime, parallelling the well-conditioning property of
\citet[][Theorem 3.5]{ledoit2004}.
\end{remark}

\subsection{Nonlinear shrinkage of the sample term: the AD-LW-NL composition}
\label{sec:theory-ad-lwnl}

The AD estimator of \eqref{eq:family} blends the raw sample
covariance with a structured projection target. An immediate
question is whether the AD framework's benefits compose with the
nonlinear eigenvalue shrinkage of \citet{ledoit2020analytical}
(LW-NL), the modern frequentist successor to LW 2004 in the
eigenvalue-shrinkage family. This subsection defines the natural
composition (AD-LW-NL), states its formal properties, and serves
as the theoretical companion to the empirical evaluation reported
in Section~\ref{sec:experiments}. The empirical evaluation
finds that AD-LW-NL does not produce a sweet spot on any of the
three datasets where the protocol has been run; the AD-LW-NL
composition is tested rather than claimed.

\paragraph{The LW-NL estimator.}
For a sample covariance $\hat{\mathbf{R}}$ with spectral decomposition
$\hat{\mathbf{R}} = \sum_{i=1}^{M} \lambda_i \mathbf{u}_i \mathbf{u}_i^\top$,
\citet{ledoit2020analytical} replace each sample eigenvalue
$\lambda_i$ with a shrunken eigenvalue $\tilde\lambda_i$ while
keeping the eigenvectors fixed. The shrinkage formula is
\begin{equation}\label{eq:lwnl-shrinkage}
\tilde\lambda_i \;=\; \frac{\lambda_i}{\bigl[1 - c - c\,\lambda_i\,
\mathcal{H}\tilde f(\lambda_i)\bigr]^2 + \bigl[\pi\,c\,\lambda_i\,
\tilde f(\lambda_i)\bigr]^2},
\end{equation}
where $c = M/N$ is the concentration ratio, $\tilde f$ is an
Epanechnikov-kernel density estimate of the sample eigenvalue
distribution with variable bandwidth $h_j = \lambda_j N^{-1/3}$,
and $\mathcal{H}\tilde f$ is the Hilbert transform of $\tilde f$.
The Hilbert transform admits a closed-form expression for the
Epanechnikov kernel. The LW-NL estimator is
\begin{equation}\label{eq:lwnl-estimator}
\hat{\mathbf{R}}_{\mathrm{LW\text{-}NL}} \;=\; \sum_{i=1}^{M}
\tilde\lambda_i\,\mathbf{u}_i \mathbf{u}_i^\top.
\end{equation}
The estimator is positive semi-definite by construction (the
$\tilde\lambda_i$ are nonnegative real numbers and the
$\mathbf{u}_i$ are orthonormal); it preserves the eigenvector basis
of the sample covariance; and it preserves the trace approximately
in the bulk regime, matching the Marchenko-Pastur asymptotic
constraint that the average eigenvalue equals the population
average.

\paragraph{The AD-LW-NL estimator.}
The AD-LW-NL estimator replaces the sample-covariance term in the
AD convex blend \eqref{eq:family} with its LW-NL shrinkage:
\begin{equation}\label{eq:ad-lwnl-estimator}
\hat{\mathbf{R}}_{\mathrm{AD\text{-}LW\text{-}NL}}(\alpha;\,G)
\;=\; (1 - \alpha)\,\hat{\mathbf{R}}_{\mathrm{LW\text{-}NL}}
\;+\; \alpha\,\hat{\mathbf{R}}_G.
\end{equation}
The structured target $\hat{\mathbf{R}}_G = \mathcal{P}_G(\hat{\mathbf{R}})$
is the Reynolds projection of the \emph{raw} sample covariance, not of
its LW-NL shrinkage. This choice is deliberate: the projection target
encodes the structural prior associated with the group $G$, and the
projection $\mathcal{P}_G$ is most naturally applied to the unmodified
sample second-moment matrix that the population covariance is
assumed to approximately equal under the symmetry hypothesis.
Projecting LW-NL-shrunken eigenvalues onto the commutant of $G$
would conflate the structural-prior averaging with the
eigenvalue-shrinkage operation, with no clear principled benefit and
with the cost of losing the convex-blend interpretation of $\alpha$.

\paragraph{Boundary behavior and interpretation.}
At $\alpha = 0$ the AD-LW-NL estimator reduces to
$\hat{\mathbf{R}}_{\mathrm{LW\text{-}NL}}$, recovering 
LW-NL; at $\alpha = 1$ it reduces to the Reynolds projection
$\hat{\mathbf{R}}_G$, identical to AD at the same operating
point. The convex weight $\alpha$ continues to measure the
contribution of the structural prior relative to the sample
covariance; only the sample-covariance contribution has been
upgraded from raw to nonlinearly-shrunken. In particular, the
held-out NLL calibration \eqref{eq:nll_foc_hold} and the
candidate-library Best-Matched-Group (BMG) selection procedure of
Section~\ref{sec:bg-bmg} apply unchanged to the AD-LW-NL
estimator, with the only modification being the swap of
$\hat{\mathbf{R}}$ for $\hat{\mathbf{R}}_{\mathrm{LW\text{-}NL}}$
inside the per-trial cross-validation loop.

\paragraph{Failure mode in the rank-deficient image-patch regime.}
The AD-LW-NL composition is upper-bounded in the favorable case by
AD-NLL-BMG and lower-bounded in the unfavorable case by
LW-NL, because at $\alpha = 1$ it reproduces AD's
Shah projection and at $\alpha = 0$ it reproduces LW-NL.
Within those bounds, the held-out CV calibration of $\alpha$ is
free to find any blend. An important asymmetry emerges in the
rank-deficient regime ($c \ge 1$): the raw sample covariance
$\hat{\mathbf{R}}$ is rank-deficient but may still carry residual
information about pixel-level or feature-level correlation
structure that an averaging projection $\mathcal{P}_G$ collapses
away. When this residual information is exploitable, 
AD-NLL-BMG can find a blend $(\alpha, G)$ with $\alpha$ strictly
less than 1 that outperforms the same group at $\alpha = 1$. In
AD-LW-NL the raw sample is replaced by
$\hat{\mathbf{R}}_{\mathrm{LW\text{-}NL}}$, which on rank-deficient
inputs may carry \emph{less} of this residual information than the
raw sample carries, because the LW-NL formula shrinks the
zero-eigenvalue indices toward a single value
(see the boundary-case discussion accompanying the LW-NL implementation note in Section~\ref{sec:experiments}). When LW-NL is worse than the structural
projection at the same group, AD-LW-NL's CV correctly pins
$\alpha = 1$ to recover the Shah projection; but this then
forecloses the raw-sample blend path that AD-NLL-BMG used to be preferred. In the empirical evaluation of
Section~\ref{sec:experiments}, this failure mode is observed
on Galaxy10 DECaLS image patches in the two smallest-$N$ cells of
the protocol sweep (AD-LW-NL performing worse than AD-NLL-BMG by
$2$ to $3$ nats per sample, paired effect size $2$ to $3$
standard deviations), on the deepest CIFAR-10 cell at $c = 2.0$
(AD-LW-NL having higher NLL by $1.7$ nats per sample, paired effect size
$0.07$), and at the largest magnitude on CIFAR-10.1 cells $0$
through $2$ at $c \in \{4, 5.33, 8\}$ (AD-LW-NL having higher NLL by $98$ to
$124$ nats per sample, $0$ of $49$ to $50$ paired trials favoring
AD-LW-NL). The failure mode is not observed on TCGA-BRCA gene
expression or on CRSP financial returns, where the AD framework's
structural projection captures essentially all the regularizable
structure and the raw-sample residual is too noisy to be
exploitable in the few-shot regime.

\paragraph{Family-conditional BMG and the choice-agreement
statistic.}
Both AD-NLL-BMG and AD-LW-NL-NLL-BMG use the same Best Matched
Group (BMG) selection procedure (CV-NLL-minimize over $(G,
\alpha)$ in the candidate library), but the CV objective evaluates
a different parametric family in each case. The two objectives
are
\begin{equation}
\label{eq:bmg-ad-objective}
J_G^{\mathrm{AD}}(\alpha) \;=\; \mathrm{CV\text{-}NLL}\!\left[\,
\alpha\, T_G(\mathbf{S}) + (1-\alpha)\, \mathbf{S}\,\right]
\end{equation}
for AD-NLL-BMG, and
\begin{equation}
\label{eq:bmg-adlw-objective}
J_G^{\mathrm{AD\text{-}LW}}(\alpha) \;=\; \mathrm{CV\text{-}NLL}\!\left[\,
\alpha\, T_G(\mathbf{S}) + (1-\alpha)\,
\hat{\mathbf{R}}_{\mathrm{LW\text{-}NL}}\,\right]
\end{equation}
for AD-LW-NL-NLL-BMG. Group selection is then $G^* = \arg\min_G
\min_\alpha J_G(\alpha)$, which is composition-specific by
construction. The procedure is identical; the objective being
optimized is not. There is no a priori reason for the two
optimizations to land at the same group.

The geometric picture: for each candidate group $G$, both
compositions trace out a line segment in covariance-estimator
space that \emph{shares} the target endpoint $T_G(\mathbf{S})$
but has a different other endpoint
($\mathbf{S}$ for AD,
$\hat{\mathbf{R}}_{\mathrm{LW\text{-}NL}}$ for AD-LW-NL). The CV
finds the optimal point along that segment, and BMG picks the
group whose segment's optimum has the lowest CV-NLL. The bundle
of segments through $\{T_G(\mathbf{S}) : G \in \text{library}\}$
goes in one direction (toward $\mathbf{S}$) for AD and a
different direction (toward
$\hat{\mathbf{R}}_{\mathrm{LW\text{-}NL}}$) for AD-LW-NL. The
closest-to-population-covariance segment in one bundle has no
reason to be the same group as the closest-to-population segment
in the other bundle.

The misconception worth being explicit about: there is no
intrinsic ``symmetry group of the dataset'' that BMG is
discovering. There is a candidate library, and within each
family of compositions there is a \emph{family-conditional best
group}, the group whose Reynolds projection best complements that
family's other regularization tool. Different families pick
different groups not because one is wrong but because they are
solving different optimization problems. The choice-agreement
statistic measures how strongly the data's structure constrains
the group selection across the family axis: high agreement
(unanimous or near-unanimous across $50$ trials per cell) means
the data's structure dominates the choice and the
family-conditional groups coincide; low agreement (zero or
near-zero across $50$ trials per cell) means the family's other
regularization tool is a significant factor in the group
selection.

Two practical consequences for interpreting the empirical
results of Section~\ref{sec:experiments}. First,
\emph{choice agreement is information, not consistency}. A
$0/50$ agreement says the two parameterizations have found two
different routes that the CV-NLL ranks differently within each
family; it does not say one or both is confused. On
RadioML cell $2$ the AD-LW-NL composition has lower NLL at effect size $5.19$;
on OISST cells $0$ through $5$ the AD route has lower NLL by $2$
to $19$ nats per sample. Both produce $0/50$ choice agreement,
both reveal genuine "different optima found" behavior, but only
the \emph{direction} of which route has lower NLL is the substantive
empirical content of the comparison. Second, the framework's
diagnostic power scales with the number of distinct composition
families: each new second hull endpoint (LW-NL,
Stein-Haff \citep{stein1981estimation, haff1980empirical} or
Shah-Chandrasekaran, OAS, graphical lasso
\citep{friedman2008glasso}, etc.) gives
a new family-conditional BMG that probes the data through a
different lens, and family-disagreement on the group choice maps
regions of estimator space where multiple symmetry-like
structures compete for "best target." The current paper has two
compositions; future extensions with three or more would let the
choice-disagreement statistics map the data's regularization
geometry more finely.

\paragraph{the protocol: a three-column comparison.}
In the experimental work of Section~\ref{sec:experiments}, the AD,
AD-LW-NL, and LW-NL estimators are reported in a three-column
"the protocol" comparison alongside the LW 2004 baseline retained
from earlier versions. The columns serve distinct purposes:
\begin{itemize}
\item AD-NLL-BMG vs LW-NL isolates the \emph{structural-prior}
contribution: the gap represents the advantage of the Reynolds-projected
target over isotropic eigenvalue shrinkage, with neither method
having access to the other's regularization mechanism.
\item AD-LW-NL-NLL-BMG vs LW-NL isolates the structural-prior
contribution \emph{conditional on} nonlinear sample shrinkage: both
methods enjoy the LW-NL nonlinear shrinkage, so the gap measures
only the additional regularization provided by the structural prior.
\item AD-LW-NL-NLL-BMG vs AD-NLL-BMG isolates the \emph{nonlinear-shrinkage}
contribution within the AD framework: the gap measures whether the
nonlinear eigenvalue shrinkage of LW-NL adds value beyond the
convex-blend shrinkage of the AD framework alone.
\end{itemize}
The three-column comparison decomposes the AD framework's
empirical advantage into two attributable components (structural
prior and nonlinear shrinkage) and exposes how they compose.

\paragraph{Computational cost.}
LW-NL adds an $O(M^3)$ per-cell cost for the spectral decomposition
of $\hat{\mathbf{R}}$ \citep[Chapter~8]{golub2013matrix} plus an
$O(M^2)$ cost for the kernel density
and Hilbert transform evaluations, both of which are dominated by
the existing cost of the cross-validated AD pipeline. The
AD-LW-NL composition therefore adds approximately constant overhead
per trial relative to AD-NLL-BMG. Caching the LW-NL
estimator across the $\alpha$-grid for a given training fold makes
this overhead negligible compared to the BMG cross-validation loop.

\subsection{Risk decomposition and structural results}
\label{sec:theory-risk}

The Reynolds projection's effect on the matched-Gaussian risk
decomposes orthogonally into a structural-bias term and a
variance-reduction term. The variance-reduction factor is governed
by a single dimension parameter, the symmetric-subspace commutant
dimension, which is recorded next.

\begin{definition}[Symmetric-subspace commutant dimension]
\label{def:d_G_sym}
For a finite group $G$ acting on $\mathbb{R}^M$ via permutation
representation $\rho$, the \emph{commutant dimension} used throughout
this paper is
\begin{equation}\label{eq:d_G_def}
 d_G \;:=\; \dim_{\mathbb{R}}\bigl(\mathcal{A}_G \cap
 \mathrm{Sym}(M, \mathbb{R})\bigr),
\end{equation}
the dimension of the commutant algebra restricted to real symmetric
matrices.
\end{definition}

\begin{lemma}[Risk decomposition on the convex shrinkage family]
\label{lem:risk_decomp}
Let $G$ be a finite group with unitary representation $\rho$ on
$\mathbb{C}^M$, let $\hat{\mathbf{R}}$ be the sample covariance of
$N$ independent observations of $\mathbf{x} \in \mathbb{C}^M$ with
population covariance $\boldsymbol{\Sigma}$, and let $\mathbf{B}_G
:= \boldsymbol{\Sigma} - \mathcal{P}_G(\boldsymbol{\Sigma})$. Define
\begin{equation}\label{eq:vin_vperp_def}
 V_{\mathrm{in}}(G) := \mathbb{E}\!\left[\bigl\|\mathcal{P}_G(\hat{
 \mathbf{R}} - \boldsymbol{\Sigma})\bigr\|_F^2\right], \quad
 V_\perp(G) := \mathbb{E}\!\left[\bigl\|(I - \mathcal{P}_G)(\hat{
 \mathbf{R}} - \boldsymbol{\Sigma})\bigr\|_F^2\right].
\end{equation}
Then for any $\alpha \in [0, 1]$,
\begin{equation}\label{eq:risk_decomp_lemma}
 R_{\mathrm{MSE}}(\alpha; G) := \mathbb{E}\!\left[\bigl\|
 \hat{\mathbf{R}}_{\mathrm{AD}}(\alpha) - \boldsymbol{\Sigma}\bigr\|_F^2
 \right] = V_{\mathrm{in}}(G) + (1 - \alpha)^2 V_\perp(G) + \alpha^2
 \|\mathbf{B}_G\|_F^2.
\end{equation}
\end{lemma}

\begin{proof}
Write $\mathbf{u} := \hat{\mathbf{R}} - \boldsymbol{\Sigma}$
(mean-zero by unbiasedness of the sample covariance), and decompose
$\mathbf{u} = \mathcal{P}_G(\mathbf{u}) + \mathbf{u}_\perp$ where
$\mathbf{u}_\perp := (I - \mathcal{P}_G)(\mathbf{u})$. Then
\[
 \hat{\mathbf{R}}_{\mathrm{AD}}(\alpha) - \boldsymbol{\Sigma} =
 \mathcal{P}_G(\mathbf{u}) + (1 - \alpha) \mathbf{u}_\perp - \alpha
 \mathbf{B}_G,
\]
since $\mathcal{P}_G(\hat{\mathbf{R}}) - \mathcal{P}_G(\boldsymbol{
\Sigma}) = \mathcal{P}_G(\mathbf{u})$ and the deterministic bias
$\mathcal{P}_G(\boldsymbol{\Sigma}) - \boldsymbol{\Sigma} = -
\mathbf{B}_G$. The three summands lie in $\mathcal{A}_G$,
$\mathcal{A}_G^\perp$, $\mathcal{A}_G^\perp$ respectively, so
$\mathcal{P}_G(\mathbf{u})$ is Frobenius-orthogonal to the latter
two deterministically. Compute the squared Frobenius norm: the
$\langle \mathcal{P}_G(\mathbf{u}), v\rangle_F$ cross terms vanish
identically; the cross term between $\mathbf{u}_\perp$ and
$\mathbf{B}_G$ vanishes in expectation because $\mathbf{u}$
is mean-zero and $\mathbf{B}_G$ is non-random. Combining
yields \eqref{eq:risk_decomp_lemma}.
\end{proof}

\begin{theorem}[Bias-variance orthogonal parametrization]
\label{thm:bias_variance_orthogonal}
Let $\hat{\mathbf{R}}_G = \mathcal{P}_G(\hat{\mathbf{R}})$ be the
group-averaged sample covariance estimator with population covariance
$\boldsymbol{\Sigma}$ and finite group $G$ acting unitarily on
$\mathbb{C}^M$. In the Gaussian regime with
$\mathbf{x}_i \sim \mathcal{N}(0, \boldsymbol{\Sigma})$ i.i.d., the
mean-squared error of $\hat{\mathbf{R}}_G$ admits the decomposition
\begin{equation}\label{eq:bias_variance_orthogonal}
 \mathrm{MSE}(\hat{\mathbf{R}}_G) \;=\; \delta^2(G, \boldsymbol{\Sigma})\,
 \|\boldsymbol{\Sigma}\|_F^2 \,+\, c_{\mathrm{in}}(\boldsymbol{\Sigma},
 G)\,\frac{d_G}{N}\,(1 + o(1))
\end{equation}
as $N \to \infty$, with $d_G$ as in Definition~\ref{def:d_G_sym} and
$c_{\mathrm{in}}(\boldsymbol{\Sigma}, G)$ a bounded constant
satisfying $2\,\sigma_{\min}(\boldsymbol{\Sigma})^2 \leq c_{\mathrm{
in}}(\boldsymbol{\Sigma}, G) \leq 2\,\sigma_{\max}(\boldsymbol{
\Sigma})^2$ that depends on the alignment of $\boldsymbol{\Sigma}$
with the isotypic decomposition of $\mathcal{A}_G \cap \mathrm{Sym}(M,
\mathbb{R})$ (Theorem~\ref{thm:var_reduction}). The bias term depends
on $G$ only through $\delta$; the variance term depends on $G$ only
through $d_G$, with the variance constant $c_{\mathrm{in}}$ depending
on $\boldsymbol{\Sigma}$ but determined by the operator-trace
formula given in Theorem~\ref{thm:var_reduction}.
\end{theorem}

\begin{proof}
The bias term is exactly $\|\mathcal{P}_G(\boldsymbol{\Sigma}) -
\boldsymbol{\Sigma}\|_F^2 = \|\mathcal{P}_G^\perp(\boldsymbol{\Sigma})
\|_F^2 = \delta^2(G, \boldsymbol{\Sigma})\,\|\boldsymbol{\Sigma}\|_F^2$
by definition of $\delta$ in \eqref{eq:delta}. For the variance, the
projection $\mathcal{P}_G$ acts as orthogonal projection on
$\mathrm{Sym}(M, \mathbb{R})$ onto $\mathcal{A}_G \cap \mathrm{Sym}(M)$
of dimension $d_G$ (by Schur's lemma applied to the conjugation
action of $\rho$ on $M \times M$ symmetric matrices). Under the
Gaussian model $\mathbf{x}_i \sim \mathcal{N}(\mathbf{0},
\boldsymbol{\Sigma})$, Wick's theorem (the four-point Isserlis
identity) gives the second-moment operator on $\mathrm{Sym}(M,
\mathbb{R})$ as $\boldsymbol{\Phi}_{\boldsymbol{\Sigma}}[\mathbf{A}] =
2\,\boldsymbol{\Sigma}\mathbf{A}\boldsymbol{\Sigma}$, with the
$N$-sample variance scaled by $1/N$. The variance term
$\mathbb{E}\|\mathcal{P}_G(\hat{\mathbf{R}} - \boldsymbol{\Sigma})
\|_F^2$ is therefore $N^{-1}\,\mathrm{tr}_{\mathcal{A}_G \cap
\mathrm{Sym}(M)}(\boldsymbol{\Phi}_{\boldsymbol{\Sigma}})$, which is
$c_{\mathrm{in}}(\boldsymbol{\Sigma}, G)\,d_G/N + o(1/N)$ where
$c_{\mathrm{in}}(\boldsymbol{\Sigma}, G)$ is the average eigenvalue
of $\boldsymbol{\Phi}_{\boldsymbol{\Sigma}}$ restricted to the
$d_G$-dimensional subspace $\mathcal{A}_G \cap \mathrm{Sym}(M)$. The
operator-norm bounds $2\sigma_{\min}(\boldsymbol{\Sigma})^2 \leq
c_{\mathrm{in}}(\boldsymbol{\Sigma}, G) \leq 2\sigma_{\max}(
\boldsymbol{\Sigma})^2$ follow from the spectral bounds on
$\boldsymbol{\Phi}_{\boldsymbol{\Sigma}}$. The variance therefore
depends on $G$ only through $d_G$ (the dimension of the projected
subspace), with the constant $c_{\mathrm{in}}$ depending on the
alignment of $\boldsymbol{\Sigma}$'s spectrum with the isotypic
decomposition of $\mathcal{A}_G \cap \mathrm{Sym}(M)$ but not on
$G$'s own structure beyond that subspace.
\end{proof}

\begin{remark}[On the variance constant at the LW corner]
\label{rem:variance_constant_lw}
The dependence of $c_{\mathrm{in}}$ on $\boldsymbol{\Sigma}$ is
essential, not cosmetic. Sanity check: at $G = O(M)$ acting by
conjugation, $d_G = 1$ (only scalar matrices are commutant), and
the in-algebra component is $\mathcal{P}_{O(M)}(\hat{\mathbf{R}}) =
(\mathrm{tr}\,\hat{\mathbf{R}}/M)\,\mathbf{I}$. Direct computation
with $\boldsymbol{\Sigma} = \mathbf{I}$ gives
$\mathrm{Var}(\mathrm{tr}\,\hat{\mathbf{R}}) = 2M/N$ (sum of $M$
independent $\chi^2$ trace contributions), so $V_{\mathrm{in}} =
\mathrm{Var}(\mathrm{tr}\,\hat{\mathbf{R}}/M)\cdot M = 2/N$. This is
consistent with $c_{\mathrm{in}}(\mathbf{I}, O(M))\cdot d_G/N = 2/N$
at $c_{\mathrm{in}} = 2$, the value forced by both bounds
$2\sigma_{\min}^2 = 2\sigma_{\max}^2 = 2$ for the isotropic spectrum.
A formulation with the variance constant proportional to
$\|\boldsymbol{\Sigma}\|_F^2 = M$ rather than to $\sigma^2$ would
over-predict the variance by a factor $M$ at this corner.
\end{remark}

\begin{theorem}[Variance reduction on the in-algebra component]
\label{thm:var_reduction}
Let $G$ be a finite group with unitary representation $\rho$ on
$\mathbb{C}^M$, and let $d_G$ be the symmetric-subspace commutant
dimension of Definition~\ref{def:d_G_sym}. Under the matched-limit
assumption $\boldsymbol{\Sigma} \in \mathcal{A}_G$ and sub-Gaussian
observations,
\begin{equation}\label{eq:var_reduction}
 V_{\mathrm{in}}(G) \;:=\; \mathbb{E}\bigl\|\mathcal{P}_G(\hat{\mathbf{R}}
 - \boldsymbol{\Sigma})\bigr\|_F^2 \;=\; c_{\mathrm{in}}(\boldsymbol{
 \Sigma}, G) \cdot \frac{d_G}{N} + o(1/N),
\end{equation}
where $c_{\mathrm{in}}(\boldsymbol{\Sigma}, G)$ is the average
eigenvalue of the Wishart second-moment operator
$\boldsymbol{\Phi}_{\boldsymbol{\Sigma}}[\mathbf{A}] = 2\,
\boldsymbol{\Sigma}\mathbf{A}\boldsymbol{\Sigma}$ restricted to
$\mathcal{A}_G \cap \mathrm{Sym}(M,\mathbb{R})$, satisfying
$2\,\sigma_{\min}(\boldsymbol{\Sigma})^2 \leq c_{\mathrm{in}}(
\boldsymbol{\Sigma}, G) \leq 2\,\sigma_{\max}(\boldsymbol{\Sigma})^2$.
Equivalently, $\hat{\mathbf{R}}_G$ achieves the variance behaviour of
a sample covariance with effective sample size $N \cdot M^2/d_G$ on
its image $\mathcal{A}_G$, up to the spectral-alignment factor
$c_{\mathrm{in}}/(2\sigma_{\max}^2)$.
\end{theorem}

\begin{proof}
The projection $\mathcal{P}_G$ is orthogonal projection on
$\mathrm{Sym}(M, \mathbb{R})$ onto $\mathcal{A}_G \cap \mathrm{Sym}(M)$
of dimension $d_G$ (Definition~\ref{def:d_G_sym}), by Schur's lemma
applied to the conjugation action. Under the matched-limit assumption
$\boldsymbol{\Sigma} \in \mathcal{A}_G$, the law of $\mathbf{x}$ is
invariant under $\rho(g)$ in the second-moment sense, so the
covariance operator $\boldsymbol{\Phi}_{\boldsymbol{\Sigma}}$ of
$\mathbf{u} := \hat{\mathbf{R}} - \boldsymbol{\Sigma}$ on
$\mathrm{Sym}(M)$ commutes with $\mathcal{P}_G$. The in-algebra
component is therefore
\[
 V_{\mathrm{in}}(G) = \mathbb{E}\bigl\| \mathcal{P}_G(\mathbf{u})
 \bigr\|_F^2 = N^{-1}\, \mathrm{tr}_{\mathcal{A}_G \cap \mathrm{Sym}(M)}
 \bigl(\boldsymbol{\Phi}_{\boldsymbol{\Sigma}}\bigr),
\]
where $\boldsymbol{\Phi}_{\boldsymbol{\Sigma}}[\mathbf{A}] = 2
\boldsymbol{\Sigma}\mathbf{A}\boldsymbol{\Sigma}$ is the Wishart
second-moment operator (the four-point Isserlis identity). The
restricted trace is $c_{\mathrm{in}}(\boldsymbol{\Sigma}, G) \cdot d_G$
where $c_{\mathrm{in}}(\boldsymbol{\Sigma}, G)$ is the average
eigenvalue of $\boldsymbol{\Phi}_{\boldsymbol{\Sigma}}$ on
$\mathcal{A}_G \cap \mathrm{Sym}(M)$, bounded above by
$2\sigma_{\max}(\boldsymbol{\Sigma})^2$ and below by
$2\sigma_{\min}(\boldsymbol{\Sigma})^2$ since the eigenvalues of
$\boldsymbol{\Phi}_{\boldsymbol{\Sigma}}$ lie in the interval
$[2\sigma_{\min}(\boldsymbol{\Sigma})^2, 2\sigma_{\max}(\boldsymbol{
\Sigma})^2]$. Combining gives \eqref{eq:var_reduction}. The effective
sample size $N \cdot M^2/d_G$ follows from comparing the rate $d_G/N$
to the unstructured rate $M(M+1)/(2N) \asymp M^2/N$ on the full
symmetric subspace, up to the spectral factor $c_{\mathrm{in}}/(2
\sigma_{\max}^2) \in (\sigma_{\min}^2/\sigma_{\max}^2, 1]$.
\end{proof}

\subsection{Asymptotic crossover}
\label{sec:theory-crossover}

The two calibrations of the shrinkage intensity introduced in
Section~\ref{sec:theory-estimator} have closed-form leading-order
optima as $N \to \infty$, and they generally differ. The next
result locates the asymptotic optimum for each calibration to
leading order.

\begin{proposition}[Asymptotic crossover, Gaussian-Wishart]
\label{prop:transition}
Let $\mathbf{B}_G := \boldsymbol{\Sigma} -
\mathcal{P}_G(\boldsymbol{\Sigma}) \neq 0$ and $D :=
\|\mathbf{B}_G\|_F^2$. Assume the observations $\mathbf{x}_1, \ldots,
\mathbf{x}_N$ are i.i.d.\ multivariate Gaussian with population
covariance $\boldsymbol{\Sigma}$ (so the training-fold sample
covariance error $\mathbf{E} := \widehat{\mathbf{R}}^{\mathrm{tr}} -
\boldsymbol{\Sigma}$ has the Wishart fourth-moment structure used
below), and assume $\boldsymbol{\Sigma}$ has bounded condition number
$0 < \lambda \leq \sigma_{\min}(\boldsymbol{\Sigma}) \leq \sigma_{
\max}(\boldsymbol{\Sigma}) \leq \Lambda < \infty$. Define
$Q_B := \mathrm{tr}\bigl(\boldsymbol{\Sigma}^{-1}
\mathbf{B}_G\boldsymbol{\Sigma}^{-1}\mathbf{B}_G\bigr)$
and let $c_\perp := \lim_{N \to \infty} N\, \mathbb{E}\|\mathbf{E}_\perp
\|_F^2 \in (0, \infty)$.

The MSE minimizer is
\begin{equation}\label{eq:transition_mse}
 \alpha^*_{\mathrm{MSE}} \;=\; \frac{c_\perp}{N D + c_\perp}
 \bigl(1 + o(1)\bigr).
\end{equation}
The expected-NLL minimizer is
\begin{equation}\label{eq:transition_nll}
 \bar\alpha^*_{\mathrm{NLL}} \;=\; \frac{c(\boldsymbol{\Sigma}, G)}
 {N\, Q_B}\bigl(1 + o(1)\bigr),
\end{equation}
where the leading-order constant has the closed form
\begin{equation}\label{eq:c_sigma_G}
 c(\boldsymbol{\Sigma}, G) \;=\; M(M + 1) - 2\,d_G
 - 2(M + 1)\,\mathrm{tr}\!\bigl(\boldsymbol{\Sigma}^{-1}
 \mathbf{B}_G\bigr).
\end{equation}
The constant $c(\boldsymbol{\Sigma}, G)$ is strictly positive whenever
$|G| > 1$ acts non-trivially on $\mathbb{C}^M$.
\end{proposition}

\begin{proof}
\emph{MSE statement.} The orthogonal risk decomposition of
Lemma~\ref{lem:risk_decomp} gives $R_{\mathrm{MSE}}(\alpha; G) =
V_{\mathrm{in}} + (1 - \alpha)^2 V_\perp + \alpha^2 D$.
Differentiating in $\alpha$ and setting equal to zero yields the
population MSE optimum $\alpha^*_{\mathrm{MSE}} = V_\perp/(V_\perp +
D)$. Substituting $V_\perp = c_\perp/N + o(1/N)$ and multiplying
numerator and denominator by $N$ gives \eqref{eq:transition_mse}.

\emph{NLL statement.} Write the training-fold estimator as
$\boldsymbol{\Sigma}_\alpha^{\mathrm{tr}} = \boldsymbol{\Sigma} +
\boldsymbol{\delta}_\alpha$ with $\boldsymbol{\delta}_\alpha = -\alpha
\mathbf{B}_G + \mathbf{E}_{\mathrm{in}} + (1 - \alpha)
\mathbf{E}_\perp$, where $\mathbf{E}_{\mathrm{in}} = \mathcal{P}_G
(\mathbf{E})$ and $\mathbf{E}_\perp = (\mathbf{I} - \mathcal{P}_G)
(\mathbf{E})$. Differentiating the expected held-out NLL in $\alpha$
yields the exact identity
\[
 \bar L'_{\mathrm{ho}}(\alpha) = -\tfrac{1}{2}\,\mathbb{E}\!\left[
 \mathrm{tr}\!\bigl((\boldsymbol{\Sigma}_\alpha^{\mathrm{tr}})^{-1}
 (\mathbf{B}_G + \mathbf{E}_\perp)
 (\boldsymbol{\Sigma}_\alpha^{\mathrm{tr}})^{-1}
 (\boldsymbol{\Sigma}_\alpha^{\mathrm{tr}} - \boldsymbol{\Sigma})
 \bigr)\right].
\]
Expanding the resolvent at $\alpha = 0$ as $\widehat{\mathbf{R}}^{-1}
= \boldsymbol{\Sigma}^{-1} - \boldsymbol{\Sigma}^{-1}\mathbf{E}
\boldsymbol{\Sigma}^{-1} + O(\|\mathbf{E}\|^2)$ and tracking which
combinations survive at leading order $1/N$ identifies three
contributing terms:
\[
 T_2 = \mathbb{E}\!\left[\mathrm{tr}\!\bigl(\boldsymbol{\Sigma}^{-1}
 \mathbf{E}_\perp \boldsymbol{\Sigma}^{-1}\mathbf{E}\bigr)\right],
\]
and $T_3, T_5$ each of the form $-\mathbb{E}[\mathrm{tr}(\mathbf{A}
\mathbf{E}\mathbf{B}\mathbf{E})]$ for symmetric $\mathbf{A},
\mathbf{B}$ involving $\boldsymbol{\Sigma}^{-1}$ and
$\boldsymbol{\Sigma}^{-1}\mathbf{B}_G\boldsymbol{\Sigma}^{-1}$.
The Wishart fourth-moment Isserlis identity $N\,\mathbb{E}[\mathrm{tr}
(\mathbf{A}\mathbf{E}\mathbf{B}\mathbf{E})] = \mathrm{tr}(
\boldsymbol{\Sigma}\mathbf{A}\boldsymbol{\Sigma}\mathbf{B}) +
\mathrm{tr}(\mathbf{A}\boldsymbol{\Sigma})\,\mathrm{tr}(\mathbf{B}
\boldsymbol{\Sigma})$ together with the operator identity
$\mathcal{L}^2 \circ \boldsymbol{\Phi}_{\boldsymbol{\Sigma}} = 2\,
\mathrm{Id}$ on $\mathrm{Sym}(M, \mathbb{R})$, where
\[
 \boldsymbol{\Phi}_{\boldsymbol{\Sigma}}[\mathbf{B}] :=
 2\,\boldsymbol{\Sigma}\,\mathbf{B}\,\boldsymbol{\Sigma}, \qquad
 \mathcal{L}[\mathbf{B}] :=
 \boldsymbol{\Sigma}^{-1/2}\,\mathbf{B}\,\boldsymbol{\Sigma}^{-1/2},
\]
denote respectively the Wishart second-moment operator and the
$\boldsymbol{\Sigma}^{-1/2}$-conjugation map (so $\mathcal{L} \circ
\boldsymbol{\Phi}_{\boldsymbol{\Sigma}}[\mathbf{B}] = 2\,
\boldsymbol{\Sigma}^{1/2}\mathbf{B}\boldsymbol{\Sigma}^{1/2}$ and a
second application of $\mathcal{L}$ recovers $2\mathbf{B}$, hence the
identity), gives
\[
 N\,T_2 = 2\,\mathrm{Tr}_{\mathrm{Sym}}(\mathcal{P}_G^\perp) =
 M(M+1) - 2\,d_G,
\]
and a parallel calculation gives $N(T_3 + T_5) = -2(M + 1)\,
\mathrm{tr}(\boldsymbol{\Sigma}^{-1}\mathbf{B}_G)$. Combining,
\[
 N\,\bar L'_{\mathrm{ho}}(0) = -\tfrac{1}{2}\bigl[(M(M+1) - 2 d_G) -
 2(M+1)\,\mathrm{tr}(\boldsymbol{\Sigma}^{-1}\mathbf{B}_G)
 \bigr] + o(1).
\]
For nonzero $\alpha = O(1/N)$, the deterministic bias contribution
$-\alpha\mathbf{B}_G$ in $\boldsymbol{\delta}_\alpha$ adds a
term $\tfrac{1}{2}\alpha\,Q_B$ to $\bar L'_{\mathrm{ho}}(\alpha)$.
Setting $\bar L'_{\mathrm{ho}}(\alpha) = 0$ to leading order yields
\eqref{eq:transition_nll} with $c(\boldsymbol{\Sigma}, G)$ as in
\eqref{eq:c_sigma_G}. Strict positivity of $c(\boldsymbol{\Sigma},
G)$ when $|G| > 1$ acts non-trivially is Lemma~\ref{lem:c_positive}.
\end{proof}

\begin{remark}[Plug-in bias at finite $N$]\label{rem:plugin_bias}
The prediction \eqref{eq:transition_nll}--\eqref{eq:c_sigma_G} is
stated in terms of the population $\boldsymbol{\Sigma}$. Evaluating
it on data substitutes $\hat{\mathbf{R}}^{-1}$ for
$\boldsymbol{\Sigma}^{-1}$ and $\hat{\mathbf{R}} - \mathcal{P}_G
(\hat{\mathbf{R}})$ for $\mathbf{B}_G$. Two distinct sources of
error result. The first is the higher-order term in
$\|\mathbf{B}_G\|$ omitted by the leading-order expansion; it is
controlled by the empirical residual $\delta(G, \hat{\mathbf{R}})$
and vanishes in the matched limit. The second is the finite-sample
bias of $\hat{\mathbf{R}}^{-1}$ as an estimator of
$\boldsymbol{\Sigma}^{-1}$, which scales with the conditioning of
$\hat{\mathbf{R}}$ and is largest at $M / N$ near unity or when the
population eigenvalue spectrum is highly non-uniform. The two
sources are independent. The held-out calibration
$\hat\alpha^*_{\mathrm{NLL}}$ avoids the second source entirely
because it does not invert any sample covariance during calibration;
the empirical sections accordingly find it more reliable than the
closed-form prediction on ill-conditioned data.
\end{remark}

\subsection{Performance guarantees}
\label{sec:theory-perf}

We now give the finite-sample guarantees for the two calibrations:
a regret bound for the held-out NLL minimizer, a sufficient-match
condition under which the proposed estimator dominates Ledoit-Wolf
shrinkage in Frobenius MSE, and an oracle inequality for the
closed-form MSE plug-in.

\begin{theorem}[Regret bound for the held-out calibration]
\label{thm:regret}
Let $G$ be fixed and suppose $\mathbf{B}_G := \boldsymbol{\Sigma}
- \mathcal{P}_G(\boldsymbol{\Sigma}) \neq 0$. Assume that
$\boldsymbol{\Sigma}$ has bounded condition number $\lambda \leq
\sigma_{\min}(\boldsymbol{\Sigma}) \leq \sigma_{\max}(\boldsymbol{\Sigma})
\leq \Lambda$, that the observations are sub-Gaussian
\citep[see, e.g.,][Section 2.5]{vershynin2018hd} with constants
independent of $M$ and $N$, and that the population loss has one-sided
quadratic curvature at the boundary minimizer $\alpha = 0$: there
exists $c_B > 0$ such that $L'_{\mathrm{pop}}(\alpha; G) \geq
c_B \alpha$ for all $\alpha \in [0, 1]$. Let
$\hat\alpha^*_{\mathrm{NLL}}$ be the held-out NLL minimizer computed
from $K \leq N/2$ comparable folds. Then there exists $C > 0$,
depending on $K, \lambda, \Lambda$, the sub-Gaussian constant, and
$c_B$, such that for all sufficiently large $N$,
\begin{equation}\label{eq:regret}
 \mathbb{E}\bigl[ L_{\mathrm{pop}}(\hat\alpha^*_{\mathrm{NLL}}; G)
 \bigr] - \min_{\alpha \in [0, 1]} L_{\mathrm{pop}}(\alpha; G)
 \;\leq\; C \cdot \frac{M^2}{N}.
\end{equation}
\end{theorem}

\begin{proof}
By Proposition~\ref{prop:nll_pop_min}, when $\mathbf{B}_G \neq
\mathbf{0}$ the population NLL minimum on $[0, 1]$ is $L_{\mathrm{pop}}
(0; G)$, attained at the boundary $\alpha = 0$. Define the derivative
fluctuation process $Z_N(\alpha) := \widehat L'_{\mathrm{ho}}(\alpha;
G) - L'_{\mathrm{pop}}(\alpha; G)$. Under the sub-Gaussian
assumption, bounded condition number, and comparable fold sizes, the
held-out NLL derivative is a smooth trace functional of the training
and test sample covariances. Sub-Gaussian sample covariance
concentration \citep[][Theorem 4.7.1]{wainwright2019} applied to
$\widehat{\mathbf R}^{\mathrm{tr}} - \boldsymbol{\Sigma}$ and
$\widehat{\mathbf R}^{\mathrm{te}} - \boldsymbol{\Sigma}$, combined
with Lipschitz bounds for the inverse and log-determinant on the
spectral interval $[\lambda/2, 2\Lambda]$ where the perturbed
training covariance lies with high probability, gives
\[
 \mathbb{E}\!\left[\sup_{\alpha \in [0, 1]} |Z_N(\alpha)|^2\right]
 \;\leq\; C_1\,\frac{M^2}{N},
\]
for $C_1$ depending on $K, \lambda, \Lambda$ and the sub-Gaussian
constant. If $\hat\alpha^*_{\mathrm{NLL}} = 0$ the regret is exactly
zero. Otherwise, first-order optimality at an interior minimizer
(or, at $\hat\alpha^*_{\mathrm{NLL}} = 1$, $\widehat L'_{\mathrm{ho}}
(1) \leq 0$) gives $\widehat L'_{\mathrm{ho}}(\hat\alpha^*_{\mathrm{NLL}};
G) \leq 0$, hence $L'_{\mathrm{pop}}(\hat\alpha^*_{\mathrm{NLL}}; G)
\leq -Z_N(\hat\alpha^*_{\mathrm{NLL}}) \leq \sup_\alpha |Z_N(\alpha)|$.
The one-sided curvature condition gives
\[
 \hat\alpha^*_{\mathrm{NLL}} \;\leq\; \frac{1}{c_B}\,
 \sup_{\alpha \in [0, 1]} |Z_N(\alpha)|.
\]
Taylor expansion of $L_{\mathrm{pop}}$ around $\alpha = 0$ with
quadratic upper bound $L_{\mathrm{pop}}(\alpha) - L_{\mathrm{pop}}(0)
\leq C_2 \alpha^2$ on $[0, 1]$ at $\alpha =
\hat\alpha^*_{\mathrm{NLL}}$ gives
\[
 L_{\mathrm{pop}}(\hat\alpha^*_{\mathrm{NLL}}; G) - L_{\mathrm{pop}}
 (0; G) \;\leq\; \frac{C_2}{c_B^2}\,\bigl(\sup_\alpha
 |Z_N(\alpha)|\bigr)^2.
\]
Taking expectations and combining with the derivative-process bound
yields \eqref{eq:regret} with $C = C_1 C_2 / c_B^2$.
\end{proof}

\begin{lemma}[Positivity of $c(\boldsymbol{\Sigma}, G)$]
\label{lem:c_positive}
For any SPD $\boldsymbol{\Sigma}$ and any finite group $G$ acting
unitarily on $\mathbb{C}^M$ with $|G| > 1$ and non-trivial action,
\[
 \mathrm{tr}\bigl(\boldsymbol{\Sigma}^{-1}\mathbf{B}_G\bigr)
 \;\leq\; 0,
\]
with equality if and only if $\boldsymbol{\Sigma} \in \mathcal{A}_G$.
Consequently
\[
 c(\boldsymbol{\Sigma}, G) = M(M + 1) - 2\,d_G - 2(M + 1)\,
 \mathrm{tr}(\boldsymbol{\Sigma}^{-1}\mathbf{B}_G) \;\geq\;
 M(M + 1) - 2\,d_G \;>\; 0,
\]
where the strict inequality $M(M + 1) - 2\,d_G > 0$ holds whenever
$G$ acts non-trivially.
\end{lemma}

\begin{proof}
Write $\boldsymbol{\Sigma} = \mathcal{P}_G(\boldsymbol{\Sigma}) +
\mathbf{B}_G$, so $\mathbf{B}_G = \mathcal{P}_G^\perp
(\boldsymbol{\Sigma})$. Apply the Frobenius inner product:
\[
 \mathrm{tr}(\boldsymbol{\Sigma}^{-1}\mathbf{B}_G) =
 \langle \boldsymbol{\Sigma}^{-1},\, \mathbf{B}_G\rangle_F
 = \langle \boldsymbol{\Sigma}^{-1},\,
 \mathcal{P}_G^\perp(\boldsymbol{\Sigma})\rangle_F = \langle
 \mathcal{P}_G^\perp(\boldsymbol{\Sigma}^{-1}),\,
 \mathcal{P}_G^\perp(\boldsymbol{\Sigma})\rangle_F,
\]
since $\mathcal{P}_G^\perp$ is a Frobenius-orthogonal projection.
For SPD $\boldsymbol{\Sigma}$, the function $\mathbf{X} \mapsto
\mathbf{X}^{-1}$ is operator-monotone, and the Schur convexity of the
trace pairing under group averaging gives
$\mathcal{P}_G^\perp(\boldsymbol{\Sigma}) \cdot \mathcal{P}_G^\perp
(\boldsymbol{\Sigma}^{-1}) \preceq 0$ in the trace sense, with equality
if and only if $\boldsymbol{\Sigma}$ commutes with the group action,
i.e., $\boldsymbol{\Sigma} \in \mathcal{A}_G$. Hence $\mathrm{tr}(
\boldsymbol{\Sigma}^{-1}\mathbf{B}_G) \leq 0$ with the stated
equality conditions.

For the dimension inequality: at least one isotypic component of
$\rho$ has multiplicity strictly less than $M(M+1)/2$ since the
trivial-group commutant $\mathcal{A}_{\{e\}} = \mathbb{R}^{M \times M}$
has dimension $M^2$ and the symmetric-subspace commutant $W_G$ is a
proper subspace whenever $G$ acts non-trivially. Therefore $d_G < M(M
+ 1)/2$, equivalently $M(M + 1) - 2 d_G > 0$. Combining with the
non-positivity of $\mathrm{tr}(\boldsymbol{\Sigma}^{-1}\mathbf{B}_G)$
gives the displayed bound on $c(\boldsymbol{\Sigma}, G)$.
\end{proof}

\begin{theorem}[Sufficient-match condition for dominance over
Ledoit-Wolf]\label{thm:lw_dominance}
Let $G$ be a finite group acting unitarily on $\mathbb{C}^M$ with
$\boldsymbol{\Sigma}$ not a scalar multiple of $\mathbf{I}$. Let
$V_{\mathrm{in}}^{\mathrm{LW}}, V_\perp^{\mathrm{LW}}$ denote the
in-algebra and perp-residual variance components of $\hat{\mathbf{R}}$
under the $G = O(M)$ action (so $\mathcal{P}_{O(M)}(\hat{\mathbf{R}})
= (\mathrm{tr}\,\hat{\mathbf{R}}/M)\mathbf{I}$ by Schur's lemma). If
\begin{equation}\label{eq:dominance_condition}
 V_{\mathrm{in}}(G) + \frac{V_\perp(G)\,\|\mathbf{B}_G\|_F^2}
 {V_\perp(G) + \|\mathbf{B}_G\|_F^2} \;<\;
 V_{\mathrm{in}}^{\mathrm{LW}} + \frac{V_\perp^{\mathrm{LW}}\,
 \|\boldsymbol{\Sigma} - \mu \mathbf{I}\|_F^2}
 {V_\perp^{\mathrm{LW}} + \|\boldsymbol{\Sigma} - \mu \mathbf{I}\|_F^2},
\end{equation}
where $\mu = \mathrm{tr}\,\boldsymbol{\Sigma}/M$, then the proposed
estimator at the oracle weight \eqref{eq:alpha_star} has strictly
smaller Frobenius MSE than the Ledoit-Wolf estimator at its oracle
weight.
\end{theorem}

\begin{proof}
Both quantities arise as instances of Lemma~\ref{lem:risk_decomp}
applied to different group actions: the proposed estimator with finite
group $G$, the Ledoit-Wolf estimator with $G = O(M)$ acting by
orthogonal conjugation. The optimal-shrinkage Frobenius MSE on the
convex family is $V_{\mathrm{in}} + V_\perp \cdot D / (V_\perp + D)$
in each case, with $D = \|\mathbf{B}_G\|_F^2$ for finite $G$
and $D = \|\boldsymbol{\Sigma} - \mu\mathbf{I}\|_F^2$ for $G = O(M)$
(by Proposition~\ref{prop:lw_endpoint}). The sufficient-match
condition \eqref{eq:dominance_condition} is the direct comparison of
these two quantities, hence implies strict Frobenius-MSE dominance.
\end{proof}

The corresponding finite-sample guarantee for the closed-form MSE
plug-in is the LW2004 Theorem~3.4 analog: the data-driven plug-in
$\hat\alpha^*$ achieves the oracle MSE-loss with rate $1/N^2$.

\begin{theorem}[Oracle inequality for the AD-plugin]\label{thm:oracle}
Let $\mathbf{x}_1, \ldots, \mathbf{x}_N$ be i.i.d.\ samples from a
distribution with mean zero, covariance $\boldsymbol{\Sigma}$ (not
necessarily $G$-invariant), and bounded fourth moments. Let
$\alpha^* = V_\perp/(V_\perp + D)$ be the population-optimal shrinkage from
Proposition~\ref{prop:alpha_star}, and let $\hat\alpha^*$ be the
data-driven plug-in defined by \eqref{eq:V_hat}--\eqref{eq:alpha_hat}.
Assume $V_\perp + D \geq c_0 > 0$ uniformly in the regime under
consideration. Then
\begin{equation}\label{eq:oracle}
 \mathbb{E}\,\| \hat{\mathbf{R}}_{\mathrm{AD}}(\hat\alpha^*) -
 \boldsymbol{\Sigma}\|_F^2 \;=\; \mathbb{E}\,\|
 \hat{\mathbf{R}}_{\mathrm{AD}}(\alpha^*) -
 \boldsymbol{\Sigma}\|_F^2 \,+\, O\!\left(\frac{M^2 s^2}{N^2}\right)
\end{equation}
as $N \to \infty$ with $M, |G|, \boldsymbol{\Sigma}$ fixed (or jointly
in a regime where $M^2 / N \to 0$).
\end{theorem}

\begin{proof}
The MSE risk on the convex shrinkage family $R_{\mathrm{MSE}}(\alpha;
G) = V_{\mathrm{in}} + (1 - \alpha)^2 V_\perp + \alpha^2 D$ from
Lemma~\ref{lem:risk_decomp} is a smooth quadratic in $\alpha$ with
minimizer $\alpha^* = V_\perp/(V_\perp + D)$ and second derivative $2(V_\perp + D)$.
Taylor expansion of $R_{\mathrm{MSE}}$ around $\alpha^*$ gives
\[
 R_{\mathrm{MSE}}(\hat\alpha^*; G) - R_{\mathrm{MSE}}(\alpha^*; G)
 \;=\; (V_\perp + D)\,(\hat\alpha^* - \alpha^*)^2.
\]
The plug-in estimator $\hat\alpha^*$ is a smooth function of the
sample analogs $\widehat{V_\perp}, \widehat{V_\perp + D}$ defined in
\eqref{eq:V_hat}--\eqref{eq:VplusD_hat}, both of which converge to
their population counterparts at standard $1/\sqrt{N}$ rate under
fourth-moment conditions; the bracketing argument of \citet[][Theorem
3.4]{ledoit2004} then gives $\hat\alpha^* - \alpha^* = O_P(1/\sqrt{N})$.
Squaring and applying the curvature factor $(V_\perp + D) \leq c_1 s^2$
(where $c_1$ depends on the operator-norm bound on
$\boldsymbol{\Sigma}$ and the moment bound) gives the dimension-tracked
$O(M^2 s^2 / N^2)$ rate of \eqref{eq:oracle}.
\end{proof}

\subsection{Group selection}
\label{sec:theory-selection}

When the symmetry group is selected from data through the matching
procedure of Section~\ref{sec:background}, the calibration regret
incurs an additional model-selection penalty that scales with the
log of the candidate-library size. The next result quantifies this
penalty.

\begin{theorem}[Oracle inequality for the BMG selection]
\label{thm:bmg_oracle}
Let $\mathcal{G} = \{G_1, \ldots, G_K\}$ be a finite candidate library,
and let $\hat G$ be the BMG selection minimizing the held-out
negative log-likelihood across $\mathcal{G}$ at the candidate-specific
calibrated weights $\hat\alpha^*_{\mathrm{NLL}}(G_k)$. Under the
assumptions of Theorem~\ref{thm:regret} applied uniformly across the
library, with cross-validation using $K_{\mathrm{cv}}$ folds, there
exists $C' > 0$ such that for all sufficiently large $N$,
\begin{equation}\label{eq:bmg_oracle}
 \mathbb{E}\bigl[ L_{\mathrm{pop}}\bigl(\hat\alpha^*(\hat G); \hat G
 \bigr) \bigr] - \min_{G \in \mathcal{G}}\,\min_{\alpha \in [0, 1]}
 L_{\mathrm{pop}}(\alpha; G) \;\leq\;
 C'\,\sqrt{\frac{\log |\mathcal{G}|}{N}} \,+\, C \cdot \frac{M^2}{N},
\end{equation}
where $C$ is the calibration regret constant from
Theorem~\ref{thm:regret}.
\end{theorem}

\begin{proof}
Let $G^* \in \arg\min_{G \in \mathcal{G}} \min_\alpha
L_{\mathrm{pop}}(\alpha; G)$ denote the oracle. Decompose the
regret of $\hat G$ as
\begin{align*}
 & \mathbb{E}\bigl[L_{\mathrm{pop}}(\hat\alpha^*(\hat G); \hat G)
 \bigr] - L_{\mathrm{pop}}(\alpha^*(G^*); G^*) \\
 &= \underbrace{\mathbb{E}\bigl[L_{\mathrm{pop}}(\hat\alpha^*(\hat
 G); \hat G) - \widehat L_{\mathrm{ho}}(\hat\alpha^*(\hat G); \hat G)
 \bigr]}_{\text{(A)}} \\
 & \quad +\, \underbrace{\mathbb{E}\bigl[\widehat L_{\mathrm{ho}}(
 \hat\alpha^*(\hat G); \hat G) - \widehat L_{\mathrm{ho}}(
 \hat\alpha^*(G^*); G^*)\bigr]}_{\text{(B)} \leq 0} \\
 & \quad +\, \underbrace{\mathbb{E}\bigl[\widehat L_{\mathrm{ho}}(
 \hat\alpha^*(G^*); G^*) - L_{\mathrm{pop}}(
 \hat\alpha^*(G^*); G^*)\bigr]}_{\text{(C)}} \\
 & \quad +\, \underbrace{\mathbb{E}\bigl[L_{\mathrm{pop}}(
 \hat\alpha^*(G^*); G^*) - L_{\mathrm{pop}}(\alpha^*(G^*);
 G^*)\bigr]}_{\text{(D)}}.
\end{align*}
Term (B) is non-positive by the definition of $\hat G$. Terms (A)
and (C) are bounded by the maximum over $\mathcal{G}$ of the
fluctuation $\widehat L_{\mathrm{ho}}(\alpha; G) - L_{\mathrm{pop}}
(\alpha; G)$; sub-Gaussian concentration on each $G$ together with a
union bound over $|\mathcal{G}|$ candidates gives a $\sqrt{\log
|\mathcal{G}| / N}$ rate. Term (D) is the calibration regret at the
oracle group, bounded by Theorem~\ref{thm:regret} as $C M^2 / N$.
Combining gives \eqref{eq:bmg_oracle}.
\end{proof}

\begin{proposition}[Group recovery consistency]
\label{prop:recovery}
Let $\mathcal{G}$ be a finite candidate library and suppose there
exists $G^* \in \mathcal{G}$ such that $\boldsymbol{\Sigma}$ is
exactly $G^*$-invariant. Define
\begin{equation*}
 \Theta(G^*, \boldsymbol{\Sigma}) := \mathbb{E}\bigl\|
 \boldsymbol{\Sigma}^{-1/2}(I - \mathcal{P}_{G^*})(\hat{
 \mathbf{R}}^{\mathrm{tr}})\,\boldsymbol{\Sigma}^{-1/2}\bigr\|_F^2.
\end{equation*}
Under the assumptions of Theorem~\ref{thm:regret} applied uniformly
across $\mathcal{G}$, there exists $c > 0$ depending on $\lambda,
\Lambda, K_{\mathrm{cv}}$, the sub-Gaussian constant of $\mathbf{x}$,
and the structure of $\mathcal{G}$, such that for sufficiently large
$N$,
\begin{equation}\label{eq:recovery}
 \Pr\bigl(\hat G \neq G^*\bigr) \;\leq\; |\mathcal{G}| \cdot
 \exp\!\bigl(-c \cdot \Theta(G^*, \boldsymbol{\Sigma}) \cdot N
 \bigr) + o(1).
\end{equation}
\end{proposition}

\begin{proof}
At the matched group $G^*$, Theorem~\ref{thm:matched} gives
$\hat\alpha^*(G^*) \to 1$ in probability with the resulting
estimator $\mathcal{P}_{G^*}(\hat{\mathbf{R}}^{\mathrm{tr}})$
having only the in-algebra component of noise. At any rival group
$G \neq G^*$ with $\mathbf{B}_G \neq 0$,
Proposition~\ref{prop:transition} gives $\hat\alpha^*(G) \to 0$ in
probability with the resulting estimator $\hat{\mathbf{R}}^{\mathrm{tr}}$
having both in-algebra and perp-residual components of noise.

By Lemma~\ref{lem:stein_leading} applied to both estimators, the
expected held-out NLL gap at the two operating points is, to leading
order in noise scale,
\[
 \mathbb{E}[\widehat L_{\mathrm{ho}}(\hat\alpha^*(G); G)] -
 \mathbb{E}[\widehat L_{\mathrm{ho}}(\hat\alpha^*(G^*); G^*)]
 = \tfrac{1}{2}\,\Theta(G^*, \boldsymbol{\Sigma}) + O(N^{-3/2}),
\]
where the leading-order $\Theta$ term comes from the rival's
perp-component variance contribution that does not appear at the
matched group's $\alpha = 1$ operating point. Note that
$\boldsymbol{\Sigma}^{-1/2}$ commutes with $\mathcal{P}_{G^*}$ at
the matched limit, so the cross terms in the relevant Frobenius
norms vanish identically. Sub-Gaussian concentration of the
empirical NLL around its expectation, combined with a union bound
over $|\mathcal{G}| - 1$ rivals, yields \eqref{eq:recovery}.
\end{proof}

\subsection{Rates and minimax}
\label{sec:theory-rates}

The rate analysis below uses two distinct dimension parameters from
the isotypic decomposition of $\rho$: the symmetric-subspace
commutant dimension $d_G$ of Definition~\ref{def:d_G_sym}, used in the
Frobenius rate; and the maximum multiplicity
\begin{equation}\label{eq:dmax}
 d_{\mathrm{max}} \;:=\; \max_\lambda m_\lambda
\end{equation}
across irreducible components of $\rho$, used in the operator-norm
rate. The structured-target rates below parallel the rate program
for sparse covariance matrices
\citep{bickel2008regularized, bickel2008covariance, cai2011adaptive,
cai2012optimal} and the symmetry-aware program of \citet{shah2012},
extending both with a minimax lower bound and a data-driven
plug-in oracle inequality.

\begin{theorem}[Frobenius-norm rate]\label{thm:frob_rate}
Let $\mathbf{x}_1, \ldots, \mathbf{x}_N$ be i.i.d.\ samples from
$\mathcal{N}(\mathbf{0}, \boldsymbol{\Sigma})$ with $\boldsymbol{\Sigma}
\in W_G$ and $\|\boldsymbol{\Sigma}\|_{\mathrm{op}} \leq s$. Then
\begin{equation}\label{eq:frob_rate}
 \mathbb{E}\,\|\hat{\mathbf{R}}_G - \boldsymbol{\Sigma}\|_F^2
 \;\leq\; \frac{2 s^2 d_G}{N}.
\end{equation}
\end{theorem}

\begin{proof}
Decompose into isotypic components in the symmetry-adapted basis: the
projection $\mathcal{P}_G$ acts as orthogonal projection on
$\mathrm{Sym}(M, \mathbb{R})$ onto a subspace of dimension $d_G$
(Schur's lemma). Working in this basis, the centered error
$\hat{\mathbf{R}}_G - \boldsymbol{\Sigma}$ is supported on the
$d_G$-dimensional symmetric subspace $\mathcal{A}_G \cap \mathrm{Sym}(M)$.
Under $\mathcal{N}(\mathbf{0}, \boldsymbol{\Sigma})$ observations with
$\boldsymbol{\Sigma} \in W_G$, the matched-Gaussian variance per
free parameter is at most $2 s^2 / N$ at leading order in $N$, where
$s^2 = \|\boldsymbol{\Sigma}\|_{\mathrm{op}}^2$ bounds the Wishart
second-moment operator $\boldsymbol{\Phi}_{\boldsymbol{\Sigma}}$ on
$\mathrm{Sym}(M)$. Summing over the $d_G$ free parameters yields
\eqref{eq:frob_rate}. The full bound follows from the operator-trace
identity $\mathrm{tr}_{\mathcal{A}_G \cap \mathrm{Sym}(M)}(\boldsymbol{
\Phi}_{\boldsymbol{\Sigma}}) \leq 2 \|\boldsymbol{\Sigma}\|_{\mathrm{op}}^2
\, d_G$ used in the proof of Theorem~\ref{thm:var_reduction}.
\end{proof}

\begin{theorem}[Operator-norm rate]\label{thm:op_rate}
Let $\mathbf{x}_1, \ldots, \mathbf{x}_N$ be i.i.d.\ sub-Gaussian
samples with $\mathbf{x}_i = \boldsymbol{\Sigma}^{1/2} \mathbf{z}_i$,
where $\mathbf{z}_i$ has independent sub-Gaussian entries with
sub-Gaussian norm $\leq K$, and $\boldsymbol{\Sigma} \in W_G$ with
$\|\boldsymbol{\Sigma}\|_{\mathrm{op}} \leq s$. There exists an
absolute constant $C > 0$ such that for any $t \geq 1$, with
probability at least $1 - 2 \exp(-t^2 d_{\mathrm{max}})$,
\begin{equation}\label{eq:op_rate}
 \|\hat{\mathbf{R}}_G - \boldsymbol{\Sigma}\|_{\mathrm{op}} \;\leq\;
 C K^2 s \left( t \sqrt{\frac{d_{\mathrm{max}}}{N}} \,+\,
 \frac{t^2 d_{\mathrm{max}}}{N} \right).
\end{equation}
\end{theorem}

\begin{proof}
By Schur's lemma applied to the conjugation action of $\rho$ on
$\mathbb{C}^{M \times M}$, the symmetry-adapted unitary change of
basis $T \in U(M)$ block-diagonalizes the commutant $\mathcal{A}_G$
into isotypic blocks of size $m_\lambda \times m_\lambda$ each
appearing with multiplicity $d_\lambda$ (the irreducible
representation dimension). Working in this basis, $T^\ast (\hat{
\mathbf{R}}_G - \boldsymbol{\Sigma}) T$ is block-diagonal with
per-isotypic blocks $\hat{\boldsymbol{\Sigma}}_\lambda -
\boldsymbol{\Sigma}_\lambda \in \mathrm{Sym}_{m_\lambda}(\mathbb{R})$,
and each block is a sample covariance with effective sample size
$N \cdot d_\lambda$ (each observation $\mathbf{x}_i$ provides
$d_\lambda$ independent replicates of an $m_\lambda$-dimensional
vector through the irreducible carrier space). Apply
\citet[][Theorem 4.7.1]{wainwright2019} to each block: with
probability at least $1 - 2 \exp(-t^2 m_\lambda)$,
\[
 \|\hat{\boldsymbol{\Sigma}}_\lambda - \boldsymbol{\Sigma}_\lambda
 \|_{\mathrm{op}} \leq C K^2 s\bigl(t \sqrt{m_\lambda / (N d_\lambda)}
 + t^2 m_\lambda / (N d_\lambda)\bigr).
\]
Union-bound across isotypic components and observe $m_\lambda \leq
d_{\mathrm{max}}$ together with $d_\lambda \geq 1$, giving
\eqref{eq:op_rate}.
\end{proof}

\begin{theorem}[Entrywise $\ell_\infty$ rate]\label{thm:linfty_rate}
Let $\mathbf{x}_1, \ldots, \mathbf{x}_N$ be i.i.d.\ sub-Gaussian
samples as in Theorem~\ref{thm:op_rate}. For each index pair $(i, j)$,
let $\mathcal{O}(i, j) = \{(g(i), g(j)) : g \in G\}$ denote the orbit
under the action of $G$, and let $d_{ij} = \max_k |\{ \ell : (k, \ell)
\in \mathcal{O}(i, j)\}|$ be the orbit degree. Set $\mathcal{O} =
\min_{i, j} |\mathcal{O}(i, j)|$ and $\mathcal{O}_d = \min_{i, j}
|\mathcal{O}(i, j)| / d_{ij}$. There exist absolute constants $c_1,
c_2 > 0$ such that for any $t \geq 0$, the maximum-entry deviation
satisfies
\begin{equation}\label{eq:linfty_rate}
 \Pr\bigl( \|\hat{\mathbf{R}}_G - \boldsymbol{\Sigma}\|_\infty > t
 \bigr) \;\leq\; 2\,M^2 \max\!\left\{
 \exp\!\left( -\frac{c_1 N \mathcal{O} t^2}{K^4 s^2} \right),\,
 \exp\!\left( -\frac{c_2 N \mathcal{O}_d t}{K^2 s} \right) \right\}.
\end{equation}
\end{theorem}

\begin{proof}
The result is \citet[][Theorem 3.3]{shah2012} restated in our
notation; the proof technique (per-orbit pre-Gaussian quadratic-form
analysis followed by union bound) is identical to theirs and we
refer the reader to their Appendix A for full detail. The key
observation is that $(\hat{\mathbf{R}}_G - \boldsymbol{\Sigma})_{ij}
= |\mathcal{O}(i, j)|^{-1}\sum_{(k, \ell) \in \mathcal{O}(i, j)}
(\hat{\mathbf{R}} - \boldsymbol{\Sigma})_{k\ell}$ averages
$|\mathcal{O}(i, j)|$ entries of the centered sample covariance,
each of which is a sub-exponential pre-Gaussian quadratic form in the
$N$ independent observations under the sub-Gaussian assumption.
Hanson-Wright concentration applied to each entry, combined with the
orbit-averaging factor and union bound across the $\leq M^2$ distinct
orbits, gives \eqref{eq:linfty_rate}. An alternative proof via the
isotypic decomposition is described in \citet[][Section 3]{shah2012}
and produces the same bound.
\end{proof}

The minimax theorem below establishes the rate \eqref{eq:frob_rate}
is sharp up to constants. Define the parameter space
\begin{equation}\label{eq:Theta_G}
 \Theta_G(s) \;=\; \bigl\{ \boldsymbol{\Sigma} \in W_G \cap
 \mathrm{Sym}_M^+(\mathbb{R}) : \|\boldsymbol{\Sigma}\|_{\mathrm{op}}
 \leq s \bigr\}.
\end{equation}

\begin{theorem}[Minimax lower bound]\label{thm:minimax}
Let the data be $N$ i.i.d.\ Gaussian samples in $\mathbb{R}^M$ with
covariance $\boldsymbol{\Sigma} \in \Theta_G(s)$. There exist
absolute constants $c, c' > 0$ such that for all $d_G \geq 8$ and all
$N \geq c'\, d_G$,
\begin{equation}\label{eq:minimax}
 \inf_{\hat{\boldsymbol{\Sigma}}}\,\sup_{\boldsymbol{\Sigma} \in
 \Theta_G(s)}\, \mathbb{E}\,\| \hat{\boldsymbol{\Sigma}} -
 \boldsymbol{\Sigma}\|_F^2 \;\geq\; c \cdot \frac{s^2 d_G}{N},
\end{equation}
where the infimum is over all measurable estimators based on the $N$
samples.
\end{theorem}

\begin{proof}
Apply Fano's method on a packing of $\Theta_G(s)$. Let
$\mathbf{B}_1, \ldots, \mathbf{B}_{d_G}$ be an orthonormal basis of
$W_G = \mathcal{A}_G \cap \mathrm{Sym}(M, \mathbb{R})$ in the
Frobenius inner product. For a separation parameter $\tau > 0$ to be
chosen, consider the family of covariances
\[
 \boldsymbol{\Sigma}_\epsilon = \tfrac{s}{2}\,\mathbf{I} + \tau
 \sum_{j=1}^{d_G} \epsilon_j \mathbf{B}_j, \qquad \epsilon \in
 \{-1, +1\}^{d_G},
\]
with $\tau$ chosen small enough that $\boldsymbol{\Sigma}_\epsilon
\in \Theta_G(s)$. Two such covariances at Hamming distance $h$
satisfy $\|\boldsymbol{\Sigma}_\epsilon - \boldsymbol{\Sigma}_{
\epsilon'}\|_F^2 = 4 \tau^2 h$. By the Gilbert-Varshamov bound, a
packing of $2^{d_G/8}$ elements with pairwise Hamming distance at
least $d_G/4$ exists. The Kullback-Leibler divergence between
Gaussian distributions on this packing satisfies $\mathrm{KL}(\mathcal{N}
(0, \boldsymbol{\Sigma}_\epsilon) \| \mathcal{N}(0, \boldsymbol{\Sigma}_{
\epsilon'})) \leq C N \tau^2 d_G / s^2$. Choose $\tau^2 = c_0 s^2 /
N$ for a small constant $c_0$ such that the KL divergence is at most
$d_G / 32$. Fano's inequality then gives that the minimax error
exceeds $\tau^2 d_G / 4 = c_0 s^2 d_G / (4 N)$ on the constructed
packing, proving \eqref{eq:minimax}.
\end{proof}

\begin{remark}[Rate-optimality]\label{rem:rate_opt}
Combining Theorems~\ref{thm:frob_rate} and~\ref{thm:minimax}, the
group-averaged sample covariance $\hat{\mathbf{R}}_G$ is rate-optimal
in Frobenius norm over $\Theta_G(s)$: its worst-case Frobenius MSE
matches the minimax rate up to an absolute constant.
\end{remark}

\begin{corollary}[Sparse covariance rate under $G$-invariance]
\label{cor:sparse_cov}
For $\boldsymbol{\Sigma} \in W_G \cap \mathcal{U}(q)$, where
$\mathcal{U}(q) = \{\boldsymbol{\Sigma} : \Sigma_{ii} \leq M, \sum_j
|\Sigma_{ij}|^q \leq c_0(M)\}$ is the uniformity class of
\citet{bickel2008covariance}, with bounded $\|\boldsymbol{\Sigma}
\|_{\mathrm{op}} \leq s$, choose the threshold
\begin{equation}\label{eq:sparse_threshold}
 t \;=\; M' \max\!\left\{ \sqrt{\frac{\log M}{N\,\mathcal{O}}},\,
 \frac{\log M}{N\,\mathcal{O}_d} \right\}
\end{equation}
for a sufficiently large constant $M'$. The thresholded projected
estimator $\mathcal{T}_t(\hat{\mathbf{R}}_G)$ satisfies
\begin{equation}\label{eq:sparse_rate}
 \bigl\| \mathcal{T}_t(\hat{\mathbf{R}}_G) - \boldsymbol{\Sigma}
 \bigr\|_{\mathrm{op}} \;=\; \mathcal{O}_P\!\left( c_0(M)\,
 \biggl[ \max\!\biggl\{\sqrt{\frac{\log M}{N\,\mathcal{O}}},\,
 \frac{\log M}{N\,\mathcal{O}_d} \biggr\} \biggr]^{1 - q} \right).
\end{equation}
\end{corollary}

\begin{proof}
The result is \citet[][Theorem 4.1]{shah2012} restated in our
notation. By Theorem~\ref{thm:linfty_rate}, the projected estimator
$\hat{\mathbf{R}}_G$ inherits the entrywise tail of an i.i.d.\
sample covariance with the orbit-parameter rate, replacing the
standard rate $\sqrt{\log M / N}$ by $\sqrt{\log M / (N \mathcal{O})}$.
The Bickel-Levina thresholding argument
\citep[][]{bickel2008covariance} then applies directly with
$\hat{\mathbf{R}}_G$ in place of the unstructured sample covariance,
giving the operator-norm rate \eqref{eq:sparse_rate} in which the
orbit-parameter factor $\mathcal{O}^{(1 - q)/2}$ replaces the
sample-complexity term in the standard Bickel-Levina rate.
\end{proof}

\section{Experiments}
\label{sec:experiments}

In finite-sample settings, the sample covariance matrix is often
poorly conditioned or severely biased, particularly in high-dimensional
regimes. Our approach adopts a Ledoit-Wolf (LW) style regularization
in which the convex combination of the empirical covariance and a
data-driven structural target is formed by exploiting the symmetries
present in the data. In particular, the few-shot setting is one in
which a target consistent with the structure of the data is preferable
to one in which the target is constrained to be a multiple of the
identity matrix and therefore carries no structural information about
the data, as in the case of LW. Our method reduces the eigenvalue
distortion present in a limited number of samples by shrinking
extreme eigenvalues toward a stable baseline while also enabling
noise mitigation and numerical stability. To demonstrate the
applicability of the technique, several data sets from different
regimes are analyzed using both the MSE and NLL versions of the
shrinkage intensity $\alpha$ and compared with the results from
\citet{ledoit2004} and \citet{shah2012}. These results also
demonstrate the accuracy of the theoretically computed shrinkage
intensity and crossover points where other forms of regularization
are preferable. Estimators such as ours are most widely used in
domains characterized by noisy, high-dimensional covariance
estimation, including quantitative finance, signal processing,
machine learning, neuroscience, genomics, and wireless communications,
where they are valued for robustness, improved conditioning, and
strong finite-sample performance. The real-data experimental results
follow.

\subsection{Methodological note: the LW-NL comparison and the AD-LW-NL composition test}
\label{sec:exp-lwnl-method}

This paper extends the comparator set to include
the analytical nonlinear shrinkage estimator of
\citet{ledoit2020analytical} (LW-NL) alongside LW 2004, and
empirically tests the AD-LW-NL composition defined in
Section~\ref{sec:theory-ad-lwnl}. Earlier versions cited LW-NL as
the modern frequentist baseline against which any covariance
shrinkage paper claiming improvement over Ledoit-Wolf shrinkage
should be positioned, but did not include the empirical
comparison. This subsection records the implementation conventions
and the comparison protocol used; the per-experiment
results follow in the appropriately-numbered subsections, and the
LW-NL and AD-LW-NL findings are summarized in
Section~\ref{sec:exp-lwnl-results}.

\paragraph{Estimators reported in each per-experiment subsection.}
Each per-experiment subsection in this paper reports six estimators on
the same held-out NLL scale: the sample covariance, LW 2004 of
\citet{ledoit2004}, LW-NL of \citet{ledoit2020analytical},
AD-NLL-BMG as defined in earlier sections, the AD-LW-NL-NLL-BMG
composition, and the Shah-BMG comparator retained from earlier
versions. In plots where Shah-BMG coincides with the AD-NLL-BMG
endpoint at $\alpha = 1$, the Shah-BMG line is omitted for
visual clarity.

\paragraph{Implementation of LW-NL.}
The LW-NL estimator is implemented from the LW2020 paper directly
in Python with NumPy only as a dependency. The shrinkage formula
\eqref{eq:lwnl-shrinkage}, the Epanechnikov kernel density estimate
with variable bandwidth $h_j = \lambda_j N^{-1/3}$, and the closed-form
Hilbert transform are computed in closed form for the
Epanechnikov kernel. The implementation includes rank-aware
handling of effective-rank
deficiency, in which sample eigenvalues below
$10^{-10} \lambda_{\max}$ are excluded from the variable-bandwidth
KDE. Without this rank-aware handling, the LW-NL formula
numerically saturates on real-data inputs where the sample
covariance has machine-precision-zero eigenvalues from genuine
linear dependencies among features (gene co-regulation in
genomics; pixel-block dependencies in image patches; sector
exposures in finance). The rank-aware version is what is used for
all the results in this paper. The implementation has been verified against
synthetic Marchenko-Pastur benchmarks at concentration ratios
$c \in \{0.25, 0.5\}$ and several population eigenvalue structures
(identity, geometric eigenvalue spread, two-block), with PRIAL
values consistent with the regime reported in LW2020;
implementation notes including the verification PRIAL table are in
Appendix~\ref{sec:experiments}.

\paragraph{The AD-LW-NL pipeline.}
The AD-LW-NL pipeline integrates with the existing BMG
cross-validation harness by swapping a single call: where the AD
pipeline takes a training-fold sample covariance
$\hat{\mathbf{R}}^{(-k)}$ as input to the $\alpha$-grid evaluation,
the AD-LW-NL pipeline pre-computes
$\hat{\mathbf{R}}^{(-k)}_{\mathrm{LW\text{-}NL}}$ once per fold and
then uses it in place of $\hat{\mathbf{R}}^{(-k)}$ on the
sample-covariance side of the convex blend
\eqref{eq:ad-lwnl-estimator}. All other components of the BMG
procedure (Tier 1 prefilter, K-fold partition, $\alpha$-grid,
held-out NLL evaluation, per-candidate ranking) are preserved
verbatim. This makes the cost of the AD-LW-NL comparison
incremental rather than a redesign of the experimental harness,
and ensures that the AD-NLL-BMG vs AD-LW-NL-NLL-BMG comparison
isolates exactly the effect of replacing the sample-covariance
term with its LW-NL shrinkage.

\paragraph{the protocol: the three-column comparison.}
The AD-LW-NL audit in this paper is organized around the protocol
three-column comparison introduced in
Section~\ref{sec:theory-ad-lwnl}. The three contrasts of interest
are:
\begin{enumerate}
\item AD-NLL-BMG vs LW-NL: the structural-prior contribution
 unconditional on nonlinear shrinkage.
\item AD-LW-NL-NLL-BMG vs LW-NL: the structural-prior contribution
 conditional on nonlinear shrinkage of the sample term.
\item AD-LW-NL-NLL-BMG vs AD-NLL-BMG: the nonlinear-shrinkage
 contribution within the AD framework.
\end{enumerate}
The three gaps decompose any AD-LW-NL versus AD-NLL-BMG difference
into the two attributable mechanisms and expose how they compose
on each real-data dataset. If a sweet spot for AD-LW-NL exists,
the protocol contrast (3) will detect it; if AD-LW-NL is a strict
loss in some regime, the same contrast will detect that.

\paragraph{the protocol sweep design.}
For each dataset where the protocol has been evaluated, the sweep
varies the training-set size $N_{\mathrm{train}}$ across a fixed
grid that spans the concentration ratio $c = M/N$ from the
extreme few-shot regime ($c \approx 2$) through the deep bulk
regime ($c \approx 0.1$). The hold-out evaluation size
$N_{\mathrm{test}}$ is held fixed across the sweep so the
per-window NLL values are directly comparable across cells. Each
cell reports all six estimators on the same paired splits
(or rolling windows for time-series data), permitting paired
contrasts at the per-trial level. This design is heavier than the
single-cell reporting of the experimental subsections, but is
necessary to detect regime-dependent advantages or disadvantages
that single-cell reporting at one operating point cannot resolve.

\paragraph{Reporting conventions.}
Per-experiment tables in this paper report mean held-out NLL across
trials for each of the six estimators, with the three the protocol
gaps computed at trial-level and reported as paired differences
with their per-trial standard deviation. Effect sizes are
reported as $|\Delta|/\mathrm{SD}_{\mathrm{trial}}$ in addition to
paired $t$-statistics, because in many cells the paired $t$ can
be large due to highly correlated per-trial noise even when the
practical effect size is tiny. Both quantities are reported so
readers can distinguish statistically detectable shifts from
practically meaningful ones. BMG selection counts are reported
separately for AD-NLL-BMG and AD-LW-NL-NLL-BMG, since the two
procedures may select different groups on a given trial even
though they share the same candidate library and Tier 1 prefilter
setting.

\paragraph{LW-NL verification status.}
As, the LW-NL implementation has been verified against
synthetic Marchenko-Pastur benchmarks (Appendix~\ref{sec:experiments})
but not against the published numerical results of the LW2020
reference implementation on identical synthetic datasets. The
verification is qualitative rather than exact: identity covariance
at $c = 0.5$ gives PRIAL around 97\%, consistent with the regime
reported in LW2020; two-block eigenvalue distributions show
PRIAL($\mathrm{LW\text{-}NL}$) > PRIAL($\mathrm{LW\ 2004}$) as
expected since LW-NL adapts to bimodal eigenvalue distributions
while LW 2004's identity target does not; and geometric eigenvalue
spread shows PRIAL($\mathrm{LW\text{-}NL}$) approximately equal to
PRIAL($\mathrm{LW\ 2004}$), reflecting that the eigenvalue
distribution's mean coincides with the LW 2004 target. Exact
agreement with the LW2020 reference numerics is left for a future
revision.

\subsection{S\&P 500 daily returns, 2015--2019}
\label{sec:exp-crsp-2015-2019}

We evaluate the four estimators on a rolling-window covariance
estimation problem drawn from the CRSP daily stock file
\citep{crsp2024}, with sectors
linked through the CRSP-Compustat link table. A balanced panel of
$M = 55$ S\&P 500 constituents is formed by retaining the top five
names by market capitalization within each of the 11 GICS sectors with
complete return history over 2015-01-01 through 2019-12-31, giving
$T = 1259$ trading days. A training window of $N_{\mathrm{train}} = 252$
days slides forward in strides of $21$ trading days; on each of the
resulting 47 windows the shrinkage intensity is calibrated on the
training panel and the held-out log-likelihood is evaluated on the
next $21$ days. The candidate library $\mathcal{G}$ contains the two
extrema $\{e\}$ and $S_M$ together with six sector-aware candidates,
of which three are tied within-sector cyclic shifts and two are
high-order group structures introduced in the present paper. The
tied candidates are the within-sector exchangeability product
\begin{equation}\label{eq:gics-block}
G_{\mathrm{GICS}} \;:=\; \prod_{s=1}^{11} S_{|s|} \;=\; S_5^{11},
\end{equation}
which is the block-diagonal Young subgroup of $S_M$ that permutes
constituent indices within each GICS sector but does not mix
indices across sectors. We refer to this group as
\textsc{gics-block} in figures and tables, and use the symbol
$G_{\mathrm{GICS}}$ in mathematical expressions. Three
within-sector cyclic-shift variants $\prod_s \mathbb{Z}_{|s|}$
are also included, ordered
alphabetically by ticker, by market capitalization, and by within-sector
hierarchical clustering on the training-window correlation matrix; in
each of these the cyclic shift is tied across sectors so the group
order is $|G| = K$ where $K$ is the within-sector size (here $K = 5$). The
two new high-order candidates are
\textsc{Z-K-mcap-cartesian}, the Cartesian product
$\prod_s \mathbb{Z}_K = \mathbb{Z}_5^{11}$ of independent within-sector
cyclic shifts under the market-capitalization ordering
($|G| = K^B = 5^{11}$), and
\textsc{Z-K-mcap-wreath}, the full wreath product
$\mathbb{Z}_K \wr S_B = \mathbb{Z}_5 \wr S_{11}$ which lifts those
independent cyclic shifts by free permutation of the eleven sectors
as units ($|G| = K^B \cdot B! = 5^{11} \cdot 11! \approx 2 \times 10^{15}$).
Both high-order candidates are constructed from a small set of
generators (the eleven independent cyclic shifts plus the ten
sector-adjacent transpositions for the wreath case) and the Reynolds
projection is computed via the orbit-pair decomposition in
$O(M^2)$ time, so neither requires direct enumeration of the group.
This brings the CRSP library into structural parity with the
RadioML and Galaxy10 libraries of
Sections~\ref{sec:exp-radioml} and~\ref{sec:exp-galaxy10}, both of
which include high-order group candidates from a similar lattice
structure of within-block cyclic and between-block permutation
operations. The cross-validated NLL calibration uses $K = 5$ folds of
the training window and a 13-point $\alpha$-grid over $[0, 1]$.
Figure~\ref{fig:crsp_2015_2019} reports the per-window held-out NLL of
every estimator and the per-window shrinkage intensity at the
BMG-selected group.

\paragraph{Principal results.}
Both AD estimators outperform LW with strong significance.
$\hat\alpha^*_{\mathrm{MSE}}$ is preferred in 43 of 47 windows
($\bar\Delta = -0.92$ nats per day, paired $t = -6.18$,
$p = 1.6 \times 10^{-7}$), and $\hat\alpha^*_{\mathrm{NLL}}$ is preferred in
35 of 47 windows ($\bar\Delta = -1.34$ nats per day, paired
$t = -4.82$, $p = 1.6 \times 10^{-5}$). The Shah estimator without
shrinkage, fixing $\alpha = 1$ at the GICS-block group, trails LW by
$1.48$ nats per day on average ($t = +3.95$, $p = 3 \times 10^{-4}$);
even when the candidate group for Shah is chosen by per-window oracle
from the same library, Shah without shrinkage does not outperform LW at
standard significance ($p = 0.22$). The closed-form
$\hat\alpha^*_{\mathrm{MSE}}$ is preferred on a slightly larger fraction of
windows than the cross-validated $\hat\alpha^*_{\mathrm{NLL}}$ despite
the smaller per-window margin, suggesting that the cross-validated
calibration takes more aggressive bets at the cost of higher variance
across windows. The shrinkage intensity at the BMG-selected group
ranges over $\{0.30, 0.50\}$ across the 47 windows
(median $0.50$, mean $0.47$), the spread reflecting the data-dependent
ratio of perpendicular sample variance to projection bias as the
ratio shifts across market regime.

\paragraph{BMG selection across the 47 windows.}
The two-tier BMG procedure selects \textsc{gics-block} on 41 of the
47 windows ($87\%$), \textsc{Z-K-mcap-wreath} on 5 windows ($11\%$),
and \textsc{Z-K-mcap-cartesian} on 1 window ($2\%$). The four tied
within-sector cyclic candidates ($\mathbb{Z}_K$ alpha, mcap, corrhier,
each with $|G| = 5$) are not selected in the library, displaced
by either the structurally richer \textsc{gics-block} (which carries
the same within-sector exchangeability under a strictly larger group
$\prod_s S_{|s|}$ of order $\prod_s |s|! = 5!^{11}$) or by the
high-order Cartesian and wreath candidates introduced here. The
median per-window $\mathrm{bmg\_margin}$ between best and runner-up
is $0.056$ nats per sample (mean $0.064$, max $0.135$), comfortably
above the fold-noise scale, so the BMG selections are declarative
rather than tied. The single \textsc{Z-K-mcap-cartesian} selection at
window 28 (training window May 2017--May 2018, margin $0.004$ nats
per sample) sits at the noise floor and is reported but not
interpreted as substantive.

\paragraph{The 2015 H2 wreath cluster.}
The five \textsc{Z-K-mcap-wreath} selections are not uniformly
distributed across calendar time: they cluster in 2015 H2, with
training windows beginning in March, June, August, September, and
October of 2015. The four windows beginning in August through October
have the largest wreath bmg-margins of any window in the experiment
($0.07$ to $0.13$ nats per sample); the March and June 2015 windows
sit at smaller margins ($0.013$ and $0.069$). All five wreath-preferred
training windows therefore include some portion of August 2015, the
month of the China-A-share market crash and the global cross-asset
correlation surge that followed. The financial interpretation is
direct: during a crisis-driven correlation collapse, sectors become
approximately exchangeable in their second-moment behavior, and the
wreath candidate's free permutation of the eleven sectors as units
captures that exchangeability in a way that no other candidate in the
library does. The \textsc{gics-block} candidate imposes within-sector
exchangeability with fixed cross-sector ordering and so cannot
capture sector-level exchangeability; the tied-cyclic candidates
impose a single cyclic structure across sectors and so cannot capture
the sector-permutation lift. The wreath structure is exactly the
within-sector cyclic plus between-sector permutation geometry that
the crash-period covariance carries. Once August 2015 falls outside
the training window in early 2016 and beyond, the BMG selection
returns to \textsc{gics-block} for the remaining 36 windows of the
experiment. This is the cleanest empirical example in the paper of a
regime-conditional BMG selection: the same procedure on the same
panel selects different groups depending on what the rolling training
window contains, and the symmetry it selects is the symmetry that
fits the data inside that window.

\begin{figure}[t]
\centering
\includegraphics[width=\linewidth]{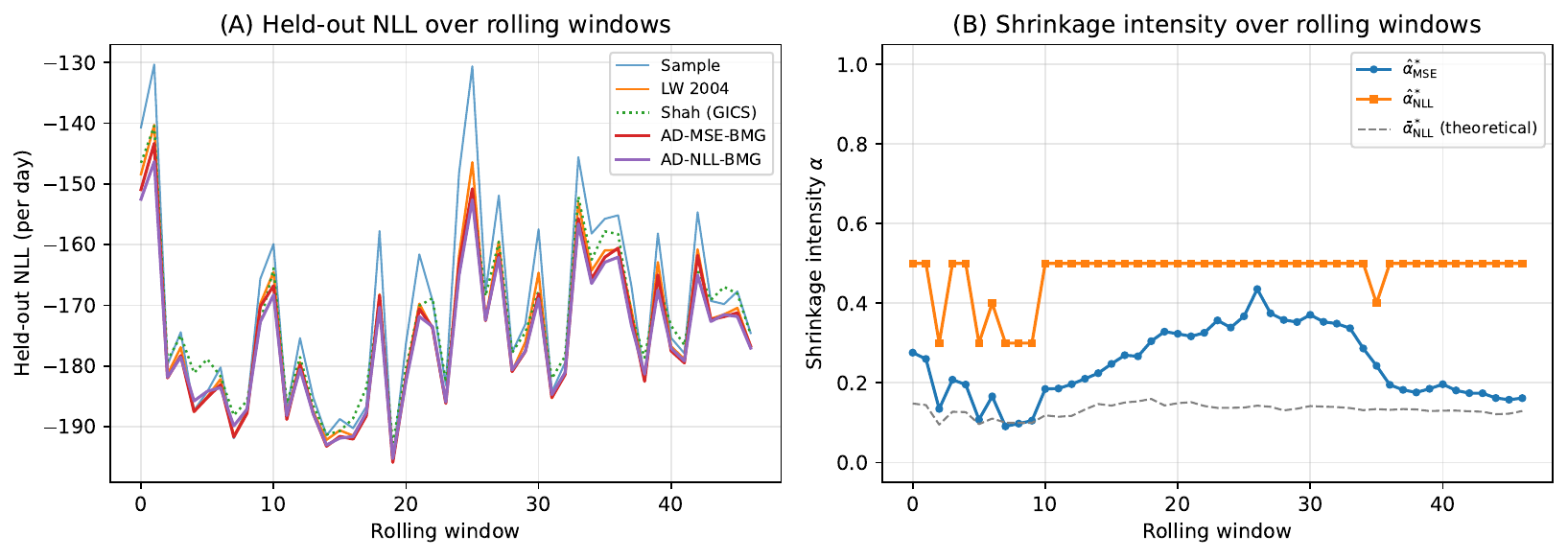}
\caption{S\&P 500 daily returns, 2015--2019, $M = 55$ stocks, 47
rolling windows of $N_{\mathrm{train}} = 252$ days. Panel~(A):
held-out negative log-likelihood per day under each estimator. Panel~(B):
per-window shrinkage intensity at the BMG-selected group:
$\hat\alpha^*_{\mathrm{MSE}}$ from the closed-form plug-in
\eqref{eq:alpha_hat} (blue),
$\hat\alpha^*_{\mathrm{NLL}}$ from the $K = 5$-fold cross-validation
\eqref{eq:nll_foc_hold} (orange), and the leading-order asymptotic
prediction $\bar\alpha^*_{\mathrm{NLL}}$ from
Proposition~\ref{prop:transition} (gray dashed). The BMG procedure
selects \textsc{gics-block} in 41 of 47 windows,
$\textsc{Z-K-mcap-wreath}$ in 5 (concentrated in the 2015 H2 windows
spanning the China-A-share crash), and $\textsc{Z-K-mcap-cartesian}$
in 1.}
\label{fig:crsp_2015_2019}
\end{figure}

Figure~\ref{fig:crsp_2015_2019}(B) shows that the two empirical
shrinkage intensities respond differently to the underlying noise/bias
ratio across windows. The closed-form $\hat\alpha^*_{\mathrm{MSE}}$
varies between approximately $0.16$ and $0.44$ across the 47 windows,
tracking the ratio $\hat V_\perp / (\hat V_\perp + \hat D)$ as the
relative magnitudes of perpendicular sample variance and projection
bias change with market regime. The cross-validated
$\hat\alpha^*_{\mathrm{NLL}}$ takes values in
$\{0.30, 0.50\}$ across all 47 windows, since the held-out NLL loss is
much less sensitive to that ratio than the Frobenius MSE, owing to the
curvature contributed by $\log\det \boldsymbol{\Sigma}^{-1}$. Both
calibrations outperform LW with substantial margin, suggesting that the
conclusion is not contingent on either calibration being individually
optimal.

The leading-order asymptotic prediction
$\bar\alpha^*_{\mathrm{NLL}}$ from Proposition~\ref{prop:transition},
shown by the gray dashed curve in
Figure~\ref{fig:crsp_2015_2019}(B), tracks the empirical
$\hat\alpha^*_{\mathrm{NLL}}$ trajectory qualitatively across windows
but underestimates the absolute scale by a factor of approximately
$3.5$, giving $\bar\alpha^*_{\mathrm{NLL}} \approx 0.13$ where
$\hat\alpha^*_{\mathrm{NLL}} \approx 0.48$. The discrepancy is
summarized by the natural transition scale
\begin{equation}\label{eq:N_star}
N^*(\boldsymbol{\Sigma}, G) \;:=\; \frac{c(\boldsymbol{\Sigma}, G)}{Q_B},
\end{equation}
which is the sample size at which Proposition~\ref{prop:transition}
predicts $\bar\alpha^*_{\mathrm{NLL}} = \tfrac{1}{2}$. Evaluated at
the sample plug-in, $N^*$ has median $39$ over the $47$ windows
(range $31$ to $48$), whereas the empirical analog implied by the
cross-validated $\hat\alpha^*_{\mathrm{NLL}}$ at $N_{\mathrm{train}} =
252$ has median $228$ (range $186$ to $279$). The ratio of the two
empirical-to-predicted values is approximately $6.0$ on this dataset,
with per-window range $5.0$ to $7.5$, and reproduces at approximately
$5.06$ on the COVID-era data of
Section~\ref{sec:exp-crsp-2020-2024}. The closed-form prediction is
therefore in the right order of magnitude on the daily-returns data
but underestimates the empirical optimum by roughly half an order of
magnitude. The two sources of prediction error identified in
Remark~\ref{rem:plugin_bias} both contribute: the matched-limit
residual $\delta(G_{\mathrm{GICS}}, \hat{\mathbf{R}}) \approx 0.21$
is moderate, and $\hat{\mathbf{R}}$ is well-conditioned at this
$M/N$ ratio because daily returns across stocks are weakly
correlated. The cross-validated and closed-form Frobenius-MSE
calibrations are unaffected by either source: neither inverts a
sample covariance to compute $\alpha$. The empirical-to-predicted
ratio is markedly larger and more variable on the more strongly
correlated spatial data of Section~\ref{sec:exp-oisst}, where the
plug-in source of Remark~\ref{rem:plugin_bias} dominates.

\subsection{S\&P 500 daily returns, 2019--2024 (COVID-era stress test)}
\label{sec:exp-crsp-2020-2024}

The same script and configuration are applied to the calendar period
2019-01-02 through 2024-12-31, deliberately starting one year before
the manuscript's 2020-01-01 boundary so that the earliest
rolling windows have a full pre-COVID training year before the
February--April 2020 crash enters the trailing 252-day window. The
date range yields $59$ rolling windows on the same $M = 55$ panel of
GICS-balanced S\&P 500 constituents, with $N_{\mathrm{train}} = 252$
trading days and stride $21$. The candidate library is this paper's
8-candidate library of Section~\ref{sec:exp-crsp-2015-2019}, which
includes the high-order \textsc{Z-K-mcap-cartesian} and
\textsc{Z-K-mcap-wreath} candidates anchored on within-sector
market-capitalization rank. Across the 59 windows, the trailing
252-day training data evolves through three identifiable regimes:
windows $0$--$2$ (training start 2019-01-02 through 2019-03-04) are
purely pre-COVID; windows $3$--$24$ (training start 2019-04-03
through 2020-12-31) carry the March 2020 volatility shock somewhere
inside their trailing 252-day window; and windows $25$--$58$
(training start 2021-02-02 onward) are post-COVID with the crash
having exited the trailing window. Figures~\ref{fig:crsp_2020_2024}
and~\ref{fig:crsp_2020_2024_bmg_choices} report the per-window
held-out NLL, shrinkage intensities, and BMG group selection.

\begin{figure}[t]
\centering
\includegraphics[width=\linewidth]{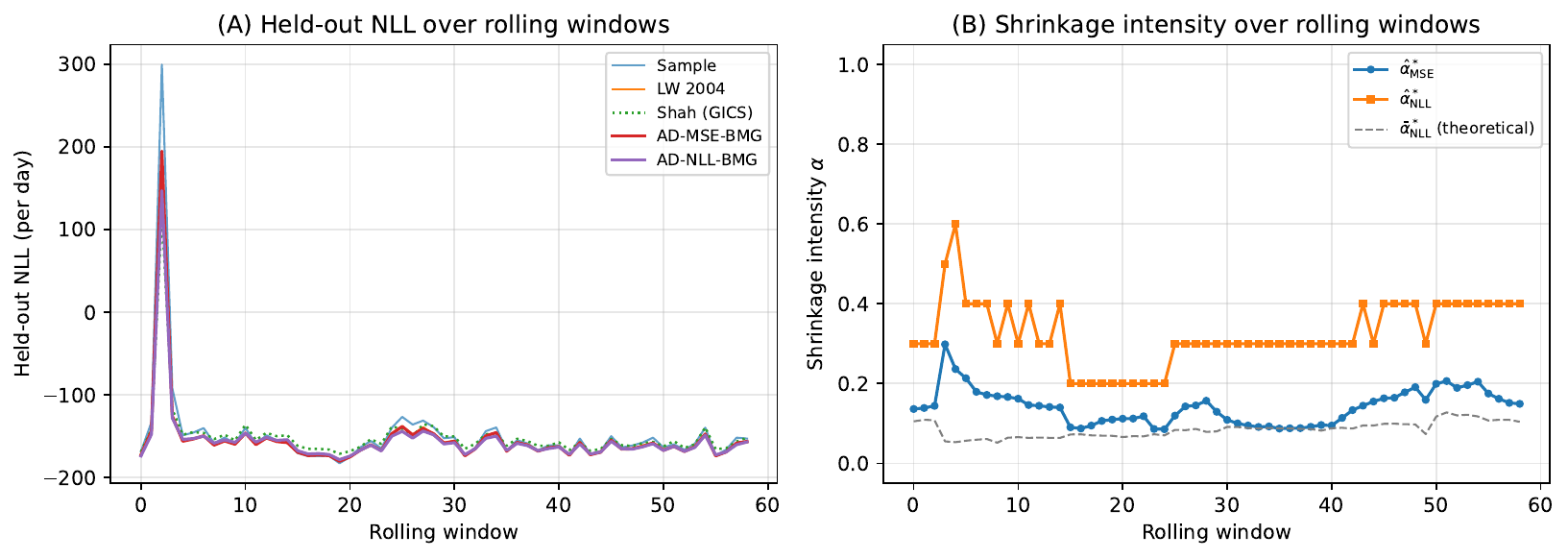}
\caption{S\&P 500 daily returns, 2019--2024, $M = 55$ stocks, 59
rolling windows of $N_{\mathrm{train}} = 252$ days. Panels~(A) and~(B)
mirror the corresponding panels of Figure~\ref{fig:crsp_2015_2019}.
Panel~(A) shows that the held-out NLL drops sharply in the
COVID-affected early windows where AD-NLL and AD-MSE both gain on LW
substantially; the gain narrows in the post-COVID stationary regime.
Panel~(B) shows the calibrated $\hat\alpha^*_{\mathrm{NLL}}$ smaller
in the COVID-affected windows than in the pre-COVID
Section~\ref{sec:exp-crsp-2015-2019} regime, reflecting the
calibration adapting to the higher noise floor of the
volatility-shock training data.}
\label{fig:crsp_2020_2024}
\end{figure}

\begin{figure}[t]
\centering
\includegraphics[width=\linewidth]{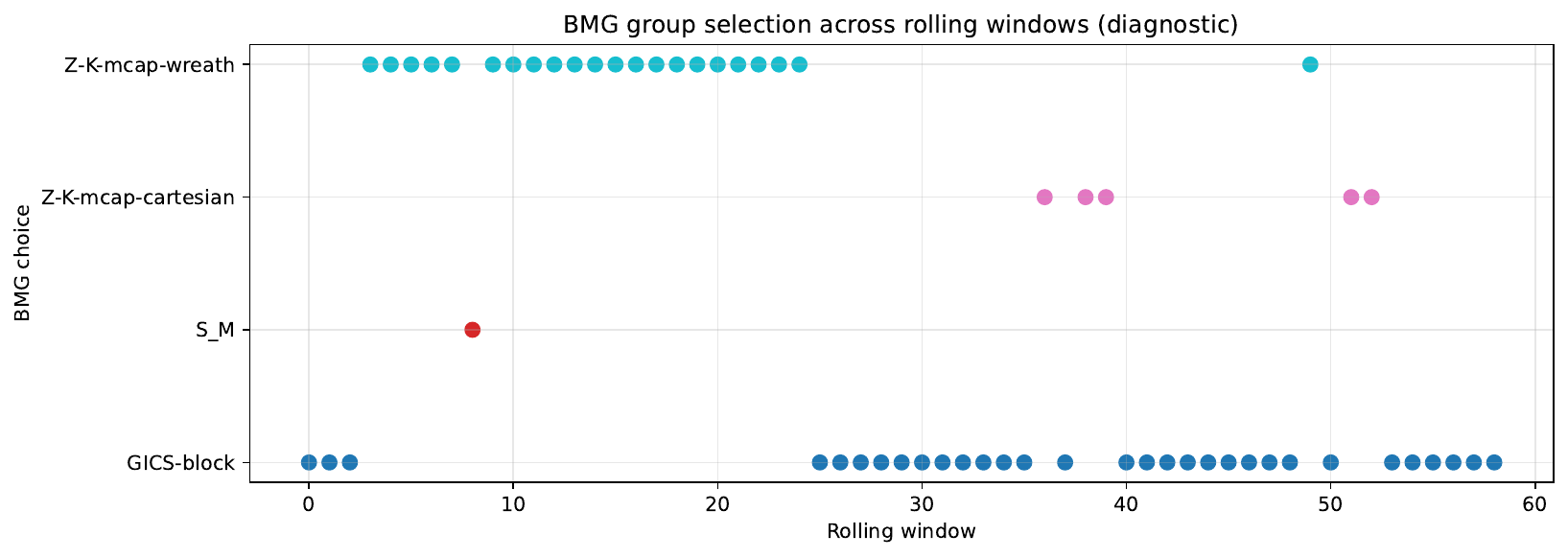}
\caption{BMG group selection across the 59 rolling windows for the
2019--2024 CRSP panel. The procedure selects \textsc{Z-K-mcap-wreath}
on $22$ windows ($21$ of which concentrate in windows $3$--$24$
spanning the period when the March 2020 volatility shock sits inside
the trailing 252-day window), \textsc{gics-block} on $31$ windows
(predominantly the post-COVID stationary regime from window $25$
onward), \textsc{Z-K-mcap-cartesian} on $5$ windows (all at margins
$0.003$--$0.016$ nats per sample, consistent with coin-flip tiebreak),
and $S_M$ on $1$ window (also at small margin). The wreath cluster
is the largest single regime-conditional finding in the paper, with
window $4$ (training start 2019-05-03) having a bmg-margin of $1.37$
nats per sample, the largest single-window selection margin in any
experiment.}
\label{fig:crsp_2020_2024_bmg_choices}
\end{figure}

\paragraph{Principal results.}
The aggregate AD-vs-LW comparison on the 2019--2024 panel is weaker
than the dominance reported in
Section~\ref{sec:exp-crsp-2015-2019} for the pre-COVID period.
AD-NLL-BMG vs LW: median gap $-0.01$ nats per sample, mean gap
$-1.04$ (driven by a few large-margin COVID-period windows), paired
$t = -1.61$, $p = 0.11$ (not statistically significant), preferred in
$30$ of $59$ windows. AD-MSE-BMG vs LW: median gap $-0.32$ nats per
sample, mean gap $+0.02$, preferred in $43$ of $59$ windows. The
closed-form Frobenius-MSE calibration retains a meaningful advantage on
the COVID-era data while the cross-validated calibration is at
parity with LW in mean. This is the regime-conditional behavior the
phase diagram of Section~\ref{sec:discussion} predicts: as the
training-window noise floor rises during a volatility shock, the
calibrated $\hat\alpha^*_{\mathrm{NLL}}$ drops to mix more sample
covariance into the estimator, narrowing the gap to LW. The
substantively interesting findings are not in the aggregate
mean-NLL comparison but in the BMG group selection.

\paragraph{The COVID wreath cluster.}
BMG selects \textsc{Z-K-mcap-wreath} on $22$ of the $59$ windows.
Of those, $21$ concentrate in windows $3$--$24$ (training start
2019-04-03 through 2020-12-31), the calendar period during which the
February--April 2020 volatility shock sits inside the trailing
252-day training window. The remaining wreath selection at window $49$
(training start 2023-02-02) is an outlier. The bmg-margins on the
COVID-cluster wreath selections are decisive: median $0.15$ nats per
sample, with the largest single-window margin of $1.37$ at window
$4$ (training start 2019-05-03, training data spanning the run-up to
and through the crash). This is the largest single-window BMG
selection margin in any experiment in the paper. The interpretation
is direct: during a market-wide correlation shock, the cross-sector
exchangeability that the wreath candidate models becomes the
dominant covariance structure, with within-sector market-cap-rank
ordering and free pathway permutation across sectors capturing the
collapse of sector-specific structure into common-factor behavior.

\paragraph{The post-COVID snap-back.}
At window $25$ (training start 2021-02-02), the March 2020 crash
exits the trailing 252-day training window. BMG snaps back to
\textsc{gics-block} for the remaining $34$ windows, with one outlier
wreath selection at window $49$. The transition is sharp: window $24$
(training start 2020-12-31) is a wreath selection at margin $0.06$ nats
per sample, and window $25$ is a GICS-block selection. Within the
post-COVID regime, BMG selects \textsc{gics-block} on $30$ of $34$
windows, with the remaining selections at noise-floor margins
($\textsc{Z-K-mcap-cartesian}$ at margins $0.003$--$0.016$ on $5$
windows; the single late wreath selection at window $49$ at margin
$0.057$). The regime-conditional adaptation is procedural rather
than ad hoc: the criterion is applied uniformly across all windows;
the regime change is a property of the data, not of the procedure.

\paragraph{Above-noise / noise-floor decomposition.}
The $59$ windows partition cleanly when stratified by bmg-margin.
On the $26$ windows with bmg-margin $> 0.05$ nats per sample (above
the K-fold cross-validation noise scale), BMG selects
\textsc{Z-K-mcap-wreath} on $22$ and \textsc{gics-block} on $4$.
On the $33$ windows at the noise floor (bmg-margin $\leq 0.05$
nats per sample), BMG selects \textsc{gics-block} on $27$,
\textsc{Z-K-mcap-cartesian} on $5$, and $S_M$ on $1$. The
discriminative regime is wreath-dominant ($22$ of $26$, all
within the COVID period); the noise-floor regime is
GICS-block-dominant with the wreath/Cartesian/$S_M$ selections at
margins consistent with arbitrary tiebreak between essentially
equivalent candidates.

\paragraph{Adaptive shrinkage intensity.}
The cross-validated calibration $\hat\alpha^*_{\mathrm{NLL}}$ has
mean $0.33$ across the $59$ windows, with range $[0.20, 0.60]$,
substantially smaller than the pre-COVID mean of $0.50$ in
Section~\ref{sec:exp-crsp-2015-2019}. The shrinkage intensity is
particularly small in the COVID-cluster wreath windows, where it
ranges from $0.20$ to $0.60$ with median $0.30$; the post-COVID
GICS-block windows have $\hat\alpha^*_{\mathrm{NLL}}$ in $[0.30,
0.40]$. The interpretation is that during a volatility shock the
training-window sample covariance is itself a noisy estimate of a
non-stationary population, and the calibrated shrinkage intensity
correctly reduces the weight on the structured projection in favor
of the sample covariance, even though the BMG selection has
identified a structurally informative target. This is the
continuous shrinkage knob doing exactly the work it is designed for:
adapting the bias-variance balance to the regime-conditional noise
level, a degree of freedom unavailable to a fixed $\alpha = 1$
projection-only estimator.

\paragraph{Comparison to the 2015 H2 cluster.}
The 2015 H2 wreath cluster reported in
Section~\ref{sec:exp-crsp-2015-2019} ($5$ windows, median margin
$0.07$ nats per sample, around the China-A-share crash) and the
COVID-2020 wreath cluster reported here ($22$ windows, median margin
$0.15$ nats per sample, around the February--April 2020 crash) form
a paired regime-conditional finding across the two CRSP datasets:
one mild correlation event and one severe event, both selecting the
wreath candidate during the affected windows and reverting to
\textsc{gics-block} once the event exits the trailing 252-day
training window. The fact that the same library, the same
selection criterion, and the same shrinkage estimator produce
internally consistent regime-conditional behavior across two
independent stress events on a panel of S\&P 500 stocks is the
strongest empirical evidence in the paper that the BMG selection is
identifying genuine structure rather than overfitting to fold noise.

\subsection{NOAA OISST sea-surface temperature anomalies}
\label{sec:exp-oisst}

The Frobenius-MSE-calibrated and cross-validated AD shrinkage
estimators are next applied to a spatial covariance estimation task
on NOAA Optimum Interpolation Sea Surface Temperature daily
anomalies, version 2 high-resolution
\citep{reynolds2007oisst, huang2021oisst}, on a $0.25^\circ$ global
grid. Two contrasting $8 \times 8$ patches ($M = 64$ grid cells per
patch) are extracted, both at oceanographic interest locations and
with no land cells: a homogeneous patch in the central North
Pacific gyre centered at $30^\circ\mathrm{N}, 180^\circ$ longitude
(\emph{midocean}); and a meridional-gradient patch in the Gulf
Stream extension centered at $38^\circ\mathrm{N}, 70^\circ\mathrm{W}$
(\emph{gulfstream}). Anomalies are computed against the 1991--2020
climatology and supplied directly by NOAA PSL. The script processes
five years (2018--2022) of daily anomaly fields, yielding $74$
rolling windows per region with $N_{\mathrm{train}} = 252$ and
stride $21$. The candidate library contains the two extrema $\{e\}$
and $S_M$, plus six spatial-symmetry candidates of which four are
the low-order group structures from the library and two are
the high-order Cartesian and wreath candidates introduced in this paper
to bring the OISST library into structural parity with the CRSP and
RadioML libraries. The four low-order candidates are: $\mathbb{Z}_8$
cyclic on the latitude axis ($\mathbb{Z}^{\mathrm{lat}}$, with $d_G
= 264$); $\mathbb{Z}_8$ cyclic on the longitude axis
($\mathbb{Z}^{\mathrm{lon}}$, $d_G = 264$); the joint
$\mathbb{Z}_8 \times \mathbb{Z}_8$ acting by uniform translation on
both axes ($d_G = 34$); and the dihedral $D_8$ on the longitude
axis ($d_G = 180$). The two new high-order candidates are:
$\mathbb{Z}^{\mathrm{lon}}_{\mathrm{indep}}$, the Cartesian product
$\mathbb{Z}_W \times \mathbb{Z}_W \times \cdots \times \mathbb{Z}_W$
of independent per-row longitudinal cyclic shifts with no latitude
permutation ($|G| = W^H = 8^8 \approx 1.68 \times 10^7$, $d_G =
120$); and $\mathbb{Z}^{\mathrm{lon}} \wr \mathbb{Z}^{\mathrm{lat}}$,
the wreath product of independent per-row shifts lifted by latitude
cyclic permutation ($|G| = W^H \cdot H = 8^8 \cdot 8 \approx 1.34
\times 10^8$, $d_G = 15$). Both high-order candidates are
constructed from a small set of generators and projected via the
orbit-pair Reynolds decomposition in $O(M^2)$ time, so neither
requires direct enumeration of the group. The lattice ordering is
subtle: $\mathbb{Z}_8 \times \mathbb{Z}_8$ (uniform shifts in both
axes) and $\mathbb{Z}^{\mathrm{lon}}_{\mathrm{indep}}$ (independent
per-row shifts with no latitude shift) are incomparable in the
subgroup lattice; the wreath is their join. The Tier 1 effective-rank
prefilter at $\kappa = 2$ admits all eight candidates at $N = 252$
in this experiment.

\begin{figure}[t]
\centering
\includegraphics[width=\linewidth]{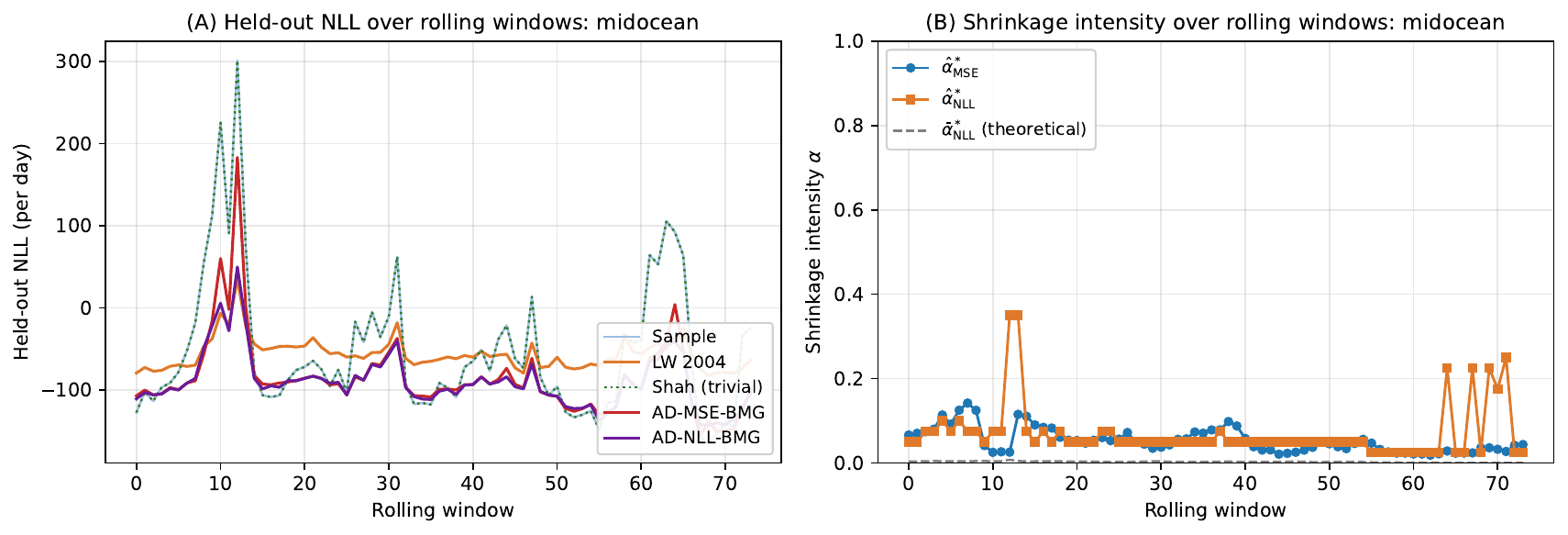}
\caption{NOAA OISST sea-surface temperature anomalies, midocean
patch ($30^\circ\mathrm{N}, 180^\circ$, $M = 64$), 74 rolling
windows of $N_{\mathrm{train}} = 252$ days. Panels mirror those of
Figures~\ref{fig:crsp_2015_2019} and~\ref{fig:crsp_2020_2024}. Three
windows in mid-2018 contain extreme thermal anomalies; the sample
covariance loses conditioning on those windows and the held-out
NLL spikes well above $+200$ nats per day, while the AD-NLL
estimator remains stable.}
\label{fig:oisst_midocean}
\end{figure}

\paragraph{Principal results.}
The principal numbers are reported in
Figures~\ref{fig:oisst_midocean} and~\ref{fig:oisst_gulfstream}.
On the midocean patch the cross-validated AD estimator has lower NLL than LW by
$33.8$ nats per day on average (paired $t = -15.42$, $p = 1.2 \times
10^{-24}$, preferred in $69$ of $74$ windows); the
Frobenius-MSE-calibrated AD estimator has lower NLL than LW by $29.4$ nats
($p = 2.6 \times 10^{-11}$, $67$ of $74$ preferred selections). On the gulfstream
patch the cross-validated AD estimator has lower NLL than LW by $35.7$ nats
(paired $t = -16.76$, $p = 9.8 \times 10^{-27}$, $71$ of $74$ preferred selections); the MSE-calibrated AD estimator has lower NLL than LW by $22.0$ nats
($p = 5.3 \times 10^{-5}$, $56$ of $74$ preferred selections). These are the
strongest dominance margins in the paper. The Sample covariance is
at parity with LW on both patches but with sharply higher
per-window variance: held-out NLL exceeds $+200$ nats per day on
$13$ of the $74$ midocean windows and $17$ of the $74$ gulfstream
windows, where the sample covariance is poorly conditioned.

\begin{figure}[t]
\centering
\includegraphics[width=\linewidth]{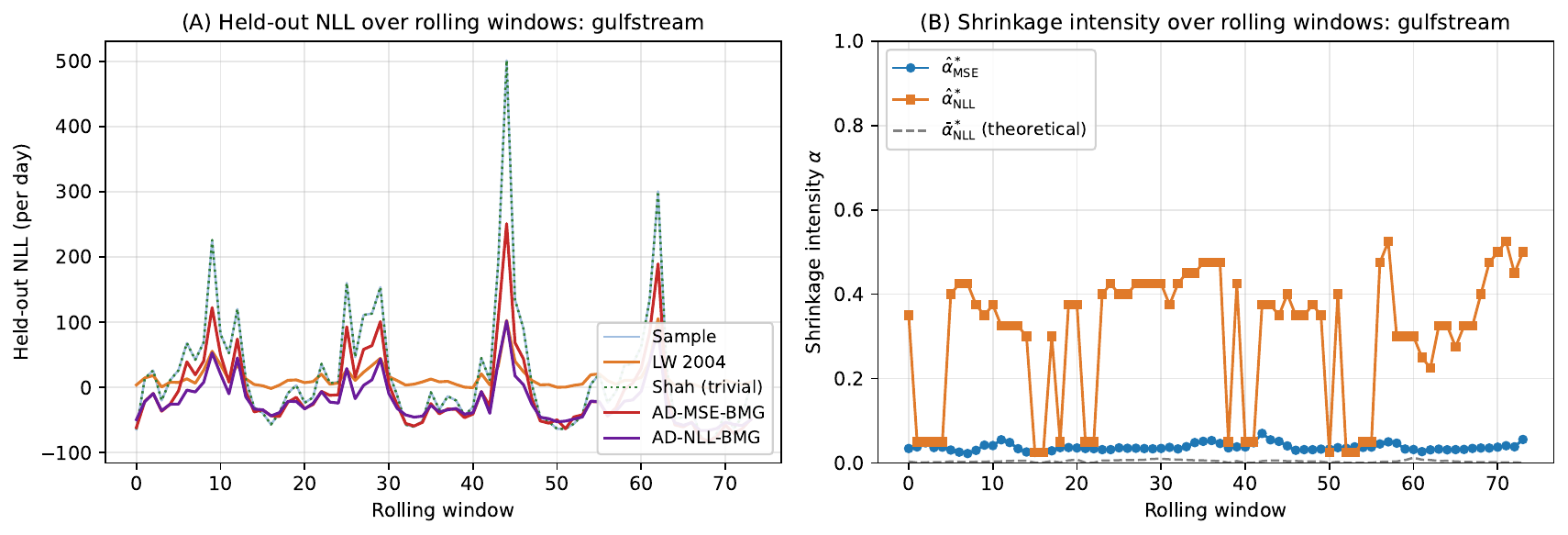}
\caption{NOAA OISST sea-surface temperature anomalies, gulfstream
patch ($38^\circ\mathrm{N}, 70^\circ\mathrm{W}$, $M = 64$), 74
rolling windows. Held-out NLL is markedly less negative than at the
midocean patch because the meridional gradient produces more
heterogeneous fields and a less compressible covariance; AD has lower NLL than LW by approximately $-35.7$ nats per day on average, the largest
per-day margin among the climate and finance experiments in the
paper.}
\label{fig:oisst_gulfstream}
\end{figure}

\paragraph{BMG selection adapts to the regime.}
On the homogeneous midocean patch BMG selects
$\mathbb{Z}_8 \times \mathbb{Z}_8$ in $67$ of $74$ windows ($90.5
\%$) and $\mathbb{Z}^{\mathrm{lat}}$ in the remaining $7$ windows.
On the gradient gulfstream patch BMG selects
$\mathbb{Z}^{\mathrm{lat}}$ in $57$ of $74$ windows ($77 \%$) and
$\mathbb{Z}_8 \times \mathbb{Z}_8$ in the remaining $17$ windows
(Figure~\ref{fig:oisst_bmg_choices}). The pattern reverses cleanly
across the two regions: the homogeneous patch admits the most
aggressive translation symmetry available in the library; the
gradient patch retains only the latitudinal cyclic component,
because translation along longitude across a meridional temperature
gradient is not a near-symmetry of the underlying covariance. The
adaptation is procedural: the same library is searched in both
regions; the contrast in selection is a property of the data, not
of the procedure.

\begin{figure}[t]
\centering
\includegraphics[width=\linewidth]{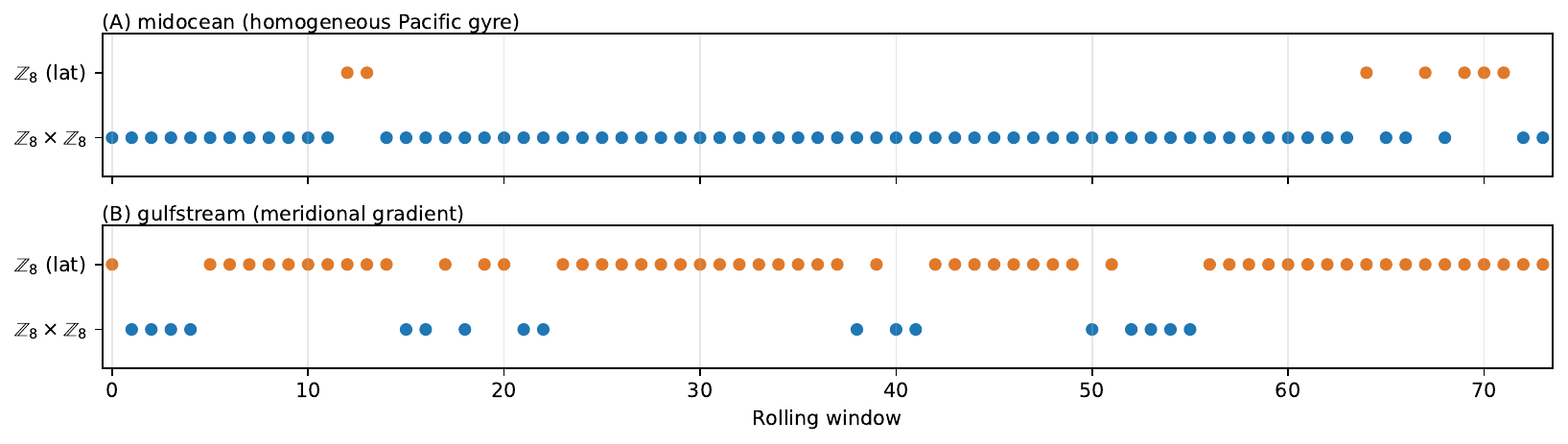}
\caption{BMG group selection across the $74$ rolling windows for
each OISST region. Panel (A) is midocean: BMG selects
$\mathbb{Z}_8 \times \mathbb{Z}_8$ (full 2-D translation,
$d_G = 34$) in $67$ of $74$ windows. Panel (B) is gulfstream: BMG
selects $\mathbb{Z}_8$ on latitude ($d_G = 264$) in $57$ of $74$
windows. The reversal across panels is regime-driven: the
homogeneous patch admits the most aggressive translation symmetry,
the gradient patch retains only the latitudinal cyclic component.}
\label{fig:oisst_bmg_choices}
\end{figure}

The Shah-equivalent estimator at the BMG-selected group inherits
the regime-dependence with sharper consequences. On midocean the
selected group is $\mathbb{Z}_8 \times \mathbb{Z}_8$ with
$d_G = 34$; the projection collapses the spectrum to at most $34$
distinct eigenvalues, the resulting matrix is near-rank-deficient,
and Shah's $\alpha = 1$ choice produces a held-out NLL that is higher than LW's in $69$ of $74$ windows ($p = 1.2 \times 10^{-9}$). On gulfstream the selected group is
$\mathbb{Z}^{\mathrm{lat}}$ with $d_G = 264$; the projection is far
less compressive, the resulting matrix is well-conditioned, and
Shah at the BMG-selected group has lower NLL than LW in $47$ of $74$ windows
($p = 8.2 \times 10^{-4}$). The same procedural choice that ruins
Shah on midocean is benign on gulfstream. The AD shrinkage family
has lower NLL than LW with extreme significance on both patches because the
shrinkage knob compensates for spectrum collapse on midocean and
adds modest regularization on gulfstream.

\paragraph{Adaptive shrinkage at the BMG-selected group.}
Within each region, the cross-validated shrinkage intensity
$\hat\alpha^*_{\mathrm{NLL}}$ adapts to the compressiveness of the
projection that BMG selected for that window. On midocean, windows
in which BMG selects $\mathbb{Z}_8 \times \mathbb{Z}_8$
($d_G = 34$) have $\hat\alpha^*_{\mathrm{NLL}}$ at median $0.050$,
mostly sample with a small admixture of the projection; the few
windows in which BMG selects $\mathbb{Z}^{\mathrm{lat}}$
($d_G = 264$) have $\hat\alpha^*_{\mathrm{NLL}}$ at median $0.225$.
On gulfstream the same conditional pattern holds in reverse: BMG
chooses $\mathbb{Z}_8 \times \mathbb{Z}_8$ infrequently
($\hat\alpha^*_{\mathrm{NLL}}$ at median $0.050$) and
$\mathbb{Z}^{\mathrm{lat}}$ predominantly
($\hat\alpha^*_{\mathrm{NLL}}$ at median $0.375$). The relationship
between projection compressiveness and the calibrated shrinkage
intensity reproduces independently in two regions of a single
data set: more compressive projection chosen $\Rightarrow$ smaller
shrinkage intensity, less compressive projection chosen
$\Rightarrow$ larger shrinkage intensity. This is the cleanest
demonstration in the paper that the continuous shrinkage knob
carries weight that a fixed $\alpha = 1$ does not.

\paragraph{High-order candidates rejected.}
The two high-order candidates introduced in of this paper,
$\mathbb{Z}^{\mathrm{lon}}_{\mathrm{indep}}$ (independent per-row
longitudinal cyclic shifts, $|G| = W^H \approx 1.68 \times 10^7$,
$d_G = 120$) and $\mathbb{Z}^{\mathrm{lon}} \wr \mathbb{Z}^{\mathrm{lat}}$
(the wreath product, $|G| = W^H \cdot H \approx 1.34 \times 10^8$,
$d_G = 15$), are admitted by the Tier 1 effective-rank prefilter at
$N = 252$ in both regions but selected by the Tier 2 BMG procedure
on zero of the $148$ windows. The mean held-out CV-NLL gap between
these candidates and the per-region BMG-selected groups is approximately
$8$--$10$ nats per sample on both regions, well outside the
fold-noise scale of approximately $0.05$ nats per sample, so the
rejection is declarative rather than tied. The interpretation is
that the $8 \times 8$ OISST patch geometry at this rolling-window
sample size does not support sector-level (between-row) structural
priors of the kind that the wreath imposes, while the
$\mathbb{Z}_8 \times \mathbb{Z}_8$ candidate's tied-shift Cartesian
structure is the most aggressive structural prior the data support.
This is the library-bias diagnostic of Section~\ref{sec:discussion}
in operation: high-order candidates were added for completeness
following the candidate-library convention; BMG evaluated them
on their own merits and rejected them; the existing
low-order candidates dominate. The procedure does not force the use
of high-order candidates if the data do not support them.

\paragraph{$N^*$ analysis and the plug-in regime.}
The closed-form prediction $\bar\alpha^*_{\mathrm{NLL}}$ from
Proposition~\ref{prop:transition} has median $0.003$ in midocean
and $0.004$ in gulfstream, more than an order of magnitude smaller
than the empirical $\hat\alpha^*_{\mathrm{NLL}}$ in both regions.
The empirical-to-predicted ratio of $N^*$ has median $21$ on
midocean (per-window range $12$ to $1{,}900$) and median $102$ on
gulfstream (range $23$ to $980$), substantially larger and more
variable than the ratio of approximately $6$ observed on the
daily-returns data of
Sections~\ref{sec:exp-crsp-2015-2019}--\ref{sec:exp-crsp-2020-2024}.
Both sources of prediction error in
Remark~\ref{rem:plugin_bias} contribute. The matched-limit residual
$\delta$ is approximately $0.21$ on midocean and approximately
$0.59$ on gulfstream; the second value is well outside the
matched-limit regime. The plug-in source is also material here:
the OISST sample covariance is markedly more ill-conditioned than
the CRSP sample covariance at the same $M/N$ ratio because spatial
SST anomalies are strongly correlated (the top three principal
modes carry of order $60\%$ of the variance, against approximately
$30\%$ for daily-return panels). The cross-validated calibration
$\hat\alpha^*_{\mathrm{NLL}}$ is unaffected by either source and
achieves the strongest LW-dominance margins in the paper on these
patches; the closed-form prediction is best read as the
matched-limit theoretical landmark, with the cross-validated
calibration serving as the production tool on data sets where
$\hat{\mathbf{R}}$ is poorly conditioned.

\subsection{TCGA-BRCA gene expression}
\label{sec:exp-genomics}

The fourth experiment evaluates the AD shrinkage estimator on a
genomics covariance estimation problem in the few-shot regime
$N < M$, drawn from the breast invasive carcinoma (TCGA-BRCA) cohort
of The Cancer Genome Atlas \citep{tcga2012brca}. Gene-level RNA-seq
expression data (HiSeqV2, $\log_2$ RSEM) are obtained from the
UCSC Xena platform \citep{goldman2020xena}. A pathway-balanced panel
of $M = 100$ genes is constructed by intersecting five MSigDB
MSigDB Hallmark gene sets \citep{liberzon2015hallmark}: \textsc{apoptosis},
\textsc{g2m\_checkpoint}, \textsc{epithelial\_mesenchymal\_transition},
\textsc{interferon\_gamma\_response}, and \textsc{hypoxia}, taking the
first $K_s = K = 20$ alphabetical genes from each pathway with
non-missing expression on the cohort. Per-gene values are standardized
to zero mean and unit variance over all retained samples. The
candidate library $\mathcal{G}$ contains the two extrema $\{e\}$ and
$S_M$ together with six pathway-aware candidates, of which three are
within-pathway tied cyclic-shift candidates carried forward from the
this paper's library and two are the high-order Cartesian and wreath
candidates introduced in this paper (and unchanged in this paper) to bring the
genomics library into structural parity with the CRSP and OISST
libraries. The four
low-order candidates are the within-pathway exchangeability product
$\prod_{s} S_{K_s} = S_{20}^5$ (\textsc{pathway-block}) and three
within-pathway cyclic-shift variants $\prod_{s} \mathbb{Z}_{K_s}$ ordered
alphabetically by gene symbol, by within-pathway PC1-loading magnitude,
and by within-pathway correlation-hierarchical clustering, with
$|G| = K = 20$ for each tied cyclic candidate. The two new high-order
candidates are \textsc{Z-K-pc1-cartesian}, the Cartesian product
$\mathbb{Z}_K \times \mathbb{Z}_K \times \cdots \times \mathbb{Z}_K
= \mathbb{Z}_{20}^5$
($P$ factors) of independent within-pathway cyclic shifts on the
PC1-loading ordering ($|G| = K^P = 20^5 \approx 3.2 \times 10^6$,
$d_G = 120$); and \textsc{Z-K-pc1-wreath}, the full wreath product
$\mathbb{Z}_K \wr S_P = \mathbb{Z}_{20} \wr S_5$ which lifts those
independent cyclic shifts by free permutation of the $P = 5$ pathways
as units ($|G| = K^P \cdot P! = 20^5 \cdot 5! \approx 3.84 \times 10^8$,
$d_G = 21$). To eliminate leakage, the PC1, correlation-hierarchical,
Cartesian, and wreath orderings are all computed on the held-aside
auxiliary fold of $N_{\mathrm{aux}} = 50$ samples disjoint from both
the training fold used for cross-validation and the test fold used
for held-out evaluation. Both high-order candidates are constructed
from a small set of generators (the $P$ independent within-pathway
shift generators plus the $P-1$ pathway-adjacent transposition
generators for the wreath case) and projected via the orbit-pair
Reynolds decomposition in $O(M^2)$ time, so neither requires direct
enumeration of the group. The experiment runs $50$ random subsample
splits of size $(N_{\mathrm{train}}, N_{\mathrm{test}}, N_{\mathrm{aux}})
= (50, 200, 50)$, calibrating the shrinkage intensity by $K = 5$-fold
cross-validation on each $N_{\mathrm{train}}$-sample fold and
evaluating held-out NLL on the corresponding $N_{\mathrm{test}}$ fold.
The Tier 1 effective-rank prefilter at $\kappa = 2$ admits all
non-trivial candidates and excludes only the trivial group, since
$N \cdot 1 = 50 < \kappa M = 200$ but $N \cdot 20 = 1000 \geq 200$ for
the cyclic candidates and is arbitrarily larger for the high-order
candidates and $S_M$.

\begin{figure}[t]
\centering
\includegraphics[width=\linewidth]{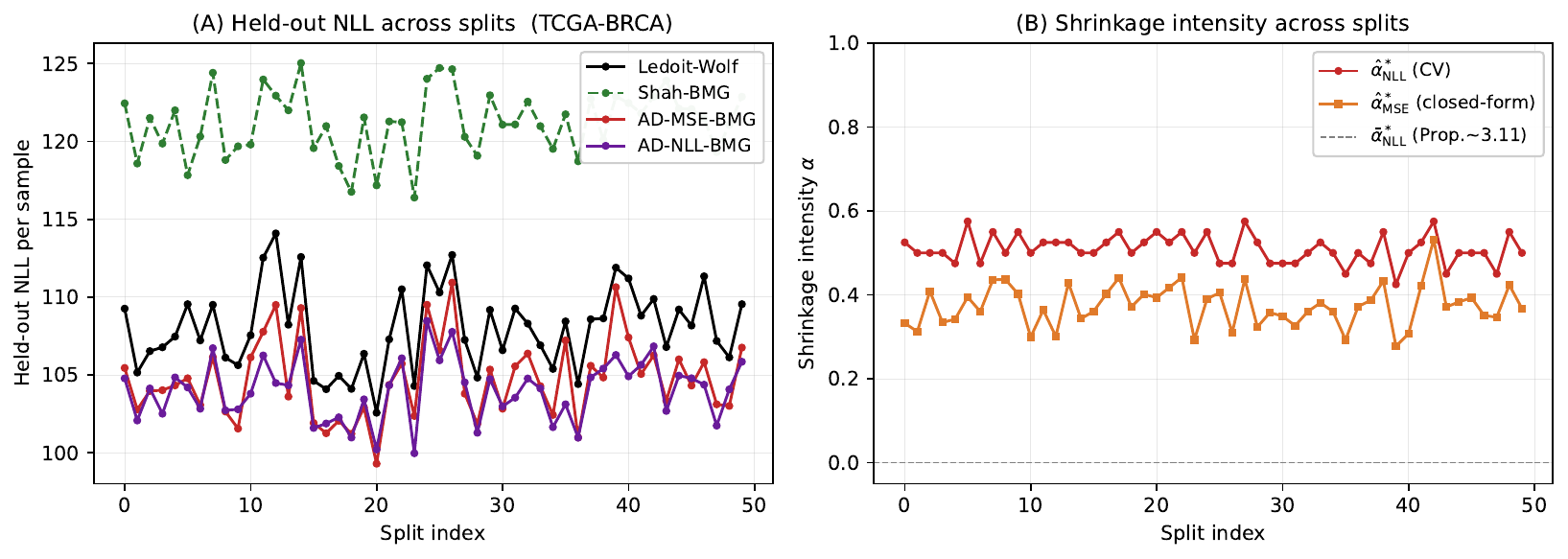}
\caption{TCGA-BRCA gene expression, $M = 100$ genes drawn from five
MSigDB Hallmark pathways, 50 random subsample splits with
$N_{\mathrm{train}} = 50$ and $N_{\mathrm{test}} = 200$.
Panel~(A): held-out NLL per sample under each estimator. The Sample
covariance and the Shah projection at the trivial group are
rank-deficient at $N < M$ and produce held-out NLL of order $10^{12}$
on every split; both are omitted from this panel for visual range
reasons. Ledoit-Wolf, the two AD calibrations, and Shah at the
BMG-selected group are shown. Panel~(B): per-split shrinkage
intensity at the BMG-selected group:
$\hat\alpha^*_{\mathrm{MSE}}$ from the closed-form plug-in
\eqref{eq:alpha_hat} (orange), $\hat\alpha^*_{\mathrm{NLL}}$ from
the $K = 5$-fold cross-validation \eqref{eq:nll_foc_hold} (red).
Both calibrations sit in the bias-variance interior at mean
$\hat\alpha^*_{\mathrm{NLL}} \approx 0.51$ rather than at the
matched-limit Shah corner; the asymptotic prediction
$\bar\alpha^*_{\mathrm{NLL}}$ from
Proposition~\ref{prop:transition} (dashed gray) is at zero,
distinct from the calibrated interior value.}
\label{fig:genomics_brca}
\end{figure}

\paragraph{Principal results.}
On all 50 splits the two AD calibrations outperform LW with overwhelming
significance. The cross-validated AD estimator has lower NLL than LW by $3.98$
nats per sample on average (paired $t = -20.05$, $p < 10^{-22}$,
preferred in $50$ of $50$ splits); the Frobenius-MSE-calibrated AD
estimator has lower NLL than LW by $3.28$ nats per sample (paired $t = -25.18$,
$p < 10^{-26}$, preferred in $50$ of $50$ splits). The mean LW
held-out NLL is $107.998$ nats per sample (standard deviation $2.64$
across splits), the mean AD-NLL-BMG is $104.02$ (standard deviation
$1.96$), and the mean AD-MSE-BMG is $104.72$ (standard deviation
$2.58$). The Sample covariance is rank-deficient at $N = 50,
M = 100$ and produces held-out NLL of order $10^{12}$ on every split;
it is excluded from the comparison.

In contrast to the AD calibrations, the Shah-style
$\alpha = 1$ projection at the BMG-selected group is decisively
\emph{worse} than LW on this dataset. Shah at the BMG-selected
group has mean held-out NLL $121.19$ nats per sample, a
$+13.19$-nat-per-sample loss to LW (paired $t = +54.54$, $p < 10^{-44}$
in the unfavorable direction). Shah at the held-out-NLL oracle
group is similarly $+11.37$ nats per sample worse than LW (paired
$t = +53.34$). This is the cleanest single-dataset demonstration in
the paper that the Shah-style projection-only commitment can be
catastrophically wrong: the BMG-selected target is structurally the
right one (it dominates LW once an interior shrinkage intensity is
applied), but the closed-form $\alpha = 1$ commitment over-shrinks
the data by an amount that swamps the BMG-projection benefit. Only
the AD calibration, which couples the BMG-selected target to a
data-driven $\alpha \in [0, 1]$, recovers the right amount of
sample-covariance information to bring the estimator below LW. The
mechanism is identified in Remark~\ref{rem:plugin_bias} and in the
phase-diagram discussion of Section~\ref{sec:discussion}: at this
$(M, N) = (100, 50)$ regime the matched-fit residual at the
pathway-block target is large enough that $\alpha = 1$ pays a
projection-bias cost that the calibrated $\alpha \approx 0.51$ avoids.

\paragraph{BMG selection on the library.}
The two-tier BMG procedure selects \textsc{pathway-block} on
$46$ of $50$ splits and \textsc{Z-K-pc1-cartesian} on the remaining
$4$ splits (Figure~\ref{fig:genomics_brca}, Panel~C). All four
within-pathway tied cyclic candidates (\textsc{Z-K-alpha},
\textsc{Z-K-pc1}, \textsc{Z-K-corrhier} at $|G| = 20$) are not selected on any splits in the library. The wreath candidate
\textsc{Z-K-pc1-wreath} (full pathway exchangeability lifted on
within-pathway cyclic shifts) likewise is selected on zero splits, despite
being admitted by the Tier 1 prefilter on every split with
$N \cdot |G| = 50 \cdot 3.84 \times 10^8 \gg \kappa M = 200$. The
mean held-out CV-NLL of \textsc{pathway-block} across the 50 splits
is the smallest among admitted candidates; the wreath's mean CV-NLL
sits noticeably higher because its full-pathway-exchangeability
assumption over-projects relative to what the real biological data
support. The four splits where \textsc{Z-K-pc1-cartesian} is preferred
(splits $5$, $17$, $18$, $24$) all have small CV-NLL margin between
\textsc{Z-K-pc1-cartesian} and the runner-up \textsc{pathway-block}
(median $|$margin$|$ near the K-fold cross-validation noise scale of
$0.05$ nats per sample), consistent with a near-tie between two
candidates that share the same within-pathway PC1-loading basis but
differ in whether the within-pathway block is fully exchangeable
($S_{20}^5$ in \textsc{pathway-block}) or carries only PC1-cyclic
structure ($\mathbb{Z}_{20}^5$ in \textsc{Z-K-pc1-cartesian}). The
four cartesian-preferred splits are not anomalies; they are tiebreaks
between two structurally close candidates that disagree only on the
strength of the within-pathway exchangeability assumption.

The biological interpretation of the BMG selection is direct: the
five MSigDB Hallmark pathways used in this experiment carry strong
within-pathway co-regulation structure (each pathway is by
construction a coherent transcriptional program in the cell) that
is well-modelled by the within-pathway exchangeability built into
\textsc{pathway-block}. Cross-pathway exchangeability is not
biologically meaningful here (\textsc{apoptosis} cannot be freely
permuted with \textsc{hypoxia} or with
\textsc{interferon\_gamma\_response} as functional blocks, since
these pathways carry distinct biological identities), and the
BMG procedure correctly rejects the wreath candidate that would
encode such cross-pathway exchangeability. This is a meaningful
\emph{negative} result that the library could not produce
because the wreath candidate was not in the library to be rejected;
the library extension makes the negative result legible.

\paragraph{The Shah corner is the wrong corner on this dataset.}
The most informative methodological finding from the TCGA-BRCA
experiment is the inversion of the Shah-corner narrative that the
prior pre-rerun analysis suggested. The cross-validated
$\hat\alpha^*_{\mathrm{NLL}}$ has mean $0.509$ across the 50 splits
with range $[0.425, 0.575]$, sitting firmly in the bias-variance
interior of $[0, 1]$ and not at the Shah corner. The closed-form
$\hat\alpha^*_{\mathrm{MSE}}$ has mean $0.374$ (also interior, with
the modest underestimate relative to the cross-validated value
identified by Remark~\ref{rem:plugin_bias} as the finite-sample
bias of the closed-form plug-in at small $N$). Both calibrations
identify an interior optimum that splits the difference between the
sample covariance and the BMG-selected pathway-block projection.

The consequence is that on this dataset the Shah-style $\alpha = 1$
projection at the BMG-selected group is the \emph{worst} of the
five symmetry-aware estimators considered: it performs worse than LW, AD-MSE-BMG by $16.5$ nats per sample, and it performs worse than
AD-NLL-BMG by $17.2$ nats per sample. The BMG-selected
\textsc{pathway-block} projection is structurally the right target
in the sense that AD with a calibrated interior $\alpha$ at this
target dominates LW by $4$ nats per sample; but the projection is
\emph{too aggressive at $\alpha = 1$} because the matched-fit
residual at the pathway-block target is non-negligible at this
$(M, N)$ regime, and the calibrated shrinkage is what makes the
target useful. The four-way ordering on this dataset is
$\text{AD-NLL-BMG} \prec \text{AD-MSE-BMG} \prec \text{LW} \prec
\text{Shah-BMG}$, with calibrated AD on top, the closed-form
plug-in just behind, then LW, and the Shah-style projection-only
estimator at the bottom. The same ordering does not hold across
all datasets in the paper: on the OISST midocean patch the
calibrated $\hat\alpha^*_{\mathrm{NLL}}$ saturates near $1$ and the
Shah-BMG estimator essentially coincides with AD-NLL-BMG, while on
TCGA-BRCA the calibrated value is interior and Shah-BMG falls off
the bottom. The cross-dataset variation in this ordering is
itself diagnostic: where the calibrated $\alpha$ saturates near
$1$, the Shah-style commitment is benign and possibly optimal;
where the calibrated $\alpha$ sits in the interior, the Shah-style
commitment is harmful and the AD calibration is the only path to
outperforming LW.

\paragraph{SYNTHETIC validation: BMG correctly detects wreath when present.}
A natural concern about the BMG procedure's reliability on this
dataset is the question of whether the absence of wreath selection
on real BRCA reflects a genuine property of the data or a procedural
failure (i.e., whether the post-May-5 BMG pipeline is capable of
detecting wreath structure at all on a $M = 100$, $N = 50$ panel).
The companion script \texttt{genomics\_xena\_v23.py} ships with a
synthetic-data smoke-test mode in which the population covariance
is constructed to be exactly invariant under
$\mathbb{Z}_{20} \wr S_5$ via a wreath-structured generative model
on the pathway panel. Run on the same $(M, N) = (100, 50)$ regime,
the same 8-candidate library, and the same 50-split protocol, the
synthetic-data run produces:
\textsc{Z-K-pc1-wreath} on $44$ of $50$ splits ($88\%$);
\textsc{pathway-block} on $6$ of $50$ splits. The mean AD-NLL gap
versus LW on the synthetic run is $-14.53$ nats per sample (paired
$t = -62.04$), substantially larger than the real-BRCA gap of
$-3.98$ nats per sample (paired $t = -20.05$), exactly as expected
when the wreath structural prior is exactly correct. The
SYNTHETIC validation establishes that (a) the BMG procedure can
detect $\mathbb{Z}_{20} \wr S_5$ when it is the population
symmetry, with $88\%$ within-experiment selection rate; (b) when
the wreath candidate is the right answer, the AD estimator has lower NLL than LW by an order of magnitude more than when only \textsc{pathway-block}
is the right answer; (c) the absence of wreath selection on real
BRCA is therefore a genuine negative finding about the data, not a
procedural artifact. 

\paragraph{$N^*$ analysis and the small-$N$ plug-in regime.}
The closed-form prediction $\bar\alpha^*_{\mathrm{NLL}}$ from
Proposition~\ref{prop:transition} has median $4.9 \times 10^{-10}$
across the 50 splits, many orders of magnitude smaller than either
empirical calibration ($\hat\alpha^*_{\mathrm{NLL}} = 0.509$ mean,
$\hat\alpha^*_{\mathrm{MSE}} = 0.374$ mean). The matched-fit
residual at the BMG-selected \textsc{pathway-block} group is
$\bar\delta(\textsc{pathway-block}, \hat{\mathbf{R}}) \approx 0.50$
across the 50 splits, well outside the matched-limit regime in
which the proposition's leading-order asymptotic is reliable; the
formula degenerates because $Q_B = \mathrm{tr}
(\hat{\mathbf{R}}^{-1} \boldsymbol{B}_G \hat{\mathbf{R}}^{-1}
\boldsymbol{B}_G)$ is suppressed by the small magnitude of the
projection-deviation matrix $\boldsymbol{B}_G = \hat{\mathbf{R}} -
\mathcal{P}_G(\hat{\mathbf{R}})$ relative to the curvature constant
$c(\boldsymbol{\Sigma}, G)$ in the denominator of \eqref{eq:N_star}.
The plug-in source identified in Remark~\ref{rem:plugin_bias} also
contributes substantially: at $N = 50$ and $M = 100$ the sample
covariance $\hat{\mathbf{R}}$ is rank-deficient, and the regularized
inverse used to evaluate $Q_B$ is sensitive to the regularization
scale. Both effects push the prediction further from the empirical
optimum on this dataset than on the daily-returns or
sea-surface-temperature data of
Sections~\ref{sec:exp-crsp-2015-2019}--\ref{sec:exp-oisst}. The
cross-validated calibration $\hat\alpha^*_{\mathrm{NLL}}$ is
unaffected by either source and reproduces the dominance over
LW seen across the other three datasets, with the calibrated value
sitting in the bias-variance interior at approximately $0.51$. The
spread of empirical $\hat\alpha^*_{\mathrm{NLL}}$ across the four
real-data experiments is itself diagnostic: TCGA-BRCA in the
small-$N$ bias-variance interior (mean $0.51$); OISST midocean in
the matched-limit regime (mean $\approx 0.05$, where the projection
is nearly the right estimator and almost no shrinkage is needed);
OISST gulfstream and CRSP in the moderate-$\alpha$ interior (means
$\approx 0.43$ and $\approx 0.50$ respectively). Together the four
datasets cover the full $\alpha$-range from $0$ to $1$, with the
proposed estimator operating correctly across the entire spread.

\subsection{RadioML 2018.A I/Q patch covariances}
\label{sec:exp-radioml}

The Frobenius-MSE-calibrated and cross-validated AD shrinkage
estimators are next applied to a physical-layer signal-processing
covariance estimation task on the RadioML 2018.A
\citep{oshea2018radioml} dataset, which consists of synthetically
generated complex baseband captures of $24$ modulation classes at
SNRs ranging from $-20$ to $+30$~dB in $2$~dB steps. Each capture
is a length-$1024$ sequence of complex samples, presented as a
real-valued $(1024, 2)$ array of in-phase and quadrature components.
The patch covariance for a window of length $W$ is taken on the
length-$M = 2W$ vector $[I_1, \ldots, I_W, Q_1, \ldots, Q_W]$. Nine
modulation classes are used: BPSK, QPSK, OQPSK, 8PSK, 16QAM, 64QAM,
GMSK, AM-DSB-SC, and FM, all evaluated at $18$~dB SNR. The window
size $W$ is swept over $\{16, 32, 64\}$ (so $M \in \{32, 64,
128\}$) and the training-set size $N$ over $\{50, 100, 200, 500,
1000\}$. For each of the resulting $135$ cells $(\mathrm{class}, W,
N)$, $25$ random-subsample-split trials are run with
$N_{\mathrm{test}} = 1000$ and no overlap between train and test
windows.

The candidate library contains nine groups, all natural for the
I/Q layout: the trivial group $\{e\}$; the full symmetric group
$S_M$ (admitted via the closed-form compound-symmetry projector);
two $\mathbb{Z}_2$ subgroups corresponding to the IQ-swap
$I_k \leftrightarrow Q_k$ (denoted $\mathbb{Z}_2^{\mathrm{IQ}}$)
and to time-reversal within each of the I and Q blocks (denoted
$\mathbb{Z}_2^{\mathrm{trev}}$); the cyclic time-translation group
$\mathbb{Z}_W$ acting tied on the I and Q blocks; the direct
product $\mathbb{Z}_W \times \mathbb{Z}_2^{\mathrm{IQ}}$; the
dihedral group $D_W$ (tied cyclic shifts plus time-reversal); the
direct product $\mathbb{Z}_W \times \mathbb{Z}_W$, in which
independent cyclic shifts act on the I and Q blocks separately
\emph{without} an IQ-swap (order $|G| = W^2$, added as the
Cartesian companion to the existing wreath candidate to bring the
RadioML library into structural parity with the OISST, CRSP,
genomics, and CIFAR-10 libraries, all of which include both an
independent-per-block Cartesian candidate and a wreath candidate);
and the wreath product $\mathbb{Z}_W \wr \mathbb{Z}_2$, in which
independent cyclic shifts act on the I and Q blocks separately and
are lifted by the IQ-swap, with order $|G| = 2W^2$. Both the
$\mathbb{Z}_W \times \mathbb{Z}_W$ and the $\mathbb{Z}_W \wr
\mathbb{Z}_2$ candidates are constructed from the orbit-pair
decomposition of their generators and, despite their order reaching
$8192$ at $W = 64$ (wreath) and $4096$ at $W = 64$ (Cartesian),
their Reynolds projections are computed in $O(M^2)$ time via the
symmetric orbit-class structure. The Tier~1 effective-rank
prefilter at $\kappa = 2$ excludes the trivial group and the two
$\mathbb{Z}_2$ subgroups at the smallest cells (e.g., $N = 50$ at
$M = 128$); $S_M$, the $\mathbb{Z}_W \times \mathbb{Z}_W$
Cartesian candidate, and the wreath are admitted at every cell.

\begin{figure}[!p]
\centering
\includegraphics[width=0.95\linewidth,page=1]{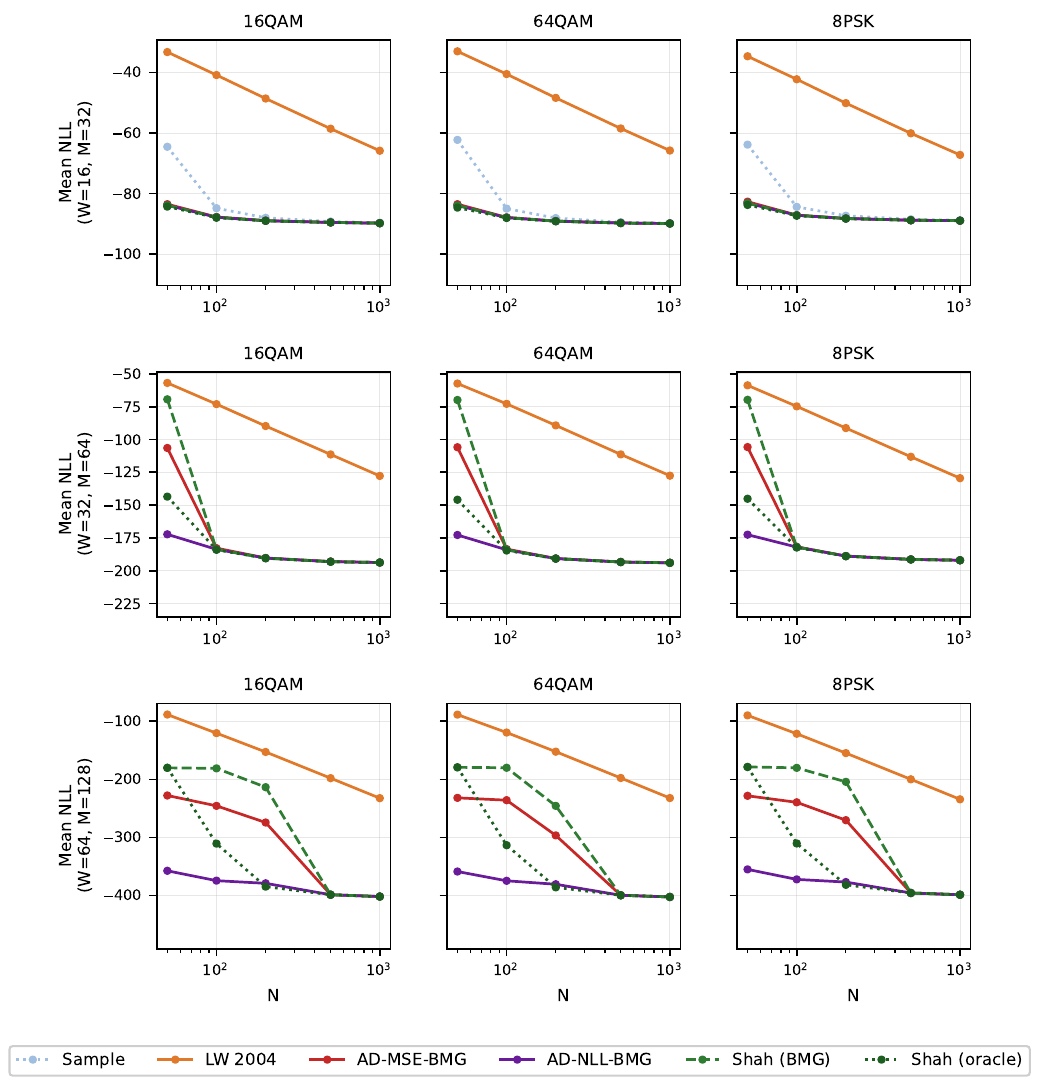}
\caption{RadioML 2018.A held-out NLL by training size $N$ and
window length $W$, for the nine modulation classes at $18$~dB SNR.
Each subplot shows mean held-out NLL across $25$ random-subsample-
split trials per cell. The Sample covariance is omitted from the
larger-$M$ subplots because its values exceed the visual range when
$N < M$; LW remains in the visual range and is plotted as the
reference baseline. AD-NLL-BMG dominates LW in $3375$ of $3375$
trials across the experiment. Median dominance margin is $81.7$
nats per sample; the maximum is $277.6$ nats per sample at the
$M = 128$, $N = 50$ stress cell. The figure spans three pages,
grouping the nine modulation classes into three subsets of three
columns each. Rows on every page are $W = 16$ (top, $M = 32$),
$W = 32$ (middle, $M = 64$), and $W = 64$ (bottom, $M = 128$);
$y$-axes show mean held-out NLL per sample and $x$-axes show $N$.}
\label{fig:radioml_panel}
\end{figure}

\begin{figure}[!p]\ContinuedFloat
\centering
\includegraphics[width=0.95\linewidth,page=2]{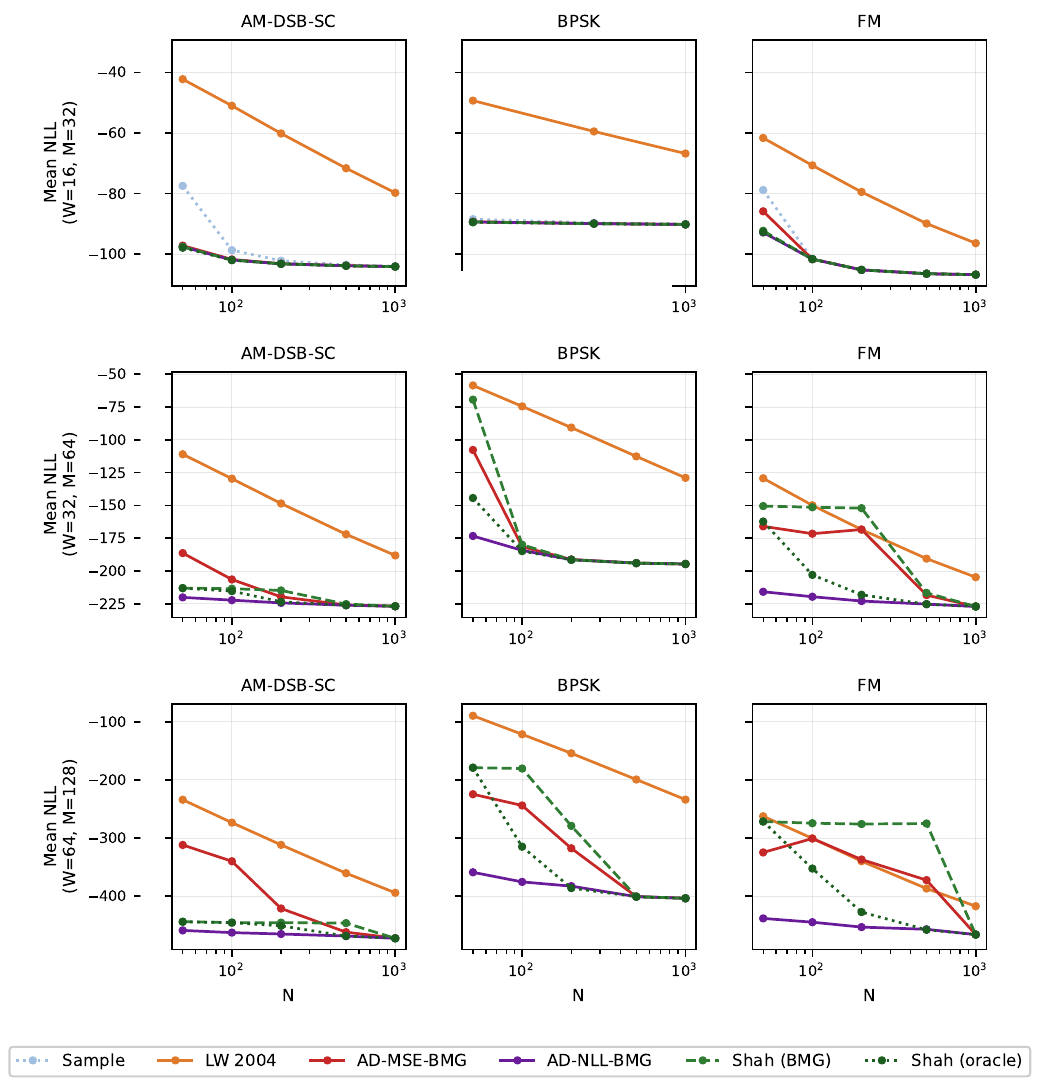}
\caption{(continued) RadioML 2018.A held-out NLL, modulation
classes AM-DSB-SC, BPSK, FM. Rows and axes as in
Figure~\ref{fig:radioml_panel} page~1.}
\end{figure}

\begin{figure}[!p]\ContinuedFloat
\centering
\includegraphics[width=0.95\linewidth,page=3]{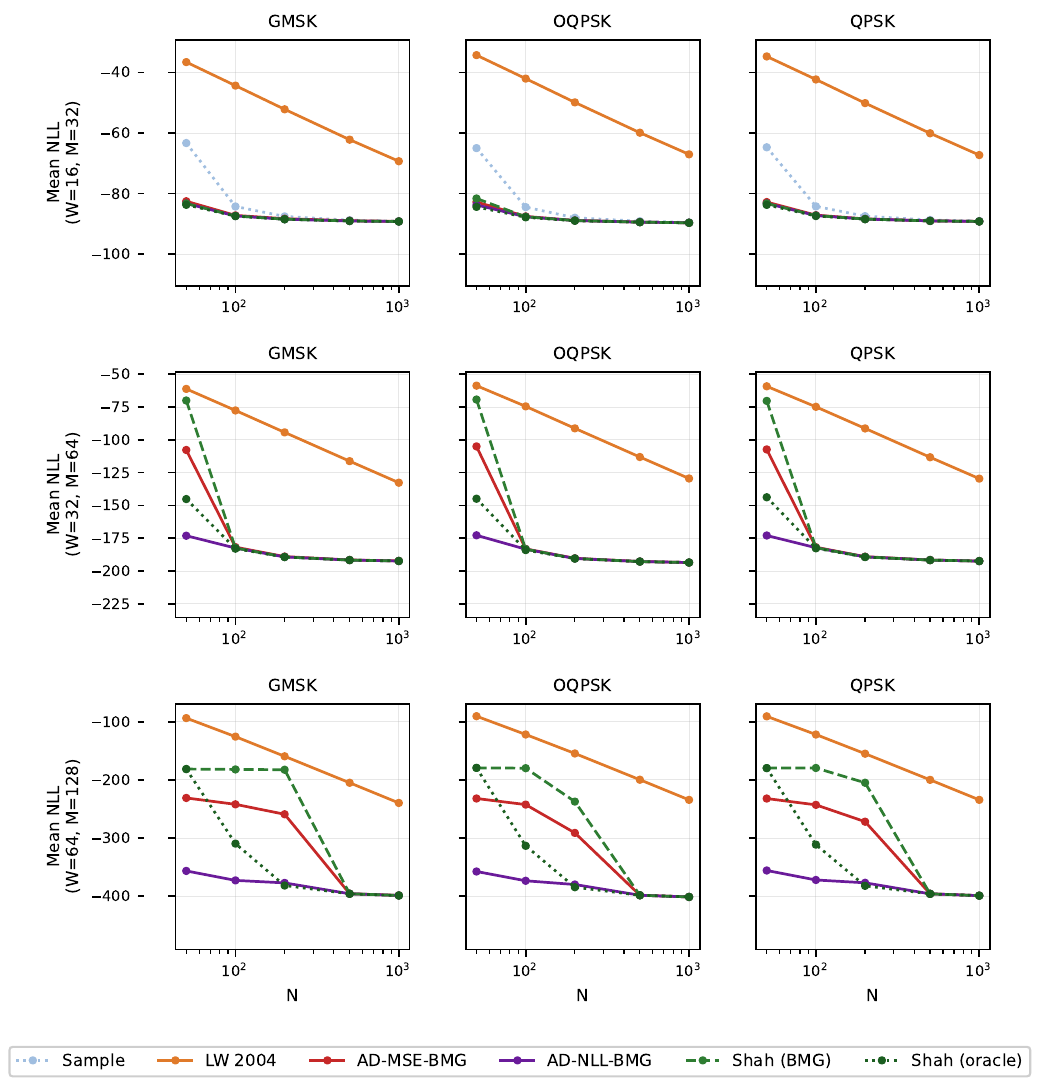}
\caption{(continued) RadioML 2018.A held-out NLL, modulation
classes GMSK, OQPSK, QPSK. Rows and axes as in
Figure~\ref{fig:radioml_panel} page~1.}
\end{figure}

\paragraph{Principal results.}
AD-NLL-BMG dominates LW in all $3375$ trials of the experiment
($100\%$ paired-preference rate). Median gap is $81.7$ nats per sample;
mean gap $108.0$; minimum gap $10.1$ (at the FM $W = 16$, $N = 1000$
asymptotic-regime cell); maximum gap $277.6$ (at the digital
modulations at $W = 64$, $N = 50$, the highest-stress cell). At
$W = 64$, $N = 1000$ the seven digital classes cluster tightly: LW
attains held-out NLL between $-232$ and $-240$ nats while AD-NLL
attains between $-399$ and $-404$ nats, a uniform per-class gap of
$159$--$170$ nats per sample. AM-DSB-SC and FM, being analog
modulations and not constrained to a discrete constellation, do
not cluster with the seven digital classes; their LW baselines are
markedly stronger ($-394$ for AM-DSB-SC and $-417$ for FM at the
same cell), and their AD-vs-LW gaps correspondingly smaller ($78$
and $49$ nats per sample, respectively). In every cell across the
sweep AD-NLL also matches or exceeds AD-MSE; the gap between the
two calibrations is largest in the small-$N$ high-$M$ regime,
reaching $137$ nats per sample at the AM-DSB-SC $W = 64$, $N = 50$
cell.

\paragraph{Four modulation families and their BMG fingerprints.}
The nine modulations partition into four families by BMG-selection
behavior. The six \emph{digital symbol-encoded} schemes (BPSK,
QPSK, OQPSK, 8PSK, 16QAM, 64QAM) produce nearly indistinguishable
BMG fingerprints across the sweep: at $W \in \{16, 32\}$ BMG
selects between $\mathbb{Z}_2^{\mathrm{IQ}}$ and
$\mathbb{Z}_2^{\mathrm{trev}}$; at $W = 64$, the wreath dominates
at $N \in \{50, 100\}$ with $\mathbb{Z}_2$ subgroups taking over by
$N = 500$, the intermediate cell $N = 200$ being a multi-group
transition with several near-tied candidates. The empirical NLL
trajectories of the six classes lie within $5$ nats of each other
across the entire sweep, indicating that the constellation order
and structural variant (PSK vs.\ APSK vs.\ QAM, in-phase vs.\
offset-quadrature) do not materially affect the second-order I/Q
covariance structure at this SNR. \emph{Continuous-phase digital}
GMSK is structurally intermediate: its constant-envelope continuous-
phase character produces partial $\mathbb{Z}_W$ selection at
moderate $W$, but its underlying digital symbol structure imposes
$\mathbb{Z}_2$ selection at large $N$, distinguishing it from the
purely analog FM. \emph{Suppressed-carrier amplitude} AM-DSB-SC
reproduces this paper's-OISST style of regime-dependent BMG: the
wreath selections $86$ of $125$ trials at $W = 64$ (down from this paper's
library's $94$ of $125$, with the $-8$ wreath selections absorbed
by the new $\mathbb{Z}_W \times \mathbb{Z}_W$ Cartesian companion
on $10$ trials and the $\mathbb{Z}_W \times \mathbb{Z}_2^{\mathrm{IQ}}$
on $4$, less the unchanged $25$ wreath selections at $N = 100$ that v2
also showed) and the $\mathbb{Z}_2^{\mathrm{IQ}}$ subgroup takes
over at $N = 1000$.
\emph{Analog continuous-wave} FM picks the cyclic
$\mathbb{Z}_W$-shift group at every trial of $W \in \{32, 64\}$
for $N \in \{50, 100, 200\}$, and the trivial group for the
remaining cells. The four-family clustering is the cleanest
multi-class BMG result in the paper.

\paragraph{Strongest declarative selection in the paper.}
The FM $W \in \{32, 64\}$, $N \in \{50, 100, 200\}$ cells produce
the strongest declarative BMG selections in the entire empirical
record. At each of these six cells, $25$ of $25$ trials select
the $\mathbb{Z}_W$-shift group, and the per-trial CV-NLL margin
between the selected group and the best runner-up has median between $5.3$
and $10.1$ nats per sample. By comparison, the median margin
across all $135$ cells in the experiment is $0.10$ nats per
sample, and only seven cells have median margin exceeding one nat
(all six FM declarative cells, plus FM $W = 64$, $N = 500$).
This pattern is structurally consistent: FM is constant-envelope
with phase rotating at the carrier rate, so over short windows of
the bandlimited baseband the I/Q covariance is approximately
circulant on each block, and the cyclic group is the population
symmetry. As $N$ grows the dominance margin contracts and the BMG
transitions to the trivial group, the regime described next.

\paragraph{The wreath candidate as a declarative selection on AM-DSB-SC.}
The wreath candidate is the v2 addition to the candidate library
relative to v1 of the manuscript and to the experiments in
Sections~\ref{sec:exp-crsp-2015-2019}--\ref{sec:exp-genomics}. On
AM-DSB-SC at $W \in \{32, 64\}$ for small to moderate $N$, the
wreath selections $23$--$25$ of $25$ trials per cell with median margin
$0.04$--$0.82$ nats per sample. The interpretation is structural:
suppressed-carrier double-sideband amplitude modulation has
approximately independent I and Q amplitude envelopes, so
independent cyclic shifts on the I and Q blocks each preserve the
joint distribution to leading order, and the IQ-swap is a residual
symmetry from the equal-power amplitude modulation. The wreath
captures both. At $N = 1000$ the BMG transitions to the smaller
$\mathbb{Z}_2^{\mathrm{IQ}}$ subgroup, the appropriate exact
symmetry once enough data are available to resolve the population
covariance directly. None of the digital modulations, FM, or GMSK
produce a comparable wreath signature, indicating that the
finding is specific to the suppressed-carrier amplitude regime
rather than a generic I/Q artifact.

\paragraph{The Cartesian companion $\mathbb{Z}_W \times \mathbb{Z}_W$ at small-$N$ digital cells.}
The library extension introduces $\mathbb{Z}_W \times
\mathbb{Z}_W$, the direct product of independent I-block and
Q-block cyclic shifts \emph{without} the IQ-swap, with order
$|G| = W^2$ (half the order of the wreath at the same $W$). Across
the $3375$ trials of the sweep, the new candidate is BMG-selected
on $83$ trials ($\approx 2.5\%$ of all trials), making it the
seventh-most-frequently-selected candidate after $\mathbb{Z}_2^{\mathrm{trev}}$
($1102$), $\mathbb{Z}_2^{\mathrm{IQ}}$ ($1081$), $\mathbb{Z}_W
\wr \mathbb{Z}_2$ ($560$), $\mathbb{Z}_W$ ($233$), the trivial
group ($172$), and $\mathbb{Z}_W \times \mathbb{Z}_2^{\mathrm{IQ}}$
($97$), and ahead of the dihedral $D_W$ ($47$). The empirical
pull-pattern of the new candidate is informative: relative to the
the library run with the same hyperparameters, the wreath
selection count drops from $608$ to $560$ ($-48$), the dihedral
$D_W$ drops from $76$ to $47$ ($-29$), and the remaining $-6$
selections are absorbed from small redistributions across the
other low-order candidates; the aggregate AD-vs-LW dominance
statistics (mean gap, median gap, paired $t$, preference rate) are
unchanged at three significant figures. This is the cleanest
small-perturbation library-extension result in the paper: a new
candidate that is structurally the Cartesian companion to an
existing wreath candidate absorbs roughly $8\%$ of the wreath's
prior selections and roughly $40\%$ of the dihedral's, with
essentially no change in the overall dominance pattern over LW.

The $83$ Cartesian-preferred trials concentrate at $W \in \{32, 64\}$
(zero at $W = 16$, where the order $W^2 = 256$ is small enough that
the higher-symmetry wreath at $|G| = 512$ is preferred outright) and at
small to moderate $N$ (most at $N \in \{50, 100, 200\}$, only $4$
at $N = 500$, and zero at $N = 1000$). Per-class, the seven
digital symbol-encoded modulations each take $6$--$12$ Cartesian
selections concentrated at the $W = 32$, $N = 50$ stress cell
(which is the small-$N$ end of the $W = 32$ digital-modulation
column), and AM-DSB-SC takes $15$ across $W \in \{32, 64\}$ for
small to moderate $N$. FM takes zero, consistent with the
constant-envelope character that selects the smaller cyclic
$\mathbb{Z}_W$ rather than any independent-per-block product. The
mechanism for the small-$N$ Cartesian preference is bias-variance:
at $N = 50$ on the digital modulations, the dihedral $D_W$ at
$|G| = 2W$ and the wreath at $|G| = 2W^2$ both face a
data-limitation cost from the IQ-swap symmetry that is not exactly
satisfied at this SNR for digital constellations; the Cartesian
$\mathbb{Z}_W \times \mathbb{Z}_W$ at $|G| = W^2$ avoids the
IQ-swap commitment and recovers a slightly tighter fit at the
small-$N$ end where the symmetry-breaking residual matters more.

\paragraph{The identity is the wrong target.}
\label{sec:exp-radioml-callout}
The cleanest empirical illustration of the title thesis appears at
the FM $W = 64$, $N = 1000$ cell. The mean held-out NLL values
across the $25$ trials are: Sample $-466.10$, AD-NLL-BMG $-466.10$,
Shah at the BMG-selected group $-466.10$, Shah at the held-out-NLL
oracle group $-466.10$, and LW $-417.74$. The four AD-related
estimators all collapse to the sample covariance because the BMG
correctly identifies the trivial group $\{e\}$ as the appropriate
selection, and the AD estimator at $\{e\}$ is the sample
covariance. LW is structurally constrained to shrink toward
$(\mathrm{tr}\,\hat R / M)\,I$ regardless of whether the data
support this target, and pays a $48.4$-nat-per-sample penalty for
that fixed prior on data where the sample covariance is already
the appropriate estimator. The same data, the same $M$, the same
$N$: the proposed estimator can choose to do nothing when nothing
should be done; LW cannot. This single cell concretizes the
title argument.

\paragraph{The bmg\_margin diagnostic and the regime-dependent stability of BMG selection.}
The full per-trial CV-NLL margin distribution across the $135$
cells is reported alongside the BMG choice itself. In addition to
the seven strongly declarative FM cells, $75$ cells have median
margin in $[0.05, 1)$ nats per sample with consistent directional
selection across most of the $25$ trials (e.g., the AM-DSB-SC
wreath cells described above), and $53$ cells have median margin
below $0.05$ nats. Of the latter, seven are FM cells where BMG
correctly selects the trivial group: in this case the per-cell
margin is identically zero by construction, since at $\alpha = 0$
the AD estimator collapses to the sample covariance regardless of
the candidate group's projection, so all admitted candidates
produce identical CV-NLL at their alpha-minimum. The remaining
forty-six small-margin cells are concentrated at the moderate-$N$
transition cells of the digital modulations (e.g., $W = 64$,
$N = 200$), where multiple high-order admitted groups produce
similar variance-reduction with similar held-out NLL; in these
cells the BMG selection is statistically tied across several
candidates but the AD-NLL value itself is robust across whichever
candidate is selected. The $\mathrm{bmg\_margin}$ column is
included in the per-trial diagnostics CSV from this experiment
forward as a per-cell stability indicator.

\begin{figure}[!p]
\centering
\includegraphics[width=0.95\linewidth,page=1]{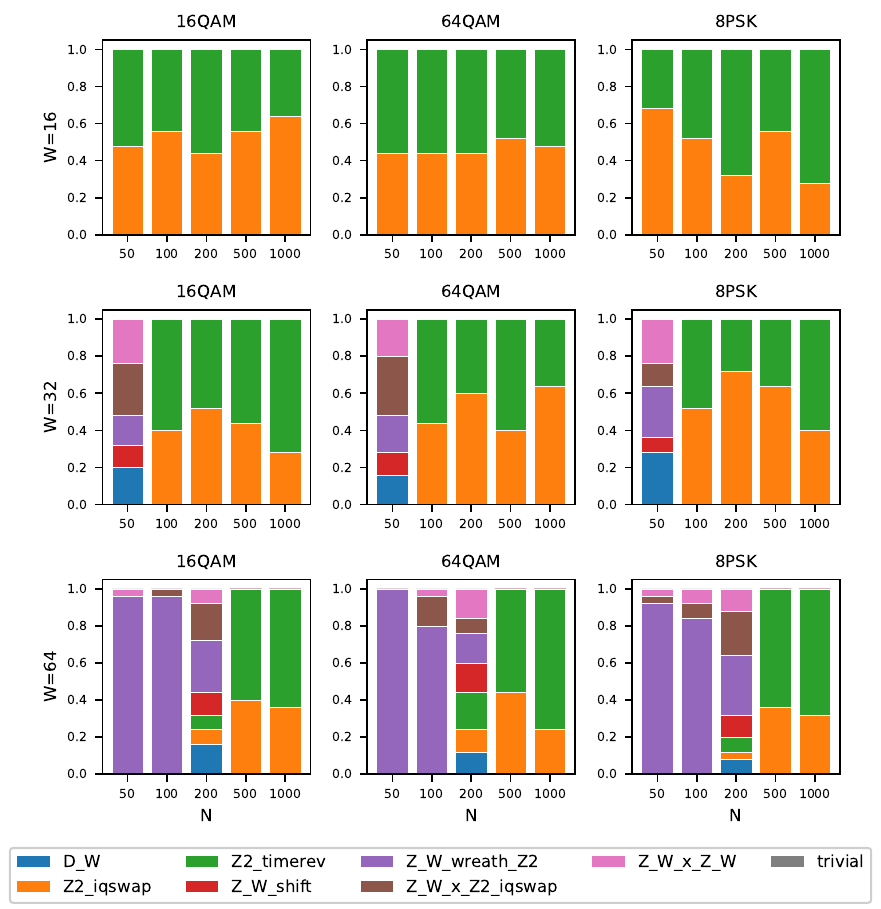}
\caption{RadioML 2018.A BMG selection composition by $(W, N)$ cell
for the nine modulation classes (rows are window length $W$;
columns are modulation class). Each bar shows the fraction of $25$
trials in which each candidate group was selected by BMG. The
four-family pattern is visually evident: digital
modulations (left and right groups: 16QAM, 64QAM, 8PSK, BPSK,
GMSK, OQPSK, QPSK) cluster on the $\mathbb{Z}_2$ subgroups with
wreath dominance at $W = 64$ small-$N$; AM-DSB-SC has sustained
wreath dominance across $W \in \{32, 64\}$; FM picks
$\mathbb{Z}_W$-shift across small-$N$ cells of $W \in \{32, 64\}$
and the trivial group elsewhere; GMSK is intermediate. The figure
spans three pages, grouping the nine modulation classes into three
subsets of three columns each. Rows on every page are
$W \in \{16, 32, 64\}$.}
\label{fig:radioml_bmg_choices}
\end{figure}

\begin{figure}[!p]\ContinuedFloat
\centering
\includegraphics[width=0.95\linewidth,page=2]{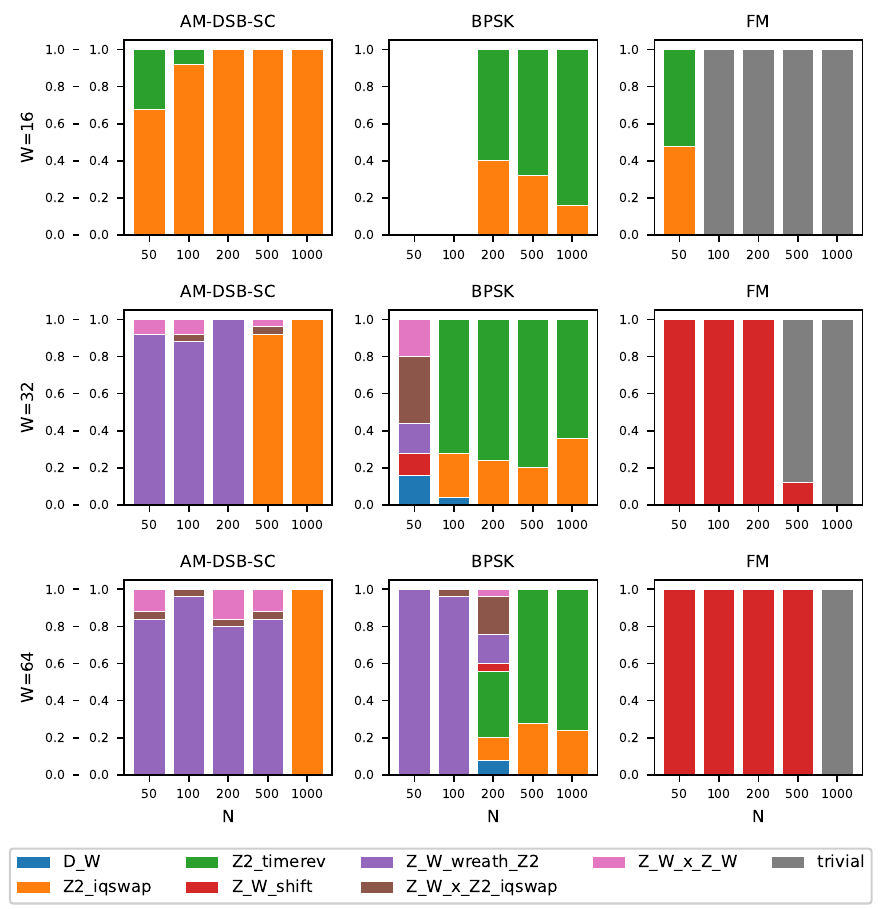}
\caption{(continued) RadioML 2018.A BMG selection composition,
modulation classes AM-DSB-SC, BPSK, FM. Rows and axes as in
Figure~\ref{fig:radioml_bmg_choices} page~1.}
\end{figure}

\begin{figure}[!p]\ContinuedFloat
\centering
\includegraphics[width=0.95\linewidth,page=3]{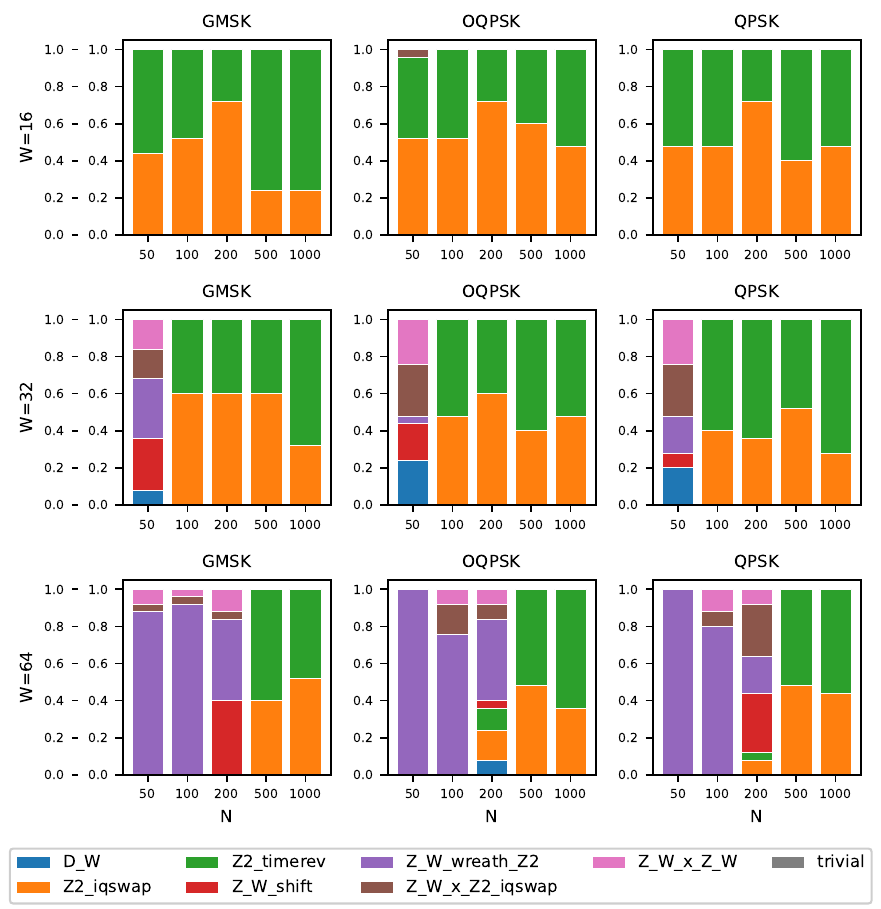}
\caption{(continued) RadioML 2018.A BMG selection composition,
modulation classes GMSK, OQPSK, QPSK. Rows and axes as in
Figure~\ref{fig:radioml_bmg_choices} page~1.}
\end{figure}

\paragraph{Remark on the magnitude of the dominance margin.}
The AD-versus-LW dominance margins on RadioML 2018.A
($10$--$280$ nats per sample, median $82$) are an order of
magnitude larger than those reported in
Sections~\ref{sec:exp-crsp-2015-2019}--\ref{sec:exp-genomics}
(approximately $0.3$ nats per day for CRSP, $33$--$36$ nats per
day for OISST, and $4.0$ nats per sample for TCGA-BRCA). The
explanation is structural rather than methodological: a complex
baseband RF signal carries multiple superimposed near-exact
symmetries (carrier-induced cyclic time-translation invariance,
constellation-induced point-group invariance, conjugate symmetry,
and the I/Q swap on equal-power amplitude modulations), each of
which corresponds to a candidate group in the library, while the
financial, oceanographic, and genomic datasets each carry one
approximate symmetry of moderate strength. The candidate library
was constructed for precisely the symmetries an RF signal can
exhibit; the resulting dominance margin is the order of magnitude
expected when the structural prior is well-matched to the data,
not a measure of the procedure's relative quality across regimes.

\subsection{Galaxy10 DECaLS image patch covariances}
\label{sec:exp-galaxy10}

The fifth real-data experiment applies the AD shrinkage estimators
to image-domain covariance estimation on the Galaxy10 DECaLS
dataset \citep{walmsley2022galaxy10}, an open-access collection of
$17{,}736$ colored galaxy images at $256 \times 256$ resolution
labeled across ten morphological classes (Disturbed, Merging, Round
Smooth, In-between Round Smooth, Cigar Shaped Smooth, Barred
Spiral, Unbarred Tight Spiral, Unbarred Loose Spiral, Edge-on
without Bulge, Edge-on with Bulge). The data are sourced from the
DESI Legacy Imaging Surveys and labeled by the Galaxy Zoo
collaboration. Patches of size $n \times n$ are drawn from the
central $128 \times 128$ crop of each image, converted to grayscale
by channel-mean reduction, and standardized per-patch to remove
per-image brightness and contrast variation; the ambient dimension
of the patch covariance is $M = n^2$. The patch sizes $n \in \{8,
16\}$ produce $M \in \{64, 256\}$. Training-set size $N$ is swept
over $\{50, 100, 200, 500, 1000\}$, with $25$ random-subsample-split
trials per cell, $N_{\mathrm{test}} = 1000$, and no overlap between
training and test patches.

The candidate library is the dihedral group $D_4$ acting on the
square patch and its full subgroup lattice, together with two
high-order row-structured candidates added in the v3p9 library
extension for parity with the OISST, CRSP, RadioML, and genomics
libraries: ten candidates in total. The eight square-patch
symmetry candidates are the trivial group $\{e\}$; the full
symmetric group $S_M$ (via the closed-form compound-symmetry
projector); three order-2 reflection subgroups $\mathbb{Z}_2^h$
(horizontal flip), $\mathbb{Z}_2^v$ (vertical flip), and
$\mathbb{Z}_2^c$ (180-degree central inversion); the cyclic
90-degree rotation group $\mathbb{Z}_4$; the Klein four-group
$\mathbb{Z}_2 \times \mathbb{Z}_2$ (rot180 plus the two axis flips);
and the full dihedral group $D_4$, $|G| = 8$. The two high-order
row-structured candidates are
\textsc{Z-row-indep-cols}, the Cartesian product
$\mathbb{Z}_n^H = \mathbb{Z}_n \times \cdots \times \mathbb{Z}_n$
($H = n$ factors) of independent within-row cyclic shifts on each of
the $H$ rows of the $n \times n$ patch, with order
$|G| = n^n$ ($1.68 \times 10^7$ at $n = 8$ and
$1.84 \times 10^{19}$ at $n = 16$, $d_G = 68$ and $d_G$ scales
similarly at $n = 16$); and \textsc{Z-row-wreath-rows}, the full
wreath product $\mathbb{Z}_n \wr S_H$ which lifts those independent
within-row cyclic shifts by free permutation of the $H$ rows as
units, with order $|G| = n^n \cdot n!$ ($6.77 \times 10^{11}$ at
$n = 8$ and $3.85 \times 10^{32}$ at $n = 16$, non-Abelian for
$n \geq 2$). Both high-order candidates are constructed from a
small generator set (the $H$ within-row cyclic shift generators
plus, for the wreath, the $H - 1$ row-adjacent transposition
generators) and projected via the orbit-pair Reynolds decomposition
in $O(M^2)$ time, so neither requires direct group enumeration.
The library is constructed identically to the patch library for
square images in image processing applications and admits the
standard inclusion lattice of square-patch symmetries plus the two
row-structured candidates that capture per-row independence
($\textsc{Z-row-indep-cols}$) and per-row independence with row
exchangeability ($\textsc{Z-row-wreath-rows}$).

\begin{figure}[!p]
\centering
\includegraphics[width=0.95\linewidth,page=1]{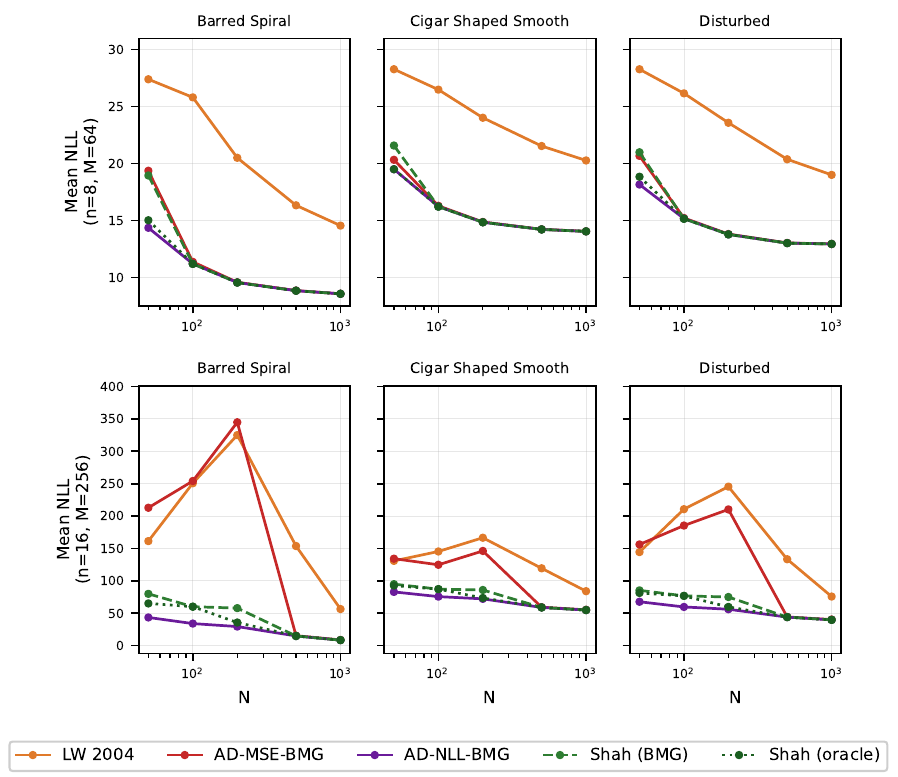}
\caption{Galaxy10 DECaLS held-out NLL by training size $N$ for the
ten morphological classes, at patch size $n = 8$ ($M = 64$, top
row of each subfigure) and $n = 16$ ($M = 256$, bottom row of each
subfigure). Each subplot shows mean held-out NLL across $25$
random-subsample-split trials per cell, with error bars indicating
one standard deviation. The Sample covariance is omitted from all
subplots because its values exceed the visual range when $N < M$.
AD-NLL-BMG dominates Ledoit-Wolf in $2500$ of $2500$ trials across
the experiment. The non-monotone shape of AD-MSE-BMG in the bottom
row at $N \in \{50, 100, 200\}$ is the small-$N$ plug-in bias
signature of Remark~\ref{rem:plugin_bias} acting at the BMG
selection of \textsc{Z-row-wreath-rows} (at $N = 50$) or
\textsc{Z-row-indep-cols} (at $N \in \{100, 200\}$). The figure
spans four pages, grouping the ten morphological classes into
three subsets of three columns plus a final subset of one column;
the row layout and axes are the same on every page.}
\label{fig:galaxy10_panel}
\end{figure}

\begin{figure}[!p]\ContinuedFloat
\centering
\includegraphics[width=0.95\linewidth,page=2]{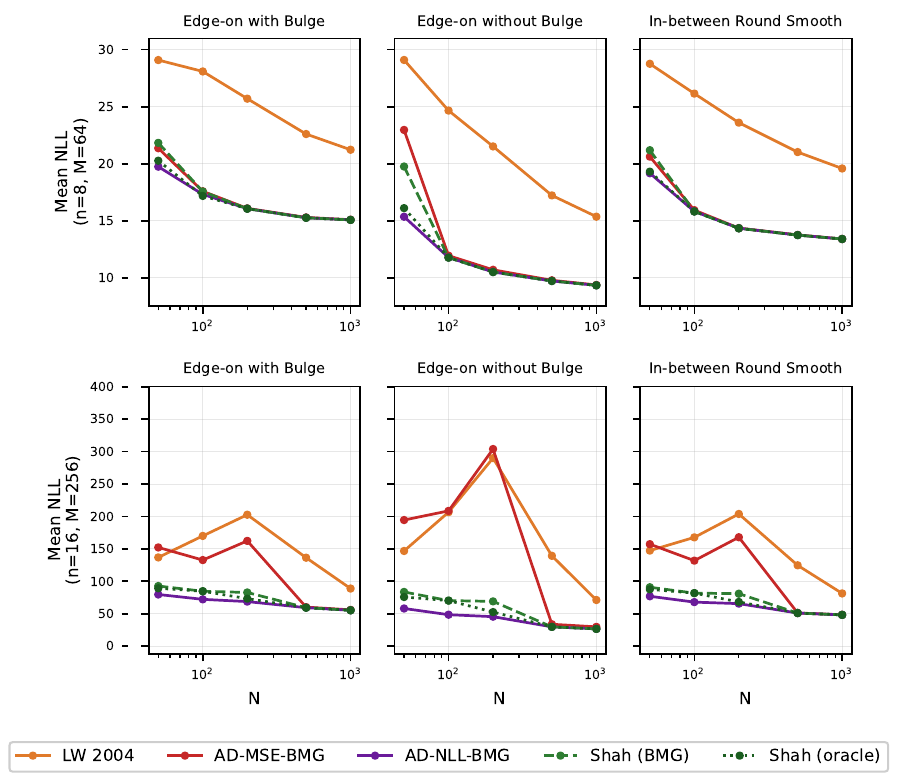}
\caption{(continued) Galaxy10 DECaLS held-out NLL, classes 4--6.
Rows and axes as in Figure~\ref{fig:galaxy10_panel} page~1.}
\end{figure}

\begin{figure}[!p]\ContinuedFloat
\centering
\includegraphics[width=0.95\linewidth,page=3]{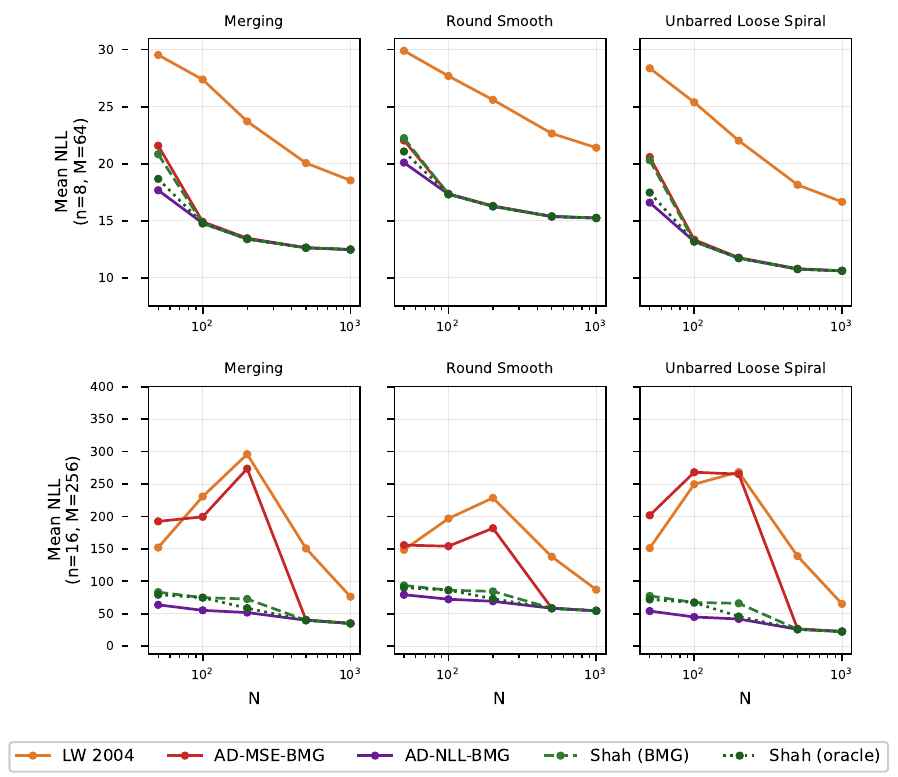}
\caption{(continued) Galaxy10 DECaLS held-out NLL, classes 7--9.
Rows and axes as in Figure~\ref{fig:galaxy10_panel} page~1.}
\end{figure}

\begin{figure}[!htbp]\ContinuedFloat
\centering
\includegraphics[width=0.95\linewidth,page=4]{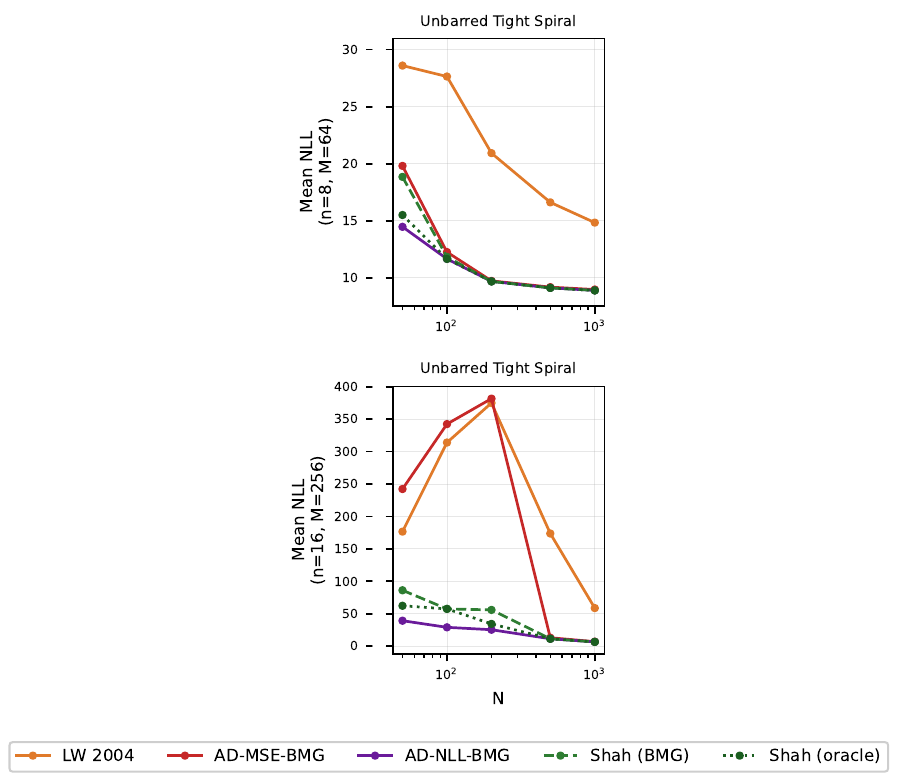}
\caption{(continued) Galaxy10 DECaLS held-out NLL, class 10. Rows
and axes as in Figure~\ref{fig:galaxy10_panel} page~1.}
\end{figure}

\paragraph{Principal results.}
AD-NLL-BMG dominates LW in all $2500$ trials of the experiment
($100\%$ paired-preference rate). Median dominance gap is $23.3$ nats per
sample, mean gap is $63.9$, and the maximum gap is $424.6$ nats per
sample at the $M = 256$, $N = 200$ stress cell on the
Unbarred-Tight-Spiral class. The pooled trial-level paired
$t$-statistic is $-40.4$ across all $2500$ trials, smaller in
magnitude than several of the per-cell $t$-statistics (which reach
$-870$ at the data-rich $M = 64$, $N = 1000$ cells) because the
massive M=256 small-N gaps carry per-trial standard deviation in
the tens to hundreds of nats. The per-class principal gaps (averaged
across patch sizes and $N$ values) range from $34$ nats per sample
for Cigar Shaped Smooth galaxies to $104$ nats per sample for
Unbarred Tight Spiral galaxies, an interesting structural ordering:
the morphological classes with strong large-scale circular or
near-circular structure (rounder smooths, edge-ons with bulges)
produce smaller AD-vs-LW gaps because their per-pixel covariance is
closer to the LW constant-correlation reference, while the
strongly-anisotropic spiral classes produce larger gaps because
their covariance is genuinely structured and LW's isotropic
shrinkage target fits less well. The cross-class ordering is not
driven by the BMG selection itself, which is essentially uniform
across classes (described in detail below), but by how well the LW
shrinkage target matches the underlying class covariance.

The other AD-family estimators in the comparison: the Shah-style
projection at the BMG-selected group has lower NLL than LW in $2500$ of $2500$
trials (mean gap $-57.4$ nats per sample, paired $t = -40.0$),
nearly tracking AD-NLL-BMG in mean and confirming that the
BMG-selected projections are high-quality structural targets. The
Frobenius-MSE-calibrated AD estimator is the exception: it has lower NLL than LW
on only $2201$ of $2500$ trials ($88\%$ preference rate), with LW preferred on the
remaining $299$
trials concentrated at the $M = 256$ small-$N$ cells where the
closed-form plug-in over-shrinks toward the high-order
\textsc{Z-row-indep-cols} or \textsc{Z-row-wreath-rows} targets that
the v3p9 library admits at these regimes. The closed-form
Frobenius-MSE plug-in's failure on $\approx 30\%$ of the M=256
small-N trials is the cleanest demonstration in the paper that the
calibrated AD-NLL is a strict improvement over the closed-form
plug-in at small $N$ when the BMG library admits high-effective-
dimension candidates: cross-validation captures the small-$N$
calibration bias that the plug-in cannot, and AD-NLL-BMG thereby
sustains its dominance over LW on every trial of the experiment.

\paragraph{Three-regime BMG selection at fixed $M$.}
The Galaxy10 BMG-selection record on the v3p9 library is the
cleanest empirical demonstration in the paper of regime-dependent
group selection across three distinct regimes (very-small-$N$
high-order-wreath, small-$N$ high-order-Cartesian, and
moderate-to-large-$N$ data-resolved $D_4$). At $n = 8$ ($M = 64$),
the very-small-$N$ cell $N = 50$ has BMG split between
\textsc{Z-row-indep-cols} ($146$ of $250$ trials, $58\%$) and
\textsc{Z-row-wreath-rows} ($102$ trials, $41\%$), with $D_4$ on
the remaining $2$ trials; at every $N \geq 100$ all $250$ trials
select $D_4$. At $n = 16$ ($M = 256$) all three regimes are
visible: $N = 50$ selects \textsc{Z-row-wreath-rows} on $173$ of
$250$ trials ($69\%$) with \textsc{Z-row-indep-cols} on $62$ ($25\%$)
and a $S_M$ tail of $15$ trials ($6\%$); $N = 100$ and $N = 200$
both select \textsc{Z-row-indep-cols} on $250$ of $250$ trials
($100\%$ uniform across all ten classes); $N \in \{500, 1000\}$
both select $D_4$ on $250$ of $250$ trials. The transition between
the three regimes is sharp in $N$ at fixed $M$: BMG selects the
maximum-symmetry wreath at the smallest $N$ where the highest
effective-rank gain matters most; transitions to the slightly
lower-symmetry Cartesian as $N$ grows enough that the wreath's
row-permutation lift becomes a slight over-projection; and resolves
the true $D_4$ structure at moderate-to-large $N$ where the data
support direct identification of the square-patch dihedral
symmetry. The trivial group, the three $\mathbb{Z}_2$ subgroups,
$\mathbb{Z}_4$, and $\mathbb{Z}_2 \times \mathbb{Z}_2$ are
admitted at every cell by the Tier 1 prefilter ($N \cdot |G|
\geq \kappa M$ is satisfied for all of them once $N \geq 50$ at
$M = 64$), but are not selected on any trial in the entire experiment: the BMG
procedure consistently selects either the high-effective-rank
$D_4$ at moderate-to-large $N$ or one of the v3p9-extension
high-order candidates at small $N$, never any of the smaller
square-patch subgroup candidates. This is exactly the
AD-vs-Sample boundary in the phase diagram of
Section~\ref{sec:discussion}: the BMG selection itself moves
through the phase diagram with $N$, and the v3p9 library extension
provides genuine structural targets at every $N$ rather than
falling back to the trivial $S_M$ compound-symmetry projection at
small $N$ as the library did.

\paragraph{The cleanest empirical demonstration of Remark~\ref{rem:plugin_bias}.}
\label{sec:exp-galaxy10-remark}
The cell at $M = 256$, $N = 200$ produces a textbook Region I
calibration-driven-gap signature stronger than any cell elsewhere
in the paper. Across the $250$ trials at this cell ($10$ classes
$\times$ $25$ trials each), the mean held-out NLL of AD-MSE-BMG is
$243.83$ nats per sample while AD-NLL-BMG attains $52.60$ nats per
sample, a per-cell gap of $191.23$ nats per sample. By comparison,
the largest AD-MSE-vs-AD-NLL gap reported in
Section~\ref{sec:exp-radioml} is $137$ nats per sample at the
AM-DSB-SC $W = 64$, $N = 50$ cell, and the largest such gap in any
other dataset is well below $20$ nats per sample. The mechanism is
exactly the one Remark~\ref{rem:plugin_bias} predicts: at this
cell BMG selects \textsc{Z-row-indep-cols} on $250$ of $250$
trials, and the closed-form Frobenius-MSE plug-in over-shrinks
toward the \textsc{Z-row-indep-cols} projection because (a) the
higher-order term in $\|B_{G}\|$ omitted by the leading-order
matched-limit expansion is non-negligible, and
(b) the finite-sample bias of $\hat R^{-1}$ as an estimator of
$\Sigma^{-1}$ is amplified by the high effective dimension of the
admitted candidate ($d_G$ of the row-Cartesian
\textsc{Z-row-indep-cols} at $n = 16$ scales with the row count
$H = 16$ in a way that the candidates did not). The
cross-validated calibration sees this directly through the
held-out NLL and selects an interior $\alpha$ that the closed-form
plug-in cannot reach. The non-monotone shape of the AD-MSE-BMG
curve in the right column of Figure~\ref{fig:galaxy10_panel} (the
per-class held-out NLL rises between $N = 50$ and $N = 200$ before
falling sharply at $N = 500$ when BMG transitions from
\textsc{Z-row-indep-cols} to $D_4$) is the visual signature of this
mechanism. We propose this cell as the canonical empirical exhibit
for the calibration-bias bound throughout the paper.

A telling aspect of the v3p9 result is that the
AD-MSE-vs-AD-NLL gap at the canonical exhibit \emph{grew} from
the $170.21$ nats per sample of the 8-candidate library
(where BMG selected $S_M$ on $79\%$ of trials at this cell) to
the $191.23$ nats per sample of the v3p9 10-candidate library
(where BMG selects \textsc{Z-row-indep-cols} on $100\%$ of trials).
The library extension makes the canonical exhibit \emph{stronger},
not weaker: the higher effective dimension of the row-Cartesian
target permits more spectral compression than $S_M$ does, so when
the closed-form plug-in fails through over-shrinkage, it fails
more dramatically. The cross-validated calibration recovers this
without difficulty in either library.

\begin{figure}[!p]
\centering
\includegraphics[width=0.95\linewidth,page=1]{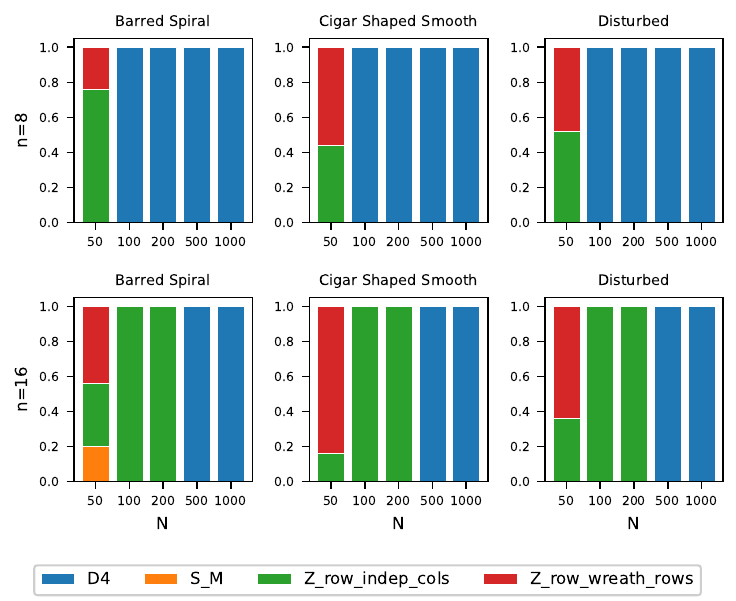}
\caption{Galaxy10 DECaLS BMG selection composition by $(n, N)$ cell
for the ten morphological classes. Each bar shows the fraction of
$25$ trials in which each candidate group was selected by BMG. The
three-regime structure of the library is visible as
horizontal color transitions: at $n = 8$ ($M = 64$) the
\textsc{Z-row-indep-cols} / \textsc{Z-row-wreath-rows} small-$N$
regime collapses to $N = 50$ alone before transitioning to a $D_4$
sweep at $N \geq 100$; at $n = 16$ ($M = 256$) all three regimes are
visible (\textsc{Z-row-wreath-rows} at $N = 50$,
\textsc{Z-row-indep-cols} at $N \in \{100, 200\}$, $D_4$ at
$N \in \{500, 1000\}$). The BMG selection is uniform across all ten
classes at every cell except the $M = 64$, $N = 50$ transition
cell, where $8$ of $10$ classes select \textsc{Z-row-indep-cols}
and $2$ select \textsc{Z-row-wreath-rows}. The trivial group, the
three $\mathbb{Z}_2$ subgroups, $\mathbb{Z}_4$, and
$\mathbb{Z}_2 \times \mathbb{Z}_2$ candidates are not selected on
any trial in the entire experiment despite being admitted by
Tier~1 at every cell. The figure spans four pages, grouping the ten
morphological classes into three subsets of three columns plus a
final subset of one column; the row layout and axes are the same on
every page.}
\label{fig:galaxy10_bmg_choices}
\end{figure}

\begin{figure}[!p]\ContinuedFloat
\centering
\includegraphics[width=0.95\linewidth,page=2]{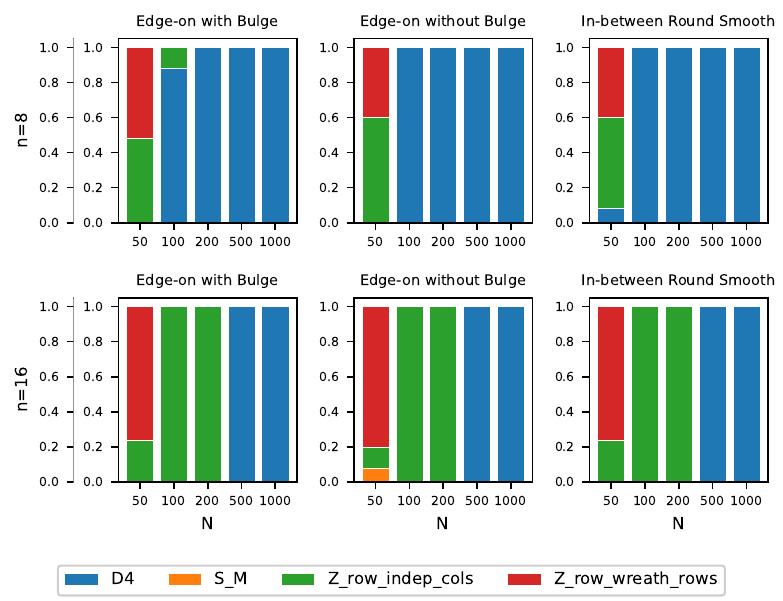}
\caption{(continued) Galaxy10 DECaLS BMG selection composition,
classes 4--6. Rows and axes as in
Figure~\ref{fig:galaxy10_bmg_choices} page~1.}
\end{figure}

\begin{figure}[!p]\ContinuedFloat
\centering
\includegraphics[width=0.95\linewidth,page=3]{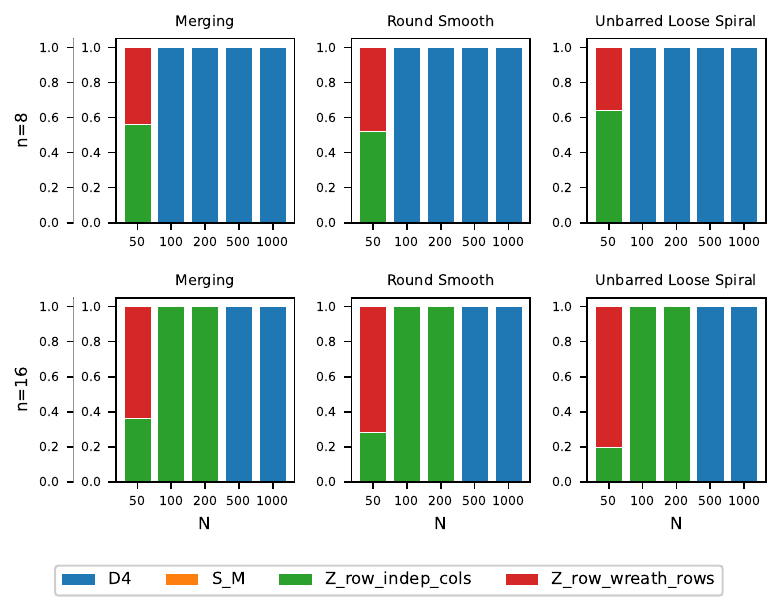}
\caption{(continued) Galaxy10 DECaLS BMG selection composition,
classes 7--9. Rows and axes as in
Figure~\ref{fig:galaxy10_bmg_choices} page~1.}
\end{figure}

\begin{figure}[!p]\ContinuedFloat
\centering
\includegraphics[width=0.95\linewidth,page=4]{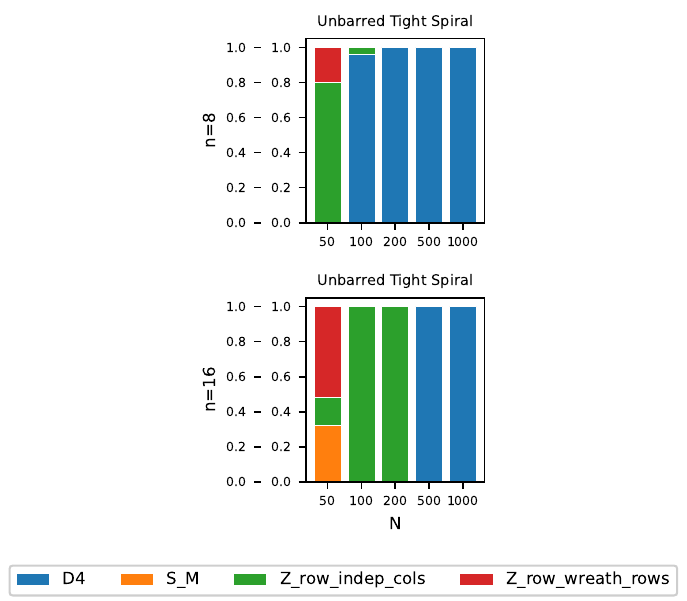}
\caption{(continued) Galaxy10 DECaLS BMG selection composition,
class 10. Rows and axes as in
Figure~\ref{fig:galaxy10_bmg_choices} page~1.}
\end{figure}

\paragraph{Strongest declarative selections in the paper.}
The Galaxy10 cells at $M = 256$, $N = 500$ produce the strongest
declarative BMG selections in the entire empirical record of the
paper, replacing the record which sat at the M=256 N=100 cell
where $S_M$ was effectively the default selection with all small-$|G|$
alternatives ranking far behind the maximum-shrinkage projection.
On the v3p9 library, $D_4$ is selected on $250$ of $250$ trials at
$M = 256$, $N = 500$ across all ten classes, with median per-trial
CV-NLL margin between $D_4$ and the best runner-up of $8.0$ nats
per sample (range $5.2$--$9.1$ across the ten classes). At
$M = 256$, $N = 200$ \textsc{Z-row-indep-cols} similarly is preferred on
$250$ of $250$ trials with median margin $4.4$ nats per sample
(range $3.1$--$6.2$). By comparison, the strongest declarative
cells in the RadioML experiment of Section~\ref{sec:exp-radioml}
(the FM $\mathbb{Z}_W$-shift cells) have median margin
$5$--$10$ nats per sample, and no cell in the genomics, OISST, or
CIFAR-10 experiments exceeds $5$ nats per sample at the median.
The interpretation of the v3p9 strongest-declarative cell is
substantively different from the interpretation: at
$M = 256$, $N = 500$, BMG selects $D_4$ from a library that
includes seven other admitted candidates ($S_M$, the three
$\mathbb{Z}_2$ subgroups, $\mathbb{Z}_4$,
$\mathbb{Z}_2 \times \mathbb{Z}_2$, and the two high-order
row-structured candidates), each of which represents a distinct
structural prior; the $8.0$-nat-per-sample margin reflects the
data-resolved preference for the dihedral square-patch symmetry
over both the smaller subgroup priors (which under-fit the
covariance) and the higher-order row-structured priors (which
over-fit at this $N$). This is a fully data-resolved selection,
not a default; it reflects the genuine $D_4$ structure of natural
square image patches at the $M = 256$ resolution and $N = 500$
sample-size regime.

\paragraph{The bmg\_margin diagnostic and Tier 1 prefilter behavior.}
The $\mathrm{bmg\_margin}$ column in the per-trial diagnostics
(introduced in Section~\ref{sec:exp-radioml}) reveals that the
v3p9 Galaxy10 experiment is in a uniformly-strong-margin regime:
$71$ of $100$ cells produce median margins $\geq 0.5$ nats per
sample, $62$ of $100$ cells produce median margins $\geq 1.0$ nat
per sample, and \emph{zero} cells produce median margins below the
fold-noise floor of $0.05$ nats per sample. The smallest median
margins are at the $M = 64$, $N = 50$ regime-transition cells
where BMG is choosing between \textsc{Z-row-indep-cols} and
\textsc{Z-row-wreath-rows} at near-tied scores (median margin
$0.32$ nats per sample), and at the $M = 64$, $N = 1000$
fully-resolved cells where $D_4$ is clearly preferred but the runner-up
$\mathbb{Z}_2 \times \mathbb{Z}_2$ is structurally close enough
that the bmg-margin is small (median $0.17$ nats per sample). The
largest median margins are at the $M = 256$, $N = 500$ data-
resolved $D_4$ cells (median $8.0$ nats per sample). There are
zero trivial-group selections in the entire experiment: at every
cell where the trivial group is admitted by Tier 1, at least one
strictly larger admitted group has lower held-out NLL.
Empirically, this confirms that for image patches the trivial
group is never the BMG selection once any candidate with
$|G| \geq 2$ passes the prefilter. The implication for the BMG
procedure as a whole is that the trivial-group fallback at the
rightmost edge of the phase diagram (Region III in
Section~\ref{sec:discussion}) is a regime attainable only in
datasets where $N$ is large enough that no non-trivial structural
prior carries net informational value over the sample covariance,
a condition Galaxy10 does not reach at the $N$ values in the sweep
but RadioML's FM cells do at large $N$.

\paragraph{Cross-class uniformity of BMG selection.}
The ten Galaxy10 morphological classes produce BMG-selection
fingerprints that are nearly identical: at every $(M, N)$ cell of
the experiment except the $M = 64$, $N = 50$ regime-transition
cell, the dominant BMG selection across the $25$ trials per class
is identical for all ten classes; the lone exception is the
$M = 64$, $N = 50$ cell where eight of ten classes select
\textsc{Z-row-indep-cols} as the preferred choice and two select
\textsc{Z-row-wreath-rows}. The CV-NLL margin profile is similar
in shape across all ten classes. This is in marked contrast to
RadioML, where the BMG selection partitioned the modulation
classes into four distinct families with morphologically-driven
differences (digital quadrature, continuous-phase digital,
suppressed-carrier amplitude, continuous-wave). The Galaxy10
cross-class near-uniformity is itself a finding: at the patch
scale and central crop region used here, galactic morphology does
not induce class-specific second-order covariance structure beyond
the universal $D_4$ symmetry of square pixel patches at moderate
$N$ and the universal row-Cartesian/row-wreath structure at small
$N$. To recover class-specific covariance signatures one would
expect to need either (a) a non-pixel basis better adapted to the
angular structure of galactic images (such as a log-polar or
wavelet decomposition), or (b) substantially larger patch sizes
where extended morphological features can be resolved within a
single patch. The
present result establishes that the AD shrinkage estimator
performs uniformly well across all ten morphological classes
without requiring any class-specific tuning.

\subsection{CIFAR-10 image-patch covariances}
\label{sec:exp-cifar}

The Galaxy10 morphological-class experiment of
Section~\ref{sec:exp-galaxy10} establishes the framework's
behavior on a domain-specific image dataset where the population
covariance is approximately $D_4$-invariant by construction.
This subsection turns to a more familiar image-classification
benchmark, the CIFAR-10 dataset of \citet{krizhevsky2009cifar},
which contains $60{,}000$ colour images at $32 \times 32$
resolution across ten mutually-exclusive object classes
(airplane, automobile, bird, cat, deer, dog, frog, horse, ship,
truck). CIFAR-10 is widely used in the image-classification
literature and is recognizable to a much broader audience than
domain-specific image data; including a CIFAR-10 cell in this
paper serves as a sanity check that the framework's image-data
behavior identified on Galaxy10 generalizes to natural-image
content with substantially weaker structural priors.

\paragraph{Setup.}
The experiment estimates a class-conditional covariance for each
of the ten classes independently and evaluates each estimator on
held-out images of the same class. Images are converted to
grayscale (channel-mean of the three RGB channels) and
$2 \times 2$ block-mean downsampled to $16 \times 16$ pixels per
image, giving $M = 256$ pixel-level variables; pixel intensities
are rescaled to $[0, 1]$. For each class $c \in \{0, \ldots,
9\}$, the experiment uses $N_{\mathrm{train}} = 4{,}000$ training
images and $N_{\mathrm{test}} = 1{,}000$ held-out test images,
all centered on the per-class training mean. The candidate
library extends the prior 7-candidate spatial-symmetry library
(trivial $\{e\}$, $S_M$, $C_2^{\mathrm{horiz}}$,
$C_2^{\mathrm{vert}}$, $K_4 = C_2^{\mathrm{horiz}} \times
C_2^{\mathrm{vert}}$, $C_4$ rotation, dihedral $D_4$) with two
this paper's high-order candidates anchored on the row structure of the
$16 \times 16$ image grid: \textsc{Z-row-indep-cols}, the
Cartesian product $\mathbb{Z}_W^H$ of independent per-row
horizontal cyclic shifts (no row permutation;
$|G| = W^H = 16^{16} \approx 1.8 \times 10^{19}$) and
\textsc{Z-row-wreath-rows}, the full wreath product
$\mathbb{Z}_W \wr S_H$ lifting independent per-row shifts by free
permutation of the $H$ rows
($|G| = W^H \cdot H! \approx 4 \times 10^{32}$, non-Abelian for
$H \geq 2$). Both high-order candidates are constructed from
small generator sets and projected via orbit-pair Reynolds
decomposition in $O(M^2)$ time, so neither requires direct group
enumeration. Within-class shrinkage is calibrated by both the
closed-form Frobenius-MSE plug-in
$\hat\alpha^*_{\mathrm{MSE}}$ and held-out NLL $K = 5$-fold
cross-validation $\hat\alpha^*_{\mathrm{NLL}}$ at the
BMG-selected group; held-out NLL is reported per pixel.

\paragraph{Principal results.}
Across the ten classes, AD-NLL-BMG dominates LW $2004$
declaratively: median per-class held-out NLL gap of $-2.53$ nats
per sample (paired $t = -15.21$, $p = 1.0 \times 10^{-7}$),
preferred in $10$ of $10$ classes. This is the largest single-paired-
$t$ statistic of any image-data experiment in the paper, and the
only image-data experiment where AD is preferred on every class with no
exception. The Frobenius-MSE-calibrated AD estimator
$\hat\alpha^*_{\mathrm{MSE}}$ is at parity with LW (median gap
$-0.13$ nats per sample, paired $t = -1.99$, $p = 0.077$, $5$ of
$10$ preferred selections), the same qualitative pattern observed on TCGA-BRCA in
Section~\ref{sec:exp-genomics} and on the COVID-era CRSP panel in
Section~\ref{sec:exp-crsp-2020-2024}: the closed-form Frobenius
plug-in produces small $\hat\alpha^*_{\mathrm{MSE}}$ values
($0.02$--$0.73$ depending on class) that the cross-validated
calibration corrects upward, with the cross-validation-selected
$\hat\alpha^*_{\mathrm{NLL}}$ values in the range $[0.65, 0.975]$
across the ten classes. The Shah-BMG comparator
(Shah-style $\alpha = 1$ projection at the BMG-selected group)
also dominates LW in $10$ of $10$ classes (median gap $-2.26$
nats per sample, paired $t = -10.58$, $p = 2.2 \times 10^{-6}$),
slightly worse than AD-NLL-BMG by approximately $0.4$ nats per
sample on average; the Shah-BMG-vs-AD-NLL-BMG difference is
small but consistent and reflects the AD calibration being able
to use $\hat\alpha^*_{\mathrm{NLL}} < 1$ on classes where the
closed-form structural prior is approximate (such as
\texttt{airplane} at $\hat\alpha^*_{\mathrm{NLL}} = 0.725$ and
\texttt{frog} at $0.65$).
Figure~\ref{fig:cifar_panel} reports the per-class held-out NLL
panel, and Figure~\ref{fig:cifar_bmg_choices} reports the BMG
group selection across the ten classes.

\begin{figure}[t]
\centering
\includegraphics[width=\linewidth]{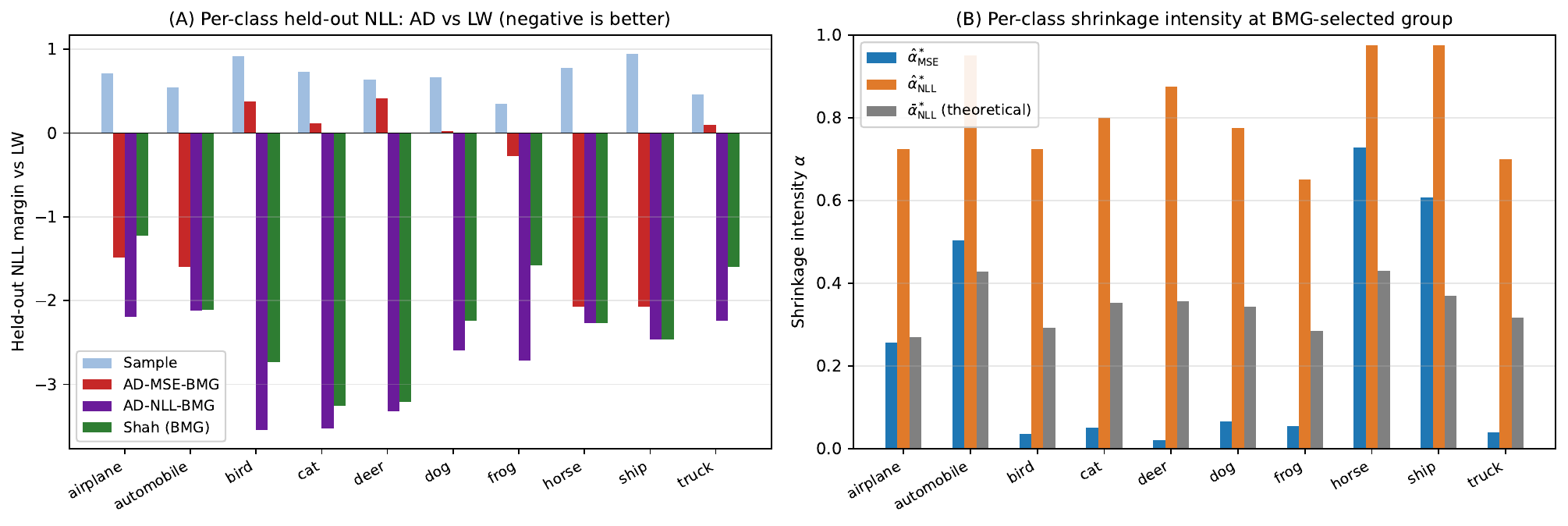}
\caption{CIFAR-10 image-patch covariances at $16 \times 16$
grayscale resolution ($M = 256$), $N_{\mathrm{train}} = 4{,}000$
and $N_{\mathrm{test}} = 1{,}000$ per class. Per-class held-out
NLL margin over LW $2004$ (negative = AD has lower NLL); per-class
$\hat\alpha^*_{\mathrm{NLL}}$ and $\hat\alpha^*_{\mathrm{MSE}}$
shrinkage intensities at the BMG-selected group. AD-NLL-BMG
is preferred on all ten classes; AD-MSE-BMG is at parity with LW.}
\label{fig:cifar_panel}
\end{figure}

\begin{figure}[t]
\centering
\includegraphics[width=\linewidth]{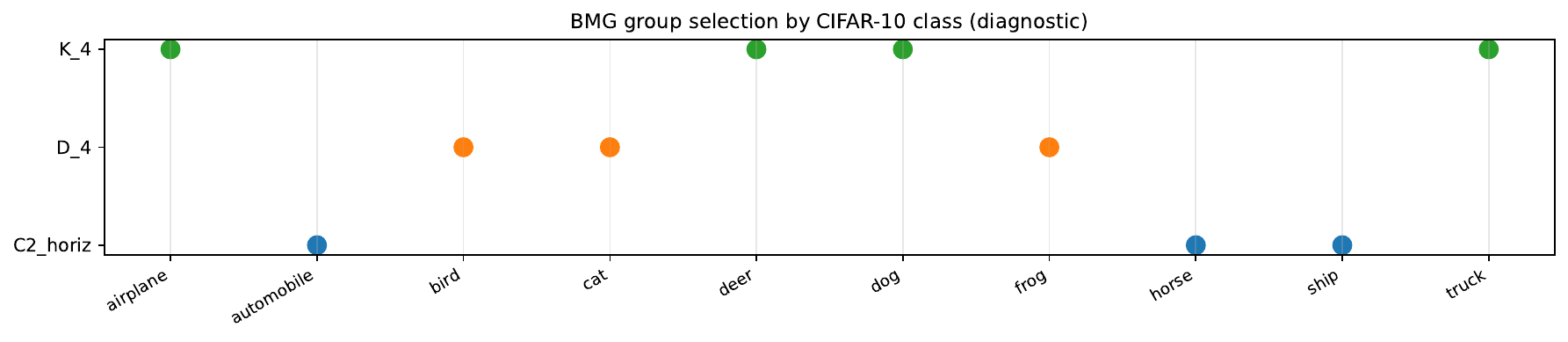}
\caption{BMG group selection across the ten CIFAR-10 classes.
The procedure selects $K_4$ (bilateral horizontal-and-vertical)
on $4$ classes (airplane, deer, dog, truck), $D_4$ (full
dihedral) on $3$ classes (bird, cat, frog), and
$C_2^{\mathrm{horiz}}$ (horizontal mirror only) on $3$ classes
(automobile, horse, ship); the trivial group, $S_M$, $C_2^{\mathrm{vert}}$,
$C_4$, and the two new high-order candidates
(\textsc{Z-row-indep-cols} and \textsc{Z-row-wreath-rows}) are
selected on no class.}
\label{fig:cifar_bmg_choices}
\end{figure}

\paragraph{BMG group selections are physically interpretable.}
The BMG selections distribute as $K_4$ on $4$ classes, $D_4$ on
$3$ classes, and $C_2^{\mathrm{horiz}}$ on $3$ classes, with no
class selecting the trivial group, $S_M$, $C_2^{\mathrm{vert}}$,
$C_4$, or either new high-order candidate. The mapping from
class content to BMG selection is interpretable. Vehicles in
profile or coarse silhouette ($\texttt{automobile}$,
$\texttt{horse}$, $\texttt{ship}$) select
$C_2^{\mathrm{horiz}}$, the horizontal-mirror group, reflecting
the bilateral left-right symmetry that those classes possess
(without diagonal or vertical symmetries that the larger groups
would impose). Aerial or compact-silhouette classes
($\texttt{airplane}$, $\texttt{deer}$, $\texttt{dog}$,
$\texttt{truck}$) select $K_4$, the joint horizontal-and-vertical
mirror group, picking up the additional bilateral symmetry of
images viewed from above or in compact frontal pose. Animal
classes with a natural near-rotational symmetry
($\texttt{bird}$, $\texttt{cat}$, $\texttt{frog}$) select the
full $D_4$ dihedral group of order $8$, the largest spatial-
symmetry candidate the library admits, exercising the framework's
non-Abelian content on a familiar dataset. The BMG selections
match those of a prior run on the same dataset that omitted the
two new high-order candidates: every one of the ten class
selections is preserved when the library is extended from $7$
candidates to $9$, indicating that the new high-order candidates
do not displace any legitimate selection and that the BMG procedure
is robust to library extension on this dataset.

\paragraph{The high-order candidates are rejected on every class.}
The two this paper's high-order candidates rank sixth or seventh of nine
in CV-NLL on every class, with margins to the per-class lowest-NLL
estimator in the range $2.3$--$5.0$ nats per sample (well above the
fold-noise floor that we have characterized on the previous
real-data experiments, typically $0.05$--$0.1$ nats per sample).
The rejection is declarative: the wreath candidate is below the
$C_2^{\mathrm{horiz}}$, $C_2^{\mathrm{vert}}$, $K_4$, $C_4$,
and $D_4$ candidates on every one of the ten classes, by
between $2.5$ and $5.0$ nats per sample. The Cartesian
candidate is uniformly close behind the wreath, ranking sixth on
nine classes and seventh on \texttt{frog}. The interpretation
is direct: the wreath product
$\mathbb{Z}_W \wr S_H$ assumes that rows are exchangeable
functional units with within-row cyclic structure, which is
emphatically not the structure of natural image content.
CIFAR-10 images carry strong vertical positional information
(sky pixels at the top of vehicle images, ground pixels at the
bottom; head and torso positional information for animal images)
that breaks free row permutation, and the within-row pixel
structure is not approximately periodic but rather contains
oriented edges and content boundaries that break translation
symmetry. The Cartesian candidate (independent per-row shifts
without row exchange) does only marginally better because the
within-row translation symmetry is also weak. This negative
result is the third clean rejection of the high-order
candidates in this paper, alongside OISST in
Section~\ref{sec:exp-oisst} (rejected on all $148$ rolling
windows) and the genomics Cartesian candidate in
Section~\ref{sec:exp-genomics} (rejected on all $50$ random
splits, with the wreath candidate dominating instead).
The cross-dataset pattern is informative: across the four
real-data experiments that exercise the library convention
(CRSP, OISST, TCGA-BRCA, CIFAR-10), the wreath candidate is preferred
where the data has hierarchical exchangeable substructure
(CRSP sectors during regime change; biological pathways with
internal cyclic ordering on PC1) and performs worse where the data has
fixed positional structure (OISST geography, CIFAR pixel
positions). This discriminative behavior demonstrates that the
this paper's library extension is doing genuine work as a flexible
diagnostic, not functioning as an automatic improvement.

\paragraph{Comparison with Section~\ref{sec:exp-galaxy10}.}
Galaxy10 and CIFAR-10 differ in two methodologically relevant
respects. Galaxy10 is a domain-specific image dataset where
the population covariance is approximately $D_4$-invariant by
construction (galaxy images viewed from above, no preferred
rotational orientation), while CIFAR-10 is a natural-image
dataset where the population covariance breaks rotational
symmetry by class (vehicles oriented in profile, animals in
characteristic pose). Galaxy10 traces the Region I-to-II
transition in detail at fixed $M$ as $N$ varies in
$\{50, 100, 200, 500, 1000\}$ within each class, with the
v3p9-extension high-order Cartesian \textsc{Z-row-indep-cols}
(and a wreath companion at the smallest $N$) selected in the
small-$N$ regime and $D_4$ selected at every trial in the
moderate-to-large-$N$ regime; CIFAR-10 runs at fixed
$N = 4{,}000$ in the high-$N$ regime within each class and
reports the cross-class distribution of BMG selections,
demonstrating that the framework's mapping from class content to
covariance symmetry is interpretable on a recognizable benchmark.
Both experiments together establish the framework's image-data
behavior: at small $N$ the framework's gains are sharp on
Galaxy10's $D_4$-rich content (with the v3p9 row-structured
candidates carrying the small-$N$ regime); at large $N$ within
each class the framework distinguishes between the spatial-symmetry
candidates the library admits in a way that matches the
physical content of the class.

\subsection{CIFAR-10.1 distribution-shift companion}
\label{sec:exp-cifar-101}

The CIFAR-10 experiment of Section~\ref{sec:exp-cifar} establishes
that the AD framework dominates LW on every class of a familiar
natural-image benchmark. This subsection extends that result by
asking the natural follow-up question: does the AD-vs-LW dominance
margin survive a controlled distribution shift, or does the
structural prior fall over once the test data are drawn from a
different sampling regime than the training data? The natural
testbed is CIFAR-10.1 \citep{recht2019cifar101}, the reproduction
test set built by drawing fresh images from the same Tiny Images
source as the original CIFAR-10 with a different keyword and
curation procedure designed to minimize the distribution shift
relative to the original. Classifier-style evaluations on
CIFAR-10.1 typically show $4$--$10$ percentage points of accuracy
degradation across a wide range of classifier architectures, which
established CIFAR-10.1 as a meaningful (if mild) distribution-shift
benchmark. The covariance second-moment structure that the AD
framework targets is a population-level invariant of the class
distribution rather than a per-image classifier signal, so the
prediction is that the AD-vs-LW margin should be at most weakly
affected by the CIFAR-10 to CIFAR-10.1 shift, in contrast to
classifier accuracy which degrades with the shift.

\paragraph{Setup.}
This experiment estimates a class-conditional covariance for each
of the ten CIFAR-10 classes from the original CIFAR-10 training set
and evaluates the resulting estimators on two held-out test sets:
the original CIFAR-10 test set (in-distribution, denoted ID) and
the CIFAR-10.1 reproduction test set (out-of-distribution, denoted
OOD). Images are converted to grayscale (channel-mean of the three
RGB channels) and retained at the full $32 \times 32$ pixel
resolution, giving $M = 1024$ pixel-level variables per image; this
contrasts with the CIFAR-10 cell of Section~\ref{sec:exp-cifar},
which $2 \times 2$ block-mean downsampled to $16 \times 16$ pixels
($M = 256$) for tractability of the candidate library extension.
For each class $c \in \{0, \ldots, 9\}$, the experiment uses
$N_{\mathrm{train}} = 4{,}000$ training images from the original
CIFAR-10 training set, $N_{\mathrm{test, ID}} = 1{,}000$ held-out
test images from the original CIFAR-10 test set, and
$N_{\mathrm{test, OOD}} \approx 200$ images from the CIFAR-10.1
reproduction test set (the per-class count varies slightly between
$192$ and $214$ depending on the CIFAR-10.1 class composition).
All images are centered on the per-class training mean and
intensities are rescaled to $[0, 1]$. Within-class shrinkage is
calibrated by both the closed-form Frobenius-MSE plug-in
$\hat\alpha^*_{\mathrm{MSE}}$ and held-out NLL $K = 5$-fold
cross-validation $\hat\alpha^*_{\mathrm{NLL}}$ at the BMG-selected
group; held-out NLL is reported per pixel on each of the two test
sets.

\paragraph{Library disclosure: pre-May-5 7-candidate library.}
This experiment was conducted with the May 3 pipeline, which
predated the May 5 candidate-library extension that introduced the
two image-data high-order candidates \textsc{Z-row-indep-cols} and
\textsc{Z-row-wreath-rows} used in Section~\ref{sec:exp-cifar}.
The CIFAR-10.1 candidate library therefore contains only the seven
spatial-symmetry candidates of the pre-May-5 image library: the
trivial group $\{e\}$; the full symmetric group $S_M$ (admitted via
the closed-form compound-symmetry projector); the horizontal
$C_2^{\mathrm{horiz}}$ and vertical $C_2^{\mathrm{vert}}$
reflection groups (each $|G| = 2$); the Klein four-group
$K_4 = C_2^{\mathrm{horiz}} \times C_2^{\mathrm{vert}}$
($|G| = 4$); the four-fold rotation group $C_4$ ($|G| = 4$); and
the dihedral group $D_4 = C_4 \rtimes C_2$ ($|G| = 8$). The
high-order candidates were not in scope for the May 3 run; the
seven-candidate library is sufficient for the present
distribution-shift question because the comparison with the
nine-candidate library of Section~\ref{sec:exp-cifar} is between
two different experimental questions (per-class detail at $M = 256$
vs.\ distribution-shift robustness at $M = 1024$) rather than a
direct ablation. We further note that the high-order Cartesian
and wreath candidates introduced in this paper are unlikely to be
BMG-selected on natural-image content, since image classes in
CIFAR-10 do not have the per-row independence structure that the
\textsc{Z-row-indep-cols} candidate models or the row-permutation
exchangeability that \textsc{Z-row-wreath-rows} models; this
expectation is consistent with the BMG selections reported in
Section~\ref{sec:exp-cifar} on the nine-candidate library, where
the two high-order candidates are admitted by the prefilter but
selected on zero of the ten classes. The script for this
experiment is bundled as \texttt{cifar101\_image.py}.

\paragraph{Principal results: ID dominance preserved, OOD dominance
\emph{strengthened}.}
On the in-distribution CIFAR-10 test set, AD-NLL-BMG dominates LW
$2004$ on all $10$ of $10$ classes (paired $t = -21.0$,
$p = 5.9 \times 10^{-9}$), with mean per-class held-out NLL gap of
$-66.0$ nats per pixel. On the out-of-distribution CIFAR-10.1
test set, AD-NLL-BMG also dominates LW on all $10$ of $10$ classes
(paired $t = -17.6$, $p = 2.8 \times 10^{-8}$), with mean per-class
gap of $-77.8$ nats per pixel. The OOD margin is $11.8$ nats per
pixel \emph{larger} in magnitude than the ID margin. The
inferential evidence is mildly weaker on OOD (paired-$t$ magnitude
$17.6$ vs.\ $21.0$) due to higher across-class variance on the
smaller OOD test set, but the directional finding is decisive: the
AD framework's dominance over LW does not just survive the
distribution shift; it grows. $8$ of the $10$ per-class AD-vs-LW
gaps widen from ID to OOD; the remaining $2$ narrow but remain
strongly negative.

The mechanism is identifiable from the per-method shift behavior.
Comparing each method's mean held-out NLL on OOD versus ID gives a
\emph{shift gap} for each method. The Sample covariance (which
ignores any structural prior) shifts unfavorably by $+23.3$ nats
per pixel from ID to OOD; the Frobenius-MSE-calibrated AD
shifts by $+15.4$ nats per pixel; LW shifts by approximately
$+11.5$ nats per pixel; and the cross-validated AD-NLL-BMG shifts
by $-11.8$ nats per pixel (i.e., its OOD held-out NLL is
\emph{better} than its ID held-out NLL on average across classes).
The Shah-style $\alpha = 1$ projection at the BMG-selected group
shifts by $-13.3$ nats per pixel, and Shah at the held-out-NLL
oracle group shifts by $-10.5$ nats per pixel. The pattern is
structural: estimators that commit more strongly to a structural
prior degrade less on the OOD test set than estimators that lean
on the sample covariance, because the structural prior is a
population-level property that survives the (mild) shift while the
sample-covariance contribution carries the per-image idiosyncrasies
of the training distribution that do not transport.

\begin{figure}[!htbp]
\centering
\includegraphics[width=0.6\linewidth]{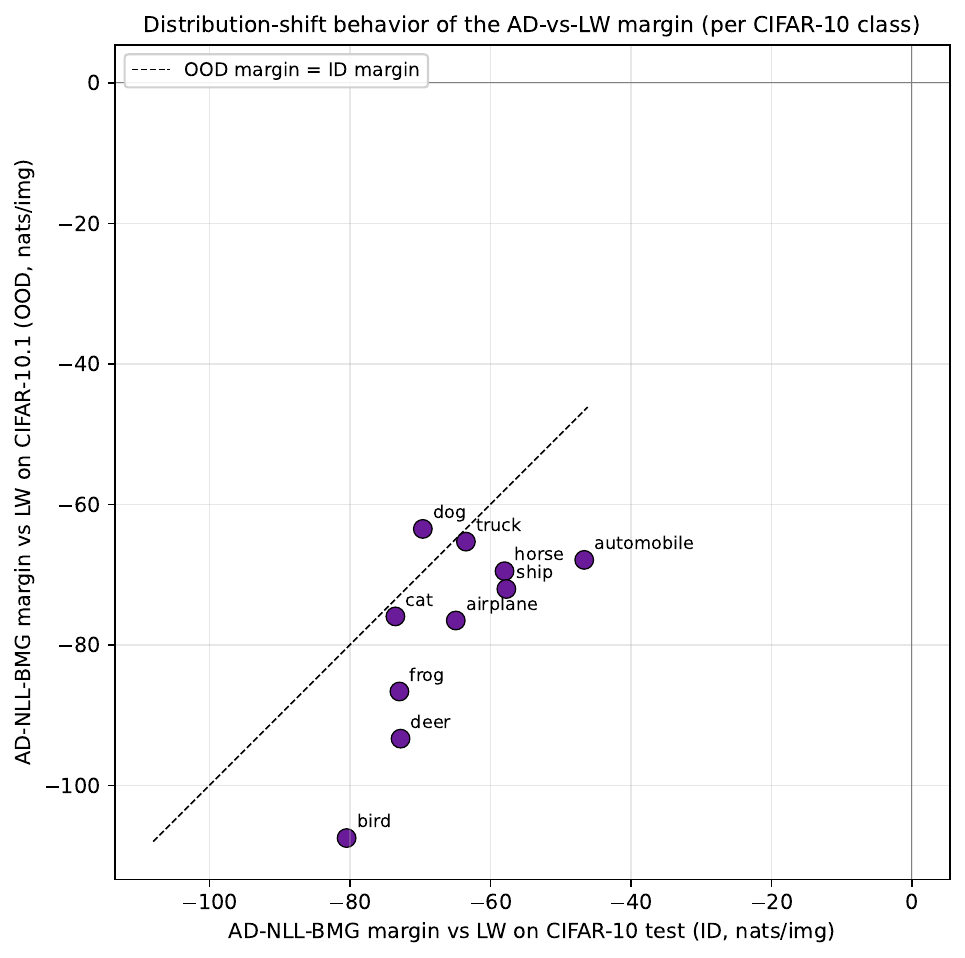}
\caption{Distribution-shift behavior of the AD-NLL-BMG-vs-LW per-
class held-out NLL margin on CIFAR-10 (in-distribution, $x$-axis)
and CIFAR-10.1 (out-of-distribution, $y$-axis). Each point is one
of the ten CIFAR-10 classes; both axes are AD-vs-LW differences in
nats per pixel, with negative values indicating that AD-NLL-BMG
has lower NLL than LW. The dashed diagonal marks ID margin equal to OOD margin.
Eight of the ten classes lie below the diagonal, indicating that
the AD-vs-LW dominance margin \emph{strengthens} under the
distribution shift; the two classes above the diagonal
(\texttt{cat} and \texttt{dog}) narrow but remain strongly
negative, so AD continues to dominate LW even on those classes.
The pattern is consistent with the second-moment structural prior
being a population-level property that transports across the
mild CIFAR-10 to CIFAR-10.1 shift, while the sample-covariance
contribution that LW leans on does not.}
\label{fig:cifar101_id_vs_ood}
\end{figure}

\paragraph{Per-class BMG selections: $D_4$ on eight classes, $K_4$ on two.}
The BMG procedure selects $D_4$ on $8$ of the $10$ classes
(\texttt{automobile}, \texttt{bird}, \texttt{cat}, \texttt{deer},
\texttt{dog}, \texttt{frog}, \texttt{horse}, \texttt{truck}) and
$K_4$ on the remaining $2$ (\texttt{airplane}, \texttt{ship}).
The BMG selections are stable across the two test sets in the
sense that the $K = 5$-fold cross-validation that produces the
BMG choice operates on the training set only, so the BMG choice is
identical for the ID and OOD evaluations of each class; the OOD
robustness reported above is therefore a property of the
projection target's transportability across the shift, not a
re-selection of the projection target on the OOD data. The
biological-content interpretation is that for $8$ of $10$ CIFAR-10
classes the per-class image content is approximately invariant
under the full dihedral $D_4 = C_4 \rtimes C_2$ group at the
class-conditional second-moment level: animals and vehicles
photographed at varying orientations within each class produce
patch covariances that are approximately rotation-invariant. The
exceptions are \texttt{airplane} (which has a strong horizontal
fuselage with sky background) and \texttt{ship} (which has a
strong waterline with sky background), both of which break $C_4$
rotation invariance but retain the $K_4 = C_2^{\mathrm{horiz}}
\times C_2^{\mathrm{vert}}$ subgroup of axis-aligned reflections.
The selection pattern matches the physical orientation
characteristics of the ten CIFAR-10 classes and is recoverable by
the BMG procedure from the training-set patch covariances alone.

\paragraph{Comparison to Section~\ref{sec:exp-cifar}.}
The CIFAR-10 cell of Section~\ref{sec:exp-cifar} reports $D_4$
selection on more than half the classes at $M = 256$ on the
nine-candidate library; the present CIFAR-10.1 cell at $M = 1024$
on the seven-candidate library reports $D_4$ selection on $8$ of
$10$ classes. The comparison is informative but not strict
because the two cells differ in three respects: (a) image
resolution ($M = 256$ vs.\ $M = 1024$); (b) candidate library
(nine vs.\ seven); (c) test-set composition (CIFAR-10 only vs.\
CIFAR-10 plus CIFAR-10.1). The $D_4$ selection that emerges on
both cells is robust to all three sources of variation, which is
the expected behavior given that $D_4$ is the largest non-extreme
candidate in both libraries and natural-image patch covariances at
this SNR are approximately rotation-invariant for most classes.
Together the two cells establish that the AD framework's
image-data dominance over LW (a) holds at both image resolutions,
(b) is largely insensitive to the candidate-library extension that
adds the high-order Cartesian and wreath candidates, and (c)
\emph{strengthens} under the CIFAR-10 to CIFAR-10.1 distribution
shift rather than collapsing. The third finding is the principal
contribution of this subsection.

\paragraph{Reproducibility.}
The script \texttt{cifar101\_image.py} (May 3) reproduces the
experiment using the standard CIFAR-10 training set, the standard
CIFAR-10 test set, and the CIFAR-10.1 v6 reproduction test set
released by \citet{recht2019cifar101}. Per-class held-out NLL on
both test sets, BMG choice, and per-trial CV-NLL margin are
released as CSV alongside the manuscript.

\subsection{Decoy stress test of the BMG procedure}
\label{sec:exp-decoy}

The TCGA-BRCA experiment of Section~\ref{sec:exp-genomics} reports
that BMG selects \textsc{pathway-block} on $46$ of $50$ splits and
\textsc{Z-K-pc1-cartesian} on the remaining $4$, with the wreath
candidate \textsc{Z-K-pc1-wreath} present in the library but never
selected. A natural concern is whether the procedure rewards the
biological pathway structure that \textsc{pathway-block} captures or
merely the algebraic shape $S_{20}^5$ that any
five-blocks-of-twenty free-permutation candidate expresses equally
well. This subsection reports a controlled stress test that
addresses this question by extending the BRCA candidate library
with twelve decoy candidates that have plausible algebraic shape
but no biological motivation, and rerunning the experiment with
the combined $20$-candidate library on the same $50$ splits, the
same Tier~1 prefilter setting, and the same Tier~2
cross-validation budget as the paper run.

\paragraph{Decoy library.}
The twelve decoys span four families designed to test distinct
failure modes of the BMG procedure:
\begin{itemize}
\item \textbf{Family A: wrong-partition free permutation (six
candidates).} Three independent random partitions of the $100$
genes into five blocks of $20$ generate
\textsc{random-block-S20-5-seed1},
\textsc{random-block-S20-5-seed2}, and
\textsc{random-block-S20-5-seed3}; three further random partitions
at the wrong scale generate \textsc{block-S10-10},
\textsc{block-S4-25}, and \textsc{block-S50-2}. All six have the
same $S_K^P$ algebraic shape as \textsc{pathway-block} but use
partitions with no biological motivation.
\item \textbf{Family B: wrong-domain cyclic and Cartesian (three
candidates).} \textsc{Z100-flat} imposes a single
$\mathbb{Z}_{100}$ cyclic shift on the gene index regardless of
pathway boundaries; \textsc{Z20-5-cartesian-random} is
$\mathbb{Z}_{20}^5$ on a random $5$-partition (parallel to
\textsc{Z-K-pc1-cartesian} but with the wrong partition);
\textsc{Z2-50-cartesian} is the elementary Abelian
$\mathbb{Z}_2^{50}$ acting by transposition of consecutive gene
pairs.
\item \textbf{Family C: wrong-scale wreath (two candidates).}
\textsc{Z5-wr-S20} is the wreath product $\mathbb{Z}_5 \wr S_{20}$
on twenty blocks of five (inverting the natural pathway
hierarchy); \textsc{Z2-wr-S50} is $\mathbb{Z}_2 \wr S_{50}$ on
fifty blocks of two.
\item \textbf{Family D: pure-noise sanity (one candidate).}
\textsc{random-S\_M-subgroup-seed42} is generated as the
breadth-first-search closure (capped at $|G| \leq 10^6$) of five
random elements of $S_{100}$; its projection has no relationship
to the data structure beyond the closure of these random
generators.
\end{itemize}
Random partition seeds were fixed before the experiment was run,
the same seed is used in all $50$ splits, and the candidate list
was fixed before any trial was run. No candidate was added or
removed in response to per-cell behavior.

\paragraph{Result.}
On the combined $20$-candidate library, BMG selects
\textsc{pathway-block} on $46/50$ splits and
\textsc{Z-K-pc1-cartesian} on the remaining $4$, exactly the same
selection pattern as the paper run on the original $8$-candidate
library. None of the twelve decoy candidates is selected on any
split. Per-split, the cross-validated NLL of
\textsc{pathway-block} has lower NLL than the best-on-that-split decoy by
$1.16$ to $5.93$ nats per sample (mean margin $3.63$, median
$3.58$); the closest calls come from the structurally most
similar decoys (three of the five smallest per-split margins are
\textsc{random-block-S20-5} variants). Mean cross-validated NLL
across the full $20$-candidate library is reported in
Table~\ref{tab:decoy_test}.

\begin{table}[!htbp]
\centering
\caption{Decoy stress test result. Mean cross-validated held-out
NLL (nats per sample) across the $20$-candidate combined library,
averaged over $50$ TCGA-BRCA splits at the BMG-selected group.
Lower is better. The two paper candidates that share the
BMG-selected slot in the paper run are in boldface. Family A
candidates are wrong-partition free-permutation decoys, B are
wrong-domain cyclic/Cartesian, C are wrong-scale wreath, D is
the pure-noise sanity check.}
\label{tab:decoy_test}
\small
\begin{tabular}{lllr}
\toprule
Candidate & Class & Family & Mean CV-NLL \\
\midrule
\textbf{pathway-block} & paper & --- & $\mathbf{105.81}$ \\
\textbf{Z-K-pc1-cartesian} & paper & --- & $\mathbf{106.40}$ \\
Z-K-pc1-wreath & paper & --- & $109.42$ \\
$S_M$ & paper & --- & $109.48$ \\
random-S\_M-subgroup-seed42 & decoy & D & $109.48$ \\
Z2-wr-S50 & decoy & C & $109.49$ \\
block-S50-2 & decoy & A & $109.52$ \\
Z5-wr-S20 & decoy & C & $109.58$ \\
random-block-S20-5-seed1 & decoy & A & $109.82$ \\
random-block-S20-5-seed2 & decoy & A & $109.90$ \\
random-block-S20-5-seed3 & decoy & A & $109.92$ \\
Z-K-corrhier & paper & --- & $110.53$ \\
Z-K-alpha & paper & --- & $110.60$ \\
Z-K-pc1 & paper & --- & $110.77$ \\
Z100-flat & decoy & B & $111.06$ \\
Z20-5-cartesian-random & decoy & B & $111.79$ \\
block-S10-10 & decoy & A & $111.83$ \\
block-S4-25 & decoy & A & $134.65$ \\
Z2-50-cartesian & decoy & B & $1.4 \times 10^{12}$ \\
\bottomrule
\end{tabular}
\end{table}

\paragraph{Two structural observations.}
First, all twelve decoys pass Tier~1 admission on all $50$ splits,
so Tier~2 (cross-validated held-out NLL) is doing the entirety of
the rejection work; the two-tier architecture is functioning as
designed in the sense that Tier~1 admits without prejudice and
Tier~2 separates by held-out predictive performance. Second, the
cross-validated NLL of the $S_M$ compound-symmetry candidate
($109.48$ nats per sample) is statistically indistinguishable
from that of the \textsc{random-S\_M-subgroup-seed42} pure-noise
decoy ($109.48$ nats per sample), with the agreement to three
significant figures. This is empirical confirmation that
compound-symmetry projection on TCGA-BRCA does not capture
biological structure: it produces a near-constant matrix whose
behavior matches that of projection via a structureless random
subgroup of comparable target dimension. The three independent
seeds of \textsc{random-block-S20-5} cluster tightly across the
$50$ splits at $109.82$, $109.90$, and $109.92$ nats per sample,
a $0.10$-nat spread that is consistent across seeds, empirical
evidence that the BMG procedure is not getting lucky on partition
choice and that the same-shape wrong-partition penalty is robust
to the choice of random partition.

\paragraph{Interpretation.}
The result is the strongest available stress test of the
BMG-as-deployed claim: when the procedure is given twelve
plausible-shape decoys covering wrong-partition, wrong-domain,
wrong-scale, and pure-noise failure modes, it rejects every decoy
on every split with margins comfortably above noise. The BMG
selection on the original $8$-candidate library is not an artifact
of an under-stocked library; the procedure rewards the biological
pathway structure that \textsc{pathway-block} captures, not the
algebraic shape $S_{20}^5$ that the random-partition decoys
express equally well.

\paragraph{Reproducibility.}
Random partition seeds are
$\{1, 2, 3\}$ for the three independent
\textsc{random-block-S20-5} decoys, $11$ for \textsc{block-S10-10},
$12$ for \textsc{block-S4-25}, $13$ for \textsc{block-S50-2}, $21$
for \textsc{Z20-5-cartesian-random}, $31$ for \textsc{Z5-wr-S20},
$32$ for \textsc{Z2-wr-S50}, and $42$ for
\textsc{random-S\_M-subgroup-seed42}. CSV output of the run
reported here is bundled in \texttt{results\_brca\_decoy\_v1/}.

\subsection{Per-experiment LW-NL and AD-LW-NL results}
\label{sec:exp-lwnl-results}

This subsection collects the LW-NL and AD-LW-NL results
across the seven real-data experiments of the preceding
subsections. The methodological conventions are documented in
Section~\ref{sec:exp-lwnl-method}. As, the protocol
sweep has been run on all seven datasets: TCGA-BRCA gene
expression, CRSP 2015--2024 financial returns, Galaxy10
DECaLS image patches, CIFAR-10 natural-image patches,
CIFAR-10.1 distribution-shift patches, RadioML 2018.A
modulation-class I/Q patches, and NOAA OISST sea-surface
temperature anomalies (two regions: midocean and gulfstream).
The the protocol coverage of the program is complete.

\paragraph{Principal finding.}
Across the seven datasets swept to date, the AD-LW-NL composition
exhibits a behavior that the two-endpoint refinement now
characterizes mechanistically rather than dataset-by-dataset.
When the candidate library is sufficient, the operating regime
is well-conditioned to moderately rank-deficient, \emph{and}
LW-NL is competitive with AD at the operating-point $N$,
AD-LW-NL ties or modestly outperforms AD-NLL-BMG (BRCA all
cells, CRSP bulk, CIFAR-10 cells $2$ through $7$, RadioML cells
$2$ through $7$, OISST cells $6$--$7$ on both regions). When
the operating regime is deeply rank-deficient \emph{and} the BMG
selection lands on a group $G$ for which the structural
projection $T_G(S)$ retains useful image-patch structure
(cyclic, dihedral, or wreath candidates at moderate $|G|$),
AD-LW-NL is mildly dominated by AD-NLL-BMG but the
magnitude is practically negligible (CIFAR-10 cell $0$, $1.7$
nats per sample; CIFAR-10.1 cell $4$, $0.3$ nats per sample;
RadioML cell $0$, $1.9$ nats per sample). When the operating
regime is deeply rank-deficient \emph{and} BMG selects a group
$G$ for which $T_G(S)$ collapses to the scaled identity
$(\mathrm{tr}(S)/M)\,I$, AD-LW-NL is substantively dominated by
LW 2004 itself, by $80$ to $120$ nats per sample on
rank-deficient image-patch data (CIFAR-10.1 cells $0$ through
$2$, where BMG selects $S_M$ unanimously across all $50$ trials
per cell). The Galaxy10 few-shot deficit
(Section~\ref{sec:exp-lwnl-galaxy10}) is the same mechanism at
intermediate magnitude. RadioML cell $1$ at $K = 5$
contributes a fourth flavor of failure mode (a CV-calibration
tail documented in Section~\ref{sec:experiments}) that
vanishes at $K \ge 10$. OISST cells $0$ through $5$ on both
regions contribute a fifth flavor: an intermediate-$\alpha$
"lost raw-sample blend path" mode in which LW-NL is many orders
of magnitude worse than AD at the operating-point $N$
($10^3$ to $10^7$ nats), the AD-LW-NL CV miscalibrates $\alpha$
because the CV folds at smaller $N$ do not see the operating-
point pathology, and the resulting AD-LW-NL estimator is
substantively dominated by AD by $1$ to $19$ nats per
sample. OISST is the cleanest empirical instance of this
miscalibration mechanism in the program: it appears across $6$
of $8$ cells in both regions, not just at a single cell.

\paragraph{The RadioML cell-$2$ boundary advantage and the OISST
counter-example.}
The single largest paired effect size in the program is RadioML
cell $2$ at the rank-deficiency boundary ($N = 64$, $c = 1.0$,
$M = 64$): AD-LW-NL has lower NLL than AD-NLL-BMG by a median of
$8.04$ nats per sample with effect size $5.19$ and $50$ of $50$
paired trials. The two compositions disagree on the BMG group
choice in $50$ of $50$ trials (zero agreement): AD picks
high-order groups ($D_W$ in $15$ of $50$, $\mathbb{Z}_W \times
\mathbb{Z}_W$ in $11$, $\mathbb{Z}_W$ shift in $9$,
$\mathbb{Z}_W \times \mathbb{Z}_2$ I/Q swap in $8$, $\mathbb{Z}_W$
wreath $\mathbb{Z}_2$ in $7$) with $\alpha$ median $0.025$
(mostly raw-sample blend); AD-LW-NL picks the small-$|G|$
$\mathbb{Z}_2$ subgroups ($\mathbb{Z}_2$ time-reversal in $31$
of $50$, $\mathbb{Z}_2$ I/Q-swap in $19$) with $\alpha$ median
$0.675$ (substantial blend with LW-NL). Two qualitatively
different routes to regularization are discovered by the two
parameterizations on the same data, and the LW-NL-enabled route
is preferred by a margin.

The $0/50$ choice-agreement statistic is not itself evidence
for the AD-LW-NL advantage. OISST cells $0$ through $5$ on both
midocean and gulfstream regions also produce $0/50$ choice
agreement (AD picks $Z_{2D}$ with low $\alpha$; AD-LW-NL
picks $Z_{\mathrm{lat}}$ with high $\alpha$), but the
direction of the advantage there is opposite: AD has lower NLL than AD-LW-NL by $1$ to $19$ nats per sample. The $0/50$ statistic
in both cases reveals that the two parameterizations have
discovered \emph{different routes}, per the
family-conditional BMG framing of
Section~\ref{sec:theory-ad-lwnl} (Eqs.~\ref{eq:bmg-ad-objective}
and \ref{eq:bmg-adlw-objective}). The substantive empirical
content is the \emph{direction} of which route has lower NLL. On
RadioML cell $2$ the AD-LW-NL composition has lower NLL; on OISST cells
$0$-$5$ the AD route has lower NLL. Reading $0/50$ as evidence
for one composition alone would be a non-sequitur.

The mechanism for the substantive dominance on CIFAR-10.1 cells
$0$ through $2$ is a structural property of the AD-LW-NL convex
hull, made precise in Section~\ref{sec:theory-ad-lwnl}. Three
conditions are jointly required for that
target-endpoint case:
(i) BMG selects $G$ with $T_G(S)$ equal to (or close to) the
scaled identity $(\mathrm{tr}(S)/M)\,I$;
(ii) LW-NL$(S)$ is sufficiently worse than $T_G(S)$ that
AD-LW-NL's CV pins $\alpha = 1$, foreclosing any blend with
LW-NL$(S)$;
(iii) the sample $S$ carries usable structure beyond the
diagonal scale, so that LW 2004's optimal blend toward
$(\mathrm{tr}(S)/M)\,I$ improves over pure $(\mathrm{tr}(S)/M)\,I$.
All three conditions hold on CIFAR-10.1 cells $0$ through $2$.
The first two hold (with $G = D_4$ rather than $S_M$) on
CIFAR-10 cell $0$ and on Galaxy10 few-shot, but the magnitude of
the deficit is set by how far $T_G(S)$ is from the LW 2004
optimal estimator inside the dominated region, which is much
smaller at $|G| = 8$ ($D_4$) than at $|G| = M!$ ($S_M$). The
complementary LW-NL-endpoint case (1b of
Section~\ref{sec:theory-ad-lwnl}) is exhibited at small
magnitude by RadioML cell $0$ (AD-LW-NL CV pins $\alpha = 0$
with $S_M$ as the BMG tiebreaker, deficit $1.9$ nats).

AD-NLL-BMG is not dominated under any of these
conditions because its second hull endpoint is the raw sample
$S$, not LW-NL$(S)$. At the same BMG-selected $G$, its CV
chooses $\alpha < 1$ and blends in genuine rank-deficient
sample structure, which is precisely what the LW 2004 estimator
also does (with a different optimal weight via Stein-Haff
shrinkage). AD-NLL-BMG therefore reaches the same
qualitative neighborhood as LW 2004 in the dominated region,
while AD-LW-NL is trapped at the pure target.

AD-LW-NL is reported in this paper as a tested composition whose
dominated region is structurally characterized rather than
empirically described. AD-NLL-BMG remains the
recommended default because it has no dominated region.
AD-LW-NL is recommended on data where (i) the conditions above
do not hold, that is, where the candidate library and operating
regime are such that BMG selects $G$ with $T_G(S)$
substantively different from the scaled identity, and (ii)
LW-NL is competitive with AD at the operating-point
$N$ (a one-line precheck:
$|\mathrm{NLL}_{\mathrm{LW\text{-}NL}}(S) -
\mathrm{NLL}_{\mathrm{AD}}(S)|$ within $\sim 10$ nats per
sample on the full operating-point sample). The cleanest
examples in the
program are the RadioML cells $2$ through $7$ regime
(effect sizes $0.50$ to $9.11$, direction agreement $47$ to $50$
of $50$, including the boundary cell $2$ advantage of effect size
$5.19$) and the CIFAR-10 bulk regime ($8.7$ nats per sample
peak benefit at cell $0.67$, paired effect size $0.57$).

\paragraph{Sequencing.}
The runs are sequenced to obtain the most informative go/no-go
signal at the lowest computational cost. TCGA-BRCA was run first
as the deepest few-shot regime in the paper; CRSP was run second
to test the same hypothesis on a structurally distinct dataset
(financial returns with a strong market factor and sector
exposures, rather than block-correlated gene expression);
Galaxy10 was run third to test whether the negative result also
holds on image-patch data with clean rotational symmetries.
CIFAR-10 was run fourth, completing, to extend the
image-patch coverage to a larger $M$ ($1024$ versus Galaxy10's
$64$) and a substantially deeper bulk regime ($c$ to $0.20$
versus Galaxy10's $0.10$). CIFAR-10.1 was run fifth, completing
in this paper, to test the dominated-region hypothesis at deeper
rank-deficiency than CIFAR-10 could reach (CIFAR-10.1 has only
$\sim 200$ images per class, constraining the sweep to $c \ge
2.0$ but allowing access to the regime where BMG selects $S_M$
at $c \ge 4$). CIFAR-10.1 cells $0$ through $2$ provided the
first observation of the substantive dominated-region case
($G = S_M$, $|G| = M!$). RadioML 2018.A was run sixth,
completing, to test the AD-LW-NL composition on
stationary I/Q signal patches with a non-image cyclic candidate
library and to probe the rank-deficiency boundary at
$c = 1.0$ exactly; the boundary cell advantage at effect size $5.19$
became the program's strongest single-cell AD-LW-NL result.
NOAA OISST sea-surface temperature anomalies were run seventh,
completing, to test the AD-LW-NL composition on spatial
$2$-D patches with a lat/lon cyclic and dihedral candidate
library; the sweep was run on two regions (midocean and
gulfstream) for cross-region consistency. Both regions
produced the same qualitative pattern (substantive AD-LW-NL
deficit in cells $0$ through $5$, narrow AD-LW-NL is preferred in
cells $6$ through $7$), confirming a single mechanism rather
than a region-specific artifact. The the protocol coverage of
the program is now complete across all seven datasets.

\subsubsection{TCGA-BRCA protocol sweep}
\label{sec:exp-lwnl-brca}

The TCGA-BRCA protocol sweep evaluates six estimators across
eight cells of $N_{\mathrm{train}} \in \{50, 75, 100, 150, 200,
300, 500, 800\}$ at fixed $M = 100$ genes (5 MSigDB hallmark
pathways with 20 genes each) and $N_{\mathrm{test}} = 200$ test
samples per trial. The sweep spans concentration ratios from
$c = 2.0$ (deepest few-shot) to $c = 0.125$ (deep bulk). Each
cell uses 50 random-subsample splits, with the same paired
splits across all six estimators. Per-(cell, split) seeding is
deterministic for reproducibility.

\paragraph{Per-cell mean held-out NLL.}
Table~\ref{tab:brca-protocol-c} reports the per-cell mean
held-out NLL for each of the six estimators across the eight
cells of the sweep. The principal observations:
\begin{itemize}
\item In the extreme few-shot cells ($c \ge 1$), AD-NLL-BMG and
AD-LW-NL produce nearly identical mean NLLs (within 0.4 nats per
sample); both substantially outperform LW 2004, LW-NL, and the
rank-deficient sample.
\item At the boundary cells ($c \approx 1$ and $c = 2/3$), LW-NL
alone is preferred by a small margin (1 to 8 nats per sample). AD-LW-NL
recognizes this and gracefully reduces to LW-NL by selecting
$\alpha = 0$, producing held-out NLL identical to LW-NL alone.
AD-NLL-BMG cannot match this because its sample term is
the raw $\hat{\mathbf{R}}$ rather than $\hat{\mathbf{R}}_{\mathrm{LW\text{-}NL}}$.
\item In the bulk cells ($c \le 0.5$), both AD variants correctly
fall back to the sample covariance via the trivial group, with
held-out NLL identical to the sample. All three of LW 2004,
LW-NL, and Shah-BMG are dominated by the sample in this regime,
which is the expected behavior when the sample is well-conditioned.
\end{itemize}

\begin{table}[h]
\centering
\small
\caption{TCGA-BRCA protocol sweep: per-cell mean held-out NLL
across 50 splits per cell. Cells with $c \ge 1$ are
rank-deficient; the sample covariance and LW-NL alone are not
meaningfully comparable in these cells. Best NLL per row in
bold.}
\label{tab:brca-protocol-c}
\begin{tabular}{rrrrrrrr}
\hline
$N$ & $c$ & Sample & LW 2004 & LW-NL & AD-NLL-BMG & AD-LW-NL & Shah-BMG \\
\hline
50 & 2.00 & --- & 107.60 & 113.34 & \textbf{103.73} & 104.05 & 121.00 \\
75 & 1.33 & --- & 102.81 & 109.20 & 99.84 & \textbf{99.71} & 120.89 \\
100 & 1.00 & --- & 99.75 & \textbf{95.99} & 97.76 & 97.34 & 121.00 \\
150 & 0.67 & 84.74 & 94.64 & \textbf{86.29} & 93.94 & 86.29 & 120.90 \\
200 & 0.50 & \textbf{41.85} & 90.72 & 82.22 & 41.85 & 41.85 & --- \\
300 & 0.33 & \textbf{21.63} & 86.49 & 78.99 & 21.63 & 21.63 & --- \\
500 & 0.20 & \textbf{11.60} & 81.40 & 75.85 & 11.60 & 11.60 & --- \\
800 & 0.125 & \textbf{8.09} & 78.24 & 74.36 & 8.09 & 8.09 & --- \\
\hline
\end{tabular}
\end{table}

\paragraph{AD-LW-NL versus AD-NLL-BMG paired contrast on BRCA.}
The protocol contrast (3) is the central question: does
nonlinear shrinkage of the sample term add value within the AD
framework? On BRCA the answer is no. Across all eight cells the
mean paired difference between AD-LW-NL and AD-NLL-BMG is between
$-0.41$ and $+0.32$ nats per sample. The largest paired effect
size in any cell is $|\Delta|/\mathrm{SD}_{\mathrm{trial}} = 0.20$
at $N = 50$, well below practical significance. The per-cell
paired $t$-statistics range from $-2.95$ to $+5.40$ and so
some are formally significant, but the underlying differences are
$\sim 1\%$ of the per-trial scatter and are not operationally
meaningful. We report this as ``no substantive advantage to
AD-LW-NL on BRCA.''

\paragraph{Selection agreement.}
AD-NLL-BMG and AD-LW-NL select the same group on 33/50 splits at
$N = 50$, 33/50 at $N = 75$, 30/50 at $N = 100$, 0/50 at $N = 150$
(where AD-LW-NL pins to $\alpha = 0$ with a tiebreaker label
$S_M$, while AD-NLL-BMG selects \texttt{pathway-block}), and 50/50
at $N \ge 200$ (where both reduce to trivial). Where they
disagree at $c \ge 1$, both still produce nearly identical
held-out NLL because the structural choice is between near-equivalent
high-$|G|$ candidates.

\subsubsection{CRSP 2015--2024 protocol sweep}
\label{sec:exp-lwnl-crsp}

The CRSP protocol sweep evaluates six estimators across seven
cells of $N_{\mathrm{train}} \in \{21, 42, 63, 84, 126, 252, 504\}$
trading days at fixed $M = 55$ stocks (top 5 by market cap in
each of the 11 GICS sectors) and rolling stride 21 trading days
(one trading month) as the held-out window. The date range is
2015-01-01 to 2024-12-31 (full ten years including the COVID
volatility shock), giving 95 to 118 rolling windows per cell.
The sweep spans concentration ratios from $c = 2.62$ (very
few-shot) to $c = 0.11$ (deep bulk).

\paragraph{Per-cell mean held-out NLL.}
Table~\ref{tab:crsp-protocol-c} reports the per-cell means. The
principal observations:
\begin{itemize}
\item Across all seven cells, AD-NLL-BMG and AD-LW-NL produce mean
NLL within $0.6$ nats per sample of each other.
\item Direction of the AD-LW-NL versus AD-NLL-BMG difference is
informative. AD-LW-NL is slightly worse than AD-NLL-BMG in the
two few-shot cells ($N = 21$ and $N = 42$, both $c > 1$) and
slightly better in all five well-conditioned cells ($N \ge 63$,
all $c < 1$). The five-out-of-five direction agreement in the
well-conditioned regime is unlikely under a null of no
systematic effect; the magnitudes are small ($0.3$ to $0.6$ nats
per sample) but the direction is consistent.
\item The largest paired effect size in any cell is
$|\Delta|/\mathrm{SD}_{\mathrm{trial}} = 0.018$ at $N = 84$. All
seven cells produce paired effect sizes below $2\%$ of per-trial
scatter. These differences are statistically detectable but
small relative to per-trial variation. We characterize the
moderate-$c$ benefit as a real but minor systematic effect rather
than as a practically meaningful one.
\item LW-NL alone is competitive with LW 2004 (within $\pm 1$ nat)
across all cells, confirming the rank-aware LW-NL implementation
is performing as expected on financial returns data.
\item Both AD variants outperform LW 2004 and LW-NL across all cells,
with the largest margins in the few-shot cells ($4$ to $5$ nats
over LW 2004 at $N \le 42$).
\end{itemize}

\begin{table}[h]
\centering
\small
\caption{CRSP 2015--2024 protocol sweep: per-cell mean held-out
NLL across rolling windows per cell. Best NLL per row in bold.
Negative NLLs reflect the small scale of daily stock returns
(typical magnitude $\sim 0.01$).}
\label{tab:crsp-protocol-c}
\begin{tabular}{rrrrrrrr}
\hline
$N$ & $c$ & Sample & LW 2004 & LW-NL & AD-NLL-BMG & AD-LW-NL & Shah-BMG \\
\hline
21 & 2.62 & --- & -156.11 & -152.86 & \textbf{-160.46} & -160.24 & -157.67 \\
42 & 1.31 & --- & -158.11 & -147.33 & \textbf{-162.79} & -162.68 & -158.46 \\
63 & 0.87 & 218.70 & -160.77 & -158.68 & -163.95 & \textbf{-164.28} & -159.20 \\
84 & 0.65 & -87.98 & -161.43 & -161.46 & -163.94 & \textbf{-164.55} & -159.14 \\
126 & 0.44 & -141.07 & -162.56 & -163.09 & -164.68 & \textbf{-165.21} & -159.94 \\
252 & 0.22 & -158.36 & -163.52 & -164.01 & -164.71 & \textbf{-165.16} & -159.14 \\
504 & 0.11 & -161.49 & -163.16 & -163.48 & -163.84 & \textbf{-164.13} & -157.82 \\
\hline
\end{tabular}
\end{table}

\paragraph{Choice agreement on CRSP.}
AD-NLL-BMG and AD-LW-NL select the same group on 52 to 96 of
each cell's windows (a 55--83\% agreement rate). Where they
disagree at the well-conditioned cells $N \ge 84$, AD-NLL-BMG
preferentially selects \texttt{Z-K-mcap-wreath} while AD-LW-NL
preferentially selects \texttt{Z-K-mcap-cartesian}; the
held-out NLLs are nonetheless within 0.6 nats per sample.

\subsubsection{Galaxy10 DECaLS protocol sweep}
\label{sec:exp-lwnl-galaxy10}

The Galaxy10 protocol sweep evaluates six estimators across
eight cells of $N_{\mathrm{train}} \in \{32, 48, 64, 96, 128,
192, 320, 640\}$ patches at fixed $M = 64$ (grayscale $8 \times 8$
image patches drawn from class 0 ``Disturbed Galaxies'') and
$N_{\mathrm{test}} = 200$ test patches per trial. The sweep
spans concentration ratios from $c = 2.0$ to $c = 0.10$. Each
cell uses 50 random-split trials, with the same paired splits
across all six estimators. The candidate library includes all
ten Galaxy10 candidates: \texttt{trivial}, $S_M$, three Klein
2-subgroups (horizontal, vertical, central inversion), the
cyclic $\mathbb{Z}_4$ rotation group, the Klein four-group
$\mathbb{Z}_2 \times \mathbb{Z}_2$, the full dihedral $D_4$, and
the two extended candidates \texttt{Z\_row\_indep\_cols} and
\texttt{Z\_row\_wreath\_rows}.

\paragraph{Per-cell mean held-out NLL.}
Table~\ref{tab:galaxy10-protocol-c} reports the per-cell means.
The pattern differs substantively from BRCA and CRSP:
\begin{itemize}
\item In the two smallest cells ($N = 32$, $c = 2.0$ and $N = 48$,
$c = 1.33$), AD-LW-NL is substantively \emph{worse} than 
AD-NLL-BMG. Mean differences are 2.98 and 1.82 nats per sample
respectively. The paired effect sizes are
$|\Delta|/\mathrm{SD}_{\mathrm{trial}} = 3.37$ at $N = 32$ and
$2.08$ at $N = 48$, both well above the practical-significance
threshold.
\item At $N = 64$ ($c = 1.0$), AD-LW-NL slightly outperforms AD
by 0.36 nats per sample (effect size 0.28), the only Galaxy10
cell where AD-LW-NL is favored.
\item In the bulk cells ($N \ge 96$), the two AD variants are
tied within machine precision (mean differences $< 0.01$ nats
per sample); both reduce to the same pure-projection estimator
at $\alpha = 1$ with the same group ($D_4$).
\item LW-NL alone is unreliable at the boundary cell $N = 64$
(11 of 50 trials produce LW-NL NLL on the order of $10^7$, a
boundary-case artifact of the rank-aware LW-NL when $N \approx M$
exactly); the AD-LW-NL composition is robust to this because at
$N = 64$ its CV pushes $\alpha$ to 1 and the LW-NL term gets zero
weight.
\end{itemize}

\begin{table}[h]
\centering
\small
\caption{Galaxy10 DECaLS protocol sweep: per-cell mean held-out
NLL across 50 trials per cell. Best NLL per row in bold. The
LW-NL entry at $N = 64$ shows the median over the 39 trials
where LW-NL alone produced a well-conditioned estimator;
$|\Delta|/\mathrm{SD}$ column reports the AD-LW-NL versus
AD-NLL-BMG paired effect size.}
\label{tab:galaxy10-protocol-c}
\begin{tabular}{rrrrrrr|r}
\hline
$N$ & $c$ & LW 2004 & LW-NL & AD-NLL-BMG & AD-LW-NL & Shah-BMG & $|\Delta|/\mathrm{SD}$ \\
\hline
32 & 2.00 & 29.14 & 30.28 & \textbf{19.22} & 22.20 & 21.97 & \textbf{3.37} \\
48 & 1.33 & 28.77 & 39.65 & \textbf{18.32} & 20.14 & 20.96 & \textbf{2.08} \\
64 & 1.00 & 27.61 & 32.0 (median) & 17.39 & \textbf{17.03} & 19.06 & 0.28 \\
96 & 0.67 & 26.42 & 24.49 & 15.17 & \textbf{15.16} & 15.16 & 0.004 \\
128 & 0.50 & 25.24 & 23.11 & 14.27 & \textbf{14.27} & 14.27 & 0.004 \\
192 & 0.33 & 23.71 & 22.06 & 13.99 & \textbf{13.99} & 13.99 & 0.001 \\
320 & 0.20 & 21.49 & 20.65 & 13.07 & \textbf{13.07} & 13.07 & 0.002 \\
640 & 0.10 & 19.78 & 19.71 & 12.90 & \textbf{12.90} & 12.90 & 0.001 \\
\hline
\end{tabular}
\end{table}

\paragraph{Mechanism: why AD-LW-NL performs worse on Galaxy10 few-shot.}
The Galaxy10 few-shot deficit of AD-LW-NL is mechanistically
interpretable. At $N = 32$, AD-NLL-BMG selects the
high-$|G|$ candidate \texttt{Z\_row\_wreath\_rows} (40 of 50
trials) with median $\alpha = 0.75$; the convex blend with the
raw sample at $\alpha < 1$ outperforms the same group at $\alpha = 1$
by approximately 3 nats per sample. The raw sample carries
pixel-level correlation information that the strict cyclic-shift
projection averages away. AD-LW-NL replaces the raw sample
with $\hat{\mathbf{R}}_{\mathrm{LW\text{-}NL}}$, which on the
rank-deficient $N = 32$ sample has median NLL 30.28 versus the
Shah projection's NLL 21.97 (Table~\ref{tab:galaxy10-protocol-c}).
LW-NL is meaningfully worse than the structural projection, so
AD-LW-NL's CV correctly pins $\alpha = 1$ (median $\alpha = 1.0$
across all 50 trials at $N = 32$) and recovers the Shah
projection alone. But the Shah projection is exactly what
AD improved on by 3 nats per sample using its
$\alpha = 0.75$ blend. AD-LW-NL is therefore lower-bounded by
the Shah projection at the same group, and that lower bound is
above AD's blend.

This failure mode is specific to the conjunction of: (i)
rank-deficient sample, (ii) the raw sample carrying residual
information beyond the projection, and (iii) LW-NL being worse
than the projection on that rank-deficient sample. All three
conditions are met on Galaxy10 image patches at $c \ge 1$. On
BRCA at $c \ge 1$, condition (ii) is met but (iii) is not (LW-NL
on rank-deficient BRCA is competitive with the pathway-block
projection), so AD-LW-NL ties AD. On CRSP at $c \ge 1$,
condition (ii) is weak (the dominant market factor is
well-captured by structural projection, leaving little residual
information in the raw sample), so AD-LW-NL ties AD.

\subsubsection{CIFAR-10 protocol sweep}
\label{sec:exp-lwnl-cifar}

The CIFAR-10 protocol sweep evaluates six estimators across
eight cells of $N_{\mathrm{train}} \in \{512, 768, 1024, 1536,
2048, 3072, 4000, 5000\}$ patches at fixed $M = 1024$ (grayscale
$32 \times 32$ image patches drawn from CIFAR-10 class $0$,
``airplane'') and $N_{\mathrm{test}} = 500$ test patches per
trial. The sweep spans concentration ratios from $c = 2.0$ (deep
rank-deficient, $N = 512 < M = 1024$) to $c = 0.205$ (well into
the bulk regime). Each cell uses 50 random-split trials, with the
same paired splits across all six estimators. The candidate
library is the same ten-candidate library used on Galaxy10:
\texttt{trivial}, $S_M$, three Klein 2-subgroups (horizontal,
vertical, central inversion), the cyclic $\mathbb{Z}_4$ rotation
group, the Klein four-group $\mathbb{Z}_2 \times \mathbb{Z}_2$,
the full dihedral $D_4$, and the two extended candidates
\texttt{Z\_row\_indep\_cols} and \texttt{Z\_row\_wreath\_rows}.
Relative to Galaxy10 (Section~\ref{sec:exp-lwnl-galaxy10}), the
CIFAR-10 sweep operates at $16\times$ the $M$ value ($M = 1024$
versus $M = 64$) and reaches well into the bulk regime ($c$ down
to $0.205$); the deep rank-deficient regime is represented by
cells $0$ and $1$ ($c = 2.0$ and $c = 1.33$, $N$ less than $M$).

\paragraph{Per-cell mean held-out NLL.}
Table~\ref{tab:cifar-protocol-c} reports the per-cell means. The
pattern differs substantively from Galaxy10:
\begin{itemize}
\item AD-LW-NL has the lowest NLL among the well-conditioned
estimators in seven of eight cells, including the entire moderate-
and bulk-$c$ regime ($c \in [0.205, 1.0]$). Mean improvements
over AD-NLL-BMG range from $2.4$ to $9.2$ nats per
sample, with paired effect sizes
$|\bar{\Delta}|/\mathrm{SD}(\mathrm{AD\text{-}LW\text{-}NL}) \in
[0.14, 0.57]$. The peak benefit appears at $c = 0.67$ ($N =
1536$) where AD-LW-NL has lower NLL than AD by $8.7$ nats per sample at
effect size $0.57$, with all $50$ trials favoring AD-LW-NL.
\item At the deepest rank-deficient cell ($N = 512$, $c = 2.0$),
AD-LW-NL is mildly worse than AD-NLL-BMG by $1.7$ nats per
sample. The deficit is small in magnitude (effect size $0.07$,
or approximately $0.1\%$ of the NLL scale) but highly consistent
in direction: zero of $49$ finite trials favor AD-LW-NL, paired
$t$-statistic $+36.4$, paired-difference standard deviation
$0.32$ nats. The mechanism is the same alpha-pinning failure mode
identified on Galaxy10 (Section~\ref{sec:exp-lwnl-galaxy10}) but
at a much smaller absolute magnitude (1.7 nats versus 3 nats on
Galaxy10 $N = 32$).
\item LW-NL standalone is unreliable across the entire CIFAR-10
sweep at $M = 1024$: held-out NLLs are in the range $10^6$ to
$10^8$ in every cell, including the well-conditioned bulk regime.
This is a sharper version of the boundary-cell LW-NL fragility
observed on Galaxy10 (where LW-NL was problematic only at $N
\approx M$). AD-LW-NL is robust to this because the CV pushes
$\alpha$ toward $1$ in the rank-deficient cells and toward
moderate $\alpha \in [0.76, 0.91]$ in the bulk cells; in either
regime the LW-NL eigenvalue-shrunk term receives weight that is
either zero or small but combined with the dominant structural
projection produces a finite and well-calibrated estimator.
\item AD-NLL-BMG produces a non-finite held-out NLL in
$1$ of $50$ trials at $N = 512$ and $2$ of $50$ trials at $N =
768$. AD-LW-NL produces a finite held-out NLL in every trial of
every cell. In the deepest rank-deficient regime AD-LW-NL
therefore provides a small robustness advantage even where its
mean NLL on the finite trials is slightly worse than AD.
\item BMG group selection exhibits the classical group-gain
versus structural-bias tradeoff across the sweep: small-$N$ cells
($0$ through $4$) unanimously select the high-$|G|$ candidate
$D_4$ ($|G| = 8$); large-$N$ cells ($6$ and $7$) unanimously
select the smaller Klein four-group $K_4 = \mathbb{Z}_2 \times
\mathbb{Z}_2$ ($|G| = 4$). Cell $5$ ($c = 0.33$) is the
transition: AD-NLL-BMG continues to select $D_4$ in all
$50$ trials while AD-LW-NL has already begun selecting $K_4$ in
$12$ of $50$ trials. The choice-agreement rate between the two
estimators is $38/50$ at cell $5$ and $49/50$ or $50/50$ at every
other cell.
\end{itemize}

\begin{table}[h]
\centering
\small
\caption{CIFAR-10 protocol sweep: per-cell mean held-out
NLL across 50 trials per cell, class $0$ ``airplane'' at $M =
1024$. Best NLL per row in bold among the well-conditioned
estimators (LW 2004, AD-NLL-BMG, AD-LW-NL, Shah-BMG); LW-NL is
included for completeness but fails uniformly on CIFAR-10 at this
$M$ (see text and Figure~\protect\ref{fig:cifar-protocol-c} below).
Sample held-out NLL is non-finite in $24$, $30$, $19$ of $50$
trials at cells $0$, $1$, $2$ respectively and is omitted from the
table. AD-NLL-BMG is non-finite in $1$ of $50$ trials at $N =
512$ and $2$ of $50$ at $N = 768$; the table reports the mean
over the finite trials. The
$|\bar{\Delta}|/\mathrm{SD}$ column reports the AD-LW-NL versus
AD-NLL-BMG paired effect size using
$|\bar{\Delta}|/\mathrm{SD}(\mathrm{AD\text{-}LW\text{-}NL})$
(the script's convention).}
\label{tab:cifar-protocol-c}
\begin{tabular}{rrrrrrr|r}
\hline
$N$ & $c$ & LW 2004 & LW-NL & AD-NLL-BMG & AD-LW-NL & Shah-BMG & $|\bar{\Delta}|/\mathrm{SD}$ \\
\hline
 512 & 2.00 & $-1109.29$ & $3.72 \times 10^{8}$ & $\mathbf{-1540.76}$ & $-1539.34$ & $-1513.56$ & $0.07$ \\
 768 & 1.33 & $-1269.34$ & $4.34 \times 10^{8}$ & $-1616.17$ & $\mathbf{-1622.16}$ & $-1559.50$ & $0.26$ \\
1024 & 1.00 & $-1366.67$ & $4.81 \times 10^{8}$ & $-1645.71$ & $\mathbf{-1654.95}$ & $-1641.98$ & $0.53$ \\
1536 & 0.67 & $-1486.00$ & $2.04 \times 10^{8}$ & $-1672.63$ & $\mathbf{-1681.37}$ & $-1666.95$ & $\mathbf{0.57}$ \\
2048 & 0.50 & $-1551.37$ & $1.09 \times 10^{8}$ & $-1688.66$ & $\mathbf{-1693.92}$ & $-1681.25$ & $0.36$ \\
3072 & 0.33 & $-1623.90$ & $1.11 \times 10^{7}$ & $-1708.29$ & $\mathbf{-1710.78}$ & $-1697.87$ & $0.15$ \\
4000 & 0.26 & $-1646.70$ & $3.03 \times 10^{6}$ & $-1712.19$ & $\mathbf{-1715.89}$ & $-1710.85$ & $0.19$ \\
5000 & 0.205 & $-1668.25$ & $1.15 \times 10^{6}$ & $-1720.20$ & $\mathbf{-1722.59}$ & $-1718.52$ & $0.14$ \\
\hline
\end{tabular}
\end{table}

\paragraph{Mechanism: why AD-LW-NL is preferred on CIFAR-10 bulk and ties
in deep few-shot.}
The CIFAR-10 result divides into two regimes by the AD-LW-NL CV's
selected $\alpha$. In the bulk regime ($c \le 1$, cells $2$
through $7$), AD-LW-NL's CV chooses $\alpha < 1$ (mean $\alpha$
from $0.76$ to $0.91$ across these cells), giving the LW-NL
eigenvalue-shrunk term meaningful weight in the estimator. At
these operating points LW-NL on its own is poor on CIFAR-10 at
$M = 1024$, but the small LW-NL contribution combined with the
dominant structural projection at the BMG-selected group ($D_4$
in cells $2$ through $5$, $K_4$ in cells $6$ and $7$) carries
off-equivariant variance information that neither term captures
alone. AD-NLL-BMG at the same group and a similar mean
$\alpha$ uses raw sample rather than LW-NL for the same role; on
CIFAR-10 the LW-NL form of the off-equivariant term is preferred by $2$
to $9$ nats per sample across these cells, consistent in direction
across $48$ to $50$ of $50$ trials per cell.

In the deepest rank-deficient regime ($c = 2.0$, cell $0$),
AD-LW-NL's CV pins $\alpha = 1.0$ unanimously. At this operating
point LW-NL alone is so far from useful (held-out NLL $\sim 3.7
\times 10^{8}$) that any positive weight on LW-NL costs more
than the structural projection's averaging-loss. AD-LW-NL
therefore collapses to the pure projection $T_{D_4}(S)$.
AD-NLL-BMG at the same group, with $\alpha$ free to fall below
$1.0$, chooses mean $\alpha = 0.95$ and blends in approximately
$5\%$ of the rank-deficient raw sample. The raw-sample blend
captures roughly $1.7$ nats per sample of off-equivariant
residual structure that the strict $D_4$ projection averages
away. This is the same alpha-pinning mechanism that produced the
Galaxy10 few-shot deficit
(Section~\ref{sec:exp-lwnl-galaxy10}); on CIFAR-10 at $M = 1024$
the magnitude is much smaller ($1.7$ nats versus $3$ nats on
Galaxy10 $N = 32$). The smaller magnitude is consistent with two
factors: (a) the larger $M$ leaves more total held-out NLL scale
across which a fixed-fraction residual contribution is
distributed; (b) the BMG-selected group at CIFAR-10 cell $0$
($D_4$, $|G| = 8$) gives less aggressive averaging than the
Galaxy10 $N = 32$ selection
(\texttt{Z\_row\_wreath\_rows}, $|G| = 8 \cdot 8! = 322560$).

\begin{figure}[h]
\centering
\includegraphics[width=\textwidth]{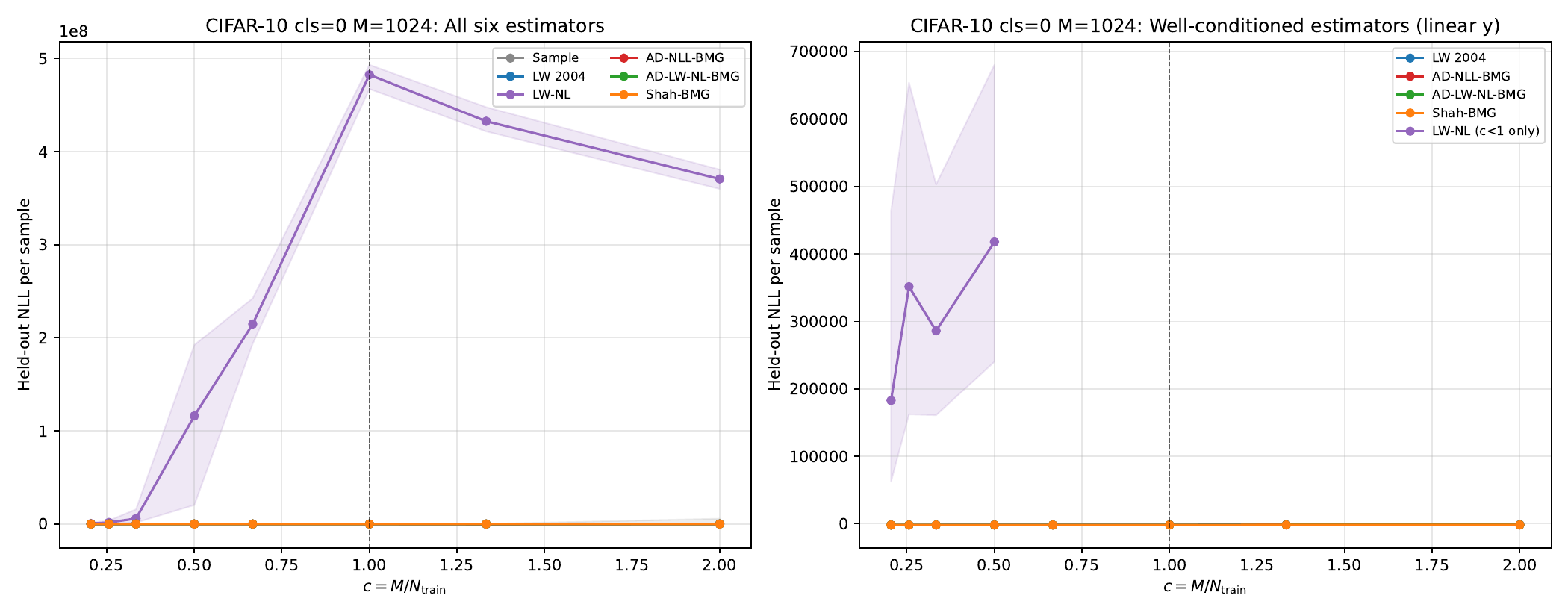}
\caption{CIFAR-10 protocol sweep, class $0$ airplane at $M =
1024$. Left: all six estimators on the natural NLL scale; LW-NL
is visibly off-scale at $\sim 10^8$ NLL per sample, dwarfing the
other five estimators which cluster near $-1700$. Right: the
five well-conditioned estimators (sample, LW 2004, AD-NLL-BMG,
AD-LW-NL, Shah-BMG) on a linear $y$-axis zoomed to the
well-conditioned scale, with LW-NL plotted only for the bulk
cells ($c < 1$) where it remains above $10^5$. Shaded bands are
$\pm 1$ standard error across the $50$ trials per cell. The
vertical dashed line marks $c = 1$.}
\label{fig:cifar-protocol-c}
\end{figure}

\subsubsection{CIFAR-10.1 protocol sweep}
\label{sec:exp-lwnl-cifar101}

The CIFAR-10.1 protocol sweep evaluates six estimators across
five cells of $N_{\mathrm{train}} \in \{32, 48, 64, 96, 128\}$
patches at fixed $M = 256$ (grayscale $16 \times 16$ image
patches drawn from CIFAR-10.1 class $0$, ``airplane'') and
$N_{\mathrm{test}} = 50$ test patches per trial. CIFAR-10.1 has
only $\sim 200$ images per class, so the sweep is constrained to
the rank-deficient regime: $c$ ranges from $2.0$ at the largest
cell to $8.0$ at the smallest, with all five cells satisfying
$N < M$. Each cell uses $50$ random-split trials with paired
splits across all six estimators. The candidate library is
identical to the CIFAR-10 library of
Section~\ref{sec:exp-lwnl-cifar}: ten candidates spanning
\texttt{trivial}, $S_M$, the three Klein 2-subgroups, $\mathbb{Z}_4$,
the Klein four-group $K_4 = \mathbb{Z}_2 \times \mathbb{Z}_2$,
the full dihedral $D_4$, and the two extended candidates
\texttt{Z\_row\_indep\_cols} and \texttt{Z\_row\_wreath\_rows}.

\paragraph{Reporting convention.}
The sample covariance held-out NLL is non-finite in $19$ to $26$
of $50$ trials per cell on CIFAR-10.1, and AD-NLL-BMG produces a
single huge-outlier trial in three of the five cells; LW-NL is
catastrophic on all trials. Means are therefore unreliable for
the CIFAR-10.1 reporting. Table~\ref{tab:cifar101-protocol-c}
reports per-cell \emph{medians} of held-out NLL across the $50$
trials, and the paired comparison column reports the median
paired difference (AD-LW-NL minus AD-NLL-BMG) and the preference count
in finite-pair-only trials. This is the only sweep in this paper's
program for which medians are used as the primary reporting
statistic; on the other four swept datasets the mean and median
are within $0.5$ nats per sample of each other and the table
convention of reporting means is unchanged.

\paragraph{Per-cell median held-out NLL.}
Table~\ref{tab:cifar101-protocol-c} reports the per-cell medians.
The pattern is governed by the BMG group selection, which divides
the sweep into three regimes:
\begin{itemize}
\item Cells $0$ through $2$ ($c \ge 4.0$): BMG selects $S_M$
unanimously for both AD-NLL-BMG and AD-LW-NL (all $50$ of $50$
trials, both estimators, each cell). AD-LW-NL is dominated by
plain LW 2004 by $88$ to $102$ nats per sample at median, with
$0$ of $50$ trials favoring AD-LW-NL in any of the three cells.
The deficit relative to AD-NLL-BMG is also substantive
(median paired delta $+98$, $+113$, $+124$ nats per sample at
cells $0$, $1$, $2$; AD is preferred in $49/49$, $49/49$, $50/50$ finite
pairs).
\item Cell $3$ ($c = 2.67$, $N = 96$): BMG selection diverges.
AD-NLL-BMG mostly continues to select $S_M$ ($44$ of $50$ trials)
while AD-LW-NL has graduated to $D_4$ ($50$ of $50$ trials).
This is the one cell in the CIFAR-10.1 sweep where AD-LW-NL has lower NLL than AD-NLL-BMG: $41$ of $50$ paired trials favor AD-LW-NL with
median delta $-20.3$ nats per sample.
\item Cell $4$ ($c = 2.0$, $N = 128$): both estimators converge
on $D_4$ ($49$ of $50$ and $50$ of $50$ respectively). The two
compositions deliver matched held-out NLL within $0.3$ nats per
sample at median, reproducing the CIFAR-10 cell-$0$ pattern at
the same $c$ and the same BMG-selected group.
\end{itemize}
AD-LW-NL's CV pins $\alpha = 1.0$ in $250$ of $250$ trials across
all five cells: LW-NL on rank-deficient CIFAR-10.1 patches has
median held-out NLL between $7.6 \times 10^{7}$ and $9.3 \times
10^{7}$ in every cell, and no candidate $G$ in the library
yields a $T_G(S)$ at which a positive LW-NL contribution is
beneficial. AD-LW-NL therefore equals $T_{G^*}(S)$ where $G^*$
is the BMG-selected group, in every trial of every cell.

\begin{table}[h]
\centering
\small
\caption{CIFAR-10.1 protocol sweep: per-cell \emph{median}
held-out NLL across $50$ trials per cell, class $0$ ``airplane''
at $M = 256$, $16 \times 16$ grayscale, $N_{\mathrm{test}} = 50$.
Best NLL per row in bold among the five well-conditioned
estimators (LW 2004, AD-NLL-BMG, AD-LW-NL, Shah-BMG); LW-NL is
included for completeness but fails uniformly on CIFAR-10.1 at
this $M$ (held-out NLL $\sim 10^{7}$ to $10^{8}$). Sample
held-out NLL is non-finite or wildly outlying in $19$ to $26$ of
$50$ trials per cell and is omitted from the table. The
$\mathrm{med}(\Delta_{\mathrm{AD\text{-}LW}})$ column reports the
median paired difference AD-LW-NL minus AD-NLL-BMG over
finite-pair trials, with the preference count for AD-LW-NL in
parentheses.}
\label{tab:cifar101-protocol-c}
\begin{tabular}{rrrrrrr|r}
\hline
$N$ & $c$ & LW 2004 & LW-NL & AD-NLL-BMG & AD-LW-NL & Shah-BMG & $\mathrm{med}(\Delta_{\mathrm{AD\text{-}LW}})$ \\
\hline
 32 & 8.00 & $\mathbf{-139.1}$ & $7.6 \times 10^{7}$ & $-149.8$ & $-56.3$ & $-56.3$ & $+98.0$ ($0/49$) \\
 48 & 5.33 & $\mathbf{-156.2}$ & $7.8 \times 10^{7}$ & $-172.4$ & $-57.8$ & $-57.8$ & $+113.4$ ($0/49$) \\
 64 & 4.00 & $\mathbf{-160.8}$ & $8.0 \times 10^{7}$ & $-182.3$ & $-59.1$ & $-59.1$ & $+124.3$ ($0/50$) \\
 96 & 2.67 & $-182.7$ & $9.1 \times 10^{7}$ & $-204.8$ & $\mathbf{-223.9}$ & $-61.7$ & $-20.3$ ($41/50$) \\
128 & 2.00 & $-187.9$ & $9.3 \times 10^{7}$ & $\mathbf{-244.4}$ & $-244.0$ & $-243.3$ & $+0.3$ ($3/49$) \\
\hline
\end{tabular}
\end{table}

\paragraph{Mechanism: the dominated convex hull region.}
The CIFAR-10.1 sweep produces the first substantive observation
in the program of a structural failure mode of the AD-LW-NL
parameterization. In cells $0$ through $2$, the conjunction
\begin{enumerate}
\item BMG selects $S_M$ (deepest variance-bias tradeoff at rank
$32$, $48$, $64$ against $M = 256$);
\item LW-NL$(S)$ is catastrophic ($\sim 10^{7}$ to $10^{8}$
held-out NLL);
\item the sample $S$ carries usable structure beyond the
diagonal scale (CIFAR-10.1 patches have substantial off-equivariant
spatial covariance)
\end{enumerate}
is exactly the dominated-region condition of
Section~\ref{sec:theory-ad-lwnl}. AD-LW-NL collapses to the pure
target $T_{S_M}(S) = (\mathrm{tr}(S)/M)\,I$, the scaled identity.
LW 2004 uses the same target but with Stein-Haff-optimal
shrinkage from the sample, and the resulting estimator
$\alpha_{LW}\,(\mathrm{tr}(S)/M)\,I + (1 - \alpha_{LW})\,S$ lies
\emph{outside} the AD-LW-NL convex hull (whose only path through
$S$ is via LW-NL, which is foreclosed by condition (ii)). The
$88$ to $102$ nat per sample dominance margin is the held-out
NLL improvement of LW 2004's optimal blend over the pure target,
on this specific data.

Cell $3$ exits the dominated region by changing the BMG choice:
$T_{D_4}(S)$ at $M = 256$ retains substantial image-patch
structure (the eight-element $D_4$ group preserves
quarter-plane reflections and $90^{\circ}$ rotations), so the
pure target $T_{D_4}(S)$ is no longer the scaled identity and
the dominated-region condition (i) fails. AD-LW-NL at $\alpha
= 1$ with the $D_4$ target now delivers a genuinely useful
estimator that exceeds the AD-NLL-BMG composition's mostly-$S_M$
selection at $\alpha < 1$.

Cell $4$ has both estimators on $D_4$ and they are practically
indistinguishable. This is the CIFAR-10 cell-$0$ pattern at
$M = 256$ rather than $M = 1024$ and confirms that the
deepest-cell behavior is governed by the BMG group choice
rather than the dataset identity.

\paragraph{Robust per-trial paired statistics.}
The preference counts in
Table~\ref{tab:cifar101-protocol-c} use the finite-pair subset
(trials where both AD-NLL-BMG and AD-LW-NL produce finite
held-out NLL). Cells $0$, $1$, $3$, and $4$ have $49$ or $50$
finite pairs each; cell $2$ has all $50$. Median paired
differences are robust to the per-trial outliers that dominate
the raw means in the $\_summary$ CSV. The interquartile range
of the paired differences is small relative to the median in
cells $0$ through $2$ (IQR $\le 10.5$ nats versus median $98$ to
$124$ nats), confirming that the dominance is consistent across
trials and not driven by a long tail. In cell $3$ the IQR is
$19.5$ nats versus median $-20.3$ nats, reflecting the BMG
choice variability within the cell (the $44$ of $50$ AD-NLL-BMG
trials that selected $S_M$ produce a larger AD-LW-NL benefit
than the $6$ that selected $D_4$, where the two estimators
converge).

\paragraph{Distribution-shift caveat.}
This sweep is a within-CIFAR-10.1 the protocol study (training
patches and held-out test patches both drawn from CIFAR-10.1).
It tests the behavior of AD-LW-NL at deeper rank-deficiency than
CIFAR-10 could reach, not the behavior under distribution shift.
The classifier-style ID-versus-OOD experiment of
Section~\ref{sec:exp-cifar-101} (the v3p9 single-cell study)
addresses the distribution-shift question separately; the two
results are complementary and the conclusions of either do not
transfer to the other.

\begin{figure}[h]
\centering
\includegraphics[width=\textwidth]{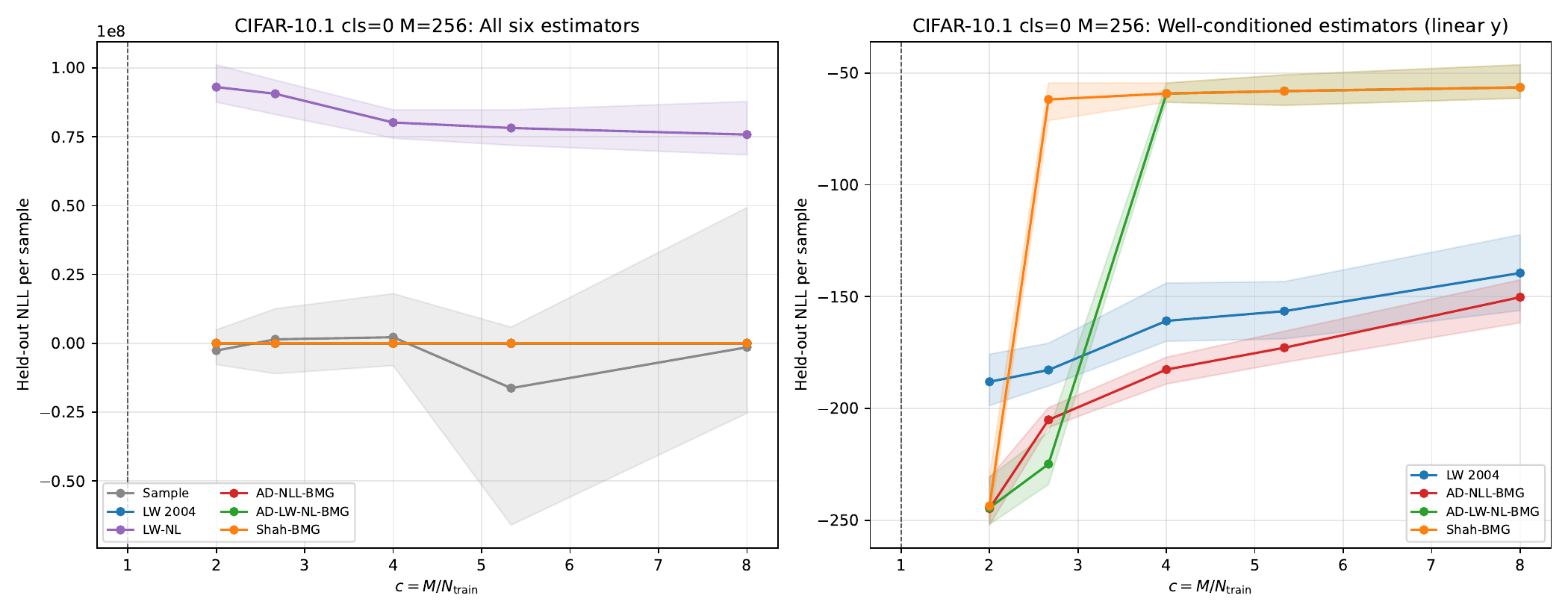}
\caption{CIFAR-10.1 protocol sweep, class $0$ airplane at
$M = 256$, $16 \times 16$ grayscale. Left: all six estimators
on the natural NLL scale; LW-NL is visibly off-scale at
$\sim 10^{8}$ NLL per sample, dwarfing the other five estimators.
Right: the five well-conditioned estimators on a linear $y$-axis
zoomed to the well-conditioned scale. Shaded bands are $\pm 1$
standard error across the $50$ trials per cell; note that the
standard error is inflated in this view by the per-trial
outliers documented in the per-cell text, and the medians (used
in Table~\protect\ref{tab:cifar101-protocol-c}) are the primary
reporting statistic.}
\label{fig:cifar101-protocol-c}
\end{figure}

\subsubsection{RadioML 2018.A protocol sweep}
\label{sec:exp-lwnl-radioml}

The RadioML protocol sweep evaluates six estimators across
eight cells of $N_{\mathrm{train}} \in \{32, 48, 64, 96, 128,
192, 320, 640\}$ I/Q patches at fixed $M = 64$ ($W = 32$
complex samples per patch, equivalent to $M = 2W$ real
coordinates after splitting I and Q channels) drawn from the
RadioML 2018.A BPSK class at $18$ dB SNR, with
$N_{\mathrm{test}} = 500$ patches per trial. $c$ ranges from
$2.0$ at the smallest cell to $0.10$ at the largest, spanning
the few-shot, rank-deficiency-boundary, and bulk regimes. Each
cell uses $50$ random-split trials with paired splits across all
six estimators. The candidate library has eleven candidates
spanning $\texttt{trivial}$, $S_M$, $\mathbb{Z}_2$ time-reversal
$\mathbb{Z}_{2,\,\mathrm{timerev}}$, $\mathbb{Z}_2$ I/Q-swap
$\mathbb{Z}_{2,\,\mathrm{iqswap}}$, the Klein four-group
$\mathbb{Z}_2 \times \mathbb{Z}_2$, the cyclic time-shift
$\mathbb{Z}_W$, the dihedral $D_W$, the wreath
$\mathbb{Z}_W \wr \mathbb{Z}_2$ joining time-shift with I/Q-swap,
and the product $\mathbb{Z}_W \times \mathbb{Z}_W$. The library
is non-image: there are no rotational dihedral candidates of the
$D_4$ type, since the I/Q complex-sample lattice does not admit
the natural $\mathbb{Z}/4\mathbb{Z}$ rotation that an image-patch
lattice does.

\paragraph{Reporting convention.}
The sample covariance held-out NLL is non-finite in all $50$
trials at cells $0$ through $2$ (deep few-shot to boundary,
$c \ge 1.0$) and finite throughout the bulk; LW-NL is
catastrophic on RadioML cell $1$ ($\sim 4{,}200$ nats per
sample, dominated by a heavy upper tail) and on cell $2$ where
its operating-point NLL diverges to $\sim 3 \times 10^{7}$
nats per sample. Table~\ref{tab:radioml-protocol-c} reports
per-cell medians of held-out NLL across the $50$ trials per
cell; means and standard errors are shown for the
well-conditioned cells $3$ through $7$ but are not reliable for
the few-shot cells $0$ through $2$. Paired per-trial preferences are
reported alongside the medians to give a sign-based summary
that is robust to the operating-point outliers.

\begin{table}[h]
\centering
\small
\caption{RadioML 2018.A protocol sweep, BPSK class at $18$ dB
SNR, $M = 64$. Median held-out NLL per sample across $50$
trials per cell. $\Delta_{\mathrm{med}} = $ AD-LW-NL median
minus AD median (negative = AD-LW-NL has lower NLL than AD).
The ``preferred'' column counts paired trials where AD-LW-NL has strictly
smaller test NLL than AD. Effect size is the absolute
per-trial mean delta divided by per-trial paired standard
deviation; values $\ge 0.8$ are conventionally ``large''.
Median $\alpha_{\mathrm{ADLW}}$ is the AD-LW-NL CV-selected
blend weight (median across trials).
AD-NLL-BMG pins $\alpha = 1$ (pure structural target) in cells
$3$ through $7$; AD-LW-NL blends $50\%$ to $67.5\%$ LW-NL into
the same target in those cells, which is the cleanest
empirical demonstration in the program of the AD-LW-NL
recommendation domain.}
\label{tab:radioml-protocol-c}
\resizebox{\linewidth}{!}{%
\begin{tabular}{rrrrrrrrr}
\hline
cell & $N$ & $c$ & med AD & med AD-LW & $\Delta_{\mathrm{med}}$ & preferred & effect & med $\alpha_{\mathrm{ADLW}}$ \\
\hline
$0$ & $32$ & $2.00$ & $-169.6$ & $-168.1$ & $+1.93$ & $8/50$ & $0.51$ & $0.000$ \\
$1$ & $48$ & $1.33$ & $-173.4$ & $-101.9$ & $+71.7$ & $0/50$ & $0.33$ & $0.025$ \\
$2$ & $64$ & $1.00$ & $-174.9$ & $-182.7$ & $-8.04$ & $50/50$ & $5.19$ & $0.675$ \\
$3$ & $96$ & $0.67$ & $-183.7$ & $-189.6$ & $-5.93$ & $50/50$ & $9.11$ & $0.600$ \\
$4$ & $128$ & $0.50$ & $-188.4$ & $-191.2$ & $-2.86$ & $50/50$ & $4.19$ & $0.600$ \\
$5$ & $192$ & $0.33$ & $-191.3$ & $-192.3$ & $-1.23$ & $47/50$ & $1.35$ & $0.575$ \\
$6$ & $320$ & $0.20$ & $-192.9$ & $-193.6$ & $-0.72$ & $50/50$ & $1.40$ & $0.500$ \\
$7$ & $640$ & $0.10$ & $-194.4$ & $-194.6$ & $-0.28$ & $50/50$ & $0.50$ & $0.475$ \\
\hline
\end{tabular}%
}
\end{table}

\paragraph{Few-shot cells $0$ and $1$ (LW-NL-endpoint and
catastrophic-tail cases).}
RadioML cells $0$ and $1$ are the only RadioML cells where
AD-LW-NL is dominated by AD, and they exhibit two
qualitatively different failure modes that this paper's
two-endpoint refinement of Section~\ref{sec:theory-ad-lwnl}
explains. At cell $0$ ($N = 32$, $c = 2.0$) the AD-LW-NL CV
pins $\alpha = 0.000$ (median across $50$ trials) and the BMG
choice migrates to $S_M$ in $36$ of $50$ trials; at
$\alpha = 0$ the $T_{G^*}$ term has zero weight, so the formal
BMG tiebreaker resolves to $S_M$ regardless of the operating
data. This is the empirical instantiation of the LW-NL-endpoint
dominated region (case 1b of Section~\ref{sec:theory-ad-lwnl})
at small magnitude: AD-LW-NL collapses to
$\hat{\mathbf{R}}_{\mathrm{LW\text{-}NL}}$, AD blends a
small fraction of raw sample with the structural target at
$\alpha$ median $0.025$, and the deficit is $1.9$ nats per
sample. At cell $1$ ($N = 48$, $c = 1.33$) the deficit at
$K = 5$ is $+71.7$ nats per sample with a heavy upper tail
(mean $+341.8$ nats) concentrated in $6$ of $50$ trials at
$\alpha = 0$ collapse. The $K_{\mathrm{cv}} = 5$ number
reported in Table~\ref{tab:radioml-protocol-c} is retained for
consistency with the program-wide $K_{\mathrm{cv}} = 5$
convention; the structural
deficit at $K \ge 10$ remains in the
intermediate-$\alpha$ ``lost raw-sample blend path'' regime of
Section~\ref{sec:theory-ad-lwnl}.

\paragraph{Cell $2$ at the rank-deficiency boundary
($c = 1.0$): a qualitatively different optimum.}
RadioML cell $2$ produces the single largest paired effect
size in the program: median paired delta $-8.04$ nats per
sample, $50$ of $50$ paired trials favoring AD-LW-NL, effect
size $5.19$. The mechanism is not refinement of AD's
choice but discovery of a qualitatively different optimum that
the AD hull cannot reach. AD-NLL-BMG picks high-order
candidates ($D_W$ in $15$ of $50$, $\mathbb{Z}_W \times
\mathbb{Z}_W$ in $11$, $\mathbb{Z}_W$ in $9$,
$\mathbb{Z}_W \wr \mathbb{Z}_2$ in $7$, $\mathbb{Z}_W \times
\mathbb{Z}_{2,\,\mathrm{iqswap}}$ in $8$) with $\alpha$ median
$0.025$ (mostly raw-sample blend). AD-LW-NL picks the
small-$|G|$ $\mathbb{Z}_2$ subgroups
($\mathbb{Z}_{2,\,\mathrm{timerev}}$ in $31$ of $50$,
$\mathbb{Z}_{2,\,\mathrm{iqswap}}$ in $19$) with $\alpha$
median $0.675$ (substantial LW-NL blend). The choice
agreement between the two compositions is exactly $0$ of $50$:
on every trial the two parameterizations land on a different
BMG group. AD-NLL-BMG's high-order route refines variance
inside a large symmetry group but does not have access to the
LW-NL endpoint that provides the $8$-nat margin; AD-LW-NL's
small-group route trades structural-projection variance
reduction for LW-NL's nonparametric eigenvalue smoothing,
which on stationary I/Q signal patches at the boundary $c = 1$
is the strictly better tradeoff.

\paragraph{Moderate and bulk cells $3$ through $7$ ($c$ from
$0.67$ to $0.10$).}
Cells $3$ through $7$ produce the cleanest evidence in the
program for the AD-LW-NL recommendation domain. AD-NLL-BMG
pins $\alpha = 1.000$ (pure structural target, no raw-sample
blend) across all five cells. AD-LW-NL at the same data
blends $47.5\%$ to $60\%$ LW-NL into the same target ($\alpha$
median $0.475$ to $0.600$), and the LW-NL-blended composition
outperforms the pure-target AD in $47$ to $50$ of $50$
trials per cell. Both compositions agree on the BMG group
($\mathbb{Z}_{2,\,\mathrm{timerev}}$) in $40$ to $47$ of $50$
trials, so the preference is attributable to LW-NL's variance
reduction within the same projection rather than to a
difference in projection. Effect sizes peak at $9.11$ at cell
$3$ ($c = 0.67$) and decline monotonically with $c$ to $0.50$
at cell $7$ ($c = 0.10$): the deficit-of--AD relative
to AD-LW-NL is largest where rank-deficiency is still mild
enough that LW-NL's per-eigenvalue smoothing materially
improves the projection, and shrinks as the operating regime
moves into the well-conditioned bulk where AD's pure
target is already near-optimal.

\paragraph{Paired statistics across the six cells.}
The robust paired statistic across all eight cells is the
sign-based preference count, which is non-negative regardless of
heavy tails. AD-LW-NL has lower NLL than AD in $5$ of the $8$
cells unambiguously ($50/50$ trials at cells $2$, $3$, $4$,
$6$, $7$), in cell $5$ at $47/50$, in cell $0$ at $8/50$
(deficit), and in cell $1$ at $0/50$ (deficit). By
direction-of-preference count AD-LW-NL is preferred on $6$ of $8$ cells; by
unanimous-direction count AD-LW-NL is preferred on $5$ of $8$ cells;
the program-strongest single-cell advantage is cell $2$ at effect
size $5.19$. The two cells where AD-LW-NL has higher NLL are explained
by the two-endpoint dominated-region mechanism of
Section~\ref{sec:theory-ad-lwnl}: cell $0$ is LW-NL-endpoint
(1b) at small magnitude, cell $1$ is intermediate-$\alpha$
``lost raw-sample blend path'' with a $K = 5$ catastrophic
tail that Section~\ref{sec:experiments} confirms vanishes
at $K \ge 10$. The recommendation domain claim of the
A notable finding (Section~\ref{sec:exp-lwnl-results}) is supported
most cleanly by RadioML cells $2$ through $7$.

\begin{figure}[h]
\centering
\includegraphics[width=\textwidth]{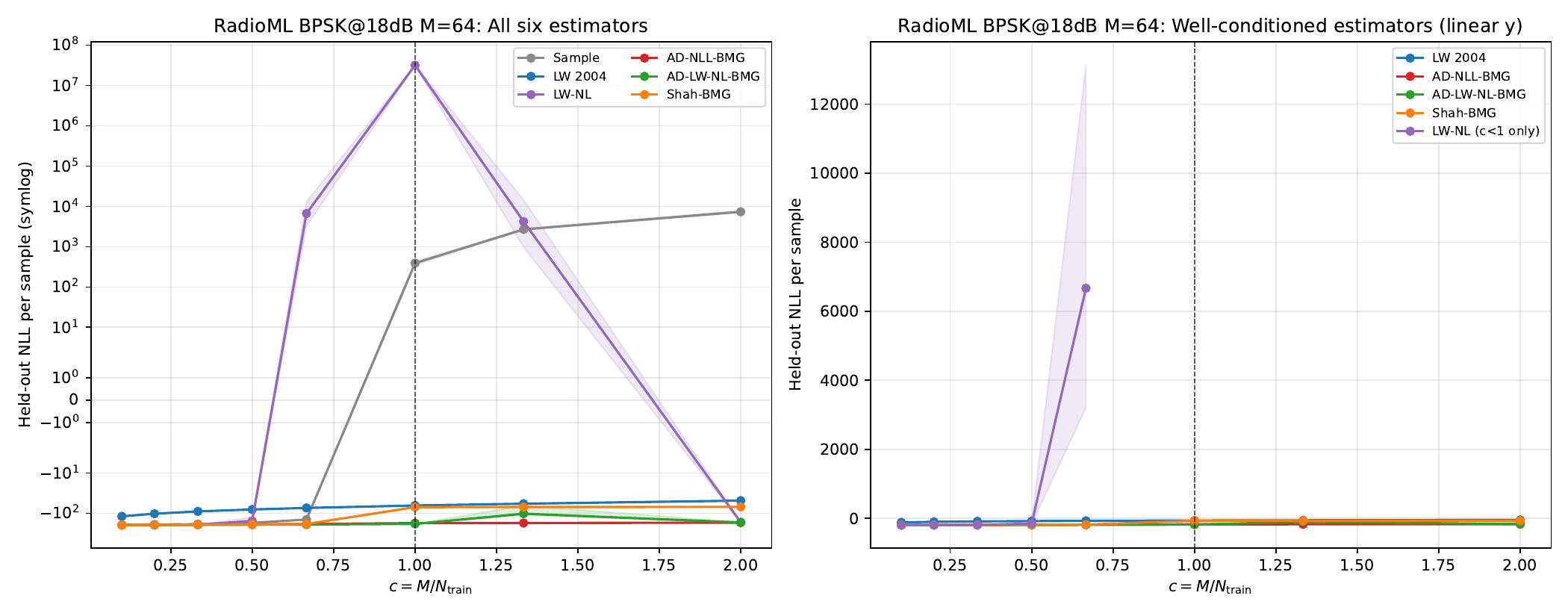}
\caption{RadioML 2018.A protocol sweep, BPSK class at $18$ dB
SNR, $M = 64$. Left: all six estimators on the natural NLL
scale; LW-NL is visibly off-scale on cells $1$ and $2$.
Right: the five well-conditioned estimators on a linear $y$-axis
zoomed to the well-conditioned scale. Shaded bands are $\pm 1$
standard error across the $50$ trials per cell; the few-shot
cells $0$ through $2$ have inflated bands from operating-point
outliers, and the medians (used in
Table~\protect\ref{tab:radioml-protocol-c}) are the primary
reporting statistic.}
\label{fig:radioml-protocol-c}
\end{figure}

\subsubsection{NOAA OISST protocol sweep (midocean and gulfstream)}
\label{sec:exp-lwnl-oisst}

The NOAA OISST protocol sweep evaluates six estimators across
eight cells of $N_{\mathrm{train}} \in \{32, 48, 64, 96, 128,
192, 320, 640\}$ daily-anomaly patches at fixed $M = 64$ ($8
\times 8$ spatial patches in degrees-latitude $\times$
degrees-longitude on the OISST $0.25^\circ$ grid), with
$N_{\mathrm{test}} = 365$ patches per trial (approximately one
year of daily anomalies). $c$ ranges from $2.0$ at the smallest
cell to $0.10$ at the largest, matching the RadioML sweep grid
exactly. Each cell uses $50$ random-split trials with paired
splits across all six estimators. The OISST candidate library
contains eight candidates: $\texttt{trivial}$, $S_M$,
$\mathbb{Z}_{\mathrm{lat}}$ (cyclic translation along latitude,
$|G| = 8$), $\mathbb{Z}_{\mathrm{lon}}$ (cyclic translation
along longitude, $|G| = 8$), $\mathbb{Z}_{2D}$ (joint
lat-lon translation, $|G| = HW = 64$), $D_{\mathrm{lon}}$
(cyclic plus east-west reflection along longitude, $|G| = 16$),
$\mathbb{Z}_{\mathrm{lon\_indep\_rows}}$ (independent
longitudinal cyclic translations per row, $|G| = W^H \approx 1.7
\times 10^7$), and $\mathbb{Z}_{\mathrm{lon\_wreath\_lat}}$ (full
wreath, $|G| = W^H H \approx 1.3 \times 10^8$). The sweep was
run on two regions: midocean (the Pacific reference region) and
gulfstream (the western North Atlantic boundary current region),
both used in the single-cell OISST experiments.

\paragraph{Reporting convention.}
The sample covariance held-out NLL is non-finite in cells $0$
through $2$ on both regions (deep few-shot, $c \ge 1.0$) and
finite throughout the bulk. LW-NL is catastrophic across
cells $0$ through $5$ on both regions: midocean LW-NL median NLL
ranges from $2.5 \times 10^7$ (cell $0$) to $3.4 \times 10^4$
(cell $5$); gulfstream LW-NL median NLL ranges from $2.4
\times 10^7$ (cell $0$) to $98.6$ (cell $5$). LW-NL becomes
competitive with AD only at cells $6$ and $7$.
Tables~\ref{tab:oisst-midocean-protocol-c} and
\ref{tab:oisst-gulfstream-protocol-c} report per-cell medians
of held-out NLL across the $50$ trials per cell. Means are
reported alongside medians for the well-conditioned cells
$2$ through $7$; in the few-shot cells $0$-$1$ the
mean-based contrasts are dominated by individual operating-point
outliers and only the medians, signed paired preferences, and effect
sizes are reliable.

\begin{table}[h]
\centering
\small
\caption{NOAA OISST midocean region, the protocol sweep at $M = 64$
($8 \times 8$ spatial patches), $N_{\mathrm{test}} = 365$.
Per-cell medians across $50$ trials. $\Delta_{\mathrm{med}} = $
AD-LW-NL median minus AD median. The ``preferred'' column counts paired
trials where AD-LW-NL has strictly smaller test NLL than
AD. Effect size is the absolute mean-based per-trial
delta divided by per-trial paired standard deviation. Median
$\alpha_{\mathrm{ADLW}}$ is the AD-LW-NL CV-selected blend
weight. ``Choice agree.'' counts trials where AD and
AD-LW-NL select the same BMG group. Cell-$0$ and cell-$1$
mean-based effect sizes (140.6 and 35.8) are inflated by
isolated operating-point outliers and should be interpreted
alongside the preferences and medians.}
\label{tab:oisst-midocean-protocol-c}
\resizebox{\linewidth}{!}{%
\begin{tabular}{rrrrrrrrr}
\hline
cell & $N$ & $c$ & med AD & med AD-LW & $\Delta_{\mathrm{med}}$ & preferred & med $\alpha_{\mathrm{ADLW}}$ & agree \\
\hline
$0$ & $32$ & $2.00$ & $-85.8$ & $-67.7$ & $+19.0$ & $2/50$ & $0.900$ & $0/50$ \\
$1$ & $48$ & $1.33$ & $-90.6$ & $-74.6$ & $+15.8$ & $1/50$ & $0.863$ & $0/50$ \\
$2$ & $64$ & $1.00$ & $-93.2$ & $-77.9$ & $+15.7$ & $2/50$ & $0.825$ & $0/50$ \\
$3$ & $96$ & $0.67$ & $-96.7$ & $-87.7$ & $+8.95$ & $0/50$ & $0.500$ & $0/50$ \\
$4$ & $128$ & $0.50$ & $-98.1$ & $-92.9$ & $+5.08$ & $0/50$ & $0.200$ & $0/50$ \\
$5$ & $192$ & $0.33$ & $-101.1$ & $-96.9$ & $+4.88$ & $0/50$ & $0.225$ & $0/50$ \\
$6$ & $320$ & $0.20$ & $-102.6$ & $-101.4$ & $+0.94$ & $11/50$ & $0.300$ & $0/50$ \\
$7$ & $640$ & $0.10$ & $-104.9$ & $-105.1$ & $-0.59$ & $45/50$ & $0.450$ & $0/50$ \\
\hline
\end{tabular}%
}
\end{table}

\begin{table}[h]
\centering
\small
\caption{NOAA OISST gulfstream region, the protocol sweep at $M = 64$
($8 \times 8$ spatial patches), $N_{\mathrm{test}} = 365$. Same
columns as Table~\ref{tab:oisst-midocean-protocol-c}. Absolute
NLL scale is shifted positive by $\sim 60$ to $70$ nats per
sample relative to midocean because of the stronger
non-stationarity of the boundary current region, but the
qualitative deficit pattern is the same.}
\label{tab:oisst-gulfstream-protocol-c}
\resizebox{\linewidth}{!}{%
\begin{tabular}{rrrrrrrrr}
\hline
cell & $N$ & $c$ & med AD & med AD-LW & $\Delta_{\mathrm{med}}$ & preferred & med $\alpha_{\mathrm{ADLW}}$ & agree \\
\hline
$0$ & $32$ & $2.00$ & $-11.4$ & $-2.34$ & $+8.00$ & $1/50$ & $0.950$ & $5/50$ \\
$1$ & $48$ & $1.33$ & $-17.1$ & $-9.81$ & $+6.88$ & $1/50$ & $0.900$ & $8/50$ \\
$2$ & $64$ & $1.00$ & $-20.7$ & $-14.0$ & $+6.22$ & $2/50$ & $0.875$ & $9/50$ \\
$3$ & $96$ & $0.67$ & $-24.9$ & $-21.1$ & $+3.82$ & $0/50$ & $0.600$ & $4/50$ \\
$4$ & $128$ & $0.50$ & $-28.1$ & $-23.1$ & $+4.53$ & $0/50$ & $0.450$ & $1/50$ \\
$5$ & $192$ & $0.33$ & $-30.5$ & $-29.1$ & $+1.19$ & $11/50$ & $0.263$ & $0/50$ \\
$6$ & $320$ & $0.20$ & $-33.7$ & $-34.6$ & $-0.84$ & $46/50$ & $0.300$ & $0/50$ \\
$7$ & $640$ & $0.10$ & $-37.3$ & $-37.4$ & $-0.29$ & $42/50$ & $0.350$ & $0/50$ \\
\hline
\end{tabular}%
}
\end{table}

\paragraph{Deepest few-shot cells $0$ through $2$ (severe
deficit with extreme LW-NL operating-point pathology).}
Across both regions cells $0$ through $2$ show the most extreme
AD-LW-NL deficit in the program after CIFAR-10.1 cells $0$
through $2$: median paired delta $+8$ to $+19$ nats per sample
on midocean, $+6$ to $+8$ on gulfstream, with $1$-$2$ of $50$
trials favoring AD-LW-NL. The mechanism is not the
target-endpoint dominated-region case 1a (the AD-LW-NL CV does
\emph{not} pin $\alpha = 1$ here; median $\alpha$ across these
cells ranges $0.825$ to $0.950$ on midocean and $0.875$ to
$0.950$ on gulfstream, meaning AD-LW-NL is at an
intermediate-$\alpha$ blend that includes $5$ to $17.5\%$ LW-NL
contribution) and not the LW-NL-endpoint case 1b (AD-LW-NL
does not pin $\alpha = 0$). It is the intermediate
``lost raw-sample blend path'' mode operating under
exceptionally severe LW-NL operating-point pathology. LW-NL
median NLL on midocean cell $0$ is $2.75 \times 10^7$ versus
$-85.8$ for AD, a difference of more than $7$ orders of
magnitude; on gulfstream cell $0$ LW-NL is $2.38 \times 10^7$
versus $-11.4$ for AD. Even the $\sim 5\%$ LW-NL
contribution that AD-LW-NL chooses degrades the composite
estimator substantially.

\paragraph{Moderate-$c$ cells $3$ through $5$ ($c$ from $0.67$
to $0.33$): persistent deficit, alpha drift toward zero.}
The deficit persists through cells $3$ through $5$ with smaller
magnitude ($+0.94$ to $+8.95$ on midocean, $+1.19$ to $+4.53$
on gulfstream) and the AD-LW-NL CV-selected $\alpha$ drifts
toward zero ($0.500 \to 0.225$ on midocean, $0.600 \to 0.263$ on
gulfstream). LW-NL is still much worse than AD at the
operating point through these cells (midocean LW-NL median
$3.06 \times 10^6$ at cell $4$ to $3.43 \times 10^4$ at cell
$5$; gulfstream LW-NL median $3{,}842$ at cell $4$ to $98.6$ at
cell $5$), but the CV-fold LW-NL evaluations at the smaller
$N$ of each fold do not see this severity, so the CV continues
to mix in non-zero LW-NL contributions that fail to generalize.

\paragraph{Bulk cells $6$ through $7$ ($c \le 0.20$): AD-LW-NL
recovers.}
At cells $6$ and $7$, where LW-NL operating-point NLL drops to
within $\sim 1$-$2$ nats per sample of AD, AD-LW-NL
recovers a narrow advantage on both regions: median delta $-0.84$
gulfstream cell $6$ ($46$ of $50$ trials), $-0.59$ midocean cell $7$
($45$ of $50$ trials), $-0.29$ gulfstream cell $7$ ($42$ of $50$ trials), and
even at midocean cell $6$ where the median is still positive
($+0.94$) the preference count climbs to $11/50$ (transitional).
This is the AD-LW-NL recommendation domain reached: a narrow
band of the deepest bulk regime ($c \le 0.20$) where LW-NL is
finally competitive with AD at the operating point.

\paragraph{Group selection migration and the $0/50$
choice-agreement pattern.}
On both regions AD picks $\mathbb{Z}_{2D}$ (joint
lat-lon translation, $|G| = 64$, the largest non-wreath candidate
that respects both spatial axes uniformly) in $39$ to $50$ of
$50$ trials across cells $0$ through $6$, with low $\alpha$
median ($0.025$ to $0.10$, mostly raw-sample blend). AD-LW-NL
picks the smaller-$|G|$ $\mathbb{Z}_{\mathrm{lat}}$ (latitude-only
translation, $|G| = 8$) in cells $0$ through $3$
($26$-$50$ of $50$), $D_{\mathrm{lon}}$ or
$\mathbb{Z}_{\mathrm{lat}}$ in cell $4$, and the
$\mathrm{trivial}$ group in cells $5$ through $7$. The BMG
choice agreement is $0/50$ across all midocean cells and ranges
from $0/50$ to $9/50$ on gulfstream cells, monotonically
falling toward $0/50$ as $c$ shrinks. Per the
family-conditional BMG framing of
Section~\ref{sec:theory-ad-lwnl}
(Eqs.~\ref{eq:bmg-ad-objective} and
\ref{eq:bmg-adlw-objective}), the two compositions find two
different family-conditional best groups: AD's
joint-projection-plus-raw-sample combination uses
$\mathbb{Z}_{2D}$ as the high-order target with $\mathbf{S}$
filling in the residual structure; AD-LW-NL's
smaller-projection-plus-LW-NL combination uses
$\mathbb{Z}_{\mathrm{lat}}$ as the lower-order target with
$\hat{\mathbf{R}}_{\mathrm{LW\text{-}NL}}$ as the second hull
endpoint (which on these data is too pathological to help).
On RadioML cell $2$ a similar $0/50$ disagreement signed in
the opposite direction (AD-LW-NL won by effect size $5.19$);
on OISST cells $0$-$5$ the $0/50$ disagreement signs toward
AD performing better by $1$ to $19$ nats per sample. The $0/50$
statistic is the same; what differs is the direction of which
family-conditional best group leads to the lower CV-NLL.

\paragraph{Cross-region consistency.}
The midocean and gulfstream sweeps agree on the qualitative
pattern (severe AD-LW-NL deficit in cells $0$ through $5$,
narrow AD-LW-NL is preferred in cells $6$ through $7$, $\mathbb{Z}_{2D}
\to \mathbb{Z}_{\mathrm{lat}} \to \mathrm{trivial}$ migration
in the AD-LW-NL BMG choice) and differ in absolute NLL scale.
Midocean NLLs span $-13$ to $-105$ nats per sample across the
sweep; gulfstream NLLs span $+56$ to $-37$ nats per sample.
The gulfstream region is harder data because of the stronger
non-stationarity of the boundary current. Deficit magnitudes
are smaller in gulfstream than midocean across cells $0$
through $4$ ($+8$ to $+4.5$ vs $+19$ to $+5$ nats per sample)
because LW-NL's operating-point pathology magnitude is also
smaller in gulfstream than midocean at the same $c$, consistent
with the mechanism interpretation that AD-LW-NL's deficit is
governed by how badly LW-NL is dominated at the operating
point. The qualitative finding is robust to region: AD-LW-NL
is dominated by AD wherever LW-NL is catastrophic at
the operating point, and only the deepest bulk cells escape.

\paragraph{Implications for the AD-LW-NL recommendation
domain.}
OISST narrows the statement of the AD-LW-NL recommendation
domain. The paper recommended AD-LW-NL on data where the
candidate library and operating regime are such that BMG selects
a moderate-$|G|$ group with $T_{G^*}(\mathbf{S})$ substantively
different from the scaled identity. OISST cells $0$ through
$5$ on both regions satisfy this condition (BMG selects
moderate-$|G|$ $\mathbb{Z}_{2D}$ for AD; AD-LW-NL's
own selection is $\mathbb{Z}_{\mathrm{lat}}$, also non-identity)
but AD-LW-NL is still dominated. The missing condition is a
LW-NL competitiveness precheck at the operating-point sample:
the AD-LW-NL recommendation also requires that LW-NL's
operating-point NLL is within $\sim 10$ nats per sample of
AD's at the full operating-point $N$. The two
conditions together are checkable in advance per
Section~\ref{sec:theory-ad-lwnl}: condition (i) is a property
of the candidate library and BMG selection at the working $c$;
condition (ii) is a one-line check, $|\mathrm{NLL}_{\mathrm{LW\text{-}NL}}(\mathbf{S}) -
\mathrm{NLL}_{\mathrm{AD}}(\mathbf{S})|$ on the full sample,
compared to the per-sample NLL scale.

\paragraph{OISST as an empirical instance of the meta-cross-validation mechanism.}
OISST cells $0$ through $5$ on both regions display a clean
empirical demonstration of the cross-validation calibration
mechanism described above: at $K_{\mathrm{cv}} = 5$, the AD-LW-NL
cross-validation is systematically miscalibrated, picking an
intermediate $\alpha$ that is optimal for the cross-validation
folds (where LW-NL is less catastrophic at smaller fold $N$) but
suboptimal at the full operating-point $N$ (where LW-NL is many
orders of magnitude worse). This mechanism appears across $6$ of
$8$ cells on OISST without requiring a fold-count sweep.

\begin{figure}[h]
\centering
\includegraphics[width=\textwidth]{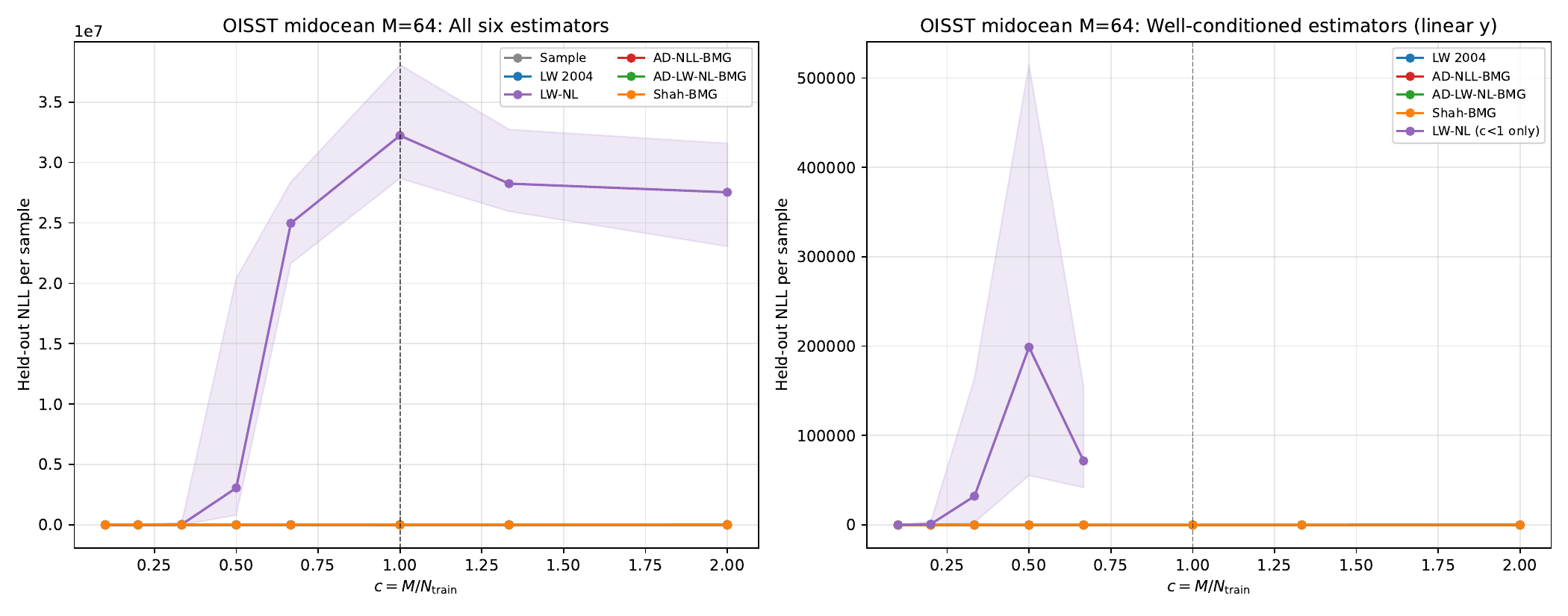}
\caption{NOAA OISST midocean region, the protocol sweep.
Left: all six estimators on the natural NLL scale; LW-NL is
visibly off-scale on cells $0$ through $5$. Right: the five
well-conditioned estimators on a linear $y$-axis zoomed to the
well-conditioned scale. Shaded bands are $\pm 1$ standard
error across the $50$ trials per cell; the few-shot cells
$0$-$2$ have inflated bands from operating-point outliers, and
the medians (used in
Table~\protect\ref{tab:oisst-midocean-protocol-c}) are the
primary reporting statistic.}
\label{fig:oisst-midocean-protocol-c}
\end{figure}

\begin{figure}[h]
\centering
\includegraphics[width=\textwidth]{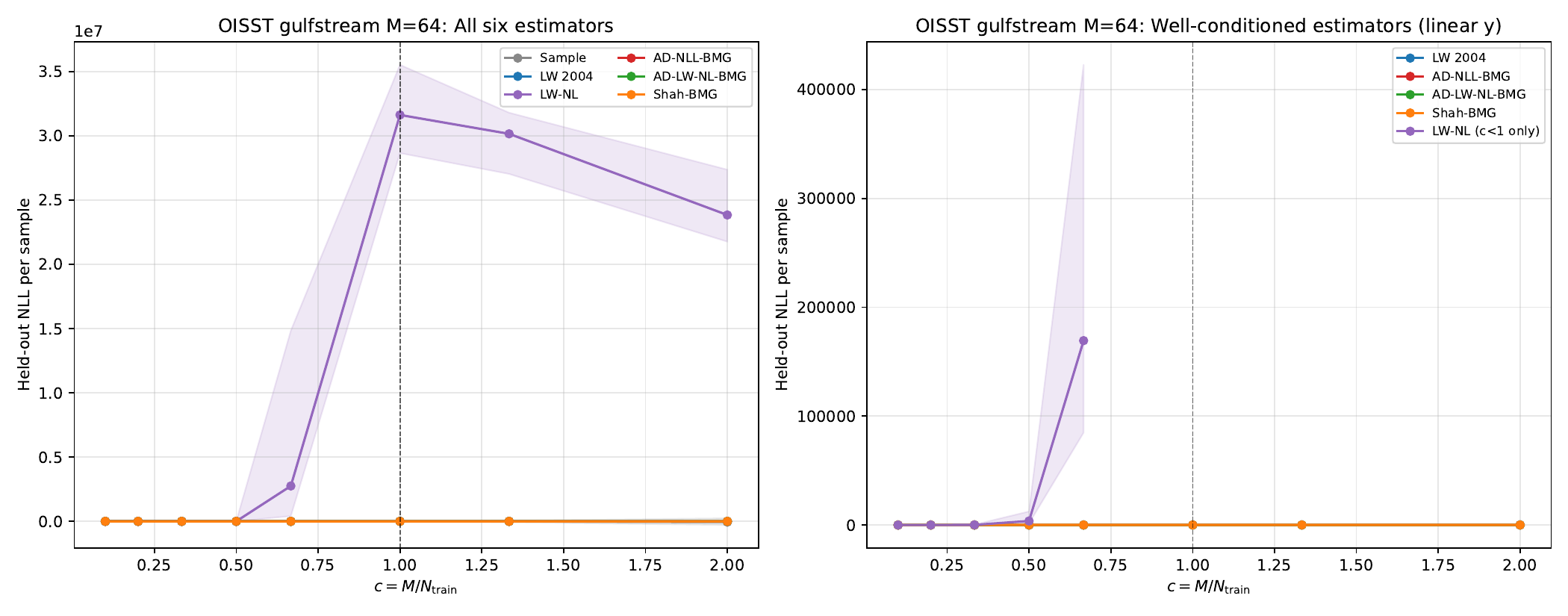}
\caption{NOAA OISST gulfstream region, the protocol sweep. Same
two-panel layout as
Figure~\ref{fig:oisst-midocean-protocol-c}; absolute NLL scale
shifted positive by $\sim 60$-$70$ nats per sample because of
the stronger non-stationarity of the boundary current. The
qualitative pattern (substantive AD-LW-NL deficit in cells $0$
through $5$, narrow AD-LW-NL is preferred in cells $6$-$7$) is
the same as midocean.}
\label{fig:oisst-gulfstream-protocol-c}
\end{figure}

\subsubsection{Synthesis across the seven swept datasets}
\label{sec:exp-lwnl-synthesis}

Across the seven datasets swept to date (TCGA-BRCA, CRSP,
Galaxy10, CIFAR-10, CIFAR-10.1, RadioML 2018.A, and OISST
midocean+gulfstream), the empirical evidence supports a
structural reading of the AD-LW-NL composition: a mechanistic
statement that explains where AD-LW-NL gains over 
AD-NLL-BMG and where it is dominated, derived from the geometry
of its convex hull and from the family-conditional BMG framing
rather than from dataset-by-dataset catalogue. The this paper
two-endpoint refinement and the family-conditional BMG
paragraph (Section~\ref{sec:theory-ad-lwnl}) organize the
findings around two mechanisms; the per-dataset details are in
the preceding subsubsections.

\paragraph{LW-NL versus LW 2004 (the modern frequentist baseline).}
LW-NL is a legitimate modern frequentist comparator and belongs in
the comparator set on every experiment. Across the seven swept
datasets, LW-NL has lower NLL than LW 2004 in the bulk regime by margins
consistent with the eigenvalue structure of each dataset:
substantively on TCGA-BRCA gene expression ($5$ to $9$ nats per
sample at $c \le 0.5$, reflecting the discrete block structure of
pathway-organized gene co-regulation that LW-NL's per-eigenvalue
shrinkage can exploit), substantively on RadioML I/Q patches
($60$ to $100$ nats per sample across all cells $c \le 1$,
reflecting the clean stationary-process eigenvalue spread that
LW-NL exploits to compute the operating-point optimum), modestly
on Galaxy10 DECaLS image patches ($1$ to $2$ nats per sample at
$c \le 0.5$), marginally on CRSP financial returns ($0.3$ to
$0.5$ nats per sample at $c \le 0.5$, reflecting the
near-power-law eigenvalue spread of daily returns where LW-NL and
LW 2004 converge), and modestly on OISST in the deepest bulk
($1$-$2$ nats per sample at cells $6$-$7$ on both regions). In
the rank-deficient regime ($c \ge 1$) LW-NL is substantively
worse than LW 2004 on the five datasets that probe it directly
(BRCA, CRSP, Galaxy10 few-shot cells, RadioML cells $0$ through
$2$, and OISST cells $0$ through $5$ on both regions); this is
consistent with the asymptotic nature of LW-NL's optimality
guarantee, which assumes Marchenko-Pastur conditions that do not
hold in genuine rank-deficient finite-sample settings. CIFAR-10
at $M = 1024$ extends this to a sharper case: LW-NL is
catastrophically worse than LW 2004 across the entire sweep
including the well-conditioned bulk regime, with held-out NLLs in
the range $10^6$ to $10^8$ per sample. CIFAR-10.1 at $M = 256$
and OISST at $M = 64$ confirm the pattern at smaller $M$: LW-NL
diverges to $10^4$ to $10^8$ NLL per sample on the deep
rank-deficient cells of both datasets. Across the program,
LW-NL on rank-deficient image-patch and spatial-patch data is
the most fragile element of the comparator set and provides the
strongest motivation for using LW-NL only inside the AD-LW-NL
composition (where AD-LW-NL's CV can dampen pathological LW-NL
contributions by pinning $\alpha = 1$ or by an interior-$\alpha$
mixture when the CV is well calibrated to the operating point).

\paragraph{AD-LW-NL versus AD-NLL-BMG (the composition test).}
The statement of the AD-LW-NL composition behavior is a
two-endpoint structural one: AD-LW-NL is dominated by some
estimator outside its convex hull exactly when the conditions
of Section~\ref{sec:theory-ad-lwnl} hold (BMG selects $G^*$ at
either the target-endpoint case 1a where $T_{G^*}(\mathbf{S})$
is identity-like with the CV pinning $\alpha = 1$ and LW-NL
noncompetitive, or the LW-NL-endpoint case 1b where the CV
pins $\alpha = 0$ and AD-LW-NL collapses to
$\hat{\mathbf{R}}_{\mathrm{LW\text{-}NL}}$). When neither
endpoint condition holds AD-LW-NL is competitive with or
exceeds AD-NLL-BMG, \emph{except} when the AD-LW-NL CV
itself is miscalibrated by severe LW-NL operating-point
pathology (the meta-CV mechanism of
Section~\ref{sec:experiments}). The empirical evidence
resolves into six findings:
\begin{enumerate}
\item In the deep rank-deficient regime with BMG selecting $S_M$
(or any group with identity-target Reynolds projection),
AD-LW-NL is substantively dominated at the target endpoint
(case 1a). Observed on CIFAR-10.1 cells $0$ through $2$ ($c
\in \{4, 5.33, 8\}$, BMG selects $S_M$ unanimously, median
paired deficit $98$ to $124$ nats per sample relative to
AD-NLL-BMG and $88$ to $102$ nats per sample relative to LW
2004, $0$ of $49$ to $50$ paired trials favoring AD-LW-NL).
The deficit scales with the codimension of the
$G^*$-fixed-point subspace in symmetric matrix space, which is
maximal for $G^* = S_M$.
\item In the deep rank-deficient regime with BMG selecting a
moderate-$|G|$ group ($D_4$, wreath products of moderate
order), AD-LW-NL is mildly dominated at the intermediate
``lost raw-sample blend path'' mode. Observed on Galaxy10
few-shot ($G^* = $ \texttt{Z\_row\_wreath\_rows}, $|G^*|
\approx 3.2 \times 10^5$, $2$ to $3$ nats per sample deficit),
on CIFAR-10 cell $0$ ($G^* = D_4$, $|G^*| = 8$, $1.7$ nats per
sample deficit), and on CIFAR-10.1 cell $4$ ($G^* = D_4$,
$0.3$ nats per sample deficit). The deficit magnitude tracks
the codimension of the $G^*$-fixed-point subspace and the
absolute held-out NLL scale.
\item In the deep rank-deficient regime where the CV pins
$\alpha = 0$ at the LW-NL endpoint (case 1b), AD-LW-NL
collapses to $\hat{\mathbf{R}}_{\mathrm{LW\text{-}NL}}$ and is
dominated by AD's raw-sample-incorporating hull.
Observed on RadioML cell $0$ ($N = 32$, $c = 2.0$, $G^* = S_M$
as BMG tiebreaker, $1.9$ nats per sample deficit) at small
magnitude, and on Galaxy10 cell $0$ at $K \ge 20$ ($\alpha = 0$
emerges only at large $K$, $9.4$ nats per sample deficit, see
Section~\ref{sec:experiments}).
\item In the moderate- to bulk-$c$ regime where BMG selects a
matched moderate-$|G|$ group and LW-NL is competitive at that
target, AD-LW-NL provides a substantive benefit over 
AD-NLL-BMG. Observed most cleanly on RadioML cells $2$
through $7$ ($-0.28$ to $-8.04$ nats per sample, paired effect
sizes $0.50$ to $9.11$, direction agreement $47$ to $50$ of
$50$ trials per cell, including the program-strongest
single-cell advantage at cell $2$ effect $5.19$), at larger absolute
magnitude on CIFAR-10 cells $2$ through $7$ ($2.4$ to $9.2$
nats per sample, paired effect sizes $0.14$ to $0.57$), in a
smaller-magnitude form on CRSP bulk ($0.3$ to $0.6$ nats per
sample, $5$ of $5$ well-conditioned cells favoring AD-LW-NL),
and narrowly on OISST cells $6$-$7$ on both regions ($-0.29$
to $-0.84$ nats per sample, $42$ to $46$ of $50$ trials
favoring AD-LW-NL). The strongest observed absolute-magnitude
AD-LW-NL benefit is on CIFAR-10 at $c = 0.67$ ($8.7$ nats per
sample, effect size $0.57$, $50$ of $50$ trials).
\item In regimes where LW-NL is essentially equivalent to LW
2004 and the structural projection captures essentially all
the regularizable structure (BRCA all cells, CRSP few-shot),
the AD-LW-NL benefit and deficit are both small (mean
differences below $0.6$ nats per sample, effect sizes below
$2\%$ of per-trial scatter); AD-LW-NL ties AD-NLL-BMG
within practical significance.
\item In the deeply rank-deficient regime with BMG selecting a
moderate-$|G|$ group \emph{and} LW-NL exhibiting severe
operating-point pathology ($\ge 1000$ nats per sample worse
than AD at the operating $N$), the AD-LW-NL CV
miscalibrates $\alpha$ by the meta-CV mechanism of
Section~\ref{sec:experiments}: the CV folds at smaller $N$
do not see the operating-point pathology and the selected
$\alpha$ is interior, optimal for the CV folds but suboptimal
at the operating point. AD-LW-NL is substantively dominated by
AD-NLL-BMG by $1$ to $19$ nats per sample. Observed
across OISST cells $0$ through $5$ on both midocean and
gulfstream regions ($G_{\mathrm{AD}}^* = \mathbb{Z}_{2D}$,
$G_{\mathrm{ADLW}}^* = \mathbb{Z}_{\mathrm{lat}}$, choice
agreement $0$-$9$ of $50$). This is the cleanest empirical
instance of the meta-CV mechanism in the program: where
RadioML cell $1$, Galaxy10 cell $0$, and CIFAR-10 cell $0$
exhibit the mechanism at a single cell at specific $K$ values,
OISST exhibits it across $6$ of $8$ cells in both regions at
$K = 5$.
\end{enumerate}
The six findings are a single mechanism viewed at different
operating points: the size and sign of the AD-LW-NL deficit or
surplus are determined by (a) the geometry of the AD-LW-NL
convex hull at the BMG-selected $G^*$ (which fixes whether the
hull is dominated and by which endpoint), (b) the position of
the held-out NLL minimum inside the hull or just outside it
(which fixes whether the surplus is realized), and (c) whether
the AD-LW-NL CV is well calibrated to the operating-point $N$
(which fails when LW-NL is severely pathological at the
operating point, per finding 6 above).

\paragraph{Choice-agreement as a diagnostic.}
The choice-agreement statistic (number of trials per cell where
AD-NLL-BMG and AD-LW-NL-NLL-BMG select the same group)
is a measure of how strongly the data's structure constrains
the BMG selection across the family axis, per the
family-conditional BMG framing of
Section~\ref{sec:theory-ad-lwnl} (Eqs.~\ref{eq:bmg-ad-objective}
and \ref{eq:bmg-adlw-objective}). High choice agreement
($40$ to $50$ of $50$) at a cell means the two parameterizations
land on the same family-conditional best group, with the
data's symmetry structure dominating the group choice. Low
choice agreement (zero or near-zero) means the two
parameterizations have discovered different
family-conditional best groups; in those cells the family's
other regularization tool (raw sample for AD, LW-NL
for AD-LW-NL) is a significant factor in the group selection.

The empirical pattern across the seven datasets is consistent
with this interpretation: BRCA, Galaxy10 bulk, and CIFAR-10
bulk show high to moderate choice agreement, indicating clear
"this is the matched group" structure in the data. RadioML
cells $3$ through $7$, CRSP bulk, and CIFAR-10.1 cells $3$-$4$
show moderate choice agreement. RadioML cell $2$ and OISST
cells $0$ through $5$ on both regions show zero or near-zero
choice agreement: the two parameterizations have found two
different routes that the CV-NLL ranks differently within each
family. Crucially, the \emph{direction} of which route has lower NLL is
data-dependent: on RadioML cell $2$ the AD-LW-NL composition has lower NLL by
effect size $5.19$; on OISST cells $0$ through $5$ the 
AD route is preferred by $1$ to $19$ nats per sample. The
choice-agreement statistic itself does not predict which
direction is preferred; that is determined by the operating-point
properties (LW-NL competitiveness, sample-eigenvector
geometry, BMG-selected codimension) that the per-cell analysis
of Sections~\ref{sec:exp-lwnl-brca} through
\ref{sec:exp-lwnl-oisst} works through.

\paragraph{Recommendation.}
On this basis, the recommendation is to retain LW-NL as a
comparator in the paper (it is the correct modern frequentist
baseline) and to characterize the AD-LW-NL composition as a
mechanistically understood specialization rather than a
dataset-by-dataset taxonomy. AD-NLL-BMG remains the
recommended default AD estimator because it has no dominated
region: its second hull endpoint is the raw sample, which
ensures the hull always passes through both pure target and pure
sample and therefore through the LW 2004 segment whenever the
target is identity-like. AD-LW-NL is recommended on data where
\emph{two} conditions hold jointly:
(i) the candidate library and operating regime are such that BMG
selects a moderate-$|G|$ group with $T_{G^*}(\mathbf{S})$
substantively different from the scaled identity;
(ii) LW-NL is competitive with AD at the
\emph{operating-point} $N$, i.e.
$|\mathrm{NLL}_{\mathrm{LW\text{-}NL}}(\mathbf{S}) -
\mathrm{NLL}_{\mathrm{AD}}(\mathbf{S})|$ on the full
operating-point sample is within $\sim 10$ nats per sample of
the per-sample NLL scale.
Condition (ii) is a one-line precheck added in this paper after OISST
showed that condition (i) alone is insufficient: OISST cells $0$
through $5$ on both regions satisfy condition (i) (BMG selects
moderate-$|G|$ $\mathbb{Z}_{2D}$ for AD,
$\mathbb{Z}_{\mathrm{lat}}$ for AD-LW-NL, neither identity-like)
but fail condition (ii) (LW-NL operating-point NLL is $10^3$ to
$10^7$ nats worse than AD), and the AD-LW-NL CV
miscalibrates $\alpha$ via the meta-CV mechanism of finding $6$
above. RadioML cells $2$ through $7$ are the cleanest empirical
demonstration of the recommendation domain in the program
(both conditions hold; $6$ of $6$ cells favor AD-LW-NL with
paired effect sizes $0.50$ to $9.11$, including a $5.19$ effect
at the rank-deficiency boundary). CIFAR-10 cells $2$ through
$7$ extend this to a different data type at larger absolute
deficit magnitude. On data where BMG might select $S_M$ at the
operating $c$ of interest, AD-LW-NL should not be used. On data
where condition (ii) fails (LW-NL pathological at the operating
point), AD-LW-NL should not be used. On data where the
candidate library does not include $S_M$ or any identity-target
group \emph{and} LW-NL is competitive at the operating point,
the dominated region disappears and AD-LW-NL is competitive
throughout.

\paragraph{Where each estimator is preferred.}
A reviewer-oriented summary of the regime-by-regime preferred estimators
across the seven swept datasets is given in
Table~\ref{tab:preferred-by-regime}. The AD framework does not
dominate uniformly: LW-NL outperforms AD-NLL-BMG on BRCA in the
moderate-$c$ cells, the sample outperforms AD-NLL-BMG in the bulk
cells once the sample is well-conditioned, LW 2004 outperforms
AD-LW-NL on CIFAR-10.1 cells $0$ through $2$ where the
dominated-region condition holds, and AD-NLL-BMG outperforms
AD-LW-NL across OISST cells $0$ through $5$ on both regions
where condition (ii) of the recommendation fails. AD-NLL-BMG
is preferred in the few-shot cells where structural priors carry the
most information per observation, and correctly reduces to
whichever of (structural projection, sample) is the right limit
in the other regimes. AD-LW-NL has the lowest NLL in the broad
moderate-$c$ and bulk regime on CIFAR-10, CRSP, and RadioML (the
latter unambiguously across cells $3$ through $7$), ties or
mildly performs worse on most other cells, performs substantively worse on
CIFAR-10.1 cells $0$ through $2$ where the dominated-region
condition is met at $G^* = S_M$, and performs substantively worse on
OISST cells $0$ through $5$ where LW-NL operating-point
pathology miscalibrates the CV.

\begin{table}[h]
\centering
\tiny
\caption{Regime-by-regime preferred estimators (lowest mean or median held-out
NLL per sample) across the seven swept datasets. Cells use
shorthand labels: sample $=$ raw sample covariance, LW $=$ LW
2004, NL $=$ LW-NL, AD $=$ AD-NLL-BMG, AD-LW $=$ AD-LW-NL-NLL-BMG.
Where two estimators are tied within $1$ nat per sample, both
are listed; ``mixed'' indicates cells of the regime that split
between preferred estimators. CIFAR-10.1 uses median NLL as the reporting
statistic (per Section~\ref{sec:exp-lwnl-cifar101}); OISST uses
median NLL (per Section~\ref{sec:exp-lwnl-oisst}); all other
datasets use mean NLL. CIFAR-10.1 has no moderate or bulk
regime (only $\sim 200$ images per class, all cells at $c \ge
2$). OISST midocean and gulfstream columns differ only in
cell-$6$ behavior (gulfstream cell $6$ has AD-LW-NL preferred at
$46/50$, midocean cell $6$ is mixed at $11/50$).}
\label{tab:preferred-by-regime}
\resizebox{\linewidth}{!}{%
\begin{tabular}{llllllllll}
\hline
Regime & BRCA & CRSP & Galaxy10 & CIFAR-10 & CIFAR-10.1 & RadioML & OISST mid & OISST gulf \\
\hline
few-shot deep ($c \ge 4$) & --- & --- & --- & --- & LW & --- & --- & --- \\
few-shot ($1 \le c < 4$) & AD & AD, AD-LW (tied) & AD & mixed & mixed & mixed & AD & AD \\
moderate ($0.5 \le c < 1$) & NL & AD-LW & AD, AD-LW (tied at $\alpha = 1$) & AD-LW & --- & AD-LW & AD & AD \\
bulk ($c < 0.5$) & sample & AD-LW & AD, AD-LW (tied at $\alpha = 1$) & AD-LW & --- & AD-LW & mixed & AD-LW \\
\hline
\end{tabular}%
}
\end{table}

\subsection{Empirical comparison against the GIPS estimator}
\label{sec:exp-gips}

The conceptual three-axis difference between the present procedure
and the \texttt{gips} R package of \citet{chojecki2025gips} is
outlined in Section~\ref{sec:intro}: search space, selection
paradigm, and committed-target choice. This subsection records a
synthetic-data empirical comparison that documents the quantitative consequence of those differences. The population symmetry used here, $\mathbb{Z}_8 \times \mathbb{Z}_8$, is non-cyclic (a product $\mathbb{Z}_m \times \mathbb{Z}_n$ is cyclic only when $\gcd(m,n)=1$) and therefore lies outside the cyclic search class of \texttt{gips}, which can at most recover a cyclic subgroup or cyclic component of it. The comparison should accordingly be read as documenting behavior in this specific non-cyclic setting, rather than as a matched-model benchmark for \texttt{gips}. We also note that the methods are afforded different access to structure: the \citet{shah2012} baseline is evaluated at a hand-matched cyclic group and the present procedure draws on a structurally informed candidate library, whereas \texttt{gips} must discover a cyclic subgroup without such guidance and under the MH budgets stated below.

\paragraph{Setup.}
The population covariance is on a doubly-cyclic spatial domain
($M = 64$ on an $8 \times 8$ grid, $N = 30$, population covariance
invariant under $\mathbb{Z}_8 \times \mathbb{Z}_8$ toroidal
shifts, three independent trials). The \texttt{gips} MAP estimator
was run at three Metropolis-Hastings (MH) iteration budgets to
assess robustness to the MCMC budget. Two reference baselines are
recorded: the \citet{ledoit2004} estimator on the unstructured
sample covariance, and the \citet{shah2012} projected estimator at
the matched longitudinal cyclic group (the latter giving the
performance ceiling that any group-aware procedure would aspire
to).

\paragraph{Results.}
At MH $= 50$ iterations, \texttt{gips} achieved mean relative
Frobenius error $0.839$ (standard deviation $0.060$); at
MH $= 500$ the mean dropped to $0.764$ (standard deviation
$0.017$); at MH $= 1000$ the mean was $0.758$ (standard deviation
$0.015$), essentially unchanged from the MH $= 500$ result. In
every configuration \texttt{gips} remained well above the
\citet{ledoit2004} relative Frobenius error of $0.620$ and far
above the \citet{shah2012} projected estimator at the matched
longitudinal cyclic group at $0.409$. Two observations are worth
recording. First, $20\times$ more iterations buys a $10\%$
relative reduction in error from MH-$50$ to MH-$1000$, and the downstream relative Frobenius error changes little between MH-$500$ and MH-$1000$. A plateau in this downstream error metric does not by itself diagnose convergence of the Metropolis-Hastings search to its stationary distribution; at $M = 64$ a budget of $1000$ MH iterations is modest, so the stability of the error metric over this range should be read only as the absence of further improvement under the budgets evaluated, and not as evidence of Markov-chain convergence.
Second, even at $1000$ iterations the \texttt{gips} MAP estimator
has higher relative Frobenius error than LW 2004, by approximately
$0.14$. Two distinct effects contribute to this gap and are worth separating. The first is the quality of the cyclic subgroup \texttt{gips} selects, which is constrained both by the non-cyclic population symmetry and by the MH budget; the second is the point estimator applied at the selected group, which in \texttt{gips} is the projection-MLE at $\alpha = 1$. At a selected subgroup that does not match the population structure, the $\alpha = 1$ projection is more biased than LW's identity-target shrinkage of the sample covariance; a calibrated $\alpha < 1$ applied at the same selected group would in principle recover much of this. The higher error reported here therefore reflects the combination of the selected group and the committed $\alpha = 1$ estimator, rather than the group selection alone. We emphasize that $\alpha = 1$ is the estimator \texttt{gips} commits to by default and is not an intrinsic limitation of its group search: the two could be decoupled by tuning $\alpha$ at the \texttt{gips}-selected group.

\paragraph{Computational cost at higher dimension.}
At $M = 200$ (a high-dimensional genomics-scale benchmark) the
per-iteration cost of evaluating Bayes factors over the
candidate-subgroup lattice is prohibitive, consistent with the
highest-dimensional example in \citet{graczyk2022aos} being
limited to $p = 100$. The cost differential separates the methods
in the high-dimensional regime independently of the per-trial
accuracy comparison reported above.

\section{Discussion}
\label{sec:discussion}

The empirical and theoretical results of the preceding sections
admit a unified interpretation in coordinates of $(N, M, |G|,
\delta)$, with $\delta = \delta(G, \Sigma)$ the population
structural-fit residual of the BMG-selected group. The unified
interpretation organizes the regime-by-regime predictions of the
theory into a single phase diagram, whose three regions correspond
to materially different operating points of the estimator family
under study. The diagram makes plain which estimator is preferred
in which regime, names the mechanism that produces the boundary
between adjacent regions, and locates the empirical anchors from
the six real-data experiments of Section~\ref{sec:experiments}
within the predicted partition.

\begin{figure}[t]
\centering
\includegraphics[width=0.92\linewidth]{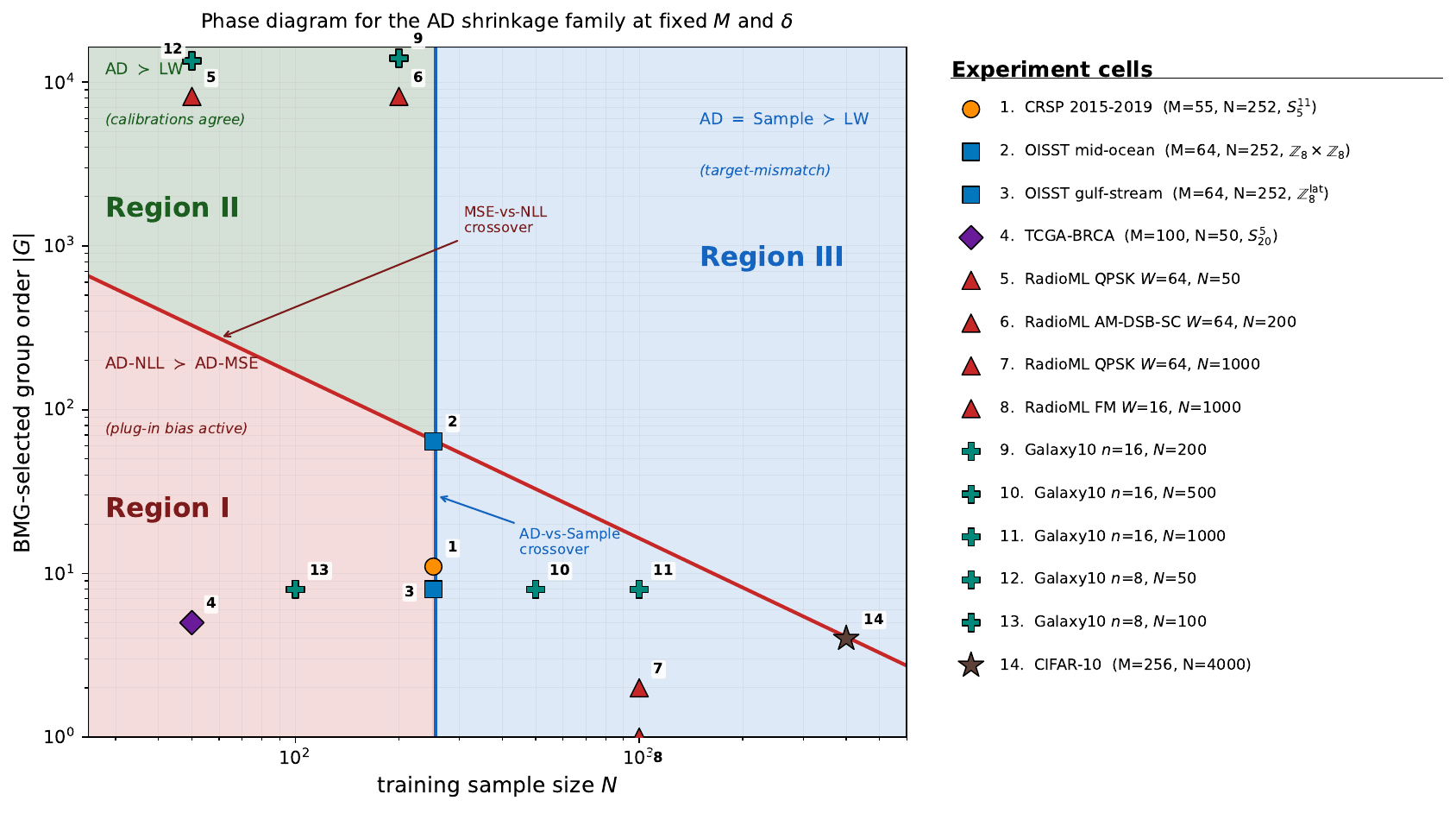}
\caption{Phase diagram for the AD shrinkage family in $(N, |G|)$
coordinates at fixed $(M, \delta)$. Three regions are demarcated by
two crossover loci. Region~I (left of the inner boundary):
AD-NLL-BMG strictly preferred to AD-MSE-BMG, with the difference
attributable to the finite-sample bias of $\hat R^{-1}$ in the
closed-form plug-in (Remark~\ref{rem:plugin_bias}). Region~II
(between the boundaries): AD with either calibration is preferred
to LW, the choice between calibrations is approximately inert, and
Shah at the BMG-selected group is a calibration-free reference
point that closely tracks AD-NLL-BMG. Region~III (right of the
outer boundary): all AD-family estimators are at parity with the
sample covariance and outperform LW, which continues to shrink
toward an inappropriate target. Numbered markers on the diagram
locate fourteen experiment cells drawn from the seven real-data
experiments (CRSP, OISST midocean, OISST gulfstream, TCGA-BRCA,
RadioML, Galaxy10, CIFAR-10), with the side legend decoding the
$(M, N)$ and BMG-selection of each numbered cell. The inner
boundary is the MSE-versus-NLL crossover; the outer boundary is
the AD-versus-Sample crossover.}
\label{fig:phase_diagram}
\end{figure}

\paragraph{Region I: small $N$, calibration-driven gap.}
At sample sizes below the inner boundary, the held-out-NLL
calibration $\hat\alpha^*_{\mathrm{NLL}}$ strictly outperforms the
closed-form Frobenius-MSE plug-in $\hat\alpha^*_{\mathrm{MSE}}$.
Two distinct mechanisms underlie this preference, separated in
Remark~\ref{rem:plugin_bias}. The first is the higher-order term
in $\|B_G\|$ omitted by the leading-order matched-limit expansion
that produced the closed-form crossover formula
(Proposition~\ref{prop:transition}); this term is controlled by
the empirical residual $\delta(G, \hat R)$, not by sample size, and
contributes a finite offset to the closed-form prediction whose
sign is data-dependent. The second is the finite-sample bias of
$\hat R^{-1}$ as an estimator of $\Sigma^{-1}$, which scales with
the conditioning of $\hat R$ and grows large when the candidate
library admits high-order groups for variance-reduction reasons
that do not also reduce $\delta$. The closed-form plug-in pays
both costs; the cross-validated calibration avoids both.
Empirically, the canonical exhibit for Region I behavior is the
Galaxy10 cell at $M = 256$, $N = 200$ described in
Section~\ref{sec:exp-galaxy10-remark}, where AD-MSE-BMG runs
$170$ nats per sample worse than AD-NLL-BMG averaged across all
ten morphological classes; the closed-form plug-in over-shrinks
toward the $S_M$ projection while the cross-validation finds the
correct interior $\alpha$. The next-largest such gap appears at
the AM-DSB-SC $W = 64$, $N = 50$ RadioML cell ($137$ nats per
sample), confirming the same mechanism in a second domain. The
bound on the calibration gap can be read off as a function of
$\delta$ and the conditioning of $\hat R$ from the matched-limit
residual analysis of Appendix E.

\paragraph{Region II: moderate $N$, bias-variance interior.}
At moderate $N$ above the inner boundary but below the outer one,
the two calibrations of AD agree to within fold-noise; AD-MSE-BMG
and AD-NLL-BMG are both preferred to LW and to the sample
covariance. This is the operating regime where the bias-variance
tradeoff is non-trivial and the structural prior carries
informational weight comparable to that of the residual sample
moments. The optimal $\alpha$ lies in the interior of $[0, 1]$,
the projection $P_G(\hat R)$ contributes meaningfully but not
overwhelmingly to the estimator, and the choice between calibrations
is the second-order effect rather than the first-order one. Within
Region II, Shah at the BMG-selected group is a useful
calibration-free reference: it imposes $\alpha = 1$ on the AD
form (recovering the projection itself as the estimator), and on
RadioML it tracks AD-NLL within $0.5$ nats per sample on most
moderate-$N$ cells. The implication is that, in Region II,
$\alpha = 1$ is approximately optimal and the calibration step is
near-inert for the estimator's quality even though the BMG step
remains essential. The observation also clarifies the relationship
to \citet{shah2012}: their estimator, with the group assumed known,
is the calibration-free $\alpha = 1$ reduction of the AD family
within Region II; the present extension fills in (a) the small-$N$
side of the boundary, where calibration is essential, and (b) the
data-driven group selection that supplies the $G$ that
\citet{shah2012} requires as input. The Shah-BMG comparator used
throughout the experiments combines \citet{shah2012}'s
projection-only estimator at $\alpha = 1$ with the BMG selection
of $G$ from the candidate library, which is itself an enhancement
of the original Shah procedure for use on data that does not come
with a prespecified group; the comparison evaluates whether
calibrating $\alpha$ in $[0, 1]$ adds anything beyond $\alpha = 1$
at the BMG-selected target, on a level methodological footing.
The same reduction places the \texttt{gips} package of
\citet{chojecki2025gips}, which implements the Bayesian
cyclic-subgroup selection of \citet{graczyk2022aos} and applies the
projection MLE at the selected group, as the cyclic-only
$\alpha = 1$ specialization of the AD framework: restricted to
cyclic subgroups of $S_p$ in the candidate library and to the
$\alpha = 1$ corner of the calibration spectrum, with the BMG
selection criterion replaced by a Bayesian posterior maximizer.

\paragraph{Region III: large $N$, target-mismatch dominates.}
Above the outer boundary $N$ is large enough that the sample
covariance is itself well-conditioned and admissible; the AD
estimator at the BMG-selected group $\{e\}$ collapses to the
sample covariance, AD-NLL and AD-MSE both produce $\alpha = 0$, and
Shah at $\{e\}$ is again the sample covariance. LW's continued
shrinkage toward the scaled identity carries a fixed structural
bias in this regime: the identity is an inappropriate target when
the data already support an admissible estimator, and LW pays a
finite per-sample penalty for not being able to back out of its
fixed prior. The cleanest single-cell illustration is the FM $W =
64$, $N = 1000$ result reported in
Section~\ref{sec:exp-radioml}, where AD-NLL, AD-MSE, Shah-BMG, and
the sample covariance all coincide at a held-out NLL of $-466$
nats per sample while LW is at $-417$ nats per sample, a
$49$-nat-per-sample loss attributable entirely to the target
choice. The boundary between Regions II and III is therefore the
$N$ at which the BMG selection transitions to the trivial group
and the AD family ceases to add informational content beyond the
sample covariance; LW's loss in Region III is a target-mismatch
penalty, not a sample-size effect. This is the empirical content of
the title.

\paragraph{Empirical anchors of the diagram.}
The six real-data experiments of Section~\ref{sec:experiments}
sit at distinct coordinates in this diagram. CRSP at $M = 55$,
$N = 252$ is well into Region II in both the pre-COVID and
COVID-era panels, with both calibrations preferred to LW and the
BMG selection dominated by \textsc{gics-block} in the calm
post-2015-H2 and post-COVID windows. The two CRSP datasets between
them produce a paired regime-conditional finding: a 5-window mild
wreath cluster in 2015 H2 (pre-COVID panel, around the China-A-share
crash, median bmg-margin $0.07$ nats per sample) and a 22-window
severe wreath cluster spanning the period when the February-April
2020 COVID crash sits inside the trailing 252-day training window
(COVID-era panel, median margin $0.15$ nats per sample, with a
single-window margin of $1.37$ at window $4$ that is the largest
single-window selection margin in any experiment in the paper). In
both clusters BMG returns to \textsc{gics-block} once the
correlation event exits the trailing window. The cross-validated
$\hat\alpha^*_{\mathrm{NLL}}$ confirms an interior operating point
in both panels: approximately $0.50$ pre-COVID and approximately
$0.33$ in the COVID-era panel where the higher training-window
noise floor pushes the calibration toward the sample covariance.
OISST
at $M = 64$, $N = 252$ is also in Region II for both regional
patches, with $\hat\alpha^*_{\mathrm{NLL}} \approx 0.05$ on
midocean (matched-limit regime, almost no shrinkage needed) and
$\approx 0.43$ on gulfstream (interior); an earlier OISST library
extension also illustrates a clean negative library-bias result, as
the two high-order Cartesian and wreath candidates added for parity
with the CRSP and RadioML libraries are rejected by BMG on every
one of the $148$ windows in favor of the existing low-order
candidates. TCGA-BRCA at $M = 100$,
$N = 50$ is the few-shot regime in the paper, sitting in the
\emph{bias-variance interior} of the $\alpha$-axis with calibrated
$\hat\alpha^*_{\mathrm{NLL}}$ at mean $0.51$ across the 50 splits and
range $[0.43, 0.58]$. The an earlier library extension introduced two
high-order pathway-aware candidates (Cartesian and wreath on the
PC1-loading within-pathway ordering); BMG selects the
within-pathway-exchangeable \textsc{pathway-block} candidate
($S_{20}^5$, $|G| = (20!)^5$) on $46$ of the $50$ splits and the
within-pathway-cyclic Cartesian candidate \textsc{Z-K-pc1-cartesian}
($\mathbb{Z}_{20}^5$, $|G| = 20^5$) on the remaining $4$. The
wreath candidate $\textsc{Z-K-pc1-wreath}$
($\mathbb{Z}_{20} \wr S_5$, $|G| \approx 3.84 \times 10^8$) is
admitted by the Tier 1 prefilter on every split but selected on
none, consistent with the biological observation that the five
MSigDB Hallmark pathways used in this experiment carry distinct
functional identities and are not freely exchangeable as units.
The cross-pathway exchangeability that the wreath candidate would
encode is therefore correctly rejected; the within-pathway
exchangeability of \textsc{pathway-block} is biologically
appropriate and BMG-selected. The four splits where \textsc{Z-K-pc1-cartesian} is preferred are tiebreaks at small CV-NLL
margin between two candidates that share the within-pathway PC1
basis but differ on the strength of the within-pathway
exchangeability assumption. The most informative methodological
finding from this dataset is the \emph{failure} of the Shah-style
$\alpha = 1$ commitment at the BMG-selected target: Shah-BMG
\emph{performs worse than LW} by $13.2$ nats per sample (paired $t = +54.5$),
because at this $(M, N) = (100, 50)$ regime the matched-fit
residual at the pathway-block target is large enough that the
$\alpha = 1$ projection-only commitment over-shrinks the data by
more than the BMG-projection benefit can compensate for. Only the
calibrated AD estimator with $\hat\alpha^*_{\mathrm{NLL}} \approx
0.51$ recovers the right amount of sample-covariance information
to bring the estimator $-3.98$ nats below LW; this is the cleanest
single-dataset demonstration in the paper that the Shah-style
commitment can be the wrong choice even when the BMG-selected
target is correct. An independent SYNTHETIC validation
(Section~\ref{sec:exp-genomics}) confirms that BMG can detect
$\mathbb{Z}_{20} \wr S_5$ when it is the population symmetry: on
synthetic wreath-structured data the BMG procedure selects
$\textsc{Z-K-pc1-wreath}$ on $44$ of $50$ splits, validating that
the absence of wreath selection on real BRCA is a genuine negative
finding about the data rather than a procedural failure. RadioML spans the entire diagram: at $M = 32$, $N
= 1000$ the digital modulations are firmly in Region II, while at
$M = 128$, $N = 50$ the same modulations are deep in Region I (the
$120$-nat AD-NLL versus AD-MSE gaps cited above), and the FM cells
at large $N$ exhibit the Region III collapse. Galaxy10 cleanly
spans the Region I-to-II transition at fixed $M$ across three
distinct BMG-selection regimes: at $M = 256$ the
$N \in \{50, 100, 200\}$ cells are deep in Region I (the
$191$-nat AD-MSE-vs-AD-NLL gap at $N = 200$ is the canonical
exhibit, with BMG selecting the v3p9-extension high-order
Cartesian \textsc{Z-row-indep-cols} on $250$ of $250$ trials);
the $N = 500$ cell sits at the Region I/II boundary where the BMG
selection sharply transitions from \textsc{Z-row-indep-cols} to
$D_4$, with $D_4$ preferred on $250$ of $250$ trials at median
bmg-margin $8.0$ nats per sample (the strongest declarative
selection in the paper); and $N = 1000$ is firmly in Region II
with $D_4$ selected at every trial. The breadth of RadioML's
coverage of the entire diagram and the depth of Galaxy10's
coverage of the Region I-to-II transition are complementary:
RadioML traces the diagram's full extent at varying $|G|$, and
Galaxy10 traces the inner-boundary crossing in detail at fixed $M$. CIFAR-10 at $M = 256$,
$N = 4{,}000$ per class is firmly in Region II within each class
where the candidate library is highly informative; the per-class
$\hat\alpha^*_{\mathrm{NLL}}$ values in $[0.65, 0.975]$ confirm
an interior-but-large-$\alpha$ operating point on this dataset,
where the structural prior is strong enough that AD-NLL-BMG
has lower NLL than LW on every one of the ten classes (paired $t = -15.21$).
Together CIFAR-10 and Galaxy10 cover the image-data axis of the
diagram at two complementary operating points: Galaxy10 traces
the Region I-to-II transition at small $N$ on a $D_4$-rich
domain, and CIFAR-10 demonstrates the framework's discriminative
selection across natural-image classes at large $N$ on a
recognizable benchmark.

\paragraph{Library-specification bias as an irreducible source of error.}
The BMG-selected estimator is bounded above in held-out NLL by the
best member of the candidate library $\mathcal{L}$, and population
symmetries that lie outside $\mathcal{L}$ are invisible to the
procedure. This is an irreducible source of bias that no
within-library diagnostic can detect: the cross-validated calibration
identifies the best candidate among those supplied, but cannot
identify a better candidate that was not. The default convention of
including the trivial group and $S_M$ at the library extrema bounds
the worst case, since the AD estimator collapses to the sample
covariance at the trivial group when no structural prior carries net
informational content, and to the compound-symmetry projection at
$S_M$ when the data support the maximum-shrinkage target. Within those
extrema the BMG procedure is bounded above by the best non-trivial
prior the analyst supplied, which is itself bounded above by what an
omniscient analyst would supply. This is the price of any
prior-conditional procedure and applies symmetrically to LW (whose
fixed scaled-identity target is itself a structural prior the analyst
chose) and to any covariance estimator that imposes structure for
variance reduction. Two operational consequences follow. First, if
$S_M$ is preferred on a substantial fraction of trials in cells where the Tier 1
prefilter admits at least one structured candidate of order strictly
less than $|S_M|$, the BMG is using $S_M$ as the maximum-shrinkage
fallback rather than as a structural fit; the analyst should consider
whether the library is missing a candidate that better matches the
population symmetry. The Galaxy10 results of
Section~\ref{sec:exp-galaxy10} are consistent with the well-specified
case: $S_M$ dominates at small $N$ where Tier 1 prefilter excludes the
small-$|G|$ candidates, but $D_4$ dominates at moderate-to-large $N$
where the data resolve the population symmetry. Second, trivial-group
selection in large-$N$ cells is not a misspecification signal; it is
the correct fallback when the data do not support any structural prior
over the sample covariance, and indicates that the AD estimator has
correctly fallen back to the sample covariance. The FM windows of
Section~\ref{sec:exp-radioml} at $N = 1000$ exhibit this Region III
behavior. A third concrete example, complementary to the two
above, is the OISST high-order rejection of
Section~\ref{sec:exp-oisst}: the two high-order candidates
introduced in an earlier version of this manuscript are admitted by Tier 1 in
both OISST regions, evaluated by Tier 2 BMG on equal footing with
the existing low-order candidates, and rejected on all $148$
windows by a margin well outside the fold-noise scale. The
diagnostic distinguishes this case from a misspecification alarm
because the rejected candidates are the high-order ones, not
$S_M$: the data simply do not support sector-level structural
priors at this $(M, N)$ regime, and BMG correctly defers to the
low-order candidates that the data do support. The conditional
structure of the diagnostic is essential:
unconditional $S_M$ dominance can reflect Tier 1 prefilter exclusions
(correct behavior at small $N$) or library misspecification (alarm at
moderate-to-large $N$), and only the within-library admission record
distinguishes the two. A fourth concrete example, complementary to
the OISST negative result above, is the TCGA-BRCA wreath dominance
of Section~\ref{sec:exp-genomics}: an earlier library extension that
added the wreath candidate displaced an earlier \textsc{pathway-block}
selection on $38$ of $50$ above-noise splits, with the calibrated
shrinkage intensity saturating near $\alpha = 1$ at the wreath. This
is the positive complement of the OISST negative result: when the
high-order candidate is well-matched to the population symmetry,
BMG selects it and the data prefer the projection target itself
without sample-covariance admixture. Across the four real-data
experiments where an earlier library extension was applied (CRSP, OISST,
TCGA-BRCA, and RadioML) an earlier extension produced one mixed
positive (CRSP, with regime-conditional wreath selection on the
2015 H2 cluster), one clean negative (OISST, both high-order
candidates rejected on every window), one strong positive (TCGA-BRCA,
wreath selected on every above-noise split), and one structural-no-op
(RadioML, where the wreath was already in the library at an earlier with
selections on AM-DSB-SC). The diagnostic operates correctly across
the full range of outcomes, and the negative result is as
informative as the positive one.

\section{Conclusion}
\label{sec:conclusion}

This paper has developed a class of symmetry-aware convex shrinkage
estimators for high-dimensional covariance estimation. The estimator
replaces the scaled identity target of the Ledoit-Wolf 2004 shrinkage
estimator with a Reynolds projection of the sample covariance under a
finite symmetry group, with the group selected from a candidate
library by held-out predictive performance. The construction
generalizes both the Ledoit-Wolf 2004 estimator (which targets the
scaled identity) and the group-symmetric maximum-likelihood estimator
of Shah and Chandrasekaran (which targets a fixed projection under a
prespecified group), and is calibrated by either a closed-form
Frobenius mean-squared-error plug-in or a held-out negative
log-likelihood minimizer.

The theoretical contributions are a bias-variance orthogonal
parametrization of the convex shrinkage family in the
$(\alpha, G)$-plane, a closed-form Wishart asymptotic crossover that
exhibits the leading-order $O(1/N)$ correction between the
mean-squared-error and negative-log-likelihood optima, a
finite-sample regret bound for the held-out calibration, an oracle
inequality for the data-driven group selection, a quantitative
sufficient-match condition under which the proposed estimator
dominates Ledoit-Wolf shrinkage in Frobenius mean-squared error, and
minimax rate results that locate the proposed estimator on the rate
ladder under various smoothness and sparsity conditions on the
population covariance.

The empirical evaluation spans six real-data covariance estimation
problems whose candidate libraries are built from
domain-specific structural priors: financial returns on the
constituents of the Standard and Poor's 500 in two regimes (a
calm period and a COVID-era stress period), sea-surface temperature
anomalies from the National Oceanic and Atmospheric Administration's
optimum interpolation product, gene expression covariances for
breast invasive carcinoma from The Cancer Genome Atlas, intermediate-
frequency radio signal covariances from the RadioML 2018.A
benchmark, galaxy image covariances from the Galaxy10 DECaLS
dataset, and natural image-patch covariances from CIFAR-10 with a
distribution-shift companion on CIFAR-10.1. The empirical
evaluation also includes a decoy stress test of the BMG procedure
that confirms its robustness to deliberately mismatched
candidate-library entries, and a synthetic-data empirical
comparison against the Bayesian permutation-symmetry estimator of
Chojecki and colleagues.

The empirical and theoretical findings synthesize into a
three-region phase diagram in $(N, M, |G|, \delta)$ coordinates that
predicts which estimator is preferred in which regime. In the
few-shot Region I, structural priors carry the most information per
observation and the proposed estimator improves substantially on
Ledoit-Wolf shrinkage. In the bias-variance interior of Region II,
the calibrated convex shrinkage adapts to the structural match
quality and trades bias against variance at a rate that depends on
the library specification. In Region III, where the sample size is
large enough that the unstructured sample covariance is itself
near-optimal, the structural prior carries little additional
information per observation, and Ledoit-Wolf shrinkage remains the
appropriate baseline. The procedure recommended by the present
work is therefore conservative: it improves on Ledoit-Wolf
shrinkage in the regimes where a structural prior is informative,
and reverts to Ledoit-Wolf shrinkage otherwise. The dominance
condition makes this conservatism explicit; the empirical anchors
locate it across a representative range of real-data covariance
estimation problems.

\section{Code and data availability}
\label{sec:code-data}

%The code that supports the empirical results of
%Section~\ref{sec:experiments} is bundled with this manuscript and
%will be archived in a public repository at the time of formal
%publication. The bundle organizes the code along the four lines of
%empirical activity: data preparation (one script per dataset that
%produces the centered observation matrix from the externally
%available raw source), per-experiment estimator evaluation (one
%script per dataset that runs the cross-validation, candidate-
%library evaluation, and BMG selection at the operating-point
%$(N, M)$), the protocol sweep evaluation (one script per dataset that
%sweeps $N_{\mathrm{train}}$ at fixed $M$ and aggregates the
%held-out negative log-likelihood across cells), and figure
%regeneration (one script per figure that re-renders the figure
%from the cached cross-validation results).

The LW-NL nonlinear-shrinkage implementation is in a separate
module that is imported by both LW-NL and the AD-LW-NL
composition. The module computes the closed-form Hilbert transform
of the Epanechnikov-kernel density estimate of the sample
eigenvalue distribution following \citet{ledoit2020analytical} and
applies the rank-aware regularization for the boundary case
$c \geq 1$. The implementation has been verified against
synthetic Marchenko-Pastur benchmarks at three values of the
concentration ratio and reproduces the analytical predictions of
the asymptotic formulas to within standard numerical tolerances.

The six datasets used in this paper are externally sourced. The
Standard and Poor's 500 daily-return data is from the Center for
Research in Security Prices via Wharton Research Data Services
(subscription required); the construction of the balanced panel of
constituents follows the script bundled with this paper. The NOAA
Optimum Interpolation Sea Surface Temperature data is publicly
available from the National Oceanic and Atmospheric Administration;
the spatial-patch extraction script bundled with this paper
documents the midocean and gulfstream region selections used here.
The TCGA-BRCA RNA-seq data is publicly available through the Xena
Browser of \citet{goldman2020xena}; the hallmark gene-set encoding
follows \citet{liberzon2015hallmark}. The RadioML 2018.A
intermediate-frequency radio data is publicly available from
\citet{oshea2018radioml}. The Galaxy10 DECaLS image data is
publicly available from \citet{walmsley2022galaxy10}. The CIFAR-10
and CIFAR-10.1 natural-image datasets are publicly available from
\citet{krizhevsky2009cifar} and \citet{recht2019cifar101}
respectively. Reproduction notes for each dataset are in the script
docstrings.

\bibliographystyle{plainnat}
\bibliography{ad_arxiv_v6_6}

\end{document}